\documentclass[usenames,dvipsnames]{lmcs} %%% last changed 2014-08-20
\pdfoutput=1
\usepackage[utf8]{inputenc}

\usepackage{lastpage}
\lmcsdoi{19}{4}{1}
\lmcsheading{}{\pageref{LastPage}}{}{}%
{Dec.~08,~2021}{Oct.~10,~2023}{}


\usepackage{xcolor}
\usepackage{adjustbox}

\usepackage{amsfonts}
\usepackage{amsmath,amsthm} 
\usepackage{xspace}
\usepackage{mathtools}
\usepackage{bussproofs}
\usepackage{stmaryrd}
\usepackage{enumitem}
\usepackage{multicol}
\usepackage{tikz}
\usetikzlibrary{matrix}
\usetikzlibrary{decorations.pathreplacing,calligraphy}
\usetikzlibrary{shapes.geometric, arrows}

\usepackage{thmtools} 
\usepackage{thm-restate}
\usepackage{hyperref}
\usepackage{cleveref}

\usepackage{hyperref}
\newcommand{\spi}{\ensuremath{\mathsf{s}\pi}\xspace}
\newcommand{\lamr}{\ensuremath{\lambda_{\oplus}}\xspace} %lam^r
\newcommand{\lamrshar}{\ensuremath{\widehat{\lambda}_{\oplus}}\xspace} 
\newcommand{\lamrfail}{\ensuremath{\lambda^{\lightning}_{\oplus}}\xspace}

\newcommand{\lamrsharfail}{\ensuremath{\widehat{\lambda}^{\lightning}_{\oplus}}\xspace} %lam^rf_s

\newcommand{\sep}{\ | \ } %|
\newcommand{\bag}[1]{\Lbag #1 \Rbag} 
\newcommand{\oneb}{\mathtt{1}}
\newcommand{\perm}[1]{\ensuremath{\mathsf{PER}(#1)}}
\newcommand{\dom}[1]{\mathtt{dom}(#1)}
\newcommand{\arrt}[2]{\ensuremath{#1 \rightarrow #2}}
\newcommand{\fail}{\mathtt{fail}}

\newcommand{\contexcat}{\ensuremath{\wedge}}

\newcommand{\size}[1]{\mathsf{size}(#1)} %size of bag
\newcommand{\headf}[1]{\ensuremath{\mathsf{head}(#1)}}
\newcommand{\headfsum}[1]{\ensuremath{\mathsf{head}_{\sum}}(#1)}

\newcommand{\dash}{\text{-}}

\newcommand{\shar}[2]{[#1\leftarrow #2]}

\def\subst#1#2{\{ \raisebox{.5ex}{\small$#1$}\! / \mbox{\small$#2$}\}} % << -/- >>
\def\linsub#1#2{\langle \raisebox{.5ex}{\small$#1$}\! / \mbox{\small$#2$}\rangle} % << -/- >>
\def\esubst#1#2{\langle\!\langle \raisebox{.5ex}{\small$#1$}\! / \mbox{\small$#2$}\rangle\!\rangle} % << -/- >>

\newcommand{\headlin}[1]{ { \{\!|} #1 { |\!\} }} %<->blue
\newcommand{\linexsub}[1]{{\langle \! |} #1 {| \! \rangle} }% <<->> red

\newcommand{\expr}[1]{\ensuremath{\mathbb{#1}}}
\newcommand{\lfv}[1]{\mathsf{fv}(#1)}

\newcommand{\mfv}[1]{\mathsf{mfv}(#1)}

\newcommand{\unit}{\mathbf{unit}}
\newcommand{\red}{\longrightarrow}
\newcommand{\tred}{\stackrel{*}{\red}}
\newcommand{\redlab}[1]{\ensuremath{\mathtt{[#1]} }}

\newcommand{\pequiv}{\equiv_\lambda}

\newcommand{\secref}[1]{$\S$\,\ref{#1}\xspace}
\newcommand{\figref}[1]{Fig.\,\ref{#1}\xspace}
\newcommand{\defref}[1]{Def.\,\ref{#1}\xspace}
\newcommand{\appref}[1]{App.\,\ref{#1}\xspace}
\newcommand{\thmref}[1]{Theorem~\ref{#1}\xspace}

\newcommand{\wfdash}{\models}
\newcommand{\core}[1]{#1^\dagger}
\newcommand{\strcore}[1]{\widehat{#1}^\dagger}
\newcommand{\outact}[2]{\overline{#1}(#2)}

\newcommand{\some}{\mathtt{some}}
\newcommand{\none}{\mathtt{none}}

\newcommand{\close}{\mathtt{close}}

\newcommand{\para}{\mathord{\;\mathbf{|}\;}}
\newcommand{\zero}{{\bf 0}}

\newcommand{\fn}[1]{\mathit{fn}(#1)}

\newcommand{\ampy}{\mathbin{\bindnasrepma}}
\newcommand{\with}{{\binampersand}}
\newcommand{\onef}{\mathbf{1}}
\newcommand{\dual}[1]{\overline{#1}}

\newcommand{\colorone}[1]{\textcolor{orange}{#1}}

\newcommand{\colorthree}[1]{\textcolor{purple}{#1}}

\newcommand{\encod}[2]{\llbracket#1\rrbracket_{#2}} % [[-]]
\newcommand{\recencodf}[1]{\colorone{\llparenthesis}  #1 \colorone{\rrparenthesis^{\bullet}}} % {{-}}

\newcommand{\recencodopenf}[1]{\colorone{\llparenthesis}  #1 \colorone{\rrparenthesis^{\circ}}}

\newcommand{\piencodf}[1]{\colorthree{\llbracket}  #1 \colorthree{\rrbracket}^{\colorthree{\lightning}}}

\newcommand{\piencod}[1]{\llbracket  #1 \rrbracket}

\newcommand{\linsetminus}{\setminus \!\! \setminus}

\newcommand{\succp}[2]{\ensuremath{#1 \Downarrow_{#2}}} 

\newcommand{\blue}[1]{\textcolor{RoyalBlue}{#1}}
\newcommand{\cred}[1]{\textcolor{BrickRed}{#1}}

\newcommand{\revo}[2]{#2}%#1  }}
\newcommand{\revd}[2]{#2}% {\small (#1)} }}

\newcommand{\revt}[2]{#2}
\newcommand{\revdaniele}[1]{#1}
 \newcommand{\secondrev}[1]{#1}
\newcommand{\srev}[1]{\secondrev{#1}}

\theoremstyle{plain} %\crefname{satz}{Satz}{S\"atze}
\begin{document}

\title[Non-Deterministic Functions as Non-Deterministic Processes]{Non-Deterministic Functions as\texorpdfstring{\\}{ }Non-Deterministic Processes}
%\titlecomment{{\lsuper*}OPTIONAL comment concerning the title, \eg,
%  if a variant or an extended abstract of the paper has appeared elsewhere.}

\author[J.~Paulus]{Joseph W. N. Paulus\lmcsorcid{0000-0002-1711-9361}}[a]	%required
\address{University of Groningen, The Netherlands}	%required
\email{j.w.n.paulus@rug.nl, j.a.perez@rug.nl}  %optional
%\thanks{thanks 1, optional.}	%optional

\author[D.~Nantes]{Daniele Nantes-Sobrinho\lmcsorcid{0000-0002-1959-8730}}[b]	%optional
%\address{University of Bras\'ilia, Brazil}
\address{Imperial College London, UK and University of Bras\'ilia, Brazil}	%optional
\email{dnantess@ic.ac.uk}  %optional
%\thanks{thanks 2, optional.}	%optional

\author[J.A.~P\'erez]{Jorge A. P\'erez\lmcsorcid{0000-0002-1452-6180}}[a]	%optional
%\address{University of Groningen, The Netherlands}
%\email{j.a.perez@rug.nl} %optional
% \urladdr{name3@url3\quad\rm{(optionally, a web-page can be specified)}}  %optional
%\thanks{thanks 3, optional.}	%optional

%% etc.

%% required for running head on odd and even pages, use suitable
%% abbreviations in case of long titles and many authors:

%%%%%%%%%%%%%%%%%%%%%%%%%%%%%%%%%%%%%%%%%%%%%%%%%%%%%%%%%%%%%%%%%%%%%%%%%%%

%% the abstract has to PRECEDE the command \maketitle:
%% be sure not to issue the \maketitle command twice!

\begin{abstract} 
  \noindent We study encodings of the $\lambda$-calculus into the $\pi$-cal\-cu\-lus in the unexplored case of calculi with \emph{non-determinism}  and \emph{failures}.
On the sequential side, we consider \lamrfail, a new  non-deterministic calculus in which intersection types control resources (terms); on the concurrent side, we consider~\spi, a $\pi$-calculus in which non-determinism and failure rest upon a Curry-Howard correspondence between linear logic and session types. 
We present a typed encoding of \lamrfail into \spi and establish its correctness. 
Our encoding precisely explains the interplay of  non-deterministic and fail-prone  evaluation in \lamrfail via  typed processes in~\spi.
In particular, it shows how failures in sequential evaluation (absence/excess of resources) can be neatly codified as interaction protocols.
%   The abstract has to precede the maketitle command.  Be
%   sure not to issue the maketitle command twice!  In the abstract,
%   mathematical expressions must be kept to the absolute minimum.
%   Otherwise it should consist of plain ASCII text, without
%   \TeX-commands, including explicit references using the
%   \texttt{\textbackslash cite} command.  Presently we are not able to
%   automatically extract an abstract containing such data and reliably
%   turn it into html code.  If you cannot meet these criteria, it is
%   your responsibility to provide us with an html-version of your
%   abstract.  Please keep the abstract fairy short to prevent it from
%   spilling over to the second page!
\end{abstract}

\maketitle

%% start the paper here:
\section*{Introduction}\label{S:one}
Milner's seminal work on encodings of the $\lambda$-calculus into the $\pi$-calculus~\cite{Milner92} explains how \emph{interaction} in $\pi$ subsumes \emph{evaluation} in~$\lambda$. 
It opened a research strand on formal connections between sequential and concurrent calculi, covering untyped and typed regimes (see, e.g.,~\cite{DBLP:journals/mscs/Sangiorgi99,DBLP:conf/birthday/BoudolL00,DBLP:conf/fossacs/BergerHY03,DBLP:conf/fossacs/ToninhoCP12,DBLP:conf/rta/HondaYB14,DBLP:conf/popl/OrchardY16,DBLP:conf/esop/ToninhoY18}). 
This paper extends this line of work by tackling a hitherto unexplored angle, namely encodability of  calculi in which computation is \emph{non-deterministic} and may be subject to \emph{failures}---two relevant features in sequential and  concurrent programming models.

We focus on \emph{typed} calculi and study how non-determinism and failures interact with \emph{resource-aware} computation. 
In sequential calculi, \emph{non-idempotent intersection types}  offer one fruitful perspective at resource-aware\-ness \revt{}{(see, e.g.,~\cite{DBLP:conf/tacs/Gardner94,DBLP:journals/logcom/Kfoury00,DBLP:journals/tcs/KfouryW04,DBLP:conf/icfp/NeergaardM04,BucciarelliKV17})}.
Because non-idempotency amounts to distinguish between types $\sigma$ and $\sigma \land \sigma$, this class of intersection types can ``count'' different resources and enforce quantitative guarantees. 
In  concurrent calculi, resource-awareness has been much studied using \emph{linear types}. Linearity ensures that process actions occur exactly once, which is key to enforce protocol correctness. 
 \revd{B3}{In particular, \emph{session types}~\cite{DBLP:conf/concur/Honda93,DBLP:conf/esop/HondaVK98} 
specify the protocols that channels must respect; 
this typing discipline exploits linearity to  
ensure absence of communication errors and stuck processes.}
To our knowledge, connections between calculi adopting these two distinct views of resource-awareness  via types are still to be established. 
We aim to develop such connections by relating models of sequential and concurrent computation.

On the sequential side, we introduce \lamrfail: a   $\lambda$-calculus with resources, non-de\-ter\-mi\-nism, and failures, which
distills key elements from  $\lambda$-calculi studied in~\cite{DBLP:conf/concur/Boudol93,PaganiR10}.
Evaluation in  \lamrfail considers \emph{bags} of resources, and determines alternative executions governed by non-determinism.
Failure results from a lack  or excess of resources (terms), and is   captured by the term $\fail^{\widetilde{x}}$, where $\widetilde{x}$ denotes a sequence of variables.
Non-determinism  in  \lamrfail  is \emph{non-collapsing} (i.e., confluent): intuitively, given   $M$ and $N$ with reductions $M \red M'$ and $N \red N'$, the non-deterministic sum $M + N$ reduces to $M' + N'$.
In contrast, under a \emph{collapsing} (i.e., non-confluent) approach, as in, e.g.,~\cite{DBLP:conf/mfcs/Dezani-CiancaglinidP93},  the non-deterministic sum  $M + N$ reduces to either $M$ or $N$. 

On the concurrent side, we consider \spi: a  session-typed $\pi$-calculus with (non-collap\-sing) non-de\-ter\-mi\-nism and failure, proposed in~\cite{CairesP17}.
\spi rests upon a Curry-Howard correspondence between session types and   (classical) linear logic, extended with modalities that express   \emph{non-deterministic protocols} that may succeed or fail. Non-determinism in \spi is non-collapsing, which ensures confluent process reductions. 

\paragraph{Contributions} 
This paper presents \secondrev{the first formal connection between a $\lambda$-calculus with non-idempotent intersection types and a $\pi$-calculus with session types. Specifically, the paper presents} the following contributions:
\begin{enumerate}
    \item \textbf{The resource calculus \lamrfail}, a new calculus that distills the distinctive elements from previous resource calculi~\cite{DBLP:conf/birthday/BoudolL00,PaganiR10}, while offering an explicit treatment of failures in a setting with non-collapsing non-determinism. 
    
    \secondrev{We develop the syntax, semantics, and essential meta-theoretical results for \lamrfail. In particular, }
    using intersection types, we  define \emph{well-typed} (fail-free) expressions and \emph{well-formed} (fail-prone) expressions in \lamrfail \secondrev{and establish their properties}. % (see below).
    
    \item \textbf{An encoding of \lamrfail into \spi},  proven correct following established  criteria \secondrev{in the realm of relative expressiveness for concurrency}~\cite{DBLP:journals/iandc/Gorla10,DBLP:journals/iandc/KouzapasPY19}. %(\secref{s:encoding}). 
     These criteria attest to an encoding's quality;  we consider
\emph{type preservation},  \emph{operational correspondence} \srev{(including \;{completeness} and \emph{soundness})},  \emph{success sensitiveness}, and \emph{compositionality}.

Thanks to these correctness properties, our encoding precisely describes how typed interaction protocols \secondrev{(given by session types)} can codify sequential evaluation in which absence and excess of resources leads to failures \secondrev{(as governed  by intersection types)}. 
\end{enumerate}

\smallskip

These contributions entail different challenges. 
The first is bridging the different mechanisms for resource-awareness involved  (i.e., intersection types in \lamrfail, session types in \spi).
A direct encoding of \lamrfail into \spi is far from obvious, as multiple occurrences of a variable in \lamrfail must  be accommodated into the linear setting of \spi. 
\srev{To overcome this challenge, 
we introduce a variant of  \lamrfail, dubbed  \lamrsharfail.
The distinctive feature of \lamrsharfail is a 
\emph{sharing} construct, which we adopt following the \emph{atomic} $\lambda$-calculus presented in~\cite{DBLP:conf/lics/GundersenHP13}.
}
Our encoding of \lamrfail  expressions into \spi processes is then in two steps.
We first define a correct encoding from \lamrfail to \lamrsharfail, which relies on the sharing construct to ``atomize''  occurrences of the same variable.
Then, we define another correct encoding, from \lamrsharfail to \spi, which extends Milner's with constructs for non-determinism. 

Another challenge is framing failures  in \lamrfail (undesirable computations)  as well-typed \spi processes. 
%This is because we wish to connect  with well-typed processes (in \spi). 
Using intersection types, we define \emph{well-formed} \lamrfail expressions,  which   can fail, in two stages. 
First, we consider \lamr, the sub-language of \lamrfail without   $\fail^{\widetilde{x}}$. 
We give an intersection type system for \lamr to regulate fail-free  evaluation. 
Well-formed expressions are then defined on top of well-typed \lamr expressions.
We show that \spi can correctly encode the fail-free \lamr but, more interestingly, also well-formed \lamrfail expressions, which are fail-prone. %\alert{changed the format of this now}

\figref{fig:summary} summarizes our approach: the encoding 
from \lamrfail to \lamrsharfail is denoted 
 $\recencodopenf{\cdot}$, whereas the encoding from \lamrsharfail to \spi is denoted $\piencodf{\cdot}_u$.

\paragraph{Organization}
%\secref{s:lambda} introduces the syntax and semantics of \lamrfail, and establishes the key properties for its intersection type system.
%\secref{sec:lamsharfail} introduces \lamrsharfail, the variant of \lamrfail with sharing, and presents its associated intersection type system.
%In \secref{s:pi} we summarize the syntax, semantics, and session type system of \spi, following~\cite{CairesP17}. 
%The encoding of \lamrfail into \spi is presented and proved correct in \secref{s:encoding}. As mentioned above, we define two correct encodings: one from \lamrfail to \lamrsharfail, and then another one from \lamrsharfail to \spi. 
%Discussion about our approach and results, and comparisons with related works {is} given in \secref{s:disc}.

 %\joe{
 Next, \secref{s:key} informally discusses  key ideas in our work. 
 \secref{s:lambda} introduces the syntax and semantics of \lamrfail, and defines its intersection type system.
\secref{sec:lamsharfail} introduces \lamrsharfail, the variant of \lamrfail with sharing. It also presents its associated intersection type system, and defines an encoding from $\lamrfail $ to $\lamrsharfail$.
%in which its correctness is given in \secref{s:encoding}.
In \secref{s:pi} we summarize the syntax, semantics, and session type system of \spi, following~\cite{CairesP17}. 
\secref{s:encoding} establishes the correctness of the encoding of $\lamrfail $ into $\lamrsharfail$ and presents and proves correct the encoding of \lamrshar into \spi. %Thus, as mentioned above, we define two correct encodings: one from \lamrfail to \lamrsharfail, and then another one from \lamrsharfail to \spi. 
\secref{s:rw} presents comparisons with related works.
\secref{s:disc} closes with a discussion about our approach and results.
% }
 
 This paper is an extended and revised version of the conference paper \cite{DBLP:conf/fscd/PaulusN021}. Here we have included full technical details, additional examples, and extended explanations. For the sake of readability, and to make the paper self-contained, we have included proof sketches in the main text; their corresponding full proofs have been collected in the appendices.
  
\begin{figure}[!t]
    \centering
\includegraphics[scale=0.95]{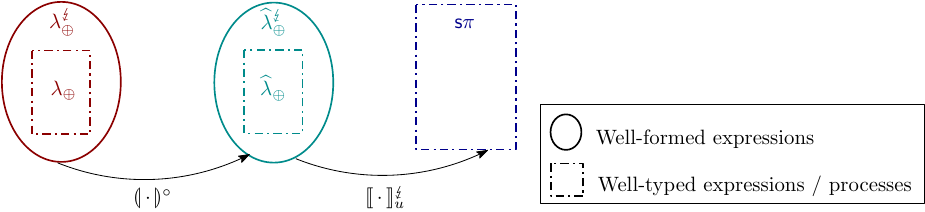}
    \caption{Overview of our approach.}
    \label{fig:summary}
\end{figure}

% \subsection*{Matters of convenience}

%   Please submit {\bf only one file} containing the
%   TeX-source of your paper!  It is a major inconvenience when certain
%   changes have to be applied in several files separately.  Of course,
%   we understand that separating a TeX source into several files has
%   advantages during the creation of a paper, but please combine all
%   parts into a single file for your submission. You personal macros
%   can of course be contained in a separate file, as can be external
%   graphics.  For the latter a dedicated subdirectory is required.

\section{Overview of Key Ideas}
\label{s:key}
% \jp{Following Rev2, here we introduce our shared calculus and the encoding into pi, giving references to prior works on lambda with sharing (as per Rev3) and explaining the working of our intersection type system. The first encoding can be mentioned at the very end.}

% \paragraph{Two-steps encoding: From $\lamrfail$ to $\spi$}
% The purpose of this work is to propose a correct translation from $\lamrfail$ to $\spi$.  The expressive power of $\spi$ is beyond the one of $\lamrfail$, even though, 
% the $\lambda$-calculus proposed here has syntactic constructs combined with a dedicated operational semantics that can implement non-collapsing non-determinism with failure.
% Nevertheless, we will present a translation from $\lamrfail$ to $\lamrsharfail$ and then to a fragment of $\spi$ which satisfies well-established encodability criteria.

\revd{}{Before embarking into our technical developments, we discuss some key ideas in the definition of \lamrfail and its correct encodability into \spi.} % and connect with relevant related works. 

\paragraph{Non-determinism.} Our source language \lamrfail has three syntactic categories: terms ($M,M'$), bags ($B,B'$) and expressions ($\expr{M}, \expr{L}$). Terms can be variables, abstractions $\lambda x. M$, 
applications $(M\ B)$, explicit substitutions $M\esubst{B}{x}$, or the explicit failure term $\fail$ (see below). 
Bags are multisets of terms (the resources); this way, e.g., $B= \bag{M_1,M_1,M_2}$ is a bag with 
three resources ($M_1$, $M_1$, and $M_2$). Expressions are sums of terms, written $M_1+M_2$; they denote a non-deterministic choice between different ways of \emph{fetching} resources from the bag.
%The idea is that terms can only  be applied to bags; similarly, an explicit substitution is also replacing a bag for variable. 

In  $\lamrfail$, reduction is lazy: first, a $\beta$-reduction evolves to an explicit substitution, which will then fetch the elements in the bag to be 
substituted for the corresponding variable, when some conditions are satisfied: we interpret this  as ``consuming a resource". 
For instance, 
given a $\lamrfail$-term $M$
with head variable $x$ and two occurrences of $x$, we have the  reduction:
\begin{eqnarray}
\lambda x. M \bag{M_1,M_2}  & \red &
 M \esubst{\bag{M_1,M_2}}{x} \nonumber
 \\
  & \red & M\headlin{M_1/x}\esubst{\bag{M_2}}{x} + M\headlin{M_2/x}\esubst{\bag{M_1}}{x} = M'
	\label{key:ex1} 
\end{eqnarray}
The resulting expression $M'$ is a sum that gathers two alternative computations: it may reduce by either (i)~first fetching $M_1$ from the bag and \secondrev{linearly} substituting it for $x$ in \secondrev{the head position of} $M$  (this is denoted with $M\headlin{M_1/x}$) and then continue with the rest of the bag ($M_2$, wrapped in an explicit substitution), or (ii)~fetching  and 
\secondrev{linearly} substituting $M_2$ \secondrev{in head position}, leaving $M_1$ in an explicit substitution.

\paragraph{Successful Reductions} We consider a computation as successful only  when the number of elements in the bag
matches the number of occurrences of the variable to be substituted; otherwise the computation fails.
 As an example, consider the previous example,  now with  $ M= x\bag{x\bag{I}}$ 
%and the bag $ B = \bag{I} \cdot \bag{ I \bag{I}}$,
where $I=\lambda z.z$ is the identity. The reduction in (\ref{key:ex1}) is then 
$$
 ( \lambda x. x\bag{x\bag{I}}) \bag{M_1,M_2}  \red^* M_1\bag{x\bag{I}}\esubst{\bag{M_2}}{x} + M_2\bag{x\bag{I}}\esubst{\bag{M_1}}{x}
 $$

%\[M\ B =(\lambda x . x \bag{x \bag{I} } )  \bag{I} \cdot \bag{ I \bag{I}} \red^* I \esubst{\oneb}{x}\]
% are applied to bags $B,B'$, which are multisets of terms. This allows, for instance, an 
% $ M= \lambda x. x\bag{x\bag{I}}$ and the bag $ B = \bag{I} \cdot \bag{ I \bag{I}}$, where $I=\lambda z.z$ is the identitiy term in our source language.
%I want to illustrate the computation from 
%\[ M \ B = (\lambda x . x \bag{x \bag{I} } )  \bag{I} \cdot \bag{ I \bag{I}} \red^* I \esubst{\oneb}{x}\]
Hence,    when $\lambda x. M$ is applied to a bag with two resources, it evolves successfully.  
However, if $\lambda x. M$ is applied to a bag with less (or more) than two resources, the computation  evolves to the \emph{explicit failure} term $\fail^{\widetilde{z}}$,  where $\widetilde{z}$ is a multiset of variables, as we explain next.

\paragraph{Explicit Failure.}
A construct for failure is present in the resource $\lambda$-calculus in~\cite{PaganiR10}. 
In this formulation, the failure term `$0$' is consumed by sums and disappears at the end of the computation; as such, it gives no information about the failed computation and its origins.

Following~\cite{PaganiR10}, a design decision in $\lamrfail$ is to have $\fail^{\widetilde{x}}$ in the syntax of terms. 
The sequence $\widetilde{x}$ denotes the variables captured by failure; this provides useful information on the origins of a failure. 
As an example, consider a term $M$ with free variables $\widetilde{y}$ and in which the number of occurrences of $x$ is  different from 2. 
Given a bag $B=\bag{M_1,M_2}$, reduction leads to a failure, as follows:
$$
    (\lambda x. M) B \red M\ \esubst{\bag{M_1,M_2}}{x}\red {}  \displaystyle\sum_{\perm{\bag{M_1,M_2}}} \fail^{\widetilde{y}} = M'
    $$
%$M$ removing $x$ and adding the multiset free variables in $B$, i.e., $ \widetilde{y} = (\mfv{M} \setminus \{x\}) \uplus \mfv{B}$\joe{I think maybe this may too repetitive , it explains the rule in too much technical details here that is repeated later }. 
In this case, $M'$ is the sum $\fail^{\widetilde{y}}+\fail^{\widetilde{y}}$, which has as many summands as the permutations of the elements of $B$. Intuitively, it means that it does not matter if one replaces the occurrence(s) of $x$ first with $M_1$ (or $M_2$), then the other occurrence (if any), with $M_2$ (or $M_1$), the result will be the same, i.e., $\fail^{\widetilde{y}}$. Here again both possibilities are expressed in a sum. The precise semantics of failure will be presented in \secref{ssec:red_sem}.

\paragraph{Typability and Well-formedness}
%The behavior presented in \eqref{eq:motiv_none} is captured in \spi via the session type system  in \secref{s:pi}. The goal is to be able to capture expressions that will produce failing behavior in $\lamrfail$ via some sort of type-system,  but also to have  a dedicated system that allows to capture fail-free behavior.

We define an intersection type system for $\lamrfail$. This choice follows a well-established tradition of coupling resource $\lambda$-calculi with intersection types~\cite{DBLP:conf/concur/Boudol93,BoudolL96,PaganiR10}. Intersection types are also adopted in related calculi~\cite{DBLP:journals/tcs/EhrhardR03}. Intersection types are a natural typing structure for resources: they have similar mathematical properties of non-idempotency and commutativity, and can help to ``count'' the number of occurrences of a variable in a term, as well as the number of components in a bag. 

In our type systems, each element of a bag must have the same type. 
This way, e.g., a well-typed bag $B=\bag{M_1,M_2,M_3}$ has type $ \sigma\wedge \sigma \wedge \sigma$, where $\sigma$ is a strict type (cf. \defref{d:typeslamrfail}). Then,  an application $M\ B$ is well-typed, say, with  type $\tau$, only if $M:\sigma\wedge  \sigma \wedge \sigma\to \tau$. 
We shall write $\sigma^k$ to denote the intersection type $\sigma\wedge \ldots \wedge \sigma$, with $k\geq 0$ copies of $\sigma$.
Notice that $\sigma^0$ denotes the empty type $\omega$.
The typing rule for application is then as expected:

\begin{prooftree}
  \AxiomC{\(  \Gamma \vdash M :  \sigma^k \to \tau \)}
    \AxiomC{\( \Gamma\vdash B : \sigma^k \)}
        \LeftLabel{\redlab{T:app}}
    \BinaryInfC{\(  \Gamma\vdash M\ B : \tau\)}
\end{prooftree}
where $\Gamma$ is a type context assigning types to variables.  

We chose to express explicit failing terms and computation.
To properly account for these computations, we define  a separate type system with  so-called \emph{well-formedness} rules, with notation `$\wfdash$'. 
Unlike rules for typability, rules for well-formedness capture computations that fail due to a mismatch  of resources (lack or excess). 
This entails some increased flexibility in selected rules. 
This way, e.g, the following is the well-formedness rule for application:
\begin{prooftree}
  \AxiomC{\( \Gamma \wfdash M :  \sigma^j \to \tau \)}
    \AxiomC{\( \Gamma \wfdash B : \sigma^k\)}
        \LeftLabel{\redlab{F:app}}
    \BinaryInfC{\( \Gamma \wfdash M\ B : \tau\)}
\end{prooftree}
Here the added flexibility is that we do not require $k=j$; hence, the rule can capture successful \emph{and} failing computations, depending on whether $k=j$ or not. As expected, the term $\fail^{\widetilde{z}}$ is not well-typed, but it is well-formed: the judgement $\Gamma\wfdash \fail^{\widetilde{z}}:\tau$ holds for an arbitrary type $\tau$ and a $\Gamma$ consisting of variable assignments 
for the variables in $\widetilde{z}$. 

Therefore, we consider two intersection type systems: one captures exclusively successful computations (see   \figref{fig:app_typingrepeat}); the other, which we call the well-formedness system (see  \figref{fig:app_wf_rules}), subsumes the first one by admitting both successful and failing computations. 
The weakening rule is admissible in both systems (see below).
Both systems enjoy subject reduction, whereas only well-typed terms satisfy subject expansion.

\paragraph{Controlling resources via sharing}
In order to better control the use of resources, i.e., substituting variables for terms with a careful form of duplication, we borrow ideas from the {\em sharing graphs} by Guerrini et\,al.~\cite{GUERRINI199999,GUERRINI2003379} and define the calculus $\lamrsharfail$.
The key idea is as follows: whenever a bound variable $x$ occurs multiple times within a term, these occurrences, say \(x_1,\ldots, x_n\), are temporarily assigned new names (think aliases). This assignment is indicated with the {\em sharing construct}   $\shar{x_1,\ldots, x_n}{x}$, which we adopt following~\cite{DBLP:conf/lics/GundersenHP13}. 
This way, for instance, the $\lamrfail$-term $\lambda x. x \bag{x}$ would correspond to $\lambda x. x_1\bag{x_2}\shar{x_1,x_2}{x}$ in  $\lamrsharfail$.

 We also carefully treat the ``erasing'' of resources: if a term has vacuous abstractions, this is also indicated with the sharing construct, where the bound variable maps to ``empty''. Hence, the $\lamrfail$-term $\lambda x. y \bag{z}$ is expressed as $\lambda x. y\bag{z}\shar{}{x}$ in  $\lamrsharfail$.
The tight control of resources in $\lamrsharfail$ turns out to be very convenient to encode $\lamrfail$ into \spi, as we discuss next.

\paragraph{Encoding $\lamrfail$ into \spi.}
The central result of our work is a correct translation of $\lamrfail$ into \spi. 
In defining our translation we use $\lamrsharfail$ as a stepping stone. 
This is advantageous, because (i)~the relation between $\lamrshar$ and $\lamrsharfail$ is fairly direct and (ii)~the sharing construct in $\lamrsharfail$ makes it explicit the variable occurrences that should be treated as linear names in \spi. 

The encoding of $\lamrfail$ into $\lamrsharfail$ is denoted $\recencodf{\cdot }$ and given in \secref{ss:auxtrans}.
The encoding of $\lamrsharfail$ into \spi, denoted $\piencodf{\cdot}_u$ and presented in \secref{ss:second_trans}, is arguably more interesting---we discuss it below. 

%We start with an auxiliary translation  that maps the source language $\lamrfail$ into $\lamrsharfail$, and   consists of replacing multiple occurrences of a variable, say $x$, in a $\lamrfail$-term $M$, for new fresh variables and recording this
%renaming in the sharing constructor. Besides, vacuous abstraction are also recorded with the empty sharing constructor. For instance, the $\lamrfail$-term 
%$\lambda x. (x\bag{x}\cdot \bag{x})$ becomes the $\lamrsharfail$-term $\lambda x. (x_1\bag{x_2}\cdot \bag{x_3})\shar{x_1,x_2,x_3}{x}$, similarly, $\lambda y. x\bag{x}\cdot \bag{x}$ becomes $\lambda y. x\bag{x}\shar{}{y}\shar{x_1,x_2}{x}$. The translation preserves well-formedness (\thmref{lem:wfpreserv_closedtrans}) of $\lamrfail$-terms and satisfies Gorla's~\cite{DBLP:journals/iandc/Gorla10} notion of  correct operational correspondence (\secref{app:ss:compsound}), which will be adopted here as a correctness criteria for our translations.

The definition of $\piencodf{\cdot}_u$ considers well-formed source terms in $\lamrsharfail$ which are translated into well-typed \spi processes.
As usual, the translation is parametric on a channel name $u$, which is used to provide the behavior of the source term. 

The calculus \spi includes a non-deterministic choice operator $P \oplus Q$ and formalizes sessions which are \emph{non-deterministically available}. Intuitively, this means that a given session protocol along a name can either be available and proceed as prescribed by the corresponding session type, or fail to be available. Clearly, such a failure may have repercussions on other sessions that depend on it. 
To this end, \spi includes prefixes $x.\overline{\some}$ and $x.\overline{\none}$, which are used to confirm the availability of $x$ and to signal its failure, respectively. 
Process $x.\some_{(w_1, \cdots, w_k)};Q$ declares the dependency of sessions $w_1, \ldots, w_k$ in $Q$ on an external session along $x$. 
The corresponding reduction rules are then:
\begin{eqnarray*}
	\label{eq:motiv_some}
x.\overline{\some} \para x.\some_{(w_1, \cdots, w_k)};Q  & \red &
Q
\\
	\label{eq:motiv_none}
x.\overline{\none} \para x.\some_{(w_1, \cdots, w_k)};Q  & \red &
w_1.\overline{\none} \para \cdots  \para w_k.\overline{\none}
\end{eqnarray*}

%The translation relies heavily on the non-detechoice for an explicit failure constructor in $\lamrfail$ (or $\lamrsharfail$) is compatible with
% the constructors $\some_{\widetilde{w}}$ and $\none$, the representation of non-deterministic sessions in  $\spi$-calculus (\secref{s:pi}). The idea is that, via synchronization across some shared channel, one can express explicitly when a process could not continue due to missing resource. Here $\widetilde{w}$ represent dependent names,  which  themselves produce $\none$ behavior along their own channel names. For instance, in a simple case in which $\widetilde{w}=(w_1,w_2,w_3)$, the synchronization is as 
%\begin{equation}  \label{eq:motiv_none}
%x.\overline{\none} \para x.\some_{(w_1, w_2, w_3)};Q   \red 
%w_1.\overline{\none} \para w_2.\overline{\none} \para w_3.\overline{\none}
%\end{equation}
%
%\noindent which expresses that all processes that were expecting to receive a confirmation of behavior along $x$, will pass along the information that no behavior is available on their own channels $w_1,w_2$ and $w_3$. The translation presented in~\secref{ss:second_trans} satisfies operational correspondence properties that connect these two failing behaviors in a natural manner.
\noindent
Following Milner, $\piencodf{\cdot}_u$ maps computation in $\lamrsharfail$ into session communication in \spi; non-deterministic sessions are used to codify the non-deterministic fetching of resources in $\lamrsharfail$.
This way, the translation of $(\lambda x. M\shar{x_1,x_2}{x}) B$ will enable synchronizations between the translations of $M\shar{x_1,x_2}{x}$ and $B$. 
More in details, the translation of a bag $B=\bag{M_1,M_2}$ is as follows: 
\begin{eqnarray*}
\piencodf{\bag{M_1}\cdot \bag{M_2}}_x & = &   x.\some_{\widetilde{z_1},\widetilde{z_2}} ; x(y_i). x.\some_{y_i,\widetilde{z_1},\widetilde{z_2}};x.\overline{\some} ; 
\\
& & \quad 
 \secondrev{\outact{x}{x_i}. (x_i.\some_{\widetilde{z_1}} ; \piencodf{M_1}_{x_i} \mid \piencodf{\bag{M_2}}_x \mid y_i. \overline{\none}) }
\end{eqnarray*}
where $\widetilde{z_1}$ and $\widetilde{z_2}$ denote the free variables of $M_1$ and $M_2$, respectively. 
Process $\piencodf{\bag{M_1}\cdot \bag{M_2}}_x$  first expects confirmation of session $x$; then, the translation of each resource $M_i$ is made available in a dedicated name $x_i$, which will be communicated to other processes.
Accordingly, the translation of  $\piencodf{M\shar{x_1,x_2}{x}}_u$ is expected to synchronize with the translation of the bag $B$: indeed, it confirms behavior along $x$, before receiving the names, one for each shared copy of $x$ that should be used throughout the synchronizations:
\[
\begin{aligned}
\piencodf{M\shar{x_1,x_2}{x}}_{u}&= x.\overline{\some}. \outact{x}{y_1}. \Big(y_1 . \some_{\emptyset} ;y_{1}.\close 
       \! \mid \! x.\overline{\some};x.\some_{u, (\lfv{M} \setminus \{x_1 ,  x_2\} )};x(x_1) . \\%[1mm]
       & \hspace{2.0em} . x.\overline{\some}. \outact{x}{y_2} . \big(y_2 . \some_{\emptyset} ; y_{2}.\close  \mid x.\overline{\some};x.\some_{u,( \lfv{M} \setminus \{x_2\} ) };x(x_2)
      \\%[1mm]
       & \hspace{2.8em} . x.\overline{\some}; \outact{x}{y_{}}. ( y_{} . \some_{u, \lfv{M} } ;y_{}.\close; \piencodf{M}_u \mid x.\overline{\none} )~ \big)  \Big) 
\end{aligned}
\]
Several confirmations take place along the channel names involved in the synchronizations; see  \secref{ss:second_trans} for details. 

Non-determinism plays a key role in the translation of an application $M' B$. In this case, we consider the permutations of the elements of $B$ using non-deterministic choice in~$\spi$. 
When $B=\bag{M_1,M_2}$,  the translation is:

\[
\begin{aligned} 
\piencodf{M'\, B}_u  =~~ & (\nu v)(\piencodf{M'}_v \mid v.\some_{u , \lfv{B}} ; \outact{v}{x} . ([v \leftrightarrow u] \mid \piencodf{\bag{M_1,M_2}}_x ) ) 
\\
& \qquad  \oplus 
\\
&  (\nu v)(\piencodf{M'}_v \mid v.\some_{u , \lfv{B}} ; \outact{v}{x} . ([v \leftrightarrow u] \mid \piencodf{\bag{M_2,M_1}}_x ) )
\end{aligned}
\]
A synchronization occurs when process $\piencodf{M'}_v$ can confirm its behavior along $v$. For instance, when $M'=\lambda x. M\shar{x_1,x_2}{x}$ the translation is as
\[   \piencodf{\lambda x.M[x_1,x_2 \leftarrow x]}_v  = v.\overline{\some}; v(x).\piencodf{M[x_2,x_2 \leftarrow x]}_v
\]
and the synchronization may be possible; it depends on the translations of $M$, $M_1$, and $M_2$. 

We close this section by observing that our translations $\recencodf{\cdot }$ and $\piencodf{\cdot}_u$ satisfy well-known  correctness criteria, as formulated by Gorla~\cite{DBLP:journals/iandc/Gorla10} and Kouzapas et al.~\cite{DBLP:journals/iandc/KouzapasPY19} (see \secref{ss:criteria} for details).

\section{ \texorpdfstring{$\lamrfail$}{λ^↯_⊕}: A \texorpdfstring{$\lambda$}{λ}-calculus with Non-Determinism and Failure}\label{s:lambda}

We define the syntax and reduction semantics of \lamrfail, our new resource calculus with non-determinism and failure. \revd{B4}{We then equip it with} non-idempotent session types, and establish the subject reduction property for well-typed and well-formed expressions (Theorems~\ref{t:app_lamrsr} and \ref{t:app_lamrfailsr}, respectively).
\revt{}{We also consider the subject expansion property, which holds for well-typed expressions (Theorem~\ref{t:app_lamrexp}) but not for well-formed ones (Theorem~\ref{t:app_lamrfailse}).}

\subsection{Syntax}\hfill

The syntax of \lamrfail combines elements from calculi introduced and studied by Boudol and Laneve~\cite{DBLP:conf/birthday/BoudolL00} and by Pagani and Ronchi della Rocca~\cite{PaganiR10}.
We use $x, y, \ldots$ to range over the set of \emph{variables}.
 We write $\widetilde{x}$ to denote the sequence of pairwise distinct variables $x_1,\ldots,x_k$, for some $k\geq 0$.
 {We write $|\widetilde{x}|$ to denote the length of $\widetilde{x}$}.
 
 \begin{defi}[Syntax of \lamrfail ]
\label{def:rsyntaxfail}
The \lamrfail calculus is defined by the following grammar:
\[
\begin{array}{l@{\hspace{10mm}}lll}
\mbox{(Terms)} &M,N, L&::=& x\sep \lambda x . M \sep (M\ B) \sep  M \esubst{B}{x} \sep \fail^{\widetilde{x}}\\
\mbox{(Bags)} &A, B&::=& \oneb \sep \bag{M} \sep A \cdot B \\
\mbox{(Expressions)} & \expr{M}, \expr{N}, \expr{L}&::=&  M \sep \expr{M}+\expr{N}\\
\end{array}
\]
\end{defi}
\noindent
We have three syntactic categories: \emph{terms} (in functional position); \emph{bags} (in argument position), which denote multisets of resources; and \emph{expressions}, which are finite formal sums that represent possible results of a computation. 
Terms are unary expressions: they can be  variables, abstractions, and applications. 
Following~\cite{DBLP:conf/concur/Boudol93,DBLP:conf/birthday/BoudolL00}, the \emph{explicit substitution} of a bag $B$ for a variable $x$ \revd{B5}{in a term $M$, written $M\esubst{B}{x}$}, is also a term.
The term $\fail^{\widetilde{x}}$ results from a reduction in which there is a lack or excess of resources to be substituted, where  $\widetilde{x}$ denotes a multiset of free variables that are encapsulated within failure.

The empty bag is denoted $\oneb$. 
The bag enclosing the term $M$ is  $\bag{M}$.
The concatenation of bags $B_1$ and $B_2$ is denoted as   $B_1 \cdot B_2$; the concatenation operator `$\cdot$' is associative and  commutative, with  $\oneb$ as its identity. 
\secondrev{To ease readability, we rely on a shorthand notation for bags: we often write $\bag{N_1, N_2}$ rather than $\bag{N_1}\cdot \bag{N_2}$.}

We treat expressions as \emph{sums}, and use notations such as $\sum_{i}^{n} N_i$ for them. Sums are associative and commutative; reordering of the terms in a sum is performed silently.

\begin{exa}
\label{ex:terms}
We give some examples of terms and expressions in \lamrfail:

\begin{multicols}{2}
    \begin{itemize}
        \item $M_1 = (\lambda x. x ) \bag{y}$
        \item $M_2 = (\lambda x. x ) (\bag{y,z} )$
        \item $M_3 = (\lambda x. x ) \oneb  $
        \item $M_4 = (\lambda x. y ) \oneb $
        \item $M_5 = \fail^{\emptyset} $
        \item $M_6 = (\lambda x. x ) \bag{y} + (\lambda x. x ) \bag{z} $
    \end{itemize}
\end{multicols}
%Some intuitions follow.
%\begin{itemize}
%\item
Terms $M_1$, $M_2$, and $M_3$ illustrate the application of the identity function  $I=\lambda x.x$ to bags with different formats: a bag with one component, two components, and the empty bag, respectively. Special attention should be given to the fact that the $x$ has only one occurrence in $I$, whereas the bags contain zero or more components (resources). This way:
\begin{itemize}
    \item $M_1$ represents a term with a \emph{correct} number of resources;
    \item $M_2$ denotes a term with an \emph{excess} of resources; and 
    \item $M_3$ denotes a term with a \emph{lack} of resources
\end{itemize}
 This resource interpretation will become clearer once the  reduction semantics is introduced in the next subsection (cf. Example~\ref{ex:reducts}).

    %\item 
    Term $M_4$ denotes the application of a vacuous abstraction on $x$  to the empty bag $\oneb$.
    %\item 
    Term $M_5$ denotes a failure term with no associated variables.
    %\item 
    Expression $M_6$ denotes the non-deterministic sum between two terms, each of which denotes an application of $I$ to a bag containing one element.
%\end{itemize}
\end{exa}

\begin{nota}[Expressions]
Notation $N \in \expr{M}$ denotes that 
$N$ is part of the sum denoted by $\expr{M}$. 
Similarly, we write $N_i \in B$ to denote that $N_i$ occurs in the bag $B$, and $B \linsetminus N_i$ to denote the   bag that is obtained by removing one occurrence of the term $N_i$ from $B$. 
\end{nota}

\subsection{Reduction Semantics}\label{ssec:red_sem}\hfill

Reduction in \lamrfail is defined in terms of the relation $\red$, defined  in \figref{fig:reductions_lamrfail}; it operates lazily on expressions,  and  will be described after introducing some auxiliary notions. 

\begin{nota}
  We write  $\perm{B}$ to denote the set of all permutations of bag $B$.
Also, $B_i(n)$ denotes the $n$-th term in the (permuted)  $B_i$.
We define $\size{B}$ to denote the number of terms in bag $B$. 
That is, $\size{\oneb} = 0$
and 
$\size{\bag{M}  \cdot B} = 1 + \size{B}$.
\end{nota}

\begin{defi}[Set and Multiset of Free Variables]
\label{def:fvfail}
The set of free variables of a term, bag, and expression, is defined as
    \[
    \begin{array}{l@{\hspace{1cm}}l}
\begin{array}{r@{\hspace{-.01mm}}l}
\lfv{x} \; &= \{ x \}   \\
\lfv{\lambda x . M} \; &= \lfv{M}\!\setminus\! \{x\} \\
\lfv{M\ B} \; &=  \lfv{M} \cup \lfv{B}\\
 \lfv{M \esubst{B}{x}} \; &= (\lfv{M}\setminus \{x\}) \cup \lfv{B} 
\end{array}
         & 
 \begin{array}{r@{\hspace{-.01mm}}l}
\lfv{\oneb} \; &= \emptyset \\
\lfv{\bag{M}} \; &= \lfv{M} \\
 \lfv{B_1 \cdot B_2} \; &= \lfv{B_1} \cup \lfv{B_2}\\
  \lfv{\fail^{x_1, \cdots , x_n}}\; & = \{ x_1, \cdots , x_n \}\\
  \lfv{\expr{M}+\expr{N}}\; & = \lfv{\expr{M}} \cup \lfv{\expr{N}}
 \end{array}
    \end{array}
    \]

We use $\mfv{M}$ or $\mfv{B}$ to denote a multiset of free variables, defined similarly. 
We sometimes treat the sequence $\widetilde{x}$ as a (multi)set.
We write $\widetilde{x}\uplus \widetilde{y}$ to denote the multiset union of $\widetilde{x}$ and $\widetilde{y}$ and $\widetilde{x} \setminus y$ to express that every occurrence of $y$ is removed from $\widetilde{x}$.
A term $M$ is \emph{closed} if $\lfv{M} = \emptyset$ (and similarly for expressions). \revdaniele{As usual, we shall consider $\lamrfail$-terms modulo $\alpha$-equivalence.}
\end{defi}

\begin{nota}
 $\#(x , M)$ denotes the number of (free) occurrences of $x$ in $M$. 
Similarly, we write $\#(x,\widetilde{y}) $ to denote the number of occurrences of $x$ in the multiset $\widetilde{y}$. 
\end{nota}

\begin{defi}[Head]
\label{def:headfailure}
Given a term $M$, we define $\headf{M}$ inductively as:
\[
\begin{array}{l@{\hspace{1cm}}l}
\begin{array}{l}
 \headf{x}  = x     \\
  \headf{\lambda x.M}  = \lambda x.M \\
  \headf{M\ B}  = \headf{M}
\end{array}
     & 
\begin{array}{l}
\headf{\fail^{\widetilde{x}}}  = \fail^{\widetilde{x}}\\
\headf{M \esubst{ B }{x}} = 
\begin{cases}
    \headf{M} & \text{if $\#(x,M) = \size{B}$}\\
    \fail^{\emptyset} & \text{otherwise}
\end{cases}
\end{array}
\end{array}
\]
\end{defi}

\begin{defi}[Linear Head Substitution]
\label{def:linsubfail}
Let $M$ be a term such that $\headf{M}=x$. 
The \emph{linear head substitution} of a term $N$ for $x$ in $M$, denoted  $M\headlin{ N/x }$,  is  defined  as:
\[
\begin{aligned}
     x \headlin{ N / x}   &= N \\
    (M\ B)\headlin{ N/x}  &= (M \headlin{ N/x })\ B \\
    (M\ \esubst{B}{y})\headlin{ N/x }  &= (M\headlin{ N/x })\ \esubst{B}{y} \qquad \text{where } x \not = y
\end{aligned}
\]
\end{defi}

\noindent 
Finally, we define contexts for terms and expressions and convenient notations:

\begin{defi}[Term and Expression Contexts]\label{def:context_lamrfail}
Contexts for terms (CTerm) and expressions (CExpr) are defined by the following grammar:
\[
\begin{array}{c@{\hspace{1cm}}c}
 \text{(CTerm)}\quad  C[\cdot] ,  C'[\cdot] ::=  ([\cdot])B \mid ([\cdot])\esubst{B}{x}     & \text{(CExpr)}  \quad  D[\cdot] , D'[\cdot] ::= M + [\cdot] \mid   [\cdot] + M 
\end{array}
\]
\end{defi}

 % \figref{fig:rPrecongruencefail} defines a precongruence on terms and expressions closed under contexts, denoted $\pcong$, which will be useful later on. We write $M \pequiv M'$ whenever both $M \pcong M'$ and $M' \pcong M$ hold.

\begin{figure}[t]
    \centering
    
\begin{prooftree}
    \AxiomC{}
    \LeftLabel{\redlab{R:Beta}}
    \UnaryInfC{\((\lambda x. M) B \red M\ \esubst{B}{x}\)}
    \end{prooftree}

\begin{prooftree}
    \AxiomC{$\headf{M} = x$}
    \AxiomC{$B = \bag{N_1, \dots , N_k} \ , \ k\geq 1 $}
    \AxiomC{$ \#(x,M) = k $}
    \LeftLabel{\redlab{R:Fetch}}
    \TrinaryInfC{\(
    M\ \esubst{ B}{x } \red M \headlin{ N_{1}/x } \esubst{ (B\linsetminus N_1)}{ x }  + \cdots + M \headlin{ N_{k}/x } \esubst{ (B\linsetminus N_k)}{x}
    \)}
\end{prooftree}   

\begin{prooftree}
    \AxiomC{$\#(x,M) \neq \size{B} \qquad \widetilde{y} = (\mfv{M} \setminus x) \uplus \mfv{B} $}
    \LeftLabel{\redlab{R:Fail}}
    \UnaryInfC{\(  M\ \esubst{ B}{x } \red {}  \displaystyle\sum_{\perm{B}} \fail^{\widetilde{y}} \)}
\end{prooftree}

\begin{prooftree}
    \AxiomC{$  \widetilde{y} = \mfv{B} $}
    \LeftLabel{$\redlab{R:Cons_1}$}
    \UnaryInfC{\(  \fail^{\widetilde{x}}\ B \red{}  \displaystyle\sum_{\perm{B}} \fail^{\widetilde{x} \uplus \widetilde{y}} \)}
\end{prooftree}

\begin{prooftree}   
\AxiomC{$\size{B} = k  \qquad \#(z , \widetilde{x}) + k  \not= 0 \qquad  \widetilde{y} = \mfv{B}$}
    \LeftLabel{$\redlab{R:Cons_2}$}
    \UnaryInfC{\( \fail^{\widetilde{x}}\ \esubst{B}{z}  \red {} \displaystyle \sum_{\perm{B}} \fail^{(\widetilde{x} \setminus z) \uplus \widetilde{y}}  \)}
\end{prooftree}

\begin{prooftree}
        \AxiomC{$   M \red M'_{1} + \cdots + M'_{k} $}
        \LeftLabel{\redlab{R:TCont}}
        \UnaryInfC{$ C[M] \red  C[M'_{1}] + \cdots +  C[M'_{k}] $}
\DisplayProof\hfill%
        \AxiomC{$ \expr{M}  \red \expr{M}'  $}
        \LeftLabel{\redlab{R:ECont}}
        \UnaryInfC{$D[\expr{M}]  \red D[\expr{M}']  $}
\end{prooftree}
%\revdaniele{
%\begin{prooftree}
%\AxiomC{\(\expr{M}\pequiv \expr{N}\)}
%\AxiomC{\(\expr{N}\red \expr{N'}\qquad \expr{N'}\pequiv \expr{M'}\)}
%\LeftLabel{\redlab{R:Sequiv}}
%\BinaryInfC{\(\expr{M}\red \expr{M'}\)}
%\end{prooftree}
%}

    \caption{Reduction Rules for $\lamrfail$}
    \label{fig:reductions_lamrfail}
    \hfill \break
\end{figure}

\revdaniele{The reduction relation on $\lamrfail$ is defined by the rules in 
%Now we are ready to describe the reduction rules in 
\figref{fig:reductions_lamrfail}. %, and works modulo structural equivalence
} Intuitively, reductions in $\lamrfail$ work as follows: A $\beta$-reduction induces an explicit substitution of a bag $B$ for a variable $x$ \revdaniele{in a term $M$}, denoted $\revdaniele{M}\esubst{B}{x}$. \revdaniele{In the case the head of the term $M$ is $x$ and the size of the bag $B$ coincides with the number of occurrences of $x$ in $M$, }this explicit substitution is expanded into a sum of terms, each of which features a \emph{linear head substitution} $M\headlin{ N_i/x }$, where $N_i$ is a term in $B$, \revdaniele{which will replace the variable $x$ occurring in the head of $M$}; the rest of the bag ($B\linsetminus N_i$) is kept in an explicit substitution. However, if there is a mismatch between the number of occurrences of the variable to be substituted and the number of resources available, then the reduction leads to the failure term. Formally,

\begin{itemize}
\item {\bf Rule~$\redlab{R:Beta}$} is standard and admits a  bag (possibly empty) as parameter. 

\item {\bf Rule~$\redlab{R:Fetch}$} transforms a term into an expression: it opens up an explicit substitution into a sum of terms with linear head  substitutions, each denoting the partial evaluation of an element from the bag, \revdaniele{considering all the possible choices for substituting an element $N_i$ of the bag for $x$}. 
Hence, the size of the bag will determine the number of summands in the resulting expression.
\end{itemize}

\noindent
There are three rules reduce to the failure term: their objective is to accumulate all (free) variables involved in failed reductions. 

\begin{itemize}
\item  {\bf Rule~$\redlab{R:Fail}$} formalizes failure in the evaluation of an explicit substitution $M\ \esubst{ B}{x }$, which occurs if there is a mismatch between the resources (terms) present in $B$ and the number of occurrences of $x$ to be substituted. 
The resulting failure preserves all free variables in $M$ and $B$ within its attached multiset $\widetilde{y}$, \revdaniele{and all possible computations that could have failed, via permutation of the bags, are captured  in a non-deterministic sum}.
%\revt{C10}{When a reduction fails all possible computations that could have failed are captured via a non-deterministic sum. }
\item {\bf Rules~$\redlab{R:Cons_1}$ and $\redlab{R:Cons_2}$} describe reductions that lazily consume the failure term, when a term has $\fail^{\widetilde{x}}$ at its head position. 
The former rule consumes bags attached to it whilst preserving all its free variables.
The latter rule is similar but for the case of explicit substitutions; its second premise ensures that either (i)~the bag in the substitution is not empty or (ii)~the number of occurrences of $z$ in the current multiset of accumulated variables is not zero. 
%When both (i) and (ii) fail (i.e. $ \#(z , \widetilde{x}) + k  = 0 $), we apply a congruence rule seen later (cf. \figref{fig:rPrecongruencefail}), rather than reduction.
\end{itemize}
\revdaniele{
Notice that our Rule~\redlab{R:Fail} rule evolves to  a sum of failure terms, where each summand accounts for a permutation of the elements of the bag. As our reduction strategy fails eagerly this may  not be evident at first; however, there is still a non-deterministic choice of elements in $B$ that are waiting to be substituted at the point of failure (see Example~\ref{exa:fail_sum}).
}

Finally, we describe the contextual rules:
\begin{itemize}
\item {\bf Rule~$\redlab{R:TCont}$} describes the reduction of sub-terms within an expression; in this rule, summations are expanded outside of term  contexts.

\item {\bf Rule $\redlab{R:ECont}$} says that reduction of expressions is closed by expression contexts.
\end{itemize}

\begin{nota}
As standard, 
 $\red$ denotes one step reduction; 
 $\red^+$ and  $\red^*$ denote the transitive and the reflexive-transitive closure of $\red$, respectively. 
 We write $\expr{N}\red_{\redlab{R}} \expr{M}$ to denote that $\redlab{R}$ is the last (non-contextual) rule used in inferring the step from $\expr{N}$ to $\expr{M}$.
\end{nota}

\begin{exa}[Cont. Example~\ref{ex:terms}]
\label{ex:reducts}
We show how the terms in Example~\ref{ex:terms} can 
 reduce: % whilst applying the \revdaniele{structural equivalence \(\pequiv\):}
 %precongruence $\pcong$:

\begin{itemize}
    \item Reduction of the term $M_1$ with an adequate number of resources:  
    \[
    \begin{aligned} 
    (\lambda x. x ) \bag{y} &\red_\redlab{R:Beta} x \esubst{\bag{y}}{x}\\
    &\red_\redlab{R:Fetch} y\esubst{\oneb}{x}    &\text{ since }  \#(x,x) =\size{\bag{y}} =1
%    \\
%    &\revdaniele{\red} y  & \text{ since }  \#(x,y)  =\size{\bag{\oneb}}=0
%%    &\pcong y  & \text{ since }  \#(x,y)  =\size{\bag{\oneb}}=0
    \end{aligned}
    \]
    \item  Reduction of  term $M_2$ with  excess  of resources:  
    \[
    \begin{aligned}
    (\lambda x. x ) (\bag{y,z}  ) &\red_\redlab{R:Beta} x \esubst{(\bag{y,z}  )}{x} \\
    &\red_{\redlab{R:Fail}} \revd{B6}{\fail^{y,z} +  \fail^{y,z}} , \ \text{ since }  \#(x,x)  =1 \neq \size{\bag{y,z}}=2
    \end{aligned}
    \]
    \item Reduction of  term $M_3$ with lack  of resources:  
    \[
    \begin{aligned} 
    (\lambda x. x ) \oneb &\red_\redlab{R:Beta} x \esubst{\oneb}{x}\\
    &\red_{\redlab{R:Fail}} \fail^{\emptyset}, \ \text{ since }  \#(x,x)  =1 \neq \size{\oneb}=0
    \end{aligned}
    \]
    \item  Reduction of  term $M_4$ which is a vacuous abstraction applied to an empty bag:
    \[
    \begin{aligned} 
    (\lambda x. y ) \oneb & \red_\redlab{R:Beta} y\esubst{\oneb}{x}\\
    \end{aligned}
    \]

   % Notice how as $x$ is substituted for the empty bag in $y$ we could purely equate the term to $y$. In fact we do this later in defining a congruence.
    
    %&\revdaniele{\pequiv} y, \ \text{ since }  \#(x,y)  =0 = \size{\oneb} 
    
    \item $M_5 = \fail^{\emptyset} $ is unable to perform any reductions, i.e., it is irreducible.
    \item Reductions of the expression $M_6 = (\lambda x. x ) \bag{y} + (\lambda x. x ) \bag{z} $ :
    
        \begin{tikzpicture}
          \matrix (m) [matrix of math nodes, row sep=0.5em, column sep=0.5em,ampersand replacement=\&]
            { 
            \node(A){ }; \& 
                \node(B){ x \esubst{\bag{y}}{x} + (\lambda x. x ) \bag{z} }; \\
            \node(C){  (\lambda x. x ) \bag{y} + (\lambda x. x ) \bag{z} }; \& 
                \node(D){ }; \& 
                \node(G){x \esubst{\bag{y}}{x} +  x \esubst{\bag{z}}{x} }; \\
            \node(E){ }; \&
                \node(F){ (\lambda x. x ) \bag{y} + x \esubst{\bag{z}}{x} }; \\};
        \path (C) edge[->](B);
        \path (C) edge[->](F);
        \path (B) edge[->](G);
        \path (F) edge[->](G);
        \end{tikzpicture} 
\end{itemize}

\end{exa}

The following example illustrates the use of $\perm{B}$ in Rule~\redlab{R:Fail}: independently of the order in which the resources in the bag are used, the computation fails. 

\begin{exa}\label{exa:fail_sum}
Let $M = ( \lambda x . x \bag{x \bag{y}} )\  B$, with 
$B = \bag{z_1,z_2,z_1}$.
We have: 
  \[
    \begin{aligned}
       M \red_\redlab{R:Beta}& x \bag{x \bag{y}} \esubst{ \bag{z_1,z_2,z_1} }{ x } \\
         \red_\redlab{R:Fail}& \sum_{\perm{B}} \fail^{y , z_1, z_2, z_1 }
%        & \revo{}{ = \fail^{y , z_1, z_2, z_1} + \fail^{y , z_1, z_2, z_1}  + \fail^{y , z_1, z_2, z_1}} \\
%        & \revo{:(A5)}{\quad + \fail^{y , z_1, z_2, z_1}  + \fail^{y , z_1, z_2, z_1}  + \fail^{y , z_1, z_2, z_1} }  
    \end{aligned}
 \]
 The number of occurrences of $x$ in the term 
obtained after  $\beta$-reduction (2) does not match the size of the bag (3). Therefore, the reduction leads to failure. 
 \revo{:(A5)}{Notice that 
 $ \sum_{\perm{B}} \fail^{y , z_1, z_2, z_1 }$
 expands to a sum between six instances of 
$\fail^{y , z_1, z_2, z_1}$,
corresponding to \revdaniele{permutation of  3 elements of} the bag $B$.}
 \end{exa}
 
 Notice that the left-hand sides of the reduction rules in $\lamrfail$  do not interfere with each other. 
Therefore, reduction in \lamrfail satisfies a \emph{diamond property}:

\begin{restatable}[Diamond Property for \lamrfail]{propo}{diamondone}
\label{prop:conf1_lamrfail}
     For all $\expr{N}$, $\expr{N}_1$, $\expr{N}_2$ in $\lamrfail$ s.t. $\expr{N} \red \expr{N}_1$, $\expr{N} \red \expr{N}_2$ { with } $\expr{N}_1 \neq \expr{N}_2$ { then } there exists $\expr{M}$ such that  $\expr{N}_1 \red \expr{M}$, $\expr{N}_2 \red \expr{M}$.
\end{restatable}

\begin{proof}[Proof (Sketch)]
By inspecting the rules of \figref{fig:reductions_lamrfail} one can check that the left-hand sides only clash in a non-variable position with Rules~\redlab{R:Fail} and \redlab{R:Cons2}. The clash does not generate a critical pair: in fact, when applied to the $\lamrfail$-term $\fail^{z,\widetilde{x}}\esubst{\oneb}{z}$ both rules reduce to $\fail^{\widetilde{x}}$.
For all the other rules, whenever they have the same shape, the side conditions of the rules determine which rule can be applied. Therefore,
   an expression can only perform a choice of reduction steps when it is a sum of terms in which multiple summands can perform independent reductions. Without loss of generality, consider an expression $\expr{N} = N + M $ such that  $N \red N'$ and $M \red M'$. Then we let $\expr{N}_1 = N' + M$ and $\expr{N}_2 = N + M'$ by Rule~$\redlab{R:ECont}$.  The result follows for $\expr{M} = N' + M' $, since $\mathbb{N}_1
   \red \mathbb{M}$ and $\mathbb{N}_2
   \red \mathbb{M}$.
\end{proof}

\begin{rem}[A Sub-calculus without Failure (\lamr)]\label{rem:lamr}
We find it convenient to define \lamr, the sub-calculus of \lamrfail without  explicit failure. 
The syntax of \lamr is obtained from Definition~\ref{def:rsyntaxfail} by excluding $\fail^{\widetilde{x}}$ from the syntax of terms.
Accordingly, the  reduction relation for \lamr is given by  Rules~\redlab{R:Beta}, \redlab{R:Fetch}, \redlab{R:ECont}, and \redlab{R:TCont} in \figref{fig:reductions_lamrfail}. 
Finally, Definition~\ref{def:headfailure} is kept unchanged with the provision that $\headf{M \esubst{ B }{x}}$ is undefined when $\#(x,M) \not = \size{B}$.
\end{rem}

%%%%%

\subsection{Well-formed \texorpdfstring{$\lamrfail$}{λ^↯_⊕}-Expressions}
\label{sec:lamfailintertypes}\hfill 

As mentioned in \secref{s:key}, we define a notion of {\em well-formed expressions}  for $\lamrfail$ by relying on a non-idempotent intersection type system, similar to the one given by Pagani and Della Rocca in~\cite{PaganiR10}. Our system for well-formed expressions will be defined in two stages:
\begin{enumerate}
\item First we define a intersection type system for the sub-language $\lamr$ (cf. Rem.~\ref{rem:lamr}), given in  \figref{fig:app_typingrepeat}. 
\revd{:B7}{Unlike the system in \cite{PaganiR10}, our type system includes a weakening rule and a rule for typing explicit substitutions.}
%This  characterizes it as \emph{non-relevant}, following criteria by Damiani and Giannini~\cite{DBLP:conf/tacs/DamianiG94}.
\item Second, we define well-formed expressions for the full language \lamrfail, via \defref{d:wellf}.
\end{enumerate}

We say that we check for ``well-formedness'' (of terms, bags, and expressions) to stress that, unlike standard type systems, our system is able to account for terms that may reduce to the failure term.

\subsubsection{Intersection Types}\label{ss:lamrfail_types}\hfill

Intersection types allow us to reason about types of resources in bags but also about every occurrence of a variable. 
That is, non-idempotent intersection types enable us to distinguish expressions not only by measuring the size of a bag but also by counting the number of times a variable occurs within a term.
\begin{defi}[Types for \lamrfail]
\label{d:typeslamrfail}
We define {\em strict} and {\em multiset types} by the   grammar:
\[
\begin{array}{c@{\hspace{1.2cm}}c}
  \text{(Strict)}\quad  \sigma, \tau, \delta ::= \unit \sep \arrt{\pi}{\sigma}   & \text{(Multiset)} \quad \pi ::=  \secondrev{\sigma^k} \sep \omega
\end{array}
\]
\secondrev{where $\sigma^k$ stands 
 for  $\sigma\wedge \cdots \wedge \sigma$
($k$ times, for some $k>0$).} %  and $\sigma^0$ coincides with $\omega$ (if $k=0$).} 
\end{defi}
A strict type can be the unit type  $\unit$  or a functional type  \arrt{\pi}{\sigma}, where 
$\pi$ is a multiset type and $\sigma$ is  a strict type. 
Multiset types can be either 
an intersection  of strict types
\secondrev{ $\sigma^k$
(if $k>0$)
or the empty type $\omega$, which would correspond to $\secondrev{\sigma^k}$ with $k = 0$. Hence, $\sigma^k$ denotes an intersection;}
 the operator $\wedge $ is commutative, associative, and non-idempotent, that is, $\sigma \wedge \sigma \neq \sigma$. The empty type is the type of the empty bag; it acts as  the identity element to~$\wedge$.

% \begin{nota}\label{not:lamrrepeat}
% \revo{}{Given $k \geq 0$, we shall write $\sigma^k$ to stand 
%  for $\omega$ (if $k=0$)
%  or
%  for $\sigma\wedge \cdots \wedge \sigma$
% ($k$ times, if $k>0$).}
% \end{nota}

%\revo{}{Variables are typed with intersection types. 
%%Multiple occurrences of a variable can occur within an assignment; they are assigned only strict types. 
%For instance,  $x:(\tau \rightarrow \tau) \wedge \tau$  is a valid type assignment: it means that $x$ can be of both type $\tau \rightarrow \tau$ and $\tau$. 
%}

\begin{defi}
\label{d:tcontsource}
\revo{}{
{\em Type contexts}  $\Gamma , \Delta, \ldots $ are sets of type assignments  $x: \pi$, as defined by the grammar:}
\revd{B9}{
\[
    \begin{aligned}
        \Gamma, \Delta & = \dash \sep \Gamma , x:\pi 
    \end{aligned}
\]}
\revo{}{
The set of variables in $\Gamma$ is denoted as $\dom{\Gamma}$.
In writing $\Gamma, x:\pi$  we assume that $x \not \in \dom{\Gamma}$. 
}
\revo{}{
We generalize \revdaniele{the operator} $\wedge$ from types to contexts, and define  $\Gamma \contexcat \Delta$ as follows:
$$(\Gamma_1 \contexcat \Gamma_2)(x) = 
\begin{cases}
    x: \pi_1 \wedge \pi_2 &  x:\pi_i \in \Gamma_i ,\  \pi_i \not = \omega , \  i \in \{ 1 , 2 \} \\
    x: \pi_i  &  x:\pi_i \in \Gamma_i, x \not \in \dom{\Gamma_j} ,\  i \not = j,\ i,j \in \{1,2\} \\
    \text{undefined} & \text{otherwise}
\end{cases}
$$
 }

\end{defi}

% and the empty type assignment is denoted $\emptyset$.
{\em Type judgements} are of the form $\Gamma \vdash \expr{M}:\sigma$, where $\Gamma$ is a type context. We write $\vdash \expr{M}:\sigma$ to denote $\dash \vdash \expr{M}:\sigma$.

\begin{defi}(Well-typed Expressions)
An expression $\expr{M} \in \lamr$ is {\em well-typed} (or typable) if there exist $\Gamma$ and  $\tau$ such that $\Gamma \vdash \expr{M} : \tau$ is entailed via the rules in \figref{fig:app_typingrepeat}.
\end{defi}

\begin{figure*}[!t]
    \centering
    
\begin{prooftree}
    \AxiomC{}
    \LeftLabel{\redlab{T:var}}
    \UnaryInfC{\( x: \sigma \vdash x : \sigma\)}
    \DisplayProof
    \hfill
    \AxiomC{\(  \)}
    \LeftLabel{\redlab{T:\oneb}}
    \UnaryInfC{\( \vdash \oneb : \omega \)}
    \DisplayProof
    \hfill
    \AxiomC{$ \Gamma \vdash M: \sigma$}
    \AxiomC{$ $}
    \LeftLabel{\redlab{T:weak}}
    \BinaryInfC{$ \Gamma, x:\omega \vdash M: \sigma $}
\end{prooftree}

\begin{prooftree}
    \AxiomC{\( \revo{}{\Gamma , {x}: \sigma^k \vdash M : \tau} \)}
    \LeftLabel{\redlab{T:abs}}
    \UnaryInfC{\( \Gamma \vdash \lambda x . M :  \sigma^k  \rightarrow \tau \)}
\DisplayProof
\hfill
  \AxiomC{\( \Gamma \vdash M : \sigma\)}
    \AxiomC{\( \Delta \vdash B : \sigma^k\)}
    \LeftLabel{\redlab{T:bag}}
    \BinaryInfC{\( \revo{}{\Gamma \contexcat \Delta \vdash \bag{M}\cdot B:\sigma^{k+1}} \)}
\end{prooftree}

\begin{prooftree}
  \AxiomC{\( \Gamma \vdash M : \pi \rightarrow \tau \)}
    \AxiomC{\( \Delta \vdash B : \pi \)}
        \LeftLabel{\redlab{T:app}}
    \BinaryInfC{\( \revo{}{\Gamma \contexcat \Delta \vdash M\ B : \tau}\)}
 \DisplayProof
 \hfill
   \AxiomC{\( \revo{}{\Gamma ,  {x}:\sigma^{k} \vdash M : \tau} \)}
         \AxiomC{\( \Delta \vdash B : \sigma^{k} \)}
    \LeftLabel{\redlab{T:ex \dash sub}}    
    \BinaryInfC{\( \revo{}{\Gamma \contexcat \Delta \vdash M \esubst{ B }{ x } : \tau} \)}
\end{prooftree}

\begin{prooftree}
    \AxiomC{$ \Gamma \vdash \expr{M} : \sigma$}
    \AxiomC{$ \Gamma \vdash \expr{N} : \sigma$}
    \LeftLabel{\redlab{T:sum}}
    \BinaryInfC{$ \Gamma \vdash \expr{M}+\expr{N}: \sigma$}
\end{prooftree}

    \caption{Typing Rules for \lamr }
    \label{fig:app_typingrepeat}
\end{figure*}

%The type system for $\lamr$ is defined by the rules in\figref{fig:app_typingrepeat}.
The rules are standard. We only consider intersections of the same strict type, say $\sigma$, since the current objective is to count the number of occurrences of a variable in a term, and measure the size of a bag. We now give a brief description of the rules in \figref{fig:app_typingrepeat}:
\begin{itemize}
\item {\bf Rules~\redlab{T{:}var}, \redlab{T{:}\oneb} and \redlab{T{:}weak}} are as expected: the first assigns a type to a variable, the second 
assigns the empty bag $\oneb$ the empty type $\omega$, and the third introduces a  useful weakening principle. 
\item {\bf Rule~\redlab{T:abs}} types an abstraction $\lambda x. M$ with $\sigma^k\to \tau$, as long as the variable assignment \revo{}{$x:\sigma^k$ has an intersection type with $\sigma $ occurring exactly $k$ times.}
\item {\bf Rule~\redlab{T:bag}} types a bag $B$ with a type $\sigma^{k+1}$ as long as every component of $B$ is typed with same type $\sigma$, a defined amount of times.
\item {\bf Rule~\redlab{T{:}app}} types an application $M\ B$ with $\tau$ as long as $M$ and $B$ match on the multiset type $\pi$, i.e., $M:\pi\to \tau$ and $B:\pi$. Intuitively, this means that $M$ expects a fixed amount of resources, and $B$ has exactly this number of resources.
\item {\bf Rule~\redlab{T{:}ex \dash sub}} types an explicit substitution $M\esubst{B}{x}$ with $\tau$ as long as the bag $B$ consists of elements of the same type as $x$ and the size of $B$ matches the number of times $x$ occurs in $M$, i.e., $B:\sigma^k$ and $x:\sigma^k$ types the assignment of $M:\tau$.
\item {\bf Rule~\redlab{T:sum}} types an expression (a sum) with a type $\sigma$, if each summand has type $\sigma$.
\end{itemize}

Notice that with the typing rules for $\lamr$ the failure term $\fail$ cannot be typed.  We could consider this set of rules as a type system for  \lamrfail, i.e. the extension of $\lamr$ with failure,  in which failure can be expressed but not typed.

\begin{exa}[Cont. Example~\ref{ex:reducts}]
\label{ex:welltyped}
We explore the typability of some of the terms given in previous examples:
\begin{enumerate}
    \item Term $M_1 = (\lambda x. x ) \bag{y} $ is typable, as we have:
  \begin{prooftree}
  \AxiomC{}
  \LeftLabel{\redlab{T:var}}
  \UnaryInfC{\( x: \sigma \vdash x : \sigma\)}
  \LeftLabel{\redlab{T:abs}}
  \UnaryInfC{\(  \vdash \lambda x . x :  \sigma \rightarrow \sigma \)}
  \AxiomC{}
  \LeftLabel{\redlab{T:var}}
  \UnaryInfC{\(  y : \sigma \vdash y : \sigma\)}
  \AxiomC{\(  \)}
  \LeftLabel{\redlab{T:\oneb}}
  \UnaryInfC{\( \vdash \oneb : \omega \)}
  \LeftLabel{\redlab{T:bag}}
  \BinaryInfC{\( y : \sigma \vdash \bag{y}\cdot \oneb:\sigma \)}
  \LeftLabel{\redlab{T:app}}
  \BinaryInfC{\( y : \sigma \vdash (\lambda x. x ) \bag{y} : \sigma\)}
  \end{prooftree}
    \item Term $M_2 = (\lambda x. x ) (\bag{y,z}) $ is not typable.
    \begin{itemize}
    \item The function $\lambda x. x$ has a functional type $\sigma \to \sigma$;
        \item  The bag has an intersection type of size two: $y:\sigma, z:\sigma \vdash(\bag{y,z}):\sigma^2$;
        \item  Rule~$\redlab{T:app}$  requires a match between the type of the bag and the left of the arrow: it can only consume a bag of type $\sigma$.
    \end{itemize}
    
    \item Similarly,  $M_3 = (\lambda x. x ) \oneb $ is not typable: since $\lambda x.x$ has type $\sigma\to \sigma$, to apply the Rule~$\redlab{T:app}$  the bag must have a type $\sigma$, but the empty bag $\oneb$ can only be typed with~$\omega$.
    \item Term $M_4 = (\lambda x. y ) \oneb $ is typable, as follows:
   \begin{prooftree}
    \AxiomC{}
   \LeftLabel{\redlab{T:var}}
  \UnaryInfC{$ y: \sigma  \vdash y : \sigma$}
  \AxiomC{$ $}
  \LeftLabel{\redlab{T:weak}}
  \BinaryInfC{$y: \sigma, x:\omega  \vdash y : \sigma  $}
  \LeftLabel{\redlab{T:abs}}
  \UnaryInfC{\( y : \sigma \vdash \lambda x . y :  \omega \rightarrow \sigma \)}
  \AxiomC{\(  \)}
  \LeftLabel{\redlab{T:\oneb}}
  \UnaryInfC{\( \vdash \oneb : \omega \)}
  \LeftLabel{\redlab{T:app}}
  \BinaryInfC{\( y : \sigma \vdash (\lambda x. y ) \oneb : \sigma\)}
  \end{prooftree}
    \end{enumerate}
\end{exa}

Our typing system for \lamr satisfies standard properties,  such as subject reduction, which follows from the {\em Linear} Substitution Lemma. We stress `linearity' because the lemma is stated in terms of the head linear substitution $\headlin{\cdot}$.

\begin{restatable}[Linear Substitution Lemma for \lamr]{lema}{subtlemfailnofail}
\label{lem:subt_lem}
%If $\Gamma , x:\sigma \vdash M: \tau$, $\headf{M} = x$, and $\Delta \vdash N : \sigma$ 
%then 
%$\Gamma , \Delta \vdash M \headlin{ N / x }$.
\revo{}{
If $\Gamma , x:\sigma^k \vdash M: \tau$ (with $k \geq 1$), $\headf{M} = x$, and $\Delta \vdash N : \sigma$ 
then 
$\Gamma \contexcat \Delta, x:\sigma^{k-1} \vdash M \headlin{ N / x }: \tau $.
}
\end{restatable}

\begin{proof}
Standard, by induction on the rule applied in $\Gamma, x:\sigma \vdash M:\tau$. 
\end{proof}

\begin{restatable}[Subject Reduction for \lamr]{thms}{subredone}
\label{t:app_lamrsr}
If $\Gamma \vdash \expr{M}:\tau$ and $\expr{M} \red \expr{M}'$ then $\Gamma \vdash \expr{M}' :\tau$.
\end{restatable}

\begin{proof} By  induction on the reduction rule (\figref{fig:reductions_lamrfail}) applied in $\expr{M}$. 
\end{proof}

\revo{}{
\begin{restatable}[Linear Anti-substitution Lemma for \lamr]{lema}{explemfailnofail}
\label{lem:antisubt_lem}
\revdaniele{Let $M$ and $N$ be $\lamr$-terms such that} $\headf{M} = x$, then we have:
\begin{itemize}
    \item $\Gamma, x:\sigma^{k-1} \vdash M \headlin{ N / x }: \tau$, with $k > 1$, then  \revdaniele{there exist} $  \Gamma_1, \Gamma_2$  such that  $\Gamma_1 , x:\sigma^k \vdash M: \tau$, and $\Gamma_2 \vdash N : \sigma$, where $\Gamma = \Gamma_1 \contexcat \Gamma_2$.
    \item $\Gamma \vdash M \headlin{ N / x }: \tau$, with $x \not \in \dom{\Gamma}$, then \revdaniele{there exist} $ \Gamma_1, \Gamma_2$  such that  $\Gamma_1 , x:\sigma \vdash M: \tau$ , and $\Gamma_2 \vdash N : \sigma$, where $\Gamma = \Gamma_1 \contexcat \Gamma_2$.
\end{itemize}
%
%
%If $\Gamma, x:\sigma^{k-1} \vdash M \headlin{ N / x }: \tau$ with $\headf{M} = x$ and $k \geq 0$ then $\exists \ \Gamma_1, \Gamma_2$  such that  $\Gamma_1 , x:\sigma^k \vdash M: \tau$ (with $k \geq 1$), and $\Gamma_2 \vdash N : \sigma$ with $\Gamma = \Gamma_1 \contexcat \Gamma_2$.
\end{restatable}
}

\begin{proof}
 \revdaniele{Standard, by structural induction.}
\end{proof}

\revo{}{
\begin{restatable}[Subject Expansion for \lamr]{thms}{subexpone}
\label{t:app_lamrexp}
If $\Gamma \vdash \expr{M}':\tau$ and $\expr{M} \red \expr{M}'$ then $\Gamma \vdash \expr{M} :\tau$.
\end{restatable}
}

\begin{proof}
 \revdaniele{Standard, by structural induction.} See \appref{app:typeshar} for details.
\end{proof}

%%%%
\subsubsection{Well-formed Expressions (in \texorpdfstring{$\lamrfail$}{})}\label{ss:lamrfail_wf}\hfill

Building upon the type system for \lamr, we now define a type system for checking {\em well-formed} \lamrfail-expressions. This approach enables us to admit expressions with a failing  computational behavior, may it be due to the mismatch in the number of resources required and available, or be due to consumption of a failing behavior by another expression.
 \secondrev{Such definition relies on the \emph{core context} which is the key to the well-formedness of failure terms:  free variables that are result of weakening will disregarded in the typing of the failure term.}
 
 %not  be expressed in the failure term as these variables contain only weakening behavior the result of their failure is inconsequential and hence disregarded in typing of the failure term.}
%\secondrev{To do this we must first define an operation on contexts that removes weakening as we wish for failure to not preserve the weakening structure. \joe{(NEED TO IMPROVE THIS SENTENCE AND FLOW. to do: better explain why we don't want to weaken here (relates to intermediate calculus having failure and weakening and encoding being independent of types. think what $x:\omega \wfdash \fail^{x} $, this shouldn't get translated to $x:\omega \wfdash \fail^{x_1} [x_1 \leftarrow x]$ but instead $x:\omega \wfdash \fail^{\empty} $ is translated to $x:\omega \wfdash \fail^{\empty}$))  }}

\secondrev{
\begin{defi}[Core Context]
    Given a context $\Gamma$, the associated  \emph{core context} is defined as 
$\core{\Gamma} = \{ x:\pi \in \Gamma \,|\, \pi \not = \omega\}$. 
%Let $\Gamma$ be a context such that $\core{\Gamma}$ contains only unary multiset types.  
\end{defi}
}

%\secondrev{We now give the well-formed rules for \lamrfail.}

\begin{defi}[Well-formed \lamrfail expressions]
\label{d:wellf}
An expression $ \expr{M}$ is \revd{B11}{\emph{well-formed}} if  there exist  $\Gamma$ and  $\tau$ such that  $ \Gamma \wfdash  \expr{M} : \tau  $ is entailed via the rules in \figref{fig:app_wf_rules}.
\end{defi}

\begin{figure}
    \centering

\begin{prooftree}
\AxiomC{\( \Gamma \vdash \expr{M} : \tau \)}
\LeftLabel{\redlab{F:wf \dash expr}}
\UnaryInfC{\( \Gamma \wfdash  \expr{M} : \tau \)}
\DisplayProof
\hfill
\AxiomC{\( \Gamma \vdash B : \pi \)}
\LeftLabel{\redlab{F:wf \dash bag}}
\UnaryInfC{\( \Gamma \wfdash  B : \pi \)}
\DisplayProof
\hfill
\AxiomC{\( \Delta  \wfdash M : \tau\)}
\LeftLabel{ \redlab{F:weak}}
\UnaryInfC{\( \Delta , x: \omega \wfdash M: \tau \)}
\end{prooftree}

 \begin{prooftree}
 \AxiomC{\( \Gamma , {x}: \sigma^n \wfdash M : \tau \quad x\notin \dom{\Gamma} \)}
 \LeftLabel{\redlab{F:abs}}
 \UnaryInfC{\( \Gamma \wfdash \lambda x . M : \sigma^n  \rightarrow \tau \)}
 \DisplayProof\hfill
 \AxiomC{\( \Gamma \wfdash M : \sigma\)}
 \AxiomC{\( \Delta \wfdash B : \sigma^k\)}
 \LeftLabel{\redlab{F:bag}}
 \BinaryInfC{\( \Gamma \contexcat \Delta \wfdash \bag{M}\cdot B:\sigma^{k+1}\)}
 \end{prooftree}
        
  \begin{prooftree}
\AxiomC{$ \Gamma \wfdash \expr{M} : \sigma$}
\AxiomC{$ \Gamma \wfdash \expr{N} : \sigma$}
\LeftLabel{\redlab{F:sum}}
\BinaryInfC{$ \Gamma \wfdash \expr{M}+\expr{N}: \sigma$}
\DisplayProof
\hfill
\AxiomC{\(\secondrev{ \dom{\core{\Gamma}} = \widetilde{x} } \)}
\LeftLabel{\redlab{F:fail}}
\UnaryInfC{\( \secondrev{ {\Gamma} \wfdash  \fail^{\widetilde{x}} : \tau } \)}
\end{prooftree}

\begin{prooftree}
\AxiomC{\( \Gamma , {x}:\sigma^{k} \wfdash M : \tau \) }
\AxiomC{\( \Delta \wfdash B : \sigma^{j} \) \ \( k, j \geq 0 \)}
\LeftLabel{\redlab{F:ex \dash sub}}  
\BinaryInfC{\( \Gamma \contexcat \Delta \wfdash M \esubst{ B }{ x } : \tau \)}
\end{prooftree}
    
    \begin{prooftree}
        \AxiomC{\( \Gamma \wfdash M : \sigma^j \rightarrow \tau \)}
        \AxiomC{\( \Delta \wfdash B : \sigma^{k} \)}
        \AxiomC{\( k, j \geq 0 \)}
            \LeftLabel{\redlab{F:app}}
        \TrinaryInfC{\( \Gamma \contexcat \Delta \wfdash M\ B : \tau\)}
    \end{prooftree}
    \caption{Well-Formed Rules for \lamrfail}\label{fig:app_wf_rules}
\end{figure}
    
Below we give a brief description of the rules in~\figref{fig:app_wf_rules}. Essentially, they differ from the ones in \figref{fig:app_typingrepeat}, by allowing mismatches between the number of copies of a variable in a functional position and the number of components in a bag. 
\begin{itemize}
\item {\bf Rules~$\redlab{F{:}wf \dash expr}$ and $\redlab{F{:}wf \dash bag}$} derive that well-typed expressions and bags in $\lamr$ are well-formed. 
\item {\bf Rules~\redlab{F{:}abs}, \redlab{F{:}bag}, and \redlab{F{:}sum}} are as in the  type system for \lamr, but extended to the system of well-formed expressions. 
\item {\bf Rules~$\redlab{F{:}ex \dash sub}$ and  $\redlab{F{:}app}$} differ from the similar typing rules as the size of the bags (as declared in their types) is no longer required to match the number of occurrences of the variable assignment in the typing context (\redlab{F:ex\dash sub}), or the type of the term in the functional position (\redlab{F:app}).
\item {\bf Rule~$\redlab{F{:}fail}$} has no analogue in the type system: we allow $\fail^{\widetilde{x}}$ to be well-formed with any strict type, provided that the core context contains the types of the variables in $\widetilde{x}$ (i.e., none of the variables in $\widetilde{x}$ is typed with $\omega$). 
\end{itemize}
Clearly, the set of {well-typed} expressions is strictly included in the set of {well-formed} expressions. 
Take, e.g., $M=x\esubst{ \bag{N_1,N_2} }{ x }$, where both $N_1$ and $N_2$ are well-typed.
    It is easy to see that  $M$  is well-formed. However, $M$ is not well-typed.

\begin{exa}[Cont. Example~\ref{ex:welltyped}]
\label{ex:wellformed}
We explore the well-formedness of some of the terms motivated in previous examples:

\begin{enumerate}
    \item Term $M_1 = (\lambda x. x ) \bag{y} $ is well-typed and also well-formed, as we have:
    \begin{prooftree}
           \AxiomC{\( y:\sigma \vdash (\lambda x. x ) \bag{y} : \sigma \)}
           \LeftLabel{\redlab{F:wf \dash expr}}
           \UnaryInfC{\( y:\sigma \wfdash  (\lambda x. x ) \bag{y} : \sigma \)}
    \end{prooftree}

    \item We saw that term $M_2 = (\lambda x. x ) (\bag{y,z}) $ is not typable; however, it is well-formed:        
    \begin{prooftree}
            \AxiomC{\(   \vdash \lambda x. x : \sigma^1 \rightarrow \sigma \)}
            \LeftLabel{\redlab{F:wf \dash expr}}
            \UnaryInfC{\(  \wfdash \lambda x. x : \sigma^1 \rightarrow \sigma \)}
            \AxiomC{\( y:\sigma, z:\sigma  \vdash \bag{y,z} : \sigma^{2} \)}
            \LeftLabel{\redlab{F:wf \dash bag}}
            \UnaryInfC{\( y:\sigma, z:\sigma \wfdash \bag{y,z} : \sigma^{2} \)}
                        \AxiomC{\( 1, 2 \geq 0 \)}
             \LeftLabel{\redlab{F:app}}
        \TrinaryInfC{\( y:\sigma, z:\sigma \wfdash (\lambda x. x ) (\bag{y,z})  : \sigma \)}
    \end{prooftree}
  Notice that   both $\vdash \lambda x. x : \sigma^1 \rightarrow \sigma$ and $ \Gamma \vdash \bag{y,z} : \sigma^{2} $ are well-typed.
    
    \item Similarly, the term $M_3 = (\lambda x. x ) \oneb $ is also well-formed. The corresponding derivation is as above, but uses an empty context as well as the well-formedness rule for bags: 
    \begin{prooftree}
            \AxiomC{\( \ \vdash \oneb : \sigma^{0} \)}
            \LeftLabel{\redlab{F:wf \dash bag}}
            \UnaryInfC{\(  \wfdash \oneb : \sigma^{0} \)}
    \end{prooftree}
    Notice how $\sigma^0 = \omega$ and that $  \wfdash \oneb : \omega$.
    
    \item Term $M_4 = (\lambda x. y ) \oneb $ is well-typed and also well-formed.
    
    \item Interestingly, term $M_5 = \fail^{\emptyset} $ is well-formed as:
    \begin{prooftree}
        \AxiomC{\(  \)}
        \LeftLabel{\redlab{F:fail}}
        \UnaryInfC{\(  \wfdash  \fail^{\emptyset} : \tau  \)}
    \end{prooftree}
\end{enumerate}
\end{exa}

\begin{exa}
Let us consider an expression that is not well-formed:
%\joe{I added spaces here to try to make it more readable but they can be removed if you dont like them here}
 $$ \lambda x . x \bag{\lambda y . y,~\lambda z . z_1 \bag{z_1\bag{z_2}}}.$$
Notice that $\lambda x.x$ is applied to bags of two different types:  
\begin{itemize}
\item The first bag containing $\lambda y. y$ is well-typed, thus well-formed. Consider the derivation $\Pi_1$:

\begin{prooftree}
    \AxiomC{}
    \LeftLabel{\redlab{T:var}}
    \UnaryInfC{\(y:\sigma\vdash y: \sigma\)}
        \LeftLabel{\redlab{T:abs}}
    \UnaryInfC{\(\vdash  \lambda y.y: \sigma\to \sigma\)}
    \AxiomC{}
    \LeftLabel{\redlab{T:\oneb}}
    \UnaryInfC{\(\vdash\oneb:\omega\)}
        \LeftLabel{\redlab{T:bag}}
    \BinaryInfC{\(\vdash \bag{\lambda y.y}\cdot \oneb:\sigma\to \sigma\)}
        \LeftLabel{\redlab{F:wf\dash bag}}
     \UnaryInfC{\(\wfdash \bag{\lambda y.y}\cdot \oneb:\sigma\to \sigma\)}
\end{prooftree}
\end{itemize}
 In the rest of the example we will omit the labels of rule applications, and concatenations with the empty bag $\oneb$ (i.e., $\bag{\lambda y. y}\cdot \oneb$ will be written simply as $\bag{\lambda y. y}$) and corresponding sub-derivations consisting of applications of Rule~\redlab{T:\oneb}.
\begin{itemize}
\item The second bag  contains $\lambda z . z_1 \bag{z_1\bag{z_2}} $  contains an abstraction that acts as a weakening as $z$ does not appear within $z_1 \bag{z_1\bag{z_2}}$. Consider the derivation $\Pi_2$:
 
\begin{prooftree}
     \AxiomC{\(z_1:\sigma\to\sigma\vdash z_1:\sigma\to \sigma\)} 
  \AxiomC{\(z_1:\sigma\to\sigma\vdash z_1:\sigma\to \sigma\)}  
     \AxiomC{\(z_2:\sigma\vdash z_2:\sigma\)} 
    %  \AxiomC{}
    %  \UnaryInfC{\(\vdash\oneb:\omega\)}
     \UnaryInfC{\(z_2:\sigma\vdash \bag{z_2}:\sigma\)}
     \BinaryInfC{\(z_1:\sigma\to\sigma, z_2:\sigma\vdash z_1\bag{z_2}:\sigma\)}
    %  \AxiomC{}
    %  \UnaryInfC{\(\vdash \oneb:\omega\)}
     \UnaryInfC{\(z_1:\sigma\to\sigma, z_2:\sigma\vdash \bag{z_1\bag{z_2}}:\sigma\)}
    \BinaryInfC{\(z_1:\sigma\to\sigma\wedge\sigma\to\sigma, z_2:\sigma\vdash z_1\bag{z_1\bag{z_2}}:\sigma\)}
\UnaryInfC{\(z_1:\sigma\to\sigma\wedge\sigma\to\sigma, z_2:\sigma, z:\omega\vdash z_1\bag{z_1\bag{z_2}}:\sigma\)}
\UnaryInfC{\(z_1:\sigma\to\sigma\wedge\sigma\to\sigma, z_2:\sigma\vdash\lambda z. z_1\bag{z_1\bag{z_2}}:\omega\to \sigma\)}
\UnaryInfC{\(z_1:\sigma\to\sigma\wedge\sigma\to\sigma, z_2:\sigma\wfdash\lambda z. z_1\bag{z_1\bag{z_2}}:\omega\to \sigma\)}
\UnaryInfC{\(\underbrace{z_1:\sigma\to\sigma \wedge \sigma\to\sigma, z_2:\sigma}_{\Gamma}\wfdash\bag{\lambda z. z_1\bag{z_1\bag{z_2}}}:\omega\to \sigma\)}
\end{prooftree}

\item The concatenation of these two bags is not well-formed since each component has a different type: $\sigma\to \sigma$ and $\omega\to \sigma$. Therefore,  $ \lambda x . x \bag{ ~ \lambda y . y ~ , ~ \lambda z . z_1 \bag{z_1\bag{z_2}} ~}$ is not well-formed.
\end{itemize}
Notice that if we change $\lambda y.y$ to $\lambda y.y_1$ in the first bag, we would have a derivation $\Pi_1'$ for $y_1:\sigma\wfdash\lambda y.y_1:\omega \to \sigma$. This would allow us to concatenate the bags with derivation $\Pi_3$:

\begin{prooftree}
     \AxiomC{$\Pi_1'$}
     \noLine
     \UnaryInfC{$y_1:\sigma\wfdash \lambda y.y_1:\omega \to \sigma$}
     \AxiomC{$\Pi_2$}
     \noLine
     \UnaryInfC{\(\Gamma\wfdash\lambda z. z_1\bag{z_1\bag{z_2}}:\omega\to \sigma\)}
      \AxiomC{}
     \UnaryInfC{\(\oneb:\omega\)}
     \BinaryInfC{\(\Gamma\wfdash\bag{\lambda z. z_1\bag{z_1\bag{z_2}}}\cdot \oneb:\omega\to \sigma\)}
     \BinaryInfC{\(\Gamma, y_1:\sigma\wfdash\bag{\lambda y.y_1}\cdot \bag{\lambda z. z_1\bag{z_1\bag{z_2}}}\cdot \oneb:(\omega\to \sigma)^2\)}
\end{prooftree}

% \daniele{I still think that the term complete term is not well-formed. }
Thus, the whole term becomes well-formed:
\begin{prooftree}
     \AxiomC{}
     \noLine
     \UnaryInfC{$x:\omega\to \sigma \vdash x:\omega\to \sigma$}
\UnaryInfC{$\vdash \lambda x. x : (\omega\to \sigma)\to \omega \to \sigma$}
\UnaryInfC{$\wfdash \lambda x. x : (\omega\to \sigma)\to \omega \to \sigma$}
\AxiomC{$\Pi_3$}
\noLine
\UnaryInfC{\(\Gamma, y_1:\sigma\wfdash\bag{ ~ \lambda y.y_1 ~,~ \lambda z. z_1\bag{z_1\bag{z_2}} ~ }:(\omega\to \sigma)^2\)}
\BinaryInfC{\(\Gamma, y_1:\sigma\wfdash\lambda x. x\bag{ ~ \lambda y.y_1 ~,~ \lambda z. z_1\bag{z_1\bag{z_2}} ~ }:\omega\to \sigma\)}
\end{prooftree}
\end{exa}

Well-formedness  rules satisfy subject reduction with respect to the rules in~\figref{fig:reductions_lamrfail} and relies on the linear substitution lemma for $\lamrfail$:

\begin{restatable}[Substitution Lemma for \lamrfail]{lema}{subtlemfail}
\label{lem:subt_lem_fail}

\revo{}{
If $\Gamma , x:\sigma^k \wfdash M: \tau$ (with $k \geq 1$), $\headf{M} = x$, and $\Delta \wfdash N : \sigma$ 
then 
$\Gamma \contexcat \Delta, x:\sigma^{k-1}  \wfdash M \headlin{ N / x }$.
}
\end{restatable}

We now show subject reduction on well formed expressions in \lamrfail. We use our results of subject reduction for well-typed \lamr (Theorem~\ref{t:app_lamrsr}) and extend them to \lamrfail.

\begin{restatable}[Subject Reduction in \lamrfail]{thms}{applamrfailsr}
\label{t:app_lamrfailsr}
If $\Gamma \wfdash \expr{M}:\tau$ and $\expr{M} \red \expr{M}'$ then $\Gamma \wfdash \expr{M}' :\tau$.
\end{restatable}

\begin{proof}[Proof (Sketch)]
By structural induction on the reduction rules. 	
See \appref{app:lamfailintertypes} for details.
\end{proof}

\revo{}{Differently from $\lamr$,  subject expansion fails for $\lamrfail$. This is due to the possibility of failure in the use of resources. In $\lamr$, if a resource is substituted within a term it is always done once, hence the term substituted must always be well-typed; however, in reductions that lead to the failure term, resources within a bag may be discarded before ever being substituted and hence, there is  no requirement to be well-formed. Formally, we have:}

\revo{}{
\begin{restatable}[Failure of Subject Expansion in \lamrfail]{thms}{applamrfailexp}
\label{t:app_lamrfailse}
If $\Gamma \wfdash \expr{M}':\tau$ and $\expr{M} \red \expr{M}'$ then \revdaniele{it is not necessarily the case that} $ \Gamma \wfdash \expr{M} :\tau$.
\end{restatable}
}

\revo{}{
\begin{proof}
A counter-example suffices here. Consider the term $  \fail^\emptyset$, which is well-formed but not well-typed, and let $ \Omega^l$  be the term $( \lambda x. x \bag{x} ) \bag{ \lambda x. x \bag{x} } $.
 \revdaniele{Notice that $\dash \wfdash \fail^\emptyset:\tau$ and 
 $ \fail^x \esubst{\bag{\Omega^l}}{x}\red\fail^\emptyset$, but
 $\fail^x \esubst{\bag{\Omega^l}}{x}$ is not well-formed (nor well-typed)}
 .
 %We will show is that there exists a term, this being $fail^\emptyset$ and $ \fail^x \esubst{\bag{\Omega^l}}{x} $ that is neither well-formed or well-typed and reduces to the term $ fail^\emptyset $. Here we take $\Omega^l$ to be the term $( \lambda x. x \bag{x} ) \bag{ \lambda x. x \bag{x} } $, however any untypeable term is fitting.
\end{proof}
}

%%%%
\section[A Resource Calculus With Sharing]{\texorpdfstring{$\lamrsharfail$}{λ̂^↯_⊕}: A Resource Calculus With Sharing}\label{sec:lamsharfail}

We define $\lamrsharfail$, a variant of $\lamrfail$ with a sharing construct, \srev{which we adopt following the {atomic} $\lambda$-calculus in~\cite{DBLP:conf/lics/GundersenHP13}.}
In $\lamrsharfail$, a variable is only allowed to appear once in a term:  
{multiple occurrences of the same variable are  atomized, i.e., they are given new different variable names. 
\srev{The ``atomization'' of variable occurrences realized in $\lamrsharfail$ via sharing will turn out to be very convenient to define our encoding into \spi.}

 Our language $\lamrsharfail$, defined in \secref{ss:syntaxshar},  includes also a form of explicit substitution, called \emph{explicit linear substitution}, which enables a refined analysis of the consumption of linear resources. Later, in \secref{ssec:lamshar_semantics}, we introduce the reduction semantics that implements a lazy evaluation.
 %, and operates modulo a pre-congruence
 In \secref{ss:typeshar}, we present a non-idempotent intersection type system to control the use of resources.
 Finally, in \secref{ss:auxtrans} we give an encoding from $\lamrfail$ into $\lamrsharfail$, denoted
  $\recencodopenf{\cdot}$, whose correctness is established in \secref{s:encoding}.
 
%These choices make \lamrsharfail a convenient and \revdaniele{well-behaved} intermediate language in our \revdaniele{translation from}  \lamrfail into \spi, given in \secref{s:encoding}.

% This approach \revdaniele{can be seen as a non-deterministic counterpart of the works developed} by Gundersen et al.~\cite{DBLP:conf/lics/GundersenHP13} with respect to atomicity, \revdaniele{ and Ghilezan et al.~\cite{GhilezanILL11} and Kesner and Lengrand \cite{DBLP:journals/iandc/KesnerL07}, as a way for accounting for the controlled use of resources}. \daniele{We need to mention also \cite{DBLP:journals/mscs/DoughertyL03} since they implement an intersection type system.}

\subsection{Syntax}
\label{ss:syntaxshar}
\hfill

The syntax of \lamrsharfail only modifies the syntax of \lamrfail-terms, which is defined by the grammar below; the syntax of bags $B$ and expressions $\expr{M}$ is  as in \defref{def:rsyntaxfail}.
\begin{align*}
\mbox{(Terms)} \quad  M,N, L ::= &
~~x 
\sep \lambda x . (M[ \widetilde{x} \leftarrow x ]) 
\sep (M\ B) 
\sep M \linexsub {N /x} 
\sep \fail^{\widetilde{x}}
\\
& \sep 
M [ \widetilde{x} \leftarrow x ] 
\sep (M[\widetilde{x} \leftarrow x])\esubst{ B }{ x } 
\end{align*}
Distinctive aspects are the \emph{sharing construct}  
$M [ \widetilde{x} \leftarrow x ]$ 
and the \emph{explicit linear substitution} 
$M \linexsub{ N /x}$.
The term  $M [ \widetilde{x} \leftarrow x ]$ defines the sharing of variables $\widetilde{x}$ occurring in $M$ using $x$. 
We shall refer to $x$ as \emph{sharing variable} and to $\widetilde{x}$ as \emph{shared variables}. 
Notice that $\widetilde{x}$ can be empty: $M[\leftarrow x]$ expresses that $x$ does not share any variables in $M$. \revdaniele{The sharing construct}  $M [ \widetilde{x} \leftarrow x ]$ \revdaniele{binds the variables in \(\widetilde{x}\);} the occurrence of $x_i$ can appear within the fail term $\fail^{\widetilde{y}}$, if $x_i \in \widetilde{y}$.
 In the explicit linear substitution $M \linexsub{ N /x}$ \revdaniele{binds $x$ in $M$}. 
 %\secondrev{In both $M\shar{\tilde{x}}{x}$ and $M\linexsub{N/x}$, term $M$ is called the {\em subject}.}
 As in $\lamrfail$, the term $\fail^{\widetilde{x}}$ explicitly accounts for failed attempts at substituting the variables $\widetilde{x}$, due to an excess or lack of resources.  \secondrev{A variable that is not explicitly sharing/shared is called {\em independent}}.

\begin{exa}\label{ex:lambshar_terms}
The following are examples of  $\lamrsharfail$-terms.
\begin{itemize}
\item (Shared identity) $\hat{{\bf I}}=\lambda x.x_1[x_1\leftarrow x]$ 
\item (Independent  variables) \secondrev{An independent variable $x$ applied to a 1-component bag (another independent variable): $x\bag{x_1}$}
\item $\hat{{\bf I}}$ applied to a 1-component bag: $ \hat{{\bf I}} \bag{y_1}[y_1\leftarrow y]$
    \item  $\hat{{\bf I}}$ applied to a 2-component bag: \(\hat{{\bf I}}(\bag{y_1,y_2}) \shar{y_1,y_2}{y}\)
    \item Shared vacuous abstraction: $(\lambda y. x_1\bag{x_2}\shar{}{y})\shar{x_1,x_2}{x}$
    \item $\hat{{\bf I}}$ applied to a bag containing an explicit substitution of a failure term that does not share the  variable $y$: \(\hat{{\bf I}} \bag{ \fail^{\emptyset}[ \leftarrow y] \esubst{\bag{N}}{y} }  \)
    \item An abstraction on $x$ of two shared occurrences of $x$: $\hat{D}=\lambda x. x_1 \bag{x_2}\shar{x_1,x_2}{x}$
    %\item (Shared Omega) $\hat{\Omega}=\hat{\Delta}\bag{\hat{\Delta}}$
\end{itemize}
\end{exa}

\srev{The syntax of terms is subject to some natural conditions on variable occurrences and on the structure of the sharing construct and the {explicit linear substitution}. We formalize these conditions as \emph{consistency}, defined as follows:}
%(see \defref{d:consistent})}.

%The \lamrsharfail-terms in the example above satisfy the following  {\em consistency} requirements:

\secondrev{
\begin{defi}[Consistent Terms, Bags, and Expressions]
\label{d:consistent}
	We say that the expression $\expr{M}$ is 
\emph{consistent} %, written $\consistent{\expr{M}}$, 
if each subterm $M_0$ of $\expr{M}$ satisfies the following conditions:
\begin{enumerate}
	\item If $M_0 = M [ \widetilde{x} \leftarrow x ]$ then: (i) 
	$\widetilde{x}$ contains pairwise distinct variables; 
	(ii)~every $x_i \in \widetilde{x}$ must occur exactly once in $M$; (iii) $x_i$ is not a sharing variable;
	(iv)~$M$ is consistent. 
	\item If $M_0 = M \linexsub{ N /x}$ then: (i) the variable $x$ must occur exactly once in $M$;
	% and it cannot occur free in $N$; 
	(ii) $x$ cannot be a sharing variable; 	(iii)~$M$ and $N$ are consistent; (iv)~$\lfv{M} \cap \lfv{N} = \emptyset$. % (v) $x$ cannot be in an explicit linear substitution occurring in $M$.
    \item Otherwise, for other forms of $M_0$, variables must occur exactly once, i.e.,:
        \begin{itemize}
            \item If $M_0 = \lambda x . (M[ \widetilde{x} \leftarrow x ])$ then: $x \not \in \lfv{M} $; $\widetilde{x}$ contains pairwise distinct variables; every $x_i \in \widetilde{x}$ must occur exactly once in $M$ and is not a sharing variable; $M$ is consistent.            
            \item If $M_0 = (M\ B)$ then $\lfv{M} \cap \lfv{B} = \emptyset$ and $M$ and $B$ are consistent.
            \item If $M_0 = \fail^{\widetilde{x}}$ then $\widetilde{x}$ contains pairwise distinct variables.
            \item If $M_0 =(M[\widetilde{x} \leftarrow x])\esubst{ B }{ x } $  then: $x \not \in \lfv{M} $; $\widetilde{x}$ contains pairwise distinct variables; every $x_i \in \widetilde{x}$ must occur exactly once in $M$ and is not a sharing variable;   $\lfv{M} \cap \lfv{B} = \emptyset$; and $M$ and $B$ are consistent.
        \end{itemize}
\end{enumerate}
Consistency extends to bags as follows. 
The bag $\oneb$ is always consistent. 
The bag $\bag{M}$ is consistent if $M$ is consistent.
The bag $A \cdot B$ is consistent  if (i)~$A$ and $B$ are consistent and (ii)~$\lfv{A} \cap \lfv{B} = \emptyset$.
\end{defi}  
}

\srev{We now discuss the consistency conditions for the sharing construct $M [ \widetilde{x} \leftarrow x ]$. Condition~1(ii) enforces that variables cannot have more than one linear occurrence in the subject of a sharing construct: this condition rules out terms such as 
    $x_1 \bag{x_1\bag{y}}[x_1 \leftarrow x]$. Condition~1(iii), which rules out terms of the form $x_1 \bag{x_2\bag{x_3\bag{y}}}[x_1 , x_2 \leftarrow x'] [x' , x_3 \leftarrow x]$,  is for convenience: by requiring that sharing occurrences appear at the top level in bindings, we can easily deduce the number of occurrences of a variable by measuring the size of $\widetilde{x}$ in $[\widetilde{x} \leftarrow x]$, rather than inductively having to measure the occurrences of each $x' \in \widetilde{x}$ in multiple sharing constructs.}
    
    \srev{Conditions on the explicit linear substitution  $M\linexsub{N/x}$ formalize our design choice:
    an explicit linear substitution is defined when the number of 
    %neither too few nor too many 
    variables to be substituted coincides with the number of available resources. 
    In particular, Condition~2(i) rules out 
    %ensures we do not produce
    terms of the form $y\linexsub{M/x}$, where an explicit linear substitution has no variable to perform a substitution. Condition~2(ii) rules  out terms such as $M[x_1,x_2 \leftarrow x]\linexsub{M/x}$, in which a term is to be linearly substituted for a single variable $x$; however, as the variable is shared twice within $M$, there are less available terms to be substituted than it is necessary. }
 
   \secondrev{ Finally, Condition 3 enforces that each variable occurs only once  in a consistent term, and also that in 
 $\fail^{\widetilde{x}}$, the $\widetilde{x}$  denotes a  set of variables (rather than a multiset), as variables can appear at most once within consistent terms. Thus, consistent terms also excludes terms such as $\fail^{x,x}$}.
   
 \secondrev{In what follows, we shall be working with consistent terms only, which we will call simply terms in 
our definitions and results. As we will see, consistency will be preserved by 
reduction (\thmref{thm:consistency_reductions})
and ensured by typing  (\thmref{thm:consistency_type}) and a structural congruence on terms~(\thmref{thm:term_consistency}).}

\subsection{Reduction Semantics}\label{ssec:lamshar_semantics}\hfill

Similarly to $\lamrfail$, the reduction semantics of \lamrsharfail is given by a relation $\red$, defined by the rules in \figref{fig:share-reductfailure}; it consists of an extension of reductions in $\lamrfail$ that deals with the sharing construct $\shar{\cdot}{\cdot }$ and  with the explicit linear substitution $\cdot \linexsub{\cdot /\cdot }$. In order to define the reduction rules formally,  we require some auxiliary notions:
the free variables of an expression/term, 
the head of a term, linear head substitution, and contexts.

\begin{defi}[Free Variables]
\label{d:fvsh}
The set of free variables of a term, bag and expressions in \lamrsharfail, is defined inductively as

       \[
       \begin{array}{l}
       \begin{array}{l@{\hspace{2cm}}l}
          \lfv{x} = \{ x \}  &  \lfv{ \fail^{\widetilde{x}}} =\{ \widetilde{x}\} \\
          \lfv{ \bag{M}} = \lfv{M} & \lfv{B_1 \cdot B_2} = \lfv{B_1} \cup \lfv{B_2} \\
          \lfv{M\ B} = \lfv{M}  \cup \lfv{B} & \lfv{\oneb} = \emptyset \\
          \lfv{M \linexsub {N /x}} = (\lfv{M}\setminus \{x\})\cup \lfv{N} &
          \lfv{M [ \widetilde{x} \leftarrow x ]} = (\lfv{M}  \setminus \{\widetilde{x}\}) \cup \{ x \}\\
          \lfv{\lambda x . (M[ \widetilde{x} \leftarrow x ])} = \lfv{M[ \widetilde{x} \leftarrow x ]}\setminus \{x\}&  \lfv{\expr{M}+\expr{N}} = \lfv{\expr{M}} \cup \lfv{\expr{N}}
          \end{array}\\
          \ \lfv{ (M[\widetilde{x} \leftarrow x])\esubst{ B }{ x }} = (\lfv{M[\widetilde{x} \leftarrow x]}\setminus \{x\})\cup \lfv{B}
    \end{array}
    \]
    
 As usual, a term $M$ is {\em closed} if $\lfv{M}=\emptyset$.    
\end{defi}

\begin{defi}[Head]
\label{d:headshar}
The head of a term $M$, denoted $\headf{M}$, is defined inductively:
\[
\begin{array}{l}
\begin{array}{l@{\hspace{3cm}}l}
\headf{x}  = x   &  \headf{\lambda x . (M[ \widetilde{x} \leftarrow x ])}  = \lambda x . (M[ \widetilde{x} \leftarrow x ])
\\
\headf{M\ B}  = \headf{M} & \headf{M \linexsub{N /x}}  = \headf{M}
\\
\headf{\fail^{\widetilde{x}}}  = \fail^{\widetilde{x}}&
\end{array}
\\
\ \headf{M[\widetilde{x} \leftarrow x]} = 
\begin{cases}
    x & \text{If $\headf{M} = y \text{ and } y \in \widetilde{x}$}\\
    \headf{M} & \text{Otherwise}
\end{cases}
\\
\ \revd{B29}{ \headf{(M[\widetilde{x} \leftarrow x])\esubst{ B }{ x }} = 
\begin{cases}
    \fail^{\emptyset} & \text{If $ | \widetilde{x} | \not = \size{B}$}\\
    \headf{M[ \leftarrow x]} & \text{If $ \widetilde{x} = \emptyset$ and $B = \oneb$}\\
    (M[\widetilde{x} \leftarrow x])\esubst{ B }{ x } & 
    \text{Otherwise} 
\end{cases} }
\end{array}
\]
\end{defi}

%\joe{ADD A SHORT EXPLANATION OF HEAD, PERHAPS CONTRASTING IT WITH THE DEFINITION OF HEAD FOR THE SOURCE LANGUAGE (NO SHARING)}

\revd{B29}{The most notable difference between $\headf{\cdot}$ in $\lamrfail$ (cf. Definition \ref{def:headfailure}) and in $\lamrsharfail$ concerns explicit substitution. Both definitions return $\fail^\emptyset$ in a mismatch of resources; in $\lamrsharfail$, the head term of an explicit substitution is only defined in the case of empty sharing (weakening). As we will see, this allows us to prioritize explicit substitution reductions over fetch reductions, as the head variable will block until an explicit substitution is separated into its linear component. %Consider $M[x_1,x_2 leftarrow x]\esubst{\bag{N_1} \cdot \bag{N_2}}{x}$ in $\lamrsharfail, here the head of this term is not a variable and hence cannot be substituted, the term must first be reduced further in order to free the head of this term to reveal the head of $M$. 
}

\begin{defi}[Linear Head Substitution]\label{def:headlinfail}
Given a term $M$ with $\headf{M} = x$, the linear substitution of a term $N$ for   $x$ in $M$, written $M\headlin{ N / x}$ is inductively defined as:
\begin{align*}
x \headlin{ N / x}   & = N 
\\
(M\ B)\headlin{ N/x }  & = (M \headlin{ N/x })\ B  
\\
(M \linexsub{L /y} ) \headlin{ N/x } &= (M\headlin{ N/x })\ \linexsub{L /y}  & x \not = y\\
((M[\widetilde{y} \leftarrow y])\esubst{ B }{ y })\headlin{ N/x } &= (M[\widetilde{y} \leftarrow y]\headlin{ N/x })\ \esubst{ B }{ y }  
& x \not = y \\
(M[\widetilde{y} \leftarrow y]) \headlin{ N/x } &=  (M\headlin{ N/x }) [\widetilde{y} \leftarrow y] & x \not = y
\end{align*}

\end{defi}

We now define contexts for terms and expressions in  \lamrsharfail.  Term contexts involve an explicit linear substitution, rather than an explicit substitution: this is due to the reduction strategy we have chosen to adopt (cf. Rule~\redlab{RS:Ex\dash Sub} in \figref{fig:share-reductfailure}), as we always wish to evaluate explicit substitutions first. Expression contexts can be seen as sums with holes. 

\begin{defi}[Term and Expression Contexts in \lamrsharfail]\label{def:ctxt_lamsharfail}
Let $[\cdot]$ denote a hole.
Contexts for terms and expressions are defined by the following grammar:
\[
\begin{array}{l@{\hspace{.7cm}}rl}
\text{(Term Contexts)}     & C[\cdot] ,  C'[\cdot] &::=([\cdot])B \mid ([\cdot])\linexsub{N/x} \mid ([\cdot])[\widetilde{x} \leftarrow x] \mid ([\cdot])[ \leftarrow x]\esubst{\oneb}{ x} \\
 \text{(Expression Contexts)}& D[\cdot] , D'[\cdot] & ::= M + [\cdot] \mid [\cdot] + M
\end{array}
\]
The substitution of a hole with a %\secondrev{consistent} 
term $M$ in a context $C[\cdot]$, denoted  $C[M]$, must be a %\secondrev{consistent} 
\lamrsharfail-term. 
\end{defi}

We assume that the terms that fill in the holes respect 
\secondrev{consistency}
%the conditions (iii)--(iv) 
(i.e., variables appear in a term only once, shared variables must occur in the context).

\begin{exa}
    This example illustrates that certain contexts cannot be filled with certain terms. Consider the hole in context $C[\cdot ]= ([\cdot])\linexsub{N/x}$. 
    \begin{itemize} 
     \item  The hole cannot be filled with $y$, since  $C[y]= y\linexsub{N/x}$ is not a consistent term.
    Indeed, $M\linexsub{N/x}$ requires that $x$ occurs exactly once within $M$.
    \item Similarly, the hole cannot be filled with $\fail^{z}$ with $z\neq x$, since $C[\fail^{z}]= (\fail^{z})\linexsub{N/x}$ and $x$ does not occur in the $\fail^z$, thus, the result is  not a consistent term.
    \end{itemize}
\end{exa}

\begin{figure}[!t]
%\hrule
\centering
  \begin{prooftree}
    \AxiomC{$\raisebox{17.0pt}{}$}
    \LeftLabel{\redlab{RS{:}Beta}}
    \UnaryInfC{\(  (\lambda x. M[\widetilde{x} \leftarrow x]) B  \red M[\widetilde{x} \leftarrow x]\esubst{ B }{ x }  \)}
 \end{prooftree}

 \begin{prooftree}
    \AxiomC{$B = \bag{M_1}
    \cdots  \bag{M_k} \qquad k \geq  1 $}
    \AxiomC{$ M \not= \fail^{\widetilde{y}} $}
    \LeftLabel{\redlab{RS{:}Ex \dash Sub}}
    \BinaryInfC{\( \!M[x_1,\ldots, x_k \leftarrow x]\esubst{ B }{ x } \red \sum_{B_i \in \perm{B}}M\linexsub{B_i(1)/x_1} \cdots \linexsub{B_i(k)/x_k}    \)}
 \end{prooftree}

 \begin{prooftree}
    \AxiomC{$ \headf{M} = x$}
     \LeftLabel{\redlab{RS{:}Lin\dash Fetch}}
     \UnaryInfC{\(  M \linexsub{N/x} \red  M \headlin{ N/x } \)}
\end{prooftree}
     
\begin{prooftree}
       \AxiomC{$ k \neq \size{B} \quad   \widetilde{y} = (\lfv{M} \setminus \{ x_1,\ldots, x_k \} ) \cup \lfv{B}$}
   \LeftLabel{\redlab{RS{:}Fail}}
   \UnaryInfC{\( M[x_1,\ldots, x_k\leftarrow x]\esubst{ B }{ x } \red \sum_{\perm{B}}  \fail^{\widetilde{y}} \)}
 \end{prooftree}

\begin{prooftree}
    \AxiomC{\( \widetilde{y} = \lfv{B} \)}
    \LeftLabel{$\redlab{RS{:}Cons_1}$}
    \UnaryInfC{\(  \fail^{\widetilde{x}} B  \red {} \hspace{-4mm} \displaystyle\sum_{\perm{B}} \hspace{-3mm}\fail^{\widetilde{x} \cup \widetilde{y}}  \)}
        \DisplayProof\hspace{-3mm}
      \AxiomC{\(  \size{B} = k \quad k  +  | \widetilde{x} | \not= 0  \quad  \widetilde{z} = \lfv{B}\)}
    \LeftLabel{$\redlab{RS{:}Cons_2}$}
    \UnaryInfC{\(  (\fail^{\widetilde{x}\cup \widetilde{y}} [ \widetilde{x} \leftarrow x])\esubst{ B }{ x }  {} \red\hspace{-3mm} \displaystyle \sum_{\perm{B}}\hspace{-3mm}\fail^{\widetilde{y}\cup \widetilde{z}} \)}
\end{prooftree}

\begin{prooftree}
    \AxiomC{\( \widetilde{z} = \lfv{N} \)}
    \LeftLabel{$\redlab{RS{:}Cons_3}$}
    \UnaryInfC{\( \fail^{\widetilde{y}\cup x} \linexsub{N/x}  {} \red  \fail^{\widetilde{y} \cup \widetilde{z}}  \)}
\end{prooftree}

\begin{prooftree}
        \AxiomC{$   M \red M'_{1} + \cdots + M'_{k} $}
        \LeftLabel{$\redlab{RS{:}TCont}$}
        \UnaryInfC{$ C[M] \red  C[M'_{1}] + \cdots +  C[M'_{k}] $}
\DisplayProof\hfill%
        \AxiomC{$ \expr{M}  \red \expr{M}'  $}
        \LeftLabel{$\redlab{RS{:}ECont}$}
        \UnaryInfC{$D[\expr{M}]  \red D[\expr{M}']  $}
\end{prooftree}

    \caption{Reduction Rules for \lamrsharfail.}
\label{fig:share-reductfailure}
\end{figure}

Now we are ready to describe the rules in ~\figref{fig:share-reductfailure}. Intuitively, the lazy reduction relation $\red$ on expressions works as follows: 
a $\beta$-reduction in \lamrsharfail results into an explicit substitution $M\shar{\widetilde{x}}{x}\esubst{ B }{ x }$, which then evolves, as an in intermediate step, to an expression consisting of explicit linear substitutions,  which are the ones reducing to a linear head substitution $\headlin{ N / x}$ (with $N \in B$) when the size of $B$ coincides with the number of occurrences of  $x$ in $M$. 
% Reduction in \lamrsharfail introduces an intermediate step whereby the explicit substitution expands into a sum of terms involving explicit linear substitutions $\linexsub {N /x}$, which are the ones to reduce into a linear head substitution. 
The term reduces to failure when there is a mismatch between the size of $B$ and the number of shared variables to be substituted.
More in details, we have:
\begin{itemize}
    \item {\bf Rule~\redlab{RS{:}Beta}} is standard and reduces to an explicit substitution.
    \item {\bf Rule~\redlab{RS{:}Ex \dash Sub}} applies when the size $k$ of the bag coincides with the length of the list  $\widetilde{x}=x_1,\ldots,x_k$. Intuitively, this rule ``distributes'' an  explicit substitution into a sum of terms involving explicit linear substitutions; it considers all possible permutations of the elements in the bag among all shared variables.
   \item  {\bf Rule~\redlab{RS{:}Lin \dash Fetch}} specifies the evaluation of a term with an explicit linear substitution into a linear head substitution.
\end{itemize}

We have three rules that reduce to the failure term---their objective is to accumulate all (free) variables involved in failed reductions. 
Accordingly:
\begin{itemize} 
\item {\bf Rule~$\redlab{RS{:}Fail}$} formalizes failure in the evaluation of an explicit substitution $M[\widetilde{x}\leftarrow x]\esubst{ B}{x }$, which occurs if there is a mismatch between the resources (terms) present in $B$ and the number of occurrences of $x$ to be substituted. 
The resulting failure term preserves all free variables in $M$ and $B$ within its attached set $\widetilde{y}$.
\item {\bf Rules~$\redlab{RS{:}Cons_1}$ and~$\redlab{RS{:}Cons_2}$}  describe reductions that lazily consume the failure term, when a term has $\fail^{\widetilde{x}}$ at its head position. 
The former rule consumes bags attached to it whilst preserving all its free variables. 
\item  {\bf  Rule~\redlab{RS{:}Cons_3}} accumulates into the failure term the free variables involved in an explicit linear substitution. 
\end{itemize}
The contextual Rules~$\redlab{RS{:}TCont}$ and $\redlab{RS{:}Econt}$ are standard. 

\begin{exa}
We show how a term can reduce using Rule~$\redlab{RS{:}Cons_2}$.

\[
      \begin{aligned}
            (\lambda x . x_1 [x_1 \leftarrow x]) \bag{ \fail^{\emptyset}[ \leftarrow y] \esubst{\bag{N}}{y} }  &\red_{\redlab{RS{:}Beta}} x_1 [x_1 \leftarrow x]  \esubst{\bag{ \fail^{\emptyset}[ \leftarrow y] \esubst{\bag{N}}{y} }}{x} \\
            & \red_{\redlab{RS{:}Ex \dash Sub}} x_1  \linexsub{ \fail^{\emptyset}[ \leftarrow y] \esubst{\bag{N}}{y}  / x_1} \\
             &\red_{\redlab{RS{:}Lin \dash Fetch}} \fail^{\emptyset}[ \leftarrow y] \esubst{\bag{N}}{y}  \\ &\red_{\redlab{RS{:}Cons_2}} \fail^{\lfv{N}}
        \end{aligned}
 \]

 \end{exa}

   \begin{exa}
We illustrate how Rule~$\redlab{RS{:}Fail} $ can introduce $\fail^{\widetilde{x}}$ into a term. It also shows how Rule~$\redlab{RS{:}Cons_3} $ consumes an explicit linear substitution:
 \[
      \begin{aligned}
            x_1 [\leftarrow y] \esubst{\bag{N}}{y}[x_1 \leftarrow x] \esubst{\bag{M}}{x}  &\red_{\redlab{RS{:}Ex \dash Sub}} x_1 [\leftarrow y] \esubst{\bag{N}}{y}\linexsub{M/x_1}  \\
            &\red_{\redlab{RS{:}Fail}} \fail^{ \{ x_1 \} \cup \lfv{N}  }\linexsub{M/x_1} \\ &\red_{\redlab{RS{:}Cons_3}} \fail^{\lfv{M} \cup \lfv{N} }
        \end{aligned}
  \]
 \end{exa}

Similarly to \lamrfail, reduction in \lamrsharfail satisfies a \emph{diamond property}. Therefore, we have the analogue of Proposition~\ref{prop:conf1_lamrfail}: 

% see~\appref{sec:diamond}. \daniele{Add reference to appendix later.}

\begin{prop}[Diamond Property for \lamrsharfail]
\label{prop:conf1_lamrsharfail}
     For all $\expr{N}$, $\expr{N}_1$, $\expr{N}_2$ in $\lamrsharfail$ s.t. $\expr{N} \red \expr{N}_1$, $\expr{N} \red \expr{N}_2$ { with } $\expr{N}_1 \neq \expr{N}_2$  then there exists   $\expr{M}$ such that $\expr{N}_1 \red \expr{M}$ and $\expr{N}_2 \red \expr{M}$.
\end{prop}

\begin{proof}
The thesis follows as in $\lamrfail$ since  the left-hand sides of the reduction rules in $\lamrsharfail$ do not interfere with each other. 
\end{proof}

%   \daniele{ToDO: Check the behavior of the term $((\lambda x. \lambda y .x) u )z$ in our sharing language to compare with the behavior presented in \cite{DBLP:journals/iandc/KesnerL07}}
% \daniele{
% \begin{exa}
% Consider the term \(M=\lambda x. (\lambda y. x_1\shar{}{y})\shar{x_1}{x}\). Notice that 
% \[
% \begin{aligned}
% (M\bag{u})&\red_\redlab{RS:Beta}  (\lambda y. x_1\shar{}{y})\shar{x_1}{x}\esubst{\bag{u}}{x}\\
% &\red_\redlab{RS:Ex\dash sub}  (\lambda y. x_1\shar{}{y})\linexsub{u/x_1}
% \end{aligned}
% \]
% By Rule~\redlab{RS:TCont} it follows that 
% \[
% \begin{aligned}
% (M\bag{u})\bag{z}&\red(\lambda y. x_1\shar{}{y})\linexsub{u/x_1}\bag{z}\\
% &\pequiv (\lambda y. x_1\shar{}{y})\bag{z}\linexsub{u/x_1}, \text{ since } x_1\notin \lfv{z}\\
% &\red_\redlab{RS:Beta} (x_1\shar{}{y})\esubst{\bag{z}}{y}\linexsub{u/x_1}\\
% \end{aligned}\]
% From the last term there are two possible behaviors (modulo $\pequiv$) which converge to the same result:
% \begin{enumerate}
%     \item \((x_1\shar{}{y})\esubst{\bag{z}}{y}\linexsub{u/x_1}\red_\redlab{RS:Fail} \fail^{z,x_1}\linexsub{u/x_1}\red_\redlab{RS:Cons_3} \fail^{z,u}\)
%     \item \((x_1\shar{}{y})\esubst{\bag{z}}{y}\linexsub{u/x_1}\pequiv (x_1\shar{}{y})\linexsub{u/x_1}\esubst{\bag{z}}{y}\red_\redlab{RS:Lin\dash Fetch} (u\shar{}{y})\esubst{\bag{z}}{y}\red_\redlab{RS:Fail} \fail^{z,u}\)
% \end{enumerate}
% \end{exa}
% }

\begin{rem}[A Calculus with Sharing but Without Failure (\lamrshar)]
\label{r:lamrshar}
As we did in Remark~\ref{rem:lamr}, we define a sub-calculus of \lamrsharfail  in which failure is not explicit. The calculus   \lamrshar is obtained from the syntax of \lamrsharfail by disallowing the term $\fail^{\widetilde{x}}$. The relevant reduction rules from \figref{fig:share-reductfailure} are \redlab{RS:Beta}, \redlab{RS{:}Ex \dash Sub}, \redlab{RS{:}Lin\dash Fetch}, and the two contextual rules. We keep \defref{d:headshar} unchanged with the provision that $\headf{M \esubst{ B }{x}}$ is undefined when $| \widetilde{x} | \not = \size{B}$. 
\end{rem}

\subsection{Non-Idempotent Intersection Types}
\label{ss:typeshar}\hfill

Similarly to $\lamrfail$, we now define \emph{well-formed} \lamrsharfail expressions and a system of rules for checking {well-formedness} by modifying the rules in~\figref{fig:app_wf_rules}. The grammar of strict and multiset types,  the notions of typing assignments, 
%typing contexts,
and judgements are  the same as in Section~\ref{sec:lamfailintertypes}.
\secondrev{We need an extension to the notion of typing context: whereas in \lamrfail variables were only assigned to multiset types, now sharing variables are assigned to multiset types, shared and independent variables  are assigned to strict types.}
\begin{defi}
\label{d:tcont}
\secondrev{
We extend the definition of typing contexts (\defref{d:tcontsource}) as follows:
\[
    \begin{aligned}
        \Gamma, \Delta & = \dash \sep \Gamma, x:\pi \sep \Gamma , x:\sigma
    \end{aligned}
\]
%The extension involves type assignments $x_i:\sigma$, where $\sigma$ is a strict type and $x_i$ is a shared variable.
 The definition of core contexts is extended accordingly, and also denoted as $\core{\Gamma}$.}
%$\pi$ is a strict type.
% \secondrev{
% \[
%     \begin{aligned}
%         \Gamma, \Delta & = \dash \sep \Gamma, x:\pi \sep \Gamma , x:\sigma
%     \end{aligned}
% \]
% The extension involves type assignments $x_i:\sigma$, where $\sigma$ is a strict type. The definition of core contexts is extended accordingly, and also denoted as $\core{\Gamma}$.
%}
\end{defi}

The presentation is in two phases:
\begin{enumerate}
\item We consider the intersection type system given in \figref{fig:typing_sharing} for which we consider the sub-calculus \lamrshar, the sharing calculus excluding failure (cf. Rem.~\ref{r:lamrshar}).
\item We define well-formed expressions for the full language \lamrsharfail, via \defref{def:wf_sharlam} (see below).
\end{enumerate}

\secondrev{To avoid ambiguities, we write $x:\sigma^1$  to make it explicit that the type assignment involves an intersection type (and a sharing variable), rather than a strict type.}
% \begin{exa}

%     Consider the context $\Gamma = x: \sigma , y : \tau^1 , z: \sigma \wedge \sigma$. Then the variable $x$ is typed with the strict type $\sigma$ , variable $y$ is typed with the multiset type $\tau$ of size one. Finally the variable $z$ is typed with a multiset type of size two consisting of $\sigma$'s. Notice that $\sigma \wedge \sigma$ can also be represented as $\sigma^2$
    
% \end{exa}

\subsubsection{Well-typed Expressions (in \texorpdfstring{$\lamrshar$}{})}\hfill 

 The typing rules in \figref{fig:typing_sharing} are essentially the same as the ones in \figref{fig:app_typingrepeat},  but now taking into account the sharing construct $M [\widetilde{x}\leftarrow x ]$ and the explicit linear substitution. We discuss selected rules:

\begin{itemize}
\item {\bf Rules~\redlab{TS{:}var}, \redlab{TS{:}\oneb}, \redlab{TS{:}bag}, \redlab{TS{:}app}, and \redlab{TS{:}sum}} are the same as in \figref{fig:app_typingrepeat}, considering sharing within the terms and bags.

\item {\bf Rule~\redlab{TS{:}weak}} deals with $k=0$, typing the term $M[\leftarrow x] $, when there are no occurrences of $x$ in $M$, as long as $M$ is typable.
%when $x$ does not share any variable in $M$.
\item {\bf Rule~\redlab{TS{:}abs\dash sh}} is as expected: it requires that the sharing variable is assigned the $k$-fold intersection type $\sigma^k$. % (Not.~\ref{not:lamrrepeat}).
\item {\bf Rule~\redlab{TS{:}ex \dash lin \dash sub}} supports explicit linear substitutions and consumes one occurrence of $x:\sigma$ from the context.
\item {\bf Rule~\redlab{TS{:}ex \dash sub}} types explicit substitutions where a bag must consist of both the same type and length of the shared variable it is being substituted for.
\item {\bf Rule~\redlab{TS{:}share}} requires that the shared variables $x_1,\ldots, x_k$ have the same type as the sharing variable $x$, for $k\neq 0$. \secondrev{This rule justifies the need for the extension of contexts with assignments of the form $x:\sigma$. This way, e.g.,   Example~\ref{ex:termwelltyped} below gives  an application of Rule~\redlab{TS:share} with $k=1$).}
\end{itemize}

\begin{defi}[Well-typed Expressions]
An expression $\expr{M} \in \lamrsharfail$ is {\em well-typed} (or typable) if there exist $\Gamma$ and  $\tau$ such that $\Gamma \vdash \expr{M} : \tau$ is entailed via the rules in \figref{fig:typing_sharing}.
\end{defi}

Again, the failure term $\fail$  in $\lamrsharfail$ is not typable via this typing system. The following examples illustrate the typing rules.

\begin{exa}%[Cont. Example~\ref{ex:terms}]
\label{ex:termwelltyped}
The term $ ((\lambda x. x_1 [x_1 \leftarrow x] ) \bag{y_1} ) [ y_1 \leftarrow y] $ is well-typed, as follows:

            \begin{prooftree}
                    \AxiomC{}
                    \LeftLabel{\redlab{TS{:}var}}
                    \UnaryInfC{\( x_1: \sigma \vdash  x_1 : \sigma\)}
                    \LeftLabel{ \redlab{TS{:}share}}
                    \UnaryInfC{\( x: \sigma^1 \vdash  x_1 [x_1 \leftarrow x] : \sigma\)}
                    \LeftLabel{\redlab{TS{:}abs\dash sh}}
                    \UnaryInfC{\(  \vdash \lambda x. x_1 [x_1 \leftarrow x] :  \sigma^1 \rightarrow \sigma \)}
                                            \AxiomC{}
                        \LeftLabel{\redlab{TS{:}var}}
                        \UnaryInfC{\(  y_1 : \sigma \vdash y_1 : \sigma\)}
                            \AxiomC{\(  \)}
                        \LeftLabel{\redlab{TS{:}\oneb}}
                        \UnaryInfC{\( \vdash \oneb : \omega \)}
                    \LeftLabel{\redlab{TS{:}bag}}
                    \BinaryInfC{\( y_1 : \sigma \vdash \bag{y_1}\cdot \oneb:\sigma^{1} \)}
                    \LeftLabel{\redlab{TS{:}app\dash sh}}
                \BinaryInfC{\( y_1 : \sigma \vdash ((\lambda x. x_1 [x_1 \leftarrow x] ) \bag{y_1} ) : \sigma\)}
                \LeftLabel{ \redlab{TS{:}share}}
                \UnaryInfC{\(  y: \sigma^1 \vdash ((\lambda x. x_1 [x_1 \leftarrow x] ) \bag{y_1} ) [ y_1 \leftarrow y] : \sigma \)}
            \end{prooftree}
       \end{exa}

\begin{figure}[t]
    \centering

\begin{prooftree}
    \AxiomC{}
    \LeftLabel{\redlab{TS{:}var}}
    \UnaryInfC{\( x: \sigma \vdash x : \sigma\)}
    \DisplayProof
    \hfill
    \AxiomC{\(  \)}
    \RightLabel{\(\)}
    \LeftLabel{\redlab{TS{:}\oneb}}
    \UnaryInfC{\( \vdash \oneb : \omega \)}
\DisplayProof
    \hfill
      \AxiomC{\( \Delta  \vdash M : \tau\)}
    \LeftLabel{ \redlab{TS{:}weak}}
    \UnaryInfC{\( \Delta , x: \omega \vdash M[\leftarrow x]: \tau \)}
\end{prooftree}

\begin{prooftree}
    \AxiomC{\( \Delta , x: \sigma^k \vdash M[\widetilde{x} \leftarrow x] : \tau \)}
    \LeftLabel{ \redlab{TS{:}abs \dash sh}}
    \UnaryInfC{\( \Delta \vdash \lambda x . (M[\widetilde{x} \leftarrow x]) : \sigma^k \rightarrow \tau \)}
    \DisplayProof
    \hfill
    \AxiomC{\( \Gamma \vdash M : \pi \rightarrow \tau \)}
    \AxiomC{\( \Delta \vdash B : \pi \)}
        \LeftLabel{\redlab{TS{:}app}}
    \BinaryInfC{\( \Gamma , \Delta \vdash M\ B : \tau\)}
    \end{prooftree}

\begin{prooftree}
     \AxiomC{\( \Gamma \vdash M : \sigma\)}
    \AxiomC{\( \Delta \vdash B : \sigma^{k}\)}
    \LeftLabel{\redlab{TS{:}bag}}
    \BinaryInfC{\( \Gamma , \Delta \vdash \bag{M}\cdot B:\sigma^{k+1} \)}
\DisplayProof\hfill
    \AxiomC{\( \Delta \vdash N : \sigma \qquad \Gamma  , x:\sigma \vdash M : \tau \)}
    \LeftLabel{\redlab{TS\!:\!ex\dash lin\dash sub}}
    \UnaryInfC{\( \Gamma , \Delta \vdash M \linexsub{N / x} : \tau \)}
\end{prooftree}

\begin{prooftree}
    \AxiomC{\( \Delta \vdash B : \sigma^k \qquad \Gamma , x:\sigma^k \vdash M [\widetilde{x} \leftarrow x]: \tau \)}
    \LeftLabel{\redlab{TS\!:ex \dash sub}}    
    \UnaryInfC{\( \Gamma , \Delta \vdash M[\widetilde{x} \leftarrow x] \esubst{ B }{ x } : \tau \)}
    \DisplayProof\hfill
     \AxiomC{$ \Gamma \vdash \expr{M} : \sigma$}
    \AxiomC{$ \Gamma \vdash \expr{N} : \sigma$}
    \LeftLabel{\redlab{TS{:}sum}}
    \BinaryInfC{$ \Gamma \vdash \expr{M}+\expr{N}: \sigma$}
\end{prooftree}

\begin{prooftree}
    \AxiomC{\( \Delta , x_1: \sigma, \cdots, x_k: \sigma \vdash M : \tau \quad x\notin \dom{\Delta} \quad k \not = 0\)}
    \LeftLabel{ \redlab{TS{:}share}}
    \UnaryInfC{\( \Delta , x: \sigma^k \vdash M[x_1 , \cdots , x_k \leftarrow x] : \tau \)}
    \end{prooftree}

    \caption{Typing Rules for \lamrshar.}
    \label{fig:typing_sharing}
\end{figure}

% \begin{lem}[Linear Substitution Lemma for \lamrshar]
% \label{lem:subt_lemshar}
% If $\Gamma , x:\sigma \vdash M: \tau$, $\headf{M} = x$, and $\Delta \vdash N : \sigma$ 
% then 
% $\Gamma , \Delta \vdash M \headlin{ N / x }$.
% \end{lem}

\secondrev{
\begin{restatable}[Consistency Stability Under $\red$]{thms}{Consistencyreductions}
\label{thm:consistency_reductions}
%\begin{prop}[Consistency by Typing and Reduction]\label{prop:term_consistency}
    %Consistency for $\lamrsharfail$-expressions (cf. \defref{d:consistent}) is preserved by reductions: 
    If $\expr{M}$ is a consistent $\lamrsharfail$-expression and $\expr{M} \red \expr{M}'$ then $\expr{M}'$ is consistent.
\end{restatable}
}

\begin{proof}
   \secondrev{By structural induction, and analyzing  the reduction rules applied in $\expr{M}$. See \Cref{app:typeshar} for details.} 
\end{proof}

As expected, the typing system satisfies the subject reduction property w.r.t. the reduction relation given in \figref{fig:share-reductfailure}, excluding rules for failure.
\begin{thm}[Subject Reduction in \lamrshar]
\label{t:app_lamrsrshar}
If $\Gamma \vdash \expr{M}:\tau$ and $\expr{M} \red \expr{M}'$ then $\Gamma \vdash \expr{M}' :\tau$.
\end{thm}

\begin{proof}
    Standard by induction on the rule applied in $\expr{M}$.
\end{proof}
%The more general property of well-formedness in \lamrsharfail will be proven later.

\revo{}{
\begin{restatable}[Linear Anti-substitution Lemma for \lamrshar]{lema}{explemfailnofailshar}
\label{lem:antisubt_lem_shar}
Let $M$ and $N$ be $\lamrshar$-terms such that $\headf{M} = x$. The following hold:
\begin{itemize}
    \item If $\Gamma, x:\sigma^{k-1} \vdash M \headlin{ N / x }: \tau$, with $k > 1$, then  there exist $  \Gamma_1, \Gamma_2$  such that  $\Gamma_1 , x:\sigma^k \vdash M: \tau$, and $\Gamma_2 \vdash N : \sigma$, where $\Gamma = \Gamma_1 \contexcat \Gamma_2$.
    \item If $\Gamma \vdash M \headlin{ N / x }: \tau$, with $x \not \in \dom{\Gamma}$, then there exist $ \Gamma_1, \Gamma_2$  such that  $\Gamma_1 , x:\sigma \vdash M: \tau$ , and $\Gamma_2 \vdash N : \sigma$, where $\Gamma = \Gamma_1 \contexcat \Gamma_2$.
\end{itemize}
\end{restatable}
}

\begin{proof}
%The proof is standard and follows by structural induction.
By structural induction on the reduction rule from Fig.~\ref{fig:typing_sharing}.
See \appref{app:typeshar} for details.
\end{proof}

\revo{}{
\begin{restatable}[Subject Expansion for \lamrshar]{thms}{subexponeshar}
\label{t:app_lamrexpshar}
If $\Gamma \vdash \expr{M}':\tau$ and $\expr{M} \red \expr{M}'$ then $\Gamma \vdash \expr{M} :\tau$.
\end{restatable}
}

\begin{proof}
Standard, by induction on the reduction rule applied. See \appref{app:typeshar} for details.
\end{proof}

\subsubsection{Well-formed Expressions (in \texorpdfstring{$\lamrsharfail$}{})}\hfill

On top of the intersection type system for \lamrshar, we  define well-formed expressions: {$\lamrsharfail$-terms whose computation may lead to failure}. 

\begin{defi}[Well-formedness in $\lamrsharfail$]\label{def:wf_sharlam}
An expression $ \expr{M}$ is well formed if there exist  $\Gamma$ and  $\tau$ such that $\Gamma \wfdash  \expr{M} : \tau $ is entailed via the rules in \figref{fig:wfsh_rules}.
\end{defi}

\begin{figure}[!t]
%\hrule
    \centering

    \begin{prooftree}
            \AxiomC{\( \Gamma \vdash \expr{M} : \tau \)}
            \LeftLabel{\redlab{FS\!:\!wf \dash expr}}
            \UnaryInfC{\( \Gamma \wfdash  \expr{M} : \tau \)}
            \DisplayProof\hfill
            \AxiomC{\( \Gamma \vdash B : \pi \)}
            \LeftLabel{\redlab{FS\!:\!wf \dash bag}}
            \UnaryInfC{\( \Gamma \wfdash  B : \pi \)}
            \DisplayProof\hfill
             \AxiomC{\( \Gamma  \wfdash M : \tau\)}
        \LeftLabel{ \redlab{FS\!:\!weak}}
        \UnaryInfC{\( \Gamma , x: \omega \wfdash M[\leftarrow x]: \tau \)}
    \end{prooftree}

\begin{prooftree}
        \AxiomC{\( \Gamma , x: \sigma^k \wfdash M[\widetilde{x} \leftarrow x] : \tau \quad x\notin \dom{\Gamma} \)}
        \LeftLabel{\redlab{FS{:}abs\dash sh}}
        \UnaryInfC{\( \Gamma \wfdash \lambda x . (M[\widetilde{x} \leftarrow x]) : \sigma^k  \rightarrow \tau \)}
         \DisplayProof
    \hfill
    \AxiomC{\( \secondrev{ \dom{\core{\Gamma}} = \widetilde{x} }\)}
        \LeftLabel{\redlab{FS{:}fail}}
        \UnaryInfC{\( \secondrev{ \Gamma \wfdash  \fail^{\widetilde{x}} : \tau } \)}
   \end{prooftree}

   \begin{prooftree}
        \AxiomC{\( \Gamma \wfdash M : \sigma^{j} \rightarrow \tau \quad \Delta \wfdash B : \sigma^{k} \)}
            \LeftLabel{\redlab{FS{:}app}}
        \UnaryInfC{\( \Gamma, \Delta \wfdash M\ B : \tau\)}
        \DisplayProof
\hfill
           \AxiomC{\( \Gamma \wfdash M : \sigma\quad \Delta \wfdash B : \sigma^{k} \)}
        \LeftLabel{\redlab{FS{:}bag}}
        \UnaryInfC{\( \Gamma, \Delta \wfdash \bag{M}\cdot B:\sigma^{k+1} \)}
    \end{prooftree}

\begin{prooftree}
\AxiomC{\( \Gamma  , x:\sigma \wfdash M : \tau \quad \Delta \wfdash N : \sigma \)}
\LeftLabel{\redlab{FS{:}ex \dash lin \dash sub}}
\UnaryInfC{\( \Gamma, \Delta \wfdash M \linexsub{N / x} : \tau \)}
 \DisplayProof
\hfill
\AxiomC{$ \Gamma \wfdash \expr{M} : \sigma \quad  \Gamma \wfdash \expr{N} : \sigma$}
\LeftLabel{\redlab{FS{:}sum}}
\UnaryInfC{$ \Gamma \wfdash \expr{M}+\expr{N}: \sigma$}
 \end{prooftree}

    \begin{prooftree}
            \AxiomC{\( \Gamma , x:\sigma^{k} \wfdash M[\widetilde{x} \leftarrow x] : \tau \)}
             \AxiomC{\( \Delta \wfdash B : \sigma^{j} \)}
        \LeftLabel{\redlab{FS{:}ex \dash sub}}    
        \BinaryInfC{\( \Gamma, \Delta \wfdash M [\widetilde{x} \leftarrow x]\esubst{ B }{ x } : \tau \)}
    \end{prooftree}

    \begin{prooftree}
        \AxiomC{\( \Gamma , x_1: \sigma, \cdots, x_k: \sigma \wfdash M : \tau \quad x\notin \dom{\Gamma} \quad k \not = 0\)}
        \LeftLabel{ \redlab{FS{:}share}}
        \UnaryInfC{\( \Gamma , x: \sigma^{k} \wfdash M[x_1 , \cdots , x_k \leftarrow x] : \tau \)}  
    \end{prooftree}

    \caption{Well-formedness Rules for \lamrsharfail.}\label{fig:wfsh_rules}
\end{figure}

Rules~$\redlab{FS{:}wf\dash expr}$ and $\redlab{FS{:}wf\dash bag}$ guarantee that every well-typed expression and bag, respectively, is well-formed. {Since our language is expressive enough to account for failing computations, we include rules for checking the structure of these ill-behaved terms---terms that can be well-formed, but not typable. For instance,} 

\begin{itemize} 
\item {\bf Rules~$\redlab{FS{:}ex \dash sub}$ and  $\redlab{FS{:}app}$} differ from similar typing rules in \figref{fig:typing_sharing}: the size of the bags (as declared in their types) is no longer required to match. 
\item {\bf Rule~$\redlab{FS{:}fail}$} has no analogue in the type system: we allow the failure term $\fail^{\widetilde{x}}$ to be well-formed with any type, provided that the core context contains types for the variables in $\widetilde{x}$.
\end{itemize}
The other rules are similar to their corresponding ones in~\figref{fig:app_wf_rules} and \figref{fig:typing_sharing}.

The following example illustrates a $\lamrsharfail$ expression that is well-formed but not well-typed.
\begin{exa}[Cont. Example~\ref{ex:shar-wf}]\label{ex:shar-wf}
The $\lamrsharfail$ expression consisting of an application of $\hat{I}$ to a bag containing a failure term 
 \(\lambda x . x_1 [x_1 \leftarrow x]) \bag{ \fail^{\emptyset}[ \leftarrow y] \esubst{ \oneb }{y} }\) is well-formed with type $\sigma$. The derivation, with omitted rule labels, is the following:
 
 \begin{prooftree}
 \AxiomC{\(x_1:\sigma \vdash x_1:\sigma\)}
 \UnaryInfC{\(x_1:\sigma \wfdash x_1:\sigma\)}
 \UnaryInfC{\(x:\sigma^1 \wfdash x_1\shar{x_1}{x}:\sigma\)}
  \UnaryInfC{\( \wfdash \lambda x. x_1\shar{x_1}{x}:\sigma\to \sigma\)}
    \AxiomC{}
    \UnaryInfC{\(\wfdash \fail^{\emptyset}:\sigma\)}
    \UnaryInfC{\(y:\omega\wfdash
     \fail^{\emptyset}\shar{}{y}:\sigma\)} 
    \AxiomC{}
    \UnaryInfC{\(\vdash \oneb:\omega\)}
    \UnaryInfC{\(\wfdash \oneb:\omega\)}
    \BinaryInfC{\(\wfdash
     \fail^{\emptyset}\shar{}{y}\esubst{1}{y}:\sigma\)}
     \AxiomC{}
    \UnaryInfC{\(\vdash \oneb:\omega\)}
     \UnaryInfC{\(\wfdash \oneb:\omega\)}
     \BinaryInfC{\(\wfdash
     \bag{\fail^{\emptyset}\shar{}{y}\esubst{1}{y}}:\sigma^1\)}
     \BinaryInfC{\( \wfdash \lambda x. x_1\shar{x_1}{x}\bag{\fail^{\emptyset}\shar{}{y}\esubst{1}{y}}:\sigma\)}
     \end{prooftree}
     
Besides, we have 
$\lambda x . x_1 [x_1 \leftarrow x]) \bag{ \fail^{\emptyset}[ \leftarrow y] \esubst{ \oneb }{y} }\red^* \fail^{\emptyset}[ \leftarrow y] \esubst{ \oneb }{y}$.
% \pequiv \fail^\emptyset
%and both $\fail^{\emptyset}[ \leftarrow y] \esubst{ \oneb }{y}$ and \(\fail^{\emptyset}\) have the same type $\sigma$.
\end{exa}

Well-formed $\lamrsharfail$ expressions satisfy the subject reduction property; as usual, the  proof relies on a linear substitution lemma for $\lamrsharfail$.

\begin{restatable}[Substitution Lemma for \lamrsharfail]{lema}{lamrsharfailsubs}
\label{l:lamrsharfailsubs}
If $\Gamma , x:\sigma \wfdash M: \tau$, $\headf{M} = x$, and $\Delta \wfdash N : \sigma$ 
then 
$\Gamma , \Delta \wfdash M \headlin{ N / x }:\tau$.
\end{restatable}
\begin{proof}
	By structural induction on $M$. See \appref{app:typeshar} for details.
\end{proof}

\begin{restatable}[Subject Reduction in \lamrsharfail]{thms}{applamrsharfailsr}
\label{t:app_lamrsharfailsr}
If $\Gamma \wfdash \expr{M}:\tau$ and $\expr{M} \red \expr{M}'$ then $\Gamma \wfdash \expr{M}' :\tau$.
\end{restatable}

\begin{proof}
By structural induction on the reduction rule from Fig.~\ref{fig:share-reductfailure}.
See \appref{app:typeshar} for details.
\end{proof}

We close this part by stating the failure of subject expansion for well-formed expressions.
 
\begin{restatable}[Failure of Subject Expansion in \lamrsharfail]{thms}{applamrfailexpshar}
\label{t:app_lamrfailsrshar}
If $\Gamma \wfdash \expr{M}':\tau$ and $\expr{M} \red \expr{M}'$ then it is not necessarily the case that $\Gamma \wfdash \expr{M} :\tau$.
\end{restatable}

\begin{proof}
We adapt the counter-example from the proof of \thmref{t:app_lamrfailse}.
 Consider the term $ \fail^{\emptyset}$, which is well-formed but not well-typed, and let $ \Omega^l$  be the term $( \lambda x. x_1 \bag{x_2}[x_1, x_2 \leftarrow x] ) \bag{ \lambda x. x_1 \bag{x_2}[x_1, x_2 \leftarrow x]} $.
 \revdaniele{Notice that  
 $ \fail^{x_1}[x_1 \leftarrow x] \esubst{\bag{\Omega^l}}{x}\red\fail^\emptyset$
 and $\dash \wfdash \fail^\emptyset:\tau$, 
  but
 $\fail^{x_1}[x_1 \leftarrow x] \esubst{\bag{\Omega^l}}{x}$ is not well-formed (nor well-typed).}
\end{proof}

\secondrev{
\begin{restatable}[Consistency enforced by typing]{thms}{consistencytype}
\label{thm:consistency_type}
%\begin{prop}[Consistency by Typing and Reduction]\label{prop:term_consistency}
    %Consistency for $\lamrsharfail$-expressions (cf. \defref{d:consistent}) is enforced by typing:
    Let $\expr{M}$ be a $\lamrsharfail$-expression. If $\Gamma \wfdash \expr{M}$ then ${\expr{M}}$ is consistent.
\end{restatable}
\begin{proof}
    By induction on the type derivation. See \Cref{app:typeshar} for details.
\end{proof}
}

\subsubsection*{Taking Stock}
Up to here, we have presented our source language \lamrfail---a new resource lambda calculus with failure---and its fail-free sub-calculus \lamr. Based on them we defined well-typed and well-formed expressions. Similarly, we defined the intermediate calculus \lamrsharfail and its sub-calculus \lamrshar. We now move on to define 
a translation of $\lamrfail$ into $\lamrsharfail$. 
%\spi, the target language of our encoding: a session-typed $\pi$-calculus with non-collapsing non-determinism and failures.

\subsection{From \texorpdfstring{\lamrfail}{λ^↯_⊕} into \texorpdfstring{\lamrsharfail}{λ̂^↯_⊕}}%\hfill % as the hfill command is commented out the inconsistency is probably intended
\label{ss:auxtrans}
\secondrev{Borrowing inspiration from translations given in~\cite{DBLP:conf/lics/GundersenHP13} for the atomic $\lambda$-calculus, we now define a translation  $\recencodopenf{\cdot}$ from well-formed expressions in $\lamrfail$ into $\lamrsharfail$.  It relies on an auxiliary translation $\recencodf{\cdot}$ on 
$\lamrfail$-terms, which depends on 
 the notion of (simultaneous) linear substitution (\defref{def:lin_subst}) which, intuitively,  forces all bound variables in \lamrfail to become shared variables in \lamrsharfail.
}
The correctness of $\recencodopenf{\cdot}$ will be addressed in \secref{ss:firststep}.

%We first introduce some notation. Given a term $M$ such that  $\#(x , M) = k$
%and a sequence of pairwise distinct fresh variables $\widetilde{x} = x_1, \ldots, x_k$ we write $M \linsub{\widetilde{x}}{x}$ or $M\linsub{x_1,\cdots, x_k}{x}$ to stand for 
%$M\linsub{x_1}{x}\cdots \linsub{x_k}{x}$. That is,  $M\linsub{\widetilde{x}}{x}$ denotes a simultaneous linear substitution whereby each distinct occurrence of $x$ in $M$ is replaced by a distinct $x_i \in \widetilde{x}$. 
%Notice that each $x_i$ has the same type as $x$.
%We use (simultaneous) linear substitutions to force all bound variables in \lamrfail to become shared variables in \lamrsharfail.

   % where $|\cdot |$ determines the number of elements in a  set and 

\begin{defi}[Linear substitution]\label{def:lin_subst}
    \secondrev{Suppose given a \lamrfail-term  $M$, a variable $x$, and a sequence of variables $\widetilde{w} = y,\widetilde{z}$.
    When \(\#(x , M) = |\widetilde{w} |\)  and  \(\{y\} \cap \widetilde{z}=\emptyset \), 
    the {\em linear substitution} $M\linsub{y,\widetilde{z}}{x}$ of  variable $x$ for variables $\widetilde{w}$ in $M$ is defined inductively as follows:}

    \secondrev{
    \[
        \begin{aligned}
            x\linsub{y}{x} & = 
                %\begin{cases}
                    y\\
                %    \text{undefined} & \text{otherwise }
                %\end{cases}\\
            (\lambda z . M )\linsub{y}{x} & =  
                %\begin{cases}
                    \lambda z . (M\linsub{y}{x}) \quad \text{if } x \in \lfv{M}\\
                %    \text{undefined } & \text{otherwise }
                %\end{cases}
                %\\
            (M\ B)\linsub{y}{x} & = 
                \begin{cases}
                    ((M\linsub{y}{x})\ B) &\text{if } x \in \lfv{M} \\ 
                    (M\ (B\linsub{y}{x})) &\text{if } x \not \in \lfv{M}, x \in \lfv{B}\\
                    %\text{undefined } & \text{otherwise }
                \end{cases}
                \\
            %M \linexsub {N /z}\linsub{y}{x} & = 
            %    \begin{cases}
            %        (M\linsub{y}{x}) \linexsub {N /z} &\text{if } x \in \lfv{M} \\ 
            %        M \linexsub {N\linsub{y}{x} /z} &\text{if } x \in \lfv{N}\\
                    %\text{undefined } & \text{otherwise }
            %    \end{cases}
            %    \\
            \fail^{\widetilde{z}}\linsub{y}{x} & = 
                %\begin{cases}
                    \fail^{\widetilde{z}',y} \quad \text{if } x \in \widetilde{z} \text{ and } \widetilde{z} = \widetilde{z}' , x\\
                %    \text{undefined } & \text{otherwise }
                %\end{cases}
                %\\
            %M [ \widetilde{z} \leftarrow z ]\linsub{y}{x} & = 
                %\begin{cases}
            %        M\linsub{y}{x} [ \widetilde{z} \leftarrow z ] \quad \text{if } x \not = z, x \in \lfv{M}\\
                %    \text{undefined } & \text{otherwise }
                %\end{cases}
                %\\
            (M\esubst{ B }{ z })\linsub{y}{x} & = 
                \begin{cases}
                    (M\linsub{y}{x})\esubst{ B }{ z } &\text{if } x \in \lfv{M} \\ 
                    M\esubst{ B\linsub{y}{x} }{ z } &\text{if }  x \not \in \lfv{M} ,x \in \lfv{B}\\
                    %\text{undefined } & \text{otherwise }
                \end{cases}
                \\
            \oneb\linsub{y}{x} & = \text{undefined }\\
            \bag{M}\linsub{y}{x} & = 
                %\begin{cases}
                    \bag{M\linsub{y}{x}} \quad \text{if } x \in \lfv{M} \\
                    %\text{undefined } & \text{otherwise }
                %\end{cases}\\
            (A \cdot B)\linsub{y}{x} & = 
                \begin{cases}
                    ((A\linsub{y}{x}) \cdot B) &\text{if } x \in \lfv{A} \\ 
                    (A \cdot (B\linsub{y}{x})) &\text{if } x \not \in \lfv{A} , x \in \lfv{B}
                    \\
                    %\text{undefined } & \text{otherwise }
                \end{cases}\\
             M\linsub{y,\widetilde{z}}{x} &= (M\linsub{y}{x})\linsub{\widetilde{z}}{x} 
        \end{aligned}
    \]
    }
    
    Otherwise, in all other cases, the substitution is undefined. 
    We write $M\linsub{z_1,z_2, \cdots, z_k}{x}$ to stand for 
    $( \cdots ((M\linsub{z_1}{x}) \linsub{z_2}{x})\cdots \linsub{z_k}{x})$.
    %provided $\#(x , M) = k$ and $x_1,x_2 \cdots, x_k$ is a non-empty sequence of pairwise distinct fresh variables, otherwise $M \linsub{\widetilde{x}}{x}$ is undefined.
    %That is,  $M\linsub{x',\widetilde{x}}{x}$ denotes a simultaneous linear substitution whereby each distinct occurrence of $x$ in $M$ is replaced by $x',\widetilde{x}$, inductively defined by $M\linsub{x',\widetilde{x}}{x} = (M\linsub{x'}{x})\linsub{\widetilde{x}}{x}$. 
    %\[
    %        M \linsub{x_1, \cdots , x_k }{x} = 
    %        \begin{cases}
    %            (\dots(M\linsub{x_1}{x}) \cdots )\linsub{x_k }{x} & \text{if } k>0 , \#(x , M) = k , i \not = j \implies x_i \not = x_j \forall i,j \in \{ 1, \cdots k\} 
    %            \\
    %            \text{undefined } & \text{otherwise }
    %        \end{cases}
    %\]
    %Notice that each $x_i$ has the same type as $x$.
    %
    %
    %
    %
    %\\
    %We first introduce some notation. Given a term $M$ such that  $\#(x , M) = k$
    %and a sequence of pairwise distinct fresh variables $\widetilde{x} = x_1, \ldots, x_k$ we write $M \linsub{\widetilde{x}}{x}$ or $M\linsub{x_1,\cdots, x_k}{x}$ to stand for 
    %$M\linsub{x_1}{x}\cdots \linsub{x_k}{x}$. .
\end{defi}

\secondrev{Notice that for a \lamrfail-term with multiple occurrences of the variable to be substituted for, this linear substitution fixes an ordering of instantiation. For example,  $\lambda x. y\bag{y,x}\linsub{z_1,z_2}{y}$  results in $\lambda x. z_1\bag{z_2,x}$, and a permutation of variables as in  $\lambda x. z_2\bag{z_1,x}$ is not accounted for. This is not restrictive; actually it is enough for our purposes since this substitution will only be used in \defref{def:enctolamrsharfail} and the variables being substituted will be bound by sharing, and therefore could be $\alpha$-renamed.}

\begin{defi}[From $\lamrfail$ to $\lamrsharfail$]
\label{def:enctolamrsharfail}
Let $M \in \lamrfail$.
Suppose $\Gamma \wfdash {M} : \tau$, with
$\dom{\Gamma} = \lfv{M}=\{x_1,\cdots,x_k\}$ and  $\#(x_i,M)=j_i$.  
We define $\recencodopenf{M}$ as
\[
 \recencodopenf{M} = 
\recencodf{M\linsub{\widetilde{y_{1}}}{x_1}\cdots \linsub{\widetilde{y_k}}{x_k}}[\widetilde{y_1}\leftarrow x_1]\cdots [\widetilde{y_k}\leftarrow x_k] 
\]
   where  $\widetilde{y_i}=y_{i_1},\cdots, y_{i_{j_i}}$
   and the translation $\recencodf{\cdot}: \lamrfail \to \lamrsharfail$  is defined in~\figref{fig:auxencfail}.
   %on closed $\lamrfail$-terms. 
   The translation $\recencodopenf{\cdot}$ extends homomorphically to expressions.
   \end{defi}

\srev{As already mentioned, } 
the translation $\recencodopenf{\cdot}$ ``atomizes'' occurrences of variables, \srev{in the spirit of~\cite{DBLP:conf/lics/GundersenHP13}}: it converts $n$ occurrences of a variable $x$ in a term into $n$ distinct variables $y_1, \ldots, y_n$.
The sharing construct coordinates the occurrences of these variables by constraining each to occur exactly once within a term. 
We proceed in two stages:
\begin{enumerate} %[label=$\triangleright$]
\item First, we use $\recencodopenf{\cdot}$ to ensure that each free variable (say, $y$) is replaced by a shared variable (say, $y_i \in \widetilde{y}$),  \revo{}{which is externally bound by the $y$ in $[\widetilde{y}\leftarrow y]$}. 
\item Second, we apply the auxiliary translation $\recencodf{\cdot}$ on the corresponding  \revo{A20}{to the sharing of bound variables.}%closed term. 
\end{enumerate}
We now describe the two cases of \figref{fig:auxencfail} that are noteworthy. 
\begin{itemize}
\item In $\recencodf{  \lambda x . M  }$, the occurrences of $x$ are replaced with fresh shared variables that only occur once  in $M$.
\item The definition of $\recencodf{  M \esubst{ B }{ x }  }$ considers two possibilities.
If the bag being translated is non-empty and the explicit substitution would not lead to failure (the number of occurrences of $x$ and the size of the bag coincide) then we translate the explicit substitution as a sum of explicit linear substitutions. 
Otherwise, the explicit substitution will lead to a failure, and the translation proceeds inductively. 
As we will see, doing this will enable a tight operational correspondence result with $\spi$.
\end{itemize}

\begin{figure*}[t]
\[
\begin{aligned}
\recencodf{ x  } & =  x  \hspace{3.2cm} \recencodf{  \oneb  } =  \oneb  \hspace{3cm}
\recencodf{  \fail^{\widetilde{x}} }  =  \fail^{\widetilde{x}}\\
\recencodf{  M\ B }  &=  \recencodf{M}\ \recencodf{B} \hspace{5.7cm}  \recencodf{\bag{M}\cdot B} =  \bag{\recencodf{M}}\cdot \recencodf{B} \\
  \recencodf{  \lambda x . M  }  &=   \lambda x . (\recencodf{M\langle \widetilde{y} / x  \rangle} [\widetilde{y} \leftarrow x]) 
    \quad  \text{ $\#(x,M) = n$, each $y_i\in\widetilde{y}$ is fresh} \\
   \recencodf{  M \esubst{ B }{ x }  } &= 
   \begin{cases}
   \displaystyle\hspace*{-1em}\sum_{\hspace*{1em}B_i \in \perm{\recencodf{ B }}}\hspace{-.6cm}\recencodf{ M \langle \widetilde{y}/ x  \rangle } \linexsub{B_i(1)/x_1} \cdots \linexsub{B_i(k)/x_k} & \hspace*{-0.5em} \#(x,M) = \size{B} = k \geq 1 \\
    \recencodf{M\langle y_1. \cdots , y_k / x  \rangle} [\widetilde{y} \leftarrow x] \esubst{ \recencodf{B} }{ x } & \hspace*{-0.5em} \text{otherwise, } \#(x,M) = k\geq 0
     \end{cases}
\end{aligned}
\]
    \caption{Auxiliary Translation: \lamrfail into \lamrsharfail.}
    \label{fig:auxencfail}
\end{figure*}

\begin{exa}[Cont. Example~\ref{ex:terms}]
\label{ex:termencoded}
We illustrate the translation $\recencodopenf{\cdot}$ on previously discussed examples. In all cases, we start by ensuring that the free variables are shared. This explains the occurrence of $[ y_1 \leftarrow y]$ in the translation of $M_1$ as well as $ [ y_1 \leftarrow y]$ and $ [ z_1 \leftarrow z] $ in the translation of $M_2$. Then, the auxiliary translation $\recencodf{\cdot}$ ensures that bound variables that are guarded by an abstraction are shared. This explains, e.g., the occurrence of  $[x_1 \leftarrow x] $ in the translation of $M_1$.

    \begin{itemize}
        \item The translation of a $\lamrfail$-term with one occurrence of a bound variable and one occurrence of a free variable: $M_1=(\lambda x. x ) \bag{y}$.
        
        \[
        \begin{aligned} 
        \recencodopenf{M_1} &=  \recencodopenf{ (\lambda x. x ) \bag{y} } \\
        &= \recencodf{ (\lambda x. x ) \bag{y_1}   } [ y_1 \leftarrow y]\\
        &= ((\lambda x. x_1 [x_1 \leftarrow x] ) \bag{y_1} ) [ y_1 \leftarrow y] 
        \end{aligned}
        \]
         \item The translation of a $\lamrfail$-term with one bound and two  different free variables: $M_2= (\lambda x. x ) (\bag{y,z} )$.
        \[
        \begin{aligned}
            \recencodopenf{M_2} &=\recencodopenf{ (\lambda x. x ) (\bag{y,z} ) }\\
            &= \recencodf{ (\lambda x. x ) (\bag{y_1,z_1} )  } [ y_1 \leftarrow y] [ z_1 \leftarrow z]\\
            & = ( (\lambda x. x_1 [x_1 \leftarrow x] ) (\bag{y_1,z_1} )  )[ y_1 \leftarrow y] [ z_1 \leftarrow z]
        \end{aligned}
        \]
        \item The translation of a $\lamrfail$-term with a vacuous abstraction: $M_4=(\lambda x. y ) \oneb$.
        \[
        \begin{aligned}
        \recencodopenf{M_4} &= \recencodopenf{(\lambda x. y ) \oneb } \\
        &= \recencodf{(\lambda x. y_1 ) \oneb  } [y_1 \leftarrow y]\\ &= ((\lambda x. y_1 [ \leftarrow x] ) \oneb )[y_1 \leftarrow y]
        \end{aligned}
        \]
      % \item $\recencodopenf{M_5} = \recencodopenf{M_1} =\recencodopenf{\fail^{\emptyset}} = \recencodf{\fail^{\emptyset}} = \fail^{\emptyset} $
        \item The translation of a $\lamrfail$-expression: $M_6=(\lambda x. x ) \bag{y} + (\lambda x. x ) \bag{z}$.
        \[
        \begin{aligned}
            \recencodopenf{M_6} &= \recencodopenf{(\lambda x. x ) \bag{y} + (\lambda x. x ) \bag{z} } \\
            &= \recencodopenf{ (\lambda x. x ) \bag{y}} + \recencodopenf{ (\lambda x. x ) \bag{z}}\\
            & = ((\lambda x. x_1 [x_1 \leftarrow x] ) \bag{y_1} ) [ y_1 \leftarrow y] 
              + ((\lambda x. x_1 [x_1 \leftarrow x] ) \bag{z_1} ) [ z_1 \leftarrow y]
        \end{aligned}
        \]
    \end{itemize}

\end{exa}

\begin{exa}
The translation of a $\lamrfail$-term with two occurrences of a bound variable and two occurrences of a free variable: $M= (\lambda x. x \bag{x}) (\bag{y,y} )$.
        \[
        \begin{aligned}
            \recencodopenf{M} &=\recencodopenf{ (\lambda x. x\bag{x} ) (\bag{y,y} ) } \\
            &= \recencodf{ (\lambda x. x\bag{x} ) (\bag{y_1,y_2} )  } [ y_1,y_2 \leftarrow y] \\
            & = ( (\lambda x. x_1 \bag{x_2}\shar{x_1,x_2}{x} ) (\bag{y_1,y_2} )  )\shar{ y_1,y_2}{y}
        \end{aligned}
        \]
\end{exa}

\begin{exa}
Now consider the translation of $y \esubst{B}{x}$, with $\lfv{B}=\emptyset$ and $y \neq x$:
 \[
 \begin{aligned} 
 \recencodopenf{y\esubst{B}{x}}&=\recencodf{y_0 \esubst{B}{x}}[y_0\leftarrow y]\\
 &=y_0[\leftarrow x]\esubst{\recencodf{B}}{x}[y_0\leftarrow y].
 \end{aligned}
 \]
Hence, the translation induces (empty) sharing on $x$, even if $x$ does not occur in the term $y$.
\end{exa}

\secondrev{
\begin{restatable}[$\recencodopenf{\cdot }$ Preserves  Consistency]{propo}{consistencyencode}
\label{thm:consistency_encod}
Let $\expr{M}$ be a $\lamrfail$-expression. Then
    $\recencodopenf{\expr{M}}$ is a consistent 
    $\lamrsharfail$-expression.
\end{restatable}
}
\begin{proof}
\secondrev{
  By induction on the structure of $\expr{M}$. See \appref{app:typeshar} for details.
  % Notice that $\recencodf{\cdot}$ ensures consistency for bound variables: it replaces all occurrences of a bound variable (say $y$) with fresh bound variables (say, $y_1, \ldots, y_k$). Thus, the following hold for bound variables: (i) they occur once within a term and (ii) they are not shared themselves, as the sharing of variables only occurs when handling binders associated to explicit substitutions and abstractions. 
  % As for free variables, the translation $\recencodopenf{\cdot}$ replaces each occurrence with a fresh variable, and does so before applying $\recencodf{\cdot}$; this ensures that free variables that are already shared are not shared again. Because of this design, the translations preserve consistency.
}
\end{proof}

\section[A Session-Typed Calculus]{\texorpdfstring{\spi}{sπ}: A Session-Typed \texorpdfstring{$\pi$}{π}-Calculus with Non-Determinism}
\label{s:pi}

The $\pi$-calculus~\cite{DBLP:journals/iandc/MilnerPW92a} is a model of concurrency in which \emph{processes} interact via \emph{names} (or \emph{channels}) to exchange values, which  can be  themselves names.
%, the $\pi$-calculus is an expressive model of message-passing concurrency.
Here we overview \spi, introduced by Caires and P\'{e}rez in~\cite{CairesP17}, in which \emph{session types}~\cite{DBLP:conf/concur/Honda93,DBLP:conf/esop/HondaVK98}  ensure that the two endpoints of a channel perform matching actions:
%processes always interact in matching pairs:
when one endpoint sends, the other receives; when an endpoint closes, the other closes too.
Following~\cite{CairesP10,DBLP:conf/icfp/Wadler12},
\spi defines a Curry-Howard correspondence between session types and a  linear logic with two dual modalities  ($\with A$ and $\oplus A$),  which define \emph{non-deterministic} sessions.
In \spi, cut elimination corresponds to process communication, proofs correspond to processes, and propositions correspond to session types. 

\subsection{Syntax and Semantics}
We use $x, y,z, w \ldots$ to denote {names}   implementing the \emph{(session) endpoints} of protocols specified by session types. 
We consider the sub-language of~\cite{CairesP17}
without labeled choices and replication, which is actually sufficient to encode $\lamrfail$.

\begin{defi}[Processes]\label{d:spi}
The syntax of \spi processes is given by the grammar in Fig.~\ref{f:spi}. 
\end{defi}

\begin{figure}[!t]
\[
\begin{array}{rcl@{\hspace{1.5cm}}l}
    P,Q &::=  & \zero & \text{(inaction)} \\
      &\sep &\overline{x}(y).P& \text{(output)}\\
      &\sep &  x(y).P & \text{(input)}\\
      &\sep &  (P \para Q) & \text{(parallel)}\\
      &\sep &  (\nu x)P & \text{(restriction)}\\
      &\sep & [x \leftrightarrow y] & \text{(forwarder)}\\
      &\sep &x.\overline{\close} & \text{(session close)}\\
      &\sep &x.\close;P & \text{(complementary close)}\\
      &\sep & x.\overline{\some};P &\text{(session confirmation)}\\
      &\sep & x.\overline{\none} & \text{(session failure)}\\
      &\sep & x.\some_{(w_1, \cdots, w_n)};P & \text{(session dependency)}\\
      &\sep & P \oplus Q & \text{(non-deterministic choice)}
\end{array}
\]
\caption{Syntax of \spi. \label{f:spi}}
\end{figure}

% \begin{align*}
% P,Q ::=  ~& \overline{x}(y).P \sep x(y).P \sep  x.\overline{\close} \sep x.\close;P \sep 
%   [x \leftrightarrow y] \sep (P \para Q) \sep (\nu x)P \sep \zero \\
% ~\sep &    x.\overline{\some};P \sep x.\overline{\none} \sep  x.\some_{(w_1, \cdots, w_n)};P \sep  P \oplus Q 
% \end{align*}

As standard, $\zero$ is the inactive process. Session communication is performed using the pair of primitives output and input: the output process $\overline{x}(y).P$ sends a fresh name $y$ along session $x$ and then continues as $P$; the input process $x(y).P$ receives a name $z$ along $x$ and then continues as  $P\subst{z}{y}$, which denotes the capture-avoiding substitution of $z$ for $y$ in $P$.
Process $P \para Q$ denotes the parallel execution of $P$ and $Q$. 
Process $(\nu x)P$ denotes the process $P$ in which name $x$ has been restricted, i.e., $x$ is kept private to $P$. The forwarder process $[x \leftrightarrow y]$ denotes a bi-directional link between sessions $x$ and $y$.
 Processes $x.\overline{\close}$ and $x.\close;P$ denote complementary actions for   closing session $x$.

The following constructs introduce non-determi\-nis\-tic sessions which, intuitively, \emph{may} provide a session protocol  \emph{or} fail. 
    \begin{itemize}
    \item Process $x.\overline{\some};P$ confirms that the session  on $x$ will execute and  continues as $P$.
    \item  Process $x.\overline{\none}$ signals the failure of implementing the session on $x$.
    
    \item Process $x.\some_{(w_1, \cdots,w_n)};P$ specifies a dependency on a non-deterministic session $x$. 
    This process can  either (i)~synchronize with an action $x.\overline{\some}$ and continue as $P$, or (ii)~synchronize with an action $x.\overline{\none}$, discard $P$, and propagate the failure on $x$ to $(w_1, \cdots, w_n)$, which are sessions implemented in $P$.
    When $x$ is the only session implemented in $P$, the tuple of dependencies is empty and so we write simply $x.\some;P$.
    
    \item $P \oplus Q$ denotes a \emph{non-deterministic choice} between $P$ and $Q$. 
            We shall often write $\bigoplus_{i \in I} P_i$ to stand for $P_1 \oplus \cdots \oplus P_n$.
            
        %As formalized by the structural congruence given next, this is a form of non-collapsing non-determinism. 
\end{itemize}
\noindent
In  $(\nu y)P$ and $x(y).P$ the distinguished occurrence of name $y$ is binding, with scope $P$.
The set of free names of $P$ is denoted by $\fn{P}$. 
We identify process up to consistent renaming of bound names, writing $\equiv_\alpha$ for this congruence. 
We omit trailing occurrences of $\zero$; this way, e.g., we write $x.\close$ instead of $x.\close;\zero$.

\emph{Structural congruence}, denoted $\equiv$, expresses basic identities on the structure of processes and the non-collapsing nature of non-determinism.   

\begin{defi}[Structural Congruence]
\label{def:spistructcong}
 Structural congruence 
is defined as the least congruence relation on processes such that:
\[
\begin{array}{ll}
\begin{array}{rcl}
P \para \zero \! &\equiv& \! \zero\\
P \para Q &\equiv& Q \para P \\
(P \para Q) \para R  &\equiv& P \para (Q \para R)\\
\left[x\leftrightarrow y\right]&\equiv&\left[y \leftrightarrow x\right]\\
  ((\nu x )P) \para Q  &\equiv& (\nu x)(P \para Q), x \not \in \fn{P}\\
(\nu x)(P \para (Q \oplus R))  &\equiv& (\nu x)(P \para Q) \oplus (\nu x)(P \para R)\\
\end{array}
&
\begin{array}{rcl}
\zero \oplus \zero  &\equiv&  \zero\\
P \oplus Q &\equiv & Q \oplus P
\\
(P \oplus Q) \oplus R  &\equiv& P \oplus (Q \oplus R)\\
(\nu x)\zero  &\equiv& \zero\\
(\nu x)(\nu y)P  &\equiv& (\nu y)(\nu x)P\\
   P \equiv_\alpha Q &\Longrightarrow& P  \equiv Q
\end{array}
\end{array}
\]

\end{defi}

\subsection{Operational Semantics}\hfill

The {operational semantics} of \spi  is given by a reduction relation, denoted $P\red Q$, which is the smallest relation on processes generated by the rules in~\figref{fig:redspi}. These rules specify the computations that a process performs on its own. We now explain each rule.

\begin{figure}[!t]

{\
\[
% \begin{array}{c}
  \begin{array}{l@{\hspace{1.5cm}}rcl}
  \redlab{Comm}&\overline{x}{(y)}.Q \para x(y).P    & \red & 
  (\nu y) (Q \para P)
  \\
  \redlab{Forw}& (\nu x)( [x \leftrightarrow y] \para P) & \red &  P \subst{y}{x}  \quad (x \neq y)
 \\
 \redlab{Close}&x.\overline{\close} \para x.\close;P   & \red &  P 
 \\
 \redlab{Some}&x.\overline{\some};P \para x.\some_{(w_1, \cdots, w_n)};Q  & \red & P \para Q
 \\
\redlab{None}&x.\overline{\none} \para x.\some_{(w_1, \cdots, w_n)};Q  & \red & 
w_1.\overline{\none} \para \cdots \para w_n.\overline{\none}\\
\redlab{Cong}& P\equiv P'\wedge P' \red Q' \wedge Q'\equiv Q &\Longrightarrow& P  \red Q\\
\redlab{Par}& Q \red Q' &\Longrightarrow& P \para Q  \red P \para Q'\\
  \redlab{Res}&P \red Q  &\Longrightarrow& (\nu y)P \red (\nu y)Q \\
  \redlab{NChoice}&Q \red Q' &\Longrightarrow& P \oplus Q  \red P \oplus Q' 
\\[2mm]
\end{array}
\]}
\caption{Reduction for \spi.}
    \label{fig:redspi}
\end{figure}

\begin{itemize}
\item {\bf Rule~\redlab{Comm}}  formalizes communication, which concerns bound names only (internal mobility): name $y$ is bound in both $\overline{x}{(y)}.Q$ and $x(y).P$.
\item  {\bf Rule~\redlab{Forw}}  implements the  forwarder process  that leads to a name  substitution.
\item {\bf Rule~\redlab{Close}} formalizes session closure and is self-explanatory. 
\item {\bf Rule~\redlab{Some}} describes the synchronization of a process, that is dependent on a non-deterministic session $x$, with the complementary  process  $x.\overline{\some}$ that confirms the availability of such non-deterministic session. 
\item {\bf Rule~\redlab{None}} applies when the non-deterministic session is not available, prefix $x.\overline{\none}$ triggers this failure to all dependent sessions $w_1, \ldots, w_n$; this may in turn trigger further failures (i.e., on sessions that depend on $w_1, \ldots, w_n$). \item {\bf Rule~\redlab{NChoice}} defines the closure of reduction w.r.t. non-collapsing non-deterministic choice.
\item {\bf Rules~\redlab{Cong}, \redlab{Par} and
\redlab{Res}} are standard and formalize that reduction is closed under structural congruence, and also contextual closure of  parallel and restriction constructs.
\end{itemize}

\begin{exa}
We illustrate confluent reductions starting in a non-deterministic process $R$ which will fail during communication due to unavailability of a session: 
\begin{align*}
R=& (\nu x) ( x.\some_{(y_1,y_2)};y_1(z).y_2(w).\zero \para  ( x.\overline{\some};P \oplus x.\overline{\none}  ) )
\\	
  \equiv  & (\nu x) ( x.\some_{(y_1,y_2)};y_1(z).y_2(w).\zero  \para\! x.\overline{\some};P )  \oplus  (\nu x) ( x.\some_{(y_1,y_2)};y_1(z).y_2(w).\zero  \!\para\! x.\overline{\none}  ) 
\end{align*}
Letting $Q = y_1(z).y_2(w).\zero$, we have:

    %  $${}
    {\small \hspace*{-1em}
        \begin{tikzpicture}
          \matrix (m) [matrix of math nodes, row sep=2em, column sep=-10em,ampersand replacement=\&]
            { 
                \node(A){ }; \& 
                \node(B){(\nu x) ( x.\some_{(y_1,y_2)};Q \para  x.\overline{\some};P)  \oplus (y_1.\overline{\none} \para y_2.\overline{\none} ) }; \\
                \node(C){R= (\nu x) ( x.\some_{(y_1,y_2)};Q \para ( x.\overline{\some};P \oplus x.\overline{\none}  ) )}; \& 
                \node(D){ }; \& 
                \node(G){ (\nu x)(Q \para P) \oplus (y_1.\overline{\none} \para y_2.\overline{\none} ) }; \\
                \node(E){ }; \&
                \node(F){ (\nu x)(Q \para P) \oplus (\nu x) ( x.\some_{(y_1,y_2)};Q \para x.\overline{\none} ) }; \\};
                
            \path (C) edge[->](B);
            \path (C) edge[->](F);
            \path (B) edge[->](G);
            \path (F) edge[->](G);
        \end{tikzpicture}
        }
        % $$
Observe that  reduction is confluent. 
The resulting term 
$(\nu x)(Q \para P) \oplus (y_1.\overline{\none} \para y_2.\overline{\none} )$
includes both alternatives for the interaction on $x$, namely the successful one (i.e., $(\nu x)(Q \para P)$) but also the failure of $x$, which is then propagated to $y_1$ and $y_2$, i.e., $y_1.\overline{\none} \para y_2.\overline{\none}$. 
\end{exa}

\subsection{Type System}\hfill
The type discipline for  $\spi$  is based on the type system given in~\cite{CairesP17}, which contains modalities $\with A$ and $\oplus A$, as dual types for non-deterministic sessions.

\begin{defi}[Session Types]
\label{d:sts}
Session types are given by 
\begin{align*}
A,B & ::=  \bot \sep   \onef \sep 
A \otimes B  \sep A \ampy B  
%\sep  A \oplus  B \sep  A \& B 
 \sep  \with A \sep \oplus A  
\end{align*} 
\end{defi}

\noindent
Types are assigned to names: an \emph{assignment} $x:A$ enforces the use of name $x$ according to the   protocol specified by $A$.
The multiplicative units  $\bot$ and  $\onef$ are used to type terminated (closed) endpoints.
 $A \otimes B$ types a name that first outputs a name of type $A$ before proceeding as specified by $B$.
Similarly, $A \ampy B $ types a name that first inputs a name of type $A$ before proceeding as specified by $B$.
Then we have the two modalities introduced in~\cite{CairesP17}.
We use $\with A$ as the type of a (non-deterministic) session that \emph{may  produce} a behavior of type $A$.
Dually, $\oplus A$ denotes the type of a session that \emph{may consume} a behavior of type $A$.

The two endpoints of a  session must be \emph{dual} to ensure  absence of communication errors. 
The dual of a type $A$ is denoted $\dual{A}$. 
Duality corresponds to negation $(\cdot)^\bot$ in linear logic:

\begin{defi}[Duality]
\label{def:duality}
The duality relation on types is given by:
\begin{align*}
\dual{\onef} & =  \bot 
&
\dual{\bot}   & =  \onef
&
\dual{A\otimes B}  & = \dual{A} \ampy \dual{B}
&
\dual{A \ampy B}  & = \dual{A} \otimes \dual{B} 
&
\dual{\oplus A}   & =    \with \dual{A}
&
\dual{\with A}  & =   \oplus \overline {A}  
\end{align*}
\end{defi}

Typing judgments are of the form $P \vdash \Delta$, where $P$ is a process and $\Delta$ is a context of 
\srev{the form $x_1:A_1, \ldots, x_n:A_n$, which defines the assignment of type $A_i$ to name $x_i$ (with $1 \leq i \leq n$); all names $x_i$ must be distinct. The context $\Delta$ is \emph{linear} in that it is subject to exchange (the ordering of assignments does not matter), but not to weakening and contraction. In writing `$\Delta, x:A$', we assume that $x$ does not occur in~$\Delta$; also, in writing `$\Delta_1, \Delta_2$', we assume that the names in $\Delta_1$ are distinct from those in $\Delta_2$.}
The empty context is denoted `$\cdot$'. 
We write $\with \Delta$ to denote that all assignments in $\Delta$ have a non-deterministic type, i.e., $\with \Delta = w_1:\with A_1, \ldots, w_n:\with A_n$, for some $A_1, \ldots, A_n$. 
The typing judgment $P \vdash \Delta$ corresponds to the logical sequent $ \vdash \Delta$ for classical linear logic, which can be recovered by erasing processes and name assignments.

\begin{figure}[!t]
\begin{prooftree}
\AxiomC{}
\LeftLabel{\redlab{T\cdot}}
\UnaryInfC{$\zero \vdash $}
\DisplayProof
\hfill
\AxiomC{}
\LeftLabel{\redlab{Tid}}
\UnaryInfC{$[x \leftrightarrow y] \vdash x{:}A, y{:}\dual{A}$}
\end{prooftree}

\begin{prooftree}
\AxiomC{$P \vdash  \Delta, y:{A} \quad Q \vdash \Delta', x:B $}
\LeftLabel{\redlab{T\otimes}}
\UnaryInfC{$\dual{x}(y). (P \mid Q) \vdash  \Delta, \Delta', x: A\otimes B$}
\DisplayProof
\hfill
\AxiomC{$P \vdash \Gamma, y:C, x:D$}
\LeftLabel{\redlab{T\ampy}}
\UnaryInfC{$x(y).P \vdash \Gamma, x: C\ampy D $}
\end{prooftree}

\begin{prooftree}
\AxiomC{\mbox{\ }}
\LeftLabel{\redlab{T\onef}}
\UnaryInfC{$x.\dual{\close} \vdash x: \onef$}
\DisplayProof
\hfill
\AxiomC{$P\vdash \Delta$}
\LeftLabel{\redlab{T\bot}}
\UnaryInfC{$x.\close;P \vdash x{:}\bot, \Delta$}
\end{prooftree}

\begin{prooftree}
\AxiomC{$P \vdash  \Delta \quad Q \vdash  \Delta'$}
\LeftLabel{\redlab{T\mid}}
\UnaryInfC{$P\mid Q \vdash \Delta, \Delta'$}
\DisplayProof
\hfill
\AxiomC{$P \vdash \Delta, x:\dual{A} \quad  Q \vdash  \Delta', x:A$}
\LeftLabel{\redlab{Tcut}}
\UnaryInfC{$(\nu x)(P \mid Q) \vdash\Delta, \Delta'$}
\end{prooftree}

\begin{prooftree}
\AxiomC{$P \vdash \Delta, x:A $}
\LeftLabel{\redlab{T\with_d^x}}
\UnaryInfC{$x.\dual{\some};P \vdash \Delta, x :\with A$}
\DisplayProof
\hfill
\AxiomC{$P \;{ \vdash} \widetilde{w}:\with\Delta, x:A$}
\LeftLabel{\redlab{T\oplus^x_{\widetilde{w}}}}
\UnaryInfC{$x.\some_{\widetilde{w}};P \vdash \widetilde{w}{:}\with\Delta, x{:}\oplus A$}
\end{prooftree}

\begin{prooftree}
\AxiomC{}
\LeftLabel{\redlab{T\with^x}}
\UnaryInfC{$x.\dual{\none} \vdash x :\with A$}
\DisplayProof
\hfill
\AxiomC{$P \vdash \with\Delta \qquad Q  \;{\vdash} \with\Delta$}
\LeftLabel{\redlab{T\with}}
\UnaryInfC{$P\oplus Q \vdash \with\Delta$}
\end{prooftree}

\caption{Typing rules for \spi.}
\label{fig:trulespifull}
\end{figure}

Typing rules for processes correspond to proof rules in the logic; see \figref{fig:trulespifull}. 
This way, Rule~$\redlab{T\cdot}$ allows us to introduce the inactive process $\zero$. Rule~$\redlab{Tid}$ interprets the identity axiom using the forwarder process. Rules~$\redlab{T\otimes}$ and $\redlab{T \ampy }$ type output and input of a name along a session, respectively. 
Rules~$\redlab{T \onef}$ and $\redlab{T \bot}$ type the process constructs for session termination.
  Rules~$\redlab{T cut}$ and  $\redlab{T \mid }$  define cut and mix principles in the logic, which induce typing rules for independent and dependent parallel composition, respectively.

The last four rules in \figref{fig:trulespifull} are used to type process constructs related to non-de\-ter\-mi\-nism and failure. 
 Rules~$ \redlab{T \with_d^x}$ and $ \redlab{T \with^x}$ introduce a session of type $\with A$, which may produce a behavior of type $A$: while the former rule covers the case in which $x:A$ is indeed available, the latter rule formalizes the case in which $x:A$ is not available (i.e., a failure).
 Rule~$\redlab{T \oplus^x_{\widetilde{w}}}$, accounts for the possibility of not being able 
to consume the session $x:A$  by considering sessions, the sequence of names $\widetilde{w} = w_1, \ldots, w_n$,  different from $x$ as potentially not available. 
 Rule~$\redlab{T \with }$ expresses non-deterministic choice of processes $P$ and $Q$ that implement non-deterministic behaviors only.

The type system enjoys type preservation, a result that
follows directly from the cut elimination property in the underlying logic; it ensures that the observable interface of a system is invariant under reduction.
The type system also ensures other properties for well-typed processes (e.g. global progress  and confluence); see~\cite{CairesP17} for details.

\begin{thm}[Type Preservation~\cite{CairesP17}]
If $P \vdash \Delta$ and $P \red Q$ then $Q \vdash \Delta$.
\end{thm}

Having defined  \spi, we now move on to define a correct translation from \lamrfail to \spi.

%%%% 

\section{A Correct Encoding}
\label{s:encoding}
\srev{Having introduced the typed sequential calculi \lamrfail and \lamrsharfail (as well as the translation~$\recencodopenf{~\cdot~} : \lamrfail \to \lamrsharfail $) and the typed concurrent calculus \spi, in this section we show how to correctly translate \lamrfail into \spi, using \lamrsharfail as a stepping stone.} 

\srev{Before delving into technical details, we briefly discuss the significance of our encoding. As in Milner's seminal work, our translation explains how interaction in $\pi$ provides a principled interpretation of evaluation in $\lambda$. We tackle the challenging case in which evaluation and interaction are fail-prone and non-deterministic, effectively generalizing previous translations. Because our encoding preserves types, our developments also delineate a new connection between non-idempotent intersection types and logically motivated session types---indeed, our translation of functions as processes goes hand-in-hand with a translation on types (\figref{fig:enc_sestypfail}), which reveals a new protocol-oriented interpretation of the non-idempotent intersections that govern functional resources. }

\srev{As already mentioned, we shall proceed in two steps. We rely on the translation $\recencodopenf{\cdot}$ from well-formed expressions in \lamrfail to well-formed expressions in \lamrsharfail given in \secref{ss:auxtrans}.
As \lamrfail and \lamrsharfail share the same syntax of types, in this case the translation of types  is the identity. 
    Then, the translation $\piencodf{\cdot}_u$ (for some name $u$) transforms well-formed expressions in \lamrsharfail to well-typed processes in \spi (cf. Fig.~\ref{f:sum}).
We first define \emph{encodability criteria} for translations, which include type preservation; these criteria lead to the notion of \emph{correct encoding} (\secref{ss:criteria}). Then, in \secref{ss:firststep} we establish the correctness of the translation  $\recencodopenf{\cdot}$ (Corollary~\ref{cor:one}); 
finally, in \secref{ss:secondstep}, we present the translation $\piencodf{\cdot}_u$ and establish its correctness (Corollary~\ref{cor:two}).
    }
    
\tikzstyle{mynode1} = [rectangle, rounded corners, minimum width=2cm, minimum height=1cm,text centered, draw=black, fill=brown!80!purple!40]
\tikzstyle{arrow} = [thick,->,>=stealth]

\tikzstyle{mynode2} = [rectangle, rounded corners, minimum width=2cm, minimum height=1cm,text centered, draw=black, fill=teal!20]

\tikzstyle{mynode3} = [rectangle, rounded corners, minimum width=2cm, minimum height=1cm,text centered, draw=black, fill=violet!20]
\tikzstyle{arrow} = [thick,->,>=stealth]

\begin{figure}[!t]
\begin{center}
\begin{tikzpicture}[node distance=2cm]
\node (source) [mynode1] {$\lamrfail$};
\node (interm) [mynode2, right of=source,  xshift=3cm] {$\lamrsharfail$};
\node (target) [mynode3, right of=interm,  xshift=3cm] {$\spi$};
\draw[arrow] (source) --  node[anchor=south] {$ \recencodopenf{\cdot}$} (interm);
\node (enc1) [right of= source, xshift=.4cm, yshift=-.5cm] {\secref{ss:auxtrans}};
\node (enc2) [right of= interm, xshift=.4cm, yshift=-.5cm] {\secref{ss:secondstep}};
\draw[arrow] (interm) -- node[anchor=south] {$ \piencodf{\cdot }_u $ } (target);
\end{tikzpicture}
\end{center}
\caption{Summary of our approach.\label{f:sum}}
\end{figure}

\subsection{Encodability Criteria}
\label{ss:criteria}\hfill

We follow most of the criteria defined by Gorla in~\cite{DBLP:journals/iandc/Gorla10}, a widely studied abstract framework for establishing the \emph{quality} of translations.
A \emph{language} $\mathcal{L}$ is defined as a pair containing a set of terms $\mathcal{M}$ and a reduction semantics $\red$ on terms (with reflexive, transitive closure denoted $\tred$). \srev{A behavioral equivalence on terms, denoted $\approx$, is also assumed.} 
Then, a \emph{correct encoding}, defined next, concerns a translation of terms of a source language $\mathcal{L}_1$ into terms of a target language  $\mathcal{L}_2$ that respects  certain criteria. 
The criteria in~\cite{DBLP:journals/iandc/Gorla10} concern \emph{untyped} languages; because we consider \emph{typed} languages,  we follow Kouzapas et al.~\cite{DBLP:journals/iandc/KouzapasPY19} in requiring also that translations preserve typability. 

\begin{defi}[\secondrev{Correct Encoding}]
\label{d:encoding}
Let $\mathcal{L}_1 = (\mathcal{M}, \red_1)$
and 
$\mathcal{L}_2 = (\mathcal{P}, \red_2)$
be two languages and let $\approx_1$ be a behavioral equivalence on terms in $\mathcal{M}$.
We use $M, M', \ldots$ and $P, P', \ldots$ to range over elements in  $\mathcal{M}$ and $\mathcal{P}$.
We say that a translation  $\encod{\cdot}{}: \mathcal{M} \to \mathcal{P}$ is a \emph{correct encoding} if it satisfies the following criteria:
\begin{enumerate}
\item {\it Type preservation:} For every well-typed $M$, it holds that $\encod{M}{}$ is well-typed.

    \item {\it Operational Completeness:} For every ${M}, {M}'$, and ${M}''$ such that ${M} \tred_1 {M}' \approx_1 {M}'' $, it holds that $\encod{{M}}{} \tred_2 \encod{{M}''}{}$.
    
    \item {\it Operational Soundness:} For every $M$ and $P$ such that $\encod{M}{} \tred_2 P$, there exist $M'$ and $M''$ such that $M \red^*_1 M' \approx_1 M''$ and $P \tred_2 \encod{M''}{}  $.
    
    \item {\it Success Sensitiveness:} Let $\checkmark_1$ and $\checkmark_2$ denote a success predicate in $\mathcal{M}$ and $\mathcal{P}$, respectively. 
For every ${M}$, it holds that $M \checkmark_1$ if and only if $\encod{M}{} \checkmark_2$.     
    %    \item Compositionality: For every $k$-ary construct $\mathtt{op}$ of $\mathcal{L}_1$ there exists a  $k$-ary context $\mathbb{C}$ in $\mathcal{L}_2$ such that $\encod{\mathtt{op}(M_1, \ldots, M_k)}{} = \mathbb{C}[\encod{M_1}{}, \cdots, \encod{M_k}{}]$.
    
\end{enumerate}
\end{defi}

We briefly describe  the criteria. First, type preservation is a natural requirement and a distinguishing aspect of our work, given that we always consider source and target calculi with types.
Operational completeness formalizes how reduction steps of a source term are mimicked by its corresponding translation in the target language; $\approx_1$ conveniently abstracts away from source terms useful in the translation but which are not meaningful in comparisons. 
Operational soundness concerns the opposite direction: it formalizes the correspondence between (i)~the reductions of a target term obtained via the translation and (ii)~the reductions of the corresponding source term. The role of $\approx_1$ can be explained as in completeness.
\srev{Our use of the equivalence $\approx_1$ for $\mathcal{M}_1$, rather than of an equivalence on $\mathcal{M}_2$, is a minor difference with respect to~\cite{DBLP:journals/iandc/Gorla10}.}
Finally, success sensitiveness complements completeness and soundness, which concern reductions and therefore do not contain information about observable behaviors. 
The so-called success predicates $\checkmark_1$ and $\checkmark_2$ serve as a minimal notion of \emph{observables}; the criterion then says that observability of success of a source term implies observability of success in the corresponding target term, and vice versa.

%\secondrev{Operational completeness and soundness differ from \cite{DBLP:journals/iandc/Gorla10}, as we consider up to behavioral equivalence in $\mathcal{L}_1$ rather then that of $\mathcal{L}_2$. We find it more convenient to consider behavioral equivalence in the source language as it clarifies behaviors in the source that differ from the target. }

Besides these semantic criteria, we also consider \emph{compositionality}, a syntactic criterion that requires that a composite source term is translated as the combination of the translations of its sub-terms. 

 \subsection{Correctness of \texorpdfstring{$\recencodopenf{\cdot}$}{⦇⋅⦈^∘}}
 \label{ss:firststep}\hfill

%Before presenting our translation $\piencodf{\cdot}_u$, 
We prove that the translation  $\recencodopenf{\cdot}$ from $\lamrfail$ into $\lamrsharfail$ in \secref{ss:auxtrans} is a correct encoding, in the sense of \defref{d:encoding}.
Because our translation $\recencodopenf{\cdot}$ is defined in terms of $\recencodf{\cdot}$, it satisfies \emph{weak compositionality}, in the sense of Parrow~\cite{DBLP:journals/entcs/Parrow08}.

\subsubsection{Type Preservation}\hfill

We now prove that 
$\recencodopenf{\cdot}$
translates well-formed $\lamrfail$-expressions
into
well-formed expressions \lamrsharfail-expressions
(Theorem~\ref{thm:preservencintolamrfail2}). 
Notice that because \lamrfail and \lamrsharfail share the same type syntax, \revo{}{there is no translation on types/contexts involved (i.e., an identity translation applies)}. 

%\begin{defi}[Encoding on Contexts] We define an translation $\recencodopenf{\cdot }$ on contexts:
%\label{d:enclamcontfail}
%\begin{align*}
%\recencodopenf{\emptyset} &= \emptyset
%\\
%\recencodopenf{ x: \tau , \Gamma }  &=   x:\tau  , \recencodopenf{ \Gamma } &  (x \not \in \dom{\Gamma}) \\
%\recencodopenf{ x: \tau , \cdots , x: \tau, \Gamma } &= x : \tau \wedge \cdots \wedge \tau , \recencodopenf{ \Gamma } &  (x \not \in \dom{\Gamma})
%    \end{align*}
%\end{defi}

Next we define well formed preservation in the translation $\recencodf{\cdot}$ from $\lamrfail$ to $\lamrsharfail$. We rely on the prerequisite proof of type preservation in the translation $\recencodf{\cdot}$ on the sub-calculi from $\lamr$ to $\lamrshar$, and also on syntactic properties of the translation such as: (i) the property below guarantees that the translation $\recencodf{\cdot }$ commutes with the linear head  substitution; (ii) preservation of typability/well-formedness w.r.t. linear substitutions in \lamrsharfail.
\begin{prop}
\label{prop:linhed_encfail}
Let $M, N$ be $\lamrfail$-terms. We have:
 \begin{enumerate}
 \item $ \recencodf{M\headlin{N/x}}=\recencodf{M}\headlin{\recencodf{N}/x}$.
 \item $ \recencodf{M\linsub{\widetilde{x}}{x}}=\recencodf{M}\linsub{\widetilde{x}}{x}$, where $\widetilde{x}=x_1,\ldots, x_k$ is sequence of pairwise distinct fresh variables.
 \end{enumerate}
\end{prop}

\begin{proof}
By induction of the structure of $M$.
\end{proof}

%We rely on the prerequisite proof of type preservation in the translation $\recencodf{\cdot}$ on the sub-calculi from $\lamr$ to $\lamrshar$.

\begin{lem}[Preservation under Linear Substitutions in $\lamrsharfail$] Let ${M} \in \lamrsharfail$.\label{lem:preser_linsub}
    \begin{enumerate}
        \item Typing: If $\Gamma, x:\sigma^{k} \vdash {M} : \tau$
    \revo{A6}{}%and $x_i : \sigma \vdash x_i : \sigma$}
    then 
     $\Gamma, x_i:\sigma^{k-1} \vdash {M}\linsub{x_i}{x} : \tau$.
        \item Well-formedness: If $\Gamma, x:\sigma^{k} \wfdash {M} : \tau$
     \revo{A6}{}%and $\Gamma \wfdash x_i : \sigma$ 
     then $\Gamma, x_i:\sigma^{k-1} \wfdash {M}\linsub{x_i}{x} : \tau$.
    \end{enumerate}
\end{lem}
\begin{proof}
Standard by induction on the rules from \figref{fig:typing_sharing} for item (1), and ~\figref{fig:wfsh_rules} for item~(2).
\end{proof}

The following example illustrates that the translation of a well-formed $\lamrfail$-expression is a well-formed $\lamrsharfail$-expression.
\begin{exa}[Cont. Example~\ref{ex:wellformed}]
Term $M_2=(\lambda x. x ) (\bag{y,z})$ is well-formed with a well-formedness judgment  \( y:\sigma, z:\sigma \wfdash (\lambda x. x ) (\bag{y,z})  : \sigma \). In Example~\ref{ex:termencoded} we showed that $\recencodopenf{M_2}=( (\lambda x. x_1 [x_1 \leftarrow x] ) (\bag{y_1,z_1} )  )[ y_1 \leftarrow y] [ z_1 \leftarrow z]$ which is well-formed with translated well-formed judgment $ y:\sigma^1, z:\sigma^1 \wfdash \recencodopenf{M_2}:\sigma$.
%, that is equal to $y:\sigma^1, z:\sigma^1 \wfdash \recencodopenf{M_2}:\sigma$.
The derivation is given below (using rules from \figref{fig:wfsh_rules});  we omit the labels of rule applications and concatenations with the empty bag, i.e., we write $\bag{y_1}$ instead of $\bag{y_1}\cdot \oneb$.
%{\small 
\begin{prooftree}
\AxiomC{}
%\LeftLabel{\(\redlab{TS:var}\)}
\UnaryInfC{\(x_1:\sigma\vdash x_1:\sigma\)}
%\LeftLabel{\(\redlab{FS:wf\dash expr}\)}
\UnaryInfC{\(x_1:\sigma\wfdash x_1:\sigma\)}
%\LeftLabel{\(\redlab{FS:share}\)}
\UnaryInfC{\(x:\sigma^1 \wfdash x\shar{x_1}{x}:\sigma\)}
%\LeftLabel{\(\redlab{FS:abs\dash sh}\)}
\UnaryInfC{\( \wfdash \lambda x.(x\shar{x_1}{x}):\sigma\to\sigma \)}
\AxiomC{}
%\LeftLabel{\(\redlab{TS:var}\)}
\UnaryInfC{\(y_1:\sigma\vdash y_1:\sigma\)}
%\LeftLabel{\redlab{FS:wf\dash bag}}
\UnaryInfC{\(y_1:\sigma\wfdash y_1:\sigma\)}
%\LeftLabel{\(\redlab{FS:bag}\)}
\UnaryInfC{\(y_1:\sigma^1\wfdash \bag{y_1}:\sigma^1\)}
\AxiomC{}
%\LeftLabel{\(\redlab{TS:var}\)}
\UnaryInfC{\(z_1:\sigma\vdash z_1:\sigma\)}
%\LeftLabel{\redlab{FS:wf\dash bag}}
\UnaryInfC{\(z_1:\sigma\wfdash z_1:\sigma\)}
%\LeftLabel{\(\redlab{FS:bag}\)}
\UnaryInfC{\(z_1:\sigma^1\wfdash \bag{z_1}:\sigma^1\)}
%\LeftLabel{\(\redlab{FS:bag}\)}
\BinaryInfC{\(y_1:\sigma^1,z_1 :\sigma^1\wfdash \bag{y_1}\cdot \bag{z_1}:\sigma^2\)}
%\LeftLabel{\(\redlab{FS:app}\)}
\BinaryInfC{\(y_1:\sigma^1,z_1:\sigma^1\wfdash \lambda x. (x_1\shar{x_1}{x})\bag{y_1,z_1}:\sigma\)}
%\LeftLabel{\(\redlab{FS:share}\)}
\UnaryInfC{\(y:\sigma^1,z_1:\sigma^1\wfdash \lambda x. (x_1\shar{x_1}{x})\bag{y_1,z_1}\shar{y_1}{y}:\sigma\)}
%\LeftLabel{\(\redlab{FS:share}\)}
\UnaryInfC{\(y:\sigma^1,z:\sigma^1\wfdash \lambda x. (x_1\shar{x_1}{x})\bag{y_1,z_1}\shar{y_1}{y}\shar{z_1}{z}:\sigma\)}
\end{prooftree}
\end{exa}

% \begin{prop}[Preservation of Typing under Linear Substitutions in \lamrsharfail] \hfill 
%     \label{p:sublemmalamrshar}
    
%     Let ${M} \in \lamrsharfail$.
%     If $\Gamma, x:\sigma \vdash {M} : \tau$
%     and $x_i : \sigma \vdash x_i : \sigma$ then 
%      $\Gamma, x_i:\sigma \vdash {M}\linsub{x_i}{x} : \tau$.
     
%  \end{prop}

%  \begin{proof}
%  Standard, by induction on the typing derivation.
%  %\qed
%  \end{proof}

% \begin{lem}[Well-formedness Preservation under Linear Substitutions in \lamrsharfail]
% \label{prop:wf_linsub}
% Let ${M} \in \lamrsharfail$.
% If $\Gamma, x:\sigma \wfdash {M} : \tau$
% and $\Gamma \wfdash x_i : \sigma$ then 
%  $\Gamma, x_i:\sigma \wfdash {M}\linsub{x_i}{x} : \tau$.
%  \end{lem}
 
%  \begin{proof}
%  Standard, by induction on the well-formedness derivation rules in Figure~\ref{fig:app_wf_rules}.
%  %\qed
%  \end{proof}

% We then have the following results, whose proofs can be found in Appendix~\ref{apen:firstencod}.
 
 As the translation $\recencodopenf{\cdot }$ for 
 \revo{}{} %open
 $\lamrfail$-terms is defined in terms of $\recencodf{\cdot}$,
 \revo{A20}{}%for closed $\lamrfail$
 it is natural that  preservation of well-formedness under $\recencodopenf{\cdot }$ (Theorem~\ref{thm:preservencintolamrfail2}) relies on the preservation of well-formedness  under $\recencodf{\cdot}$, given next.
 
 \secondrev{To state well-formedness preservation, we use $\core{\Gamma}$, the core context of $\Gamma$ (\defref{d:tcont}). In the following property, we use an additional condition on $\core{\Gamma}$, which reflects the fact that intersection types get ``flattened'' by virtue of the translation. The condition, denoted~$\strcore{\Gamma}$, is defined whenever $\core{\Gamma}$ contains only unary multisets as follows: if
$x: \sigma^1 \in \core{\Gamma}$ for all $x \in \dom{\core{\Gamma}}$, then $x: \sigma \in \strcore{\Gamma}$. }
%($\strcore{\Gamma}$ is undefined otherwise.) }
 % \revo{A9}{To state well-formedness preservation, we use the core contexts of $\Gamma$ (denoted $\core{\Gamma}$ and $\strcore{\Gamma}$), as defined in \defref{d:tcont}. In the following property, the use of  $\strcore{\Gamma}$ reflects the fact that intersection types get ``flattened'' by virtue of the encoding.}

\begin{restatable}[Well-formedness preservation for $\recencodf{\cdot}$]{lema}{preservencintolamrfail}
\label{lem:wfpreserv_closedtrans}
\revo{A20,A21,A22}{
Let $B$ and  $\expr{M}$  be a  bag and an expression in $\lamrfail$, respectively. \secondrev{Also, let $\Gamma$ be a context such that $\strcore{\Gamma}$ is defined. We have:}
\begin{enumerate}
\item
    \revo{A8}{If $\Gamma \wfdash B:\pi$  
then $\strcore{\Gamma} \wfdash \recencodf{B}:\pi$.}
    \item 
    \revo{A9}{If $\Gamma \wfdash \expr{M}:\sigma$ 
then $\strcore{\Gamma} \wfdash \recencodf{\expr{M}}:\sigma$.}
\end{enumerate}
}
\end{restatable}

\begin{proof}[Proof (Sketch)]
By mutual induction on the typing derivations $\Gamma\wfdash B:\sigma$ and $\Gamma\wfdash \expr{M}:\sigma$. The proof of item (1) follows mostly by induction hypothesis, by analyzing the rule applied (\figref{fig:app_wf_rules}).
The proof of item (2), also follows by analyzing the rule applied, but it is more delicate, especially when treating  cases involving Rules~\redlab{FS:app} or \redlab{FS:ex\dash sub}, for which the size of the bag does not match the number of occurrences of variables in the expression. See  \appref{app:encodingprop} for full details.
\end{proof}

\begin{restatable}[Well-formedness Preservation for $\recencodopenf{\cdot}$]{thms}{preservencintolamrfailtwo}
\label{thm:preservencintolamrfail2}
Let $B$ and  $\expr{M}$  be a bag and an expression in $\lamrfail$, respectively. 
\begin{enumerate}

\item
    \revo{A10}{If $\Gamma \wfdash B:\pi$ 
then $\core{\Gamma} \wfdash \recencodopenf{B}:\pi$}.

    \item 
    \revo{A11}{If $\Gamma \wfdash \expr{M}:\sigma$ 
then $\core{\Gamma}\wfdash \recencodopenf{\expr{M}}:\sigma$}.

\end{enumerate}
\end{restatable}

\begin{proof}[Proof (Sketch)]
By mutual induction on the typing derivations $\Gamma\wfdash B:\sigma$ and $\Gamma\wfdash \expr{M}:\sigma$. Note that for a bag $B$, since the first part of translation consists in sharing the free variables of $B$, we will work with the translated bag $\recencodopenf{B}=\recencodf{B\linsub{\widetilde{x_1}}{x_1}\ldots \linsub{\widetilde{x_k}}{x_k}}\shar{\widetilde{x_1}}{x_1}\ldots \shar{\widetilde{x_k}}{x_k}$, and the rest of the proof depends on Proposition~\ref{prop:linhed_encfail} that moves linear substitutions outside $\recencodf{\cdot}$, then Lemma~\ref{lem:preser_linsub} that guarantees preservation of typability/well-formedness under linear substitutions,  and Lemma~\ref{lem:wfpreserv_closedtrans} for treating the closed translation. The dependency extends to the proof of item (2), for expressions. The full proof can be found in \appref{app:encodingprop}.
\end{proof}

\subsubsection{Operational Correspondence: Completeness and Soundness}\label{app:ss:compsound}\hfill

Def.~\ref{d:encoding} states operational completeness and soundness over the reflexive, transitive closure of the reduction rules. However, in the case of $\recencodopenf{\cdot}$, we prove completeness and soundness for a single reduction step (cf. Fig.~\ref{f:opcf}).  This is sufficient: by the diamond property (Proposition~\ref{prop:conf1_lamrfail}) a result stated for $\red$ can be extended easily to $\tred$, by induction on the length of the reduction sequence. (The result is immediate when the length is zero.)

\revdaniele{We rely on a {\em structural equivalence} over $\lamrfail$-expressions, denoted $\pequiv$, which is the least congruence satisfying $\alpha$-conversion and satisfying the identities in ~\figref{fig:rPrecongruencefail}. This congruence allows us to move explicit substitutions to the right of the term and to ignore explicit substitutions of a variable $x$ for empty bags in a term that does not contain $x$. }

\begin{exa}
   Consider the failure term $M= \fail^{y,y,z} \esubst{\oneb}{x}$.
   Since $\size{\oneb} = 0$, the term $M$ cannot reduce using Rule~$\redlab{R:Cons_2}$, which requires that the size of the bag is greater than 0. Instead, \revdaniele{we use the structural equivalence identity in \figref{fig:rPrecongruencefail}: 
   $\fail^{y,y,z}\esubst{\oneb}{x} \pequiv \fail^{y,y,z}$.}
 \end{exa}

\begin{figure}[!t]
    \[
    \begin{array}{rcl@{\hspace{0.5cm}}l}
      M \esubst{\oneb}{x} &\revdaniele{\pequiv}& M & \text{(if  $x \not \in \lfv{M}$)} 
    \\
    M B_1\esubst{B_2}{x} & \pequiv & (M\esubst{B_2}{x})B_1 & \text{(if $x \not \in \lfv{B_1}$)}
    \\
        M\esubst{B_1}{y}\esubst{B_2}{x} & \pequiv& (M\esubst{B_2}{x})\esubst{B_1}{y} & \revt{C12}{\text{(if  $ x\neq y, x \not \in \lfv{B_1}$ and $y \not \in \lfv{B_2}$)}}\\
         M \revdaniele{\pequiv} M' &\Rightarrow&  C[M] \revdaniele{\pequiv} C[M']\\
    \expr{M} \revdaniele{\pequiv} \expr{M}' &\Rightarrow&      D[\expr{M}]  \revdaniele{\pequiv} D[\expr{M}'] 
        \end{array}
    \]
\caption{Congruence in \lamrfail}
    \label{fig:rPrecongruencefail}
\end{figure}

 %We then have the following completeness and soundness results, whose proofs can be found in Appendix~\ref{app:compandsoundone}.

\begin{figure}[!t]
\begin{tikzpicture}[scale=.9pt]
%\draw[help lines] (0,0) grid (15,9);
% \draw[rounded corners, color=black] (0,0) rectangle (15,8);
\node (opcom) at (4.3,7.0){Operational Completeness};
\draw[rounded corners, color=brown!80!purple!90!black, fill=brown!80!purple!40] (0,5.2) rectangle (15.5,6.6);
\draw[rounded corners, color=teal, fill=teal!20] (0,1.4) rectangle (15.5,2.8);
\node (lamrfail) at (1.2,6) {$\lamrfail$:};
\node (expr1) [right of=lamrfail, xshift=.3cm] {$\mathbb{N}$};
\node (expr2) [right of=expr1, xshift=2cm] {$\mathbb{M} \pequiv \mathbb{M}'$};
\node (expr3) [right of=expr2, xshift=-0.5cm, yshift = -0.1cm] {};
\draw[arrow] (expr1) -- (expr2);
 \node at (5.0,5.6) {\redlab{R}};
\node (lamrsharfail) at (1.2,2) {$\lamrsharfail$:};
\node (transl1) [right of=lamrsharfail, xshift=.3cm] {$\recencodopenf{\mathbb{N}}$};
\node (transl2) [right of=transl1, xshift=2.5cm] {$\recencodopenf{\mathbb{M'}}$};
\draw[arrow, dotted] (transl1) -- node[anchor= south] {$*$ }(transl2);
%\node at (5.5,6.2) {\small $\pequiv$ };
\node (enc1) at (2,4) {$\recencodopenf{\cdot }$};
\node (refcomp) [right of=enc1, xshift=1.4cm]{Thm~\ref{l:app_completenessone}};
\node  at (7.0,4) {$\recencodopenf{\cdot }$};
\draw[arrow, dotted] (expr1) -- (transl1);
\draw[arrow, dotted] (expr3) -- (transl2);
%%%
\node (opcom) at (11,7.0){Operational Soundness};
\node (expr1shar) [right of=expr2, xshift=1.8cm] {$\mathbb{N}$}; % Moved this node to the right to align with the header as on page 51
\node (expr2shar) [right of=expr1shar, xshift=2.8cm] {$\mathbb{N'}\pequiv \mathbb{N}''$};
\node (expr3shar) [right of=expr2shar, xshift=-0.5cm, yshift = -0.1cm] {};
\draw[arrow, dotted] (expr1shar) -- (expr2shar);
\node at (12.5,5.6) {\redlab{R}};
\node (transl1shar) [right of=transl2, xshift=1.3cm] {$\recencodopenf{\mathbb{N}}$};
\node (transl2shar) [right of=transl1shar, xshift=1.25cm] {$\mathbb{L}$};
\node (transl3shar) [right of=transl2shar, xshift=1.05cm]{$\recencodopenf{\mathbb{N'}}$};
\draw[arrow,dotted] (expr3shar) -- (transl3shar);
\draw[arrow,dotted] (expr1shar) -- (transl1shar);
\node (enc2) at (8.5,4) {$\recencodopenf{\cdot }$};
\node (refsound) [right of=enc2, xshift=1.5cm]{Thm~\ref{l:soundnessone}};
\node  at (14.5,4) {$\recencodopenf{\cdot }$};
\draw[arrow] (transl1shar) -- (transl2shar);
\draw[arrow,dotted] (transl2shar) -- node[anchor= south] {*}(transl3shar);
%\node at (12.2,1.7) {\small $\pequiv$ };
\end{tikzpicture}
\caption{Operational correspondence for $\recencodopenf{\cdot }$ \label{f:opcf}}	
\end{figure}

\begin{restatable}[Operational Completeness]{thms}{appcompletenessone}
\label{l:app_completenessone}
Let $\expr{M}, \expr{N}$ be well-formed $\lamrfail$ expressions. 
Suppose $\expr{N}\red_{\redlab{R}} \expr{M}$.
\begin{enumerate}
\item If $\redlab{R} =  \redlab{R:Beta}$  then $ \recencodopenf{\expr{N}}  \red^{\leq 2}\recencodopenf{\expr{M}}$;

\item If $\redlab{R} = \redlab{R:Fetch}$   then   $ \recencodopenf{\expr{N}}  \red^+ \recencodopenf{\expr{M}'}$, for some $ \expr{M}'$ such that  $\expr{M} \pequiv \expr{M}'$. 
\item If $\redlab{R} \neq \redlab{R:Beta}$    and $\redlab{R} \neq \redlab{R:Fetch}$ then   $ \recencodopenf{\expr{N}}  \red \recencodopenf{\expr{M}}$.
\end{enumerate}
\end{restatable}

\begin{proof}[Proof (Sketch)]
By induction on the rules from Figure~\ref{fig:reductions_lamrfail} applied to infer $\expr{N}\red_{\redlab{R}} \expr{M}$. We analyse the reduction depending on whether $\redlab{R}$ is either $\redlab{R:Beta}$, or $\redlab{R:Fetch}$,  or neither. In the case the rule applied is $\redlab{Beta}$, then $\mathbb{N}=(\lambda x. M')\bag{B}$ and $\mathbb{M}=M'\esubst{B}{x}$.  When  applying the translation $\recencodopenf{\cdot}$ to $\mathbb{N}$ and $\mathbb{M}$ we obtain:
\begin{itemize} 
\item $\recencodopenf{\mathbb{N}}=((\lambda x.\recencodf{ M^{''}\langle{\widetilde{y}/x}\rangle}[\widetilde{y}\leftarrow x])\recencodf{B'})[\widetilde{x_1}\leftarrow x_1]\ldots [\widetilde{x_k}\leftarrow x_k]$
\item $\recencodopenf{\mathbb{M}}=\recencodf{M^{''}\esubst{B'}{x}} [\widetilde{x_1}\leftarrow x_1]\ldots [\widetilde{x_k}\leftarrow x_k]$
\end{itemize}
\revd{B17}{where $B'$ and $M^{''}$ stand for the renamings of $B$ and $M'$, respectively,} after sharing the multiple occurrences of their free/bound variables (\defref{def:enctolamrsharfail}). Note that 
\[\recencodopenf{\mathbb{N}}\red{}_{\redlab{RS:Beta}}(\recencodf{ M^{''} \langle{\widetilde{y}/x} \rangle} [\widetilde{y} \leftarrow x] \esubst{\recencodf{B'}}{x}) [\widetilde{x_1}\leftarrow x_1]\ldots [\widetilde{x_k}\leftarrow x_k]:=\mathbb{L},\] and according to rules in \figref{fig:share-reductfailure}, the remaining reduction depends upon the characteristics of the bag $\recencodf{B'}$:
\begin{enumerate}[(i)]
    \item $\size{\recencodf{B'}}=\#(x,M^{''})=k\geq 1$. 
    Then, 
    $\recencodopenf{\mathbb{N}}\red{}_{\redlab{RS:Beta}}\mathbb{L}\red_{\redlab{RS:ex\dash sub}}\recencodopenf{\mathbb{M}}$.
    \item Otherwise, $\mathbb{L}$ can be further expanded, the ``otherwise case'' of the translation of explicit substitutions,  such that 
    \[\recencodopenf{\mathbb{N}}\red{}_{\redlab{RS:Beta}}(\recencodf{ M^{''} \langle{\widetilde{y}/x} \rangle} [\widetilde{y} \leftarrow x] \esubst{\recencodf{B'}}{x}) [\widetilde{x_1}\leftarrow x_1]\ldots [\widetilde{x_k}\leftarrow x_k]=\mathbb{L}=\recencodopenf{\mathbb{M}}.\]
\end{enumerate}
In the case the rule applied is $\redlab{R:Fetch}$, the proof depends on the size $n$ of the bag. The interesting case is when the bag $B$  has only one component (i.e., $n=1$):  from $\mathbb{N}\red_{\redlab{F:Fetch}} \mathbb{N}$  we have that  $\mathbb{N}=M\esubst{\bag{N_1}}{x}$ and $\mathbb{M}=M\headlin{N_1/x}\esubst{1}{x}$. We need to use the congruence $\pequiv$ to obtain $\mathbb{M}=M\headlin{N_1/x}\esubst{1}{x}\pequiv M\headlin{N_1/x}:=\mathbb{M'}$ and then conclude that $\recencodopenf{\mathbb{N}}\red\recencodopenf{\mathbb{M'}}$. 
The analysis for the other cases is also done by inspecting the structure of expressions and bags. The full proof can be found in \appref{app:compandsoundone}.
\end{proof}

We establish soundness for a single reduction step. As  we discussed for completeness, the property generalizes to multiple steps.

\begin{restatable}[Operational Soundness]{thms}{soundnessone}
\label{l:soundnessone}
Let $\expr{N}$ be a well-formed $\lamrfail$ expression. 
Suppose $ \recencodopenf{\expr{N}}  \red \expr{L}$. Then, there exists $ \expr{N}' $ such that $ \expr{N}  \red_{\redlab{R}} \expr{N}'$ and 

\begin{enumerate}
 \item If $\redlab{R} = \redlab{R:Beta}$ then $\expr{ L } \red^{\leq 1} \recencodopenf{\expr{N}'}$;

   \item If $\redlab{R} \neq \redlab{R:Beta}$ then $\expr{ L } \red^*  \recencodopenf{\expr{N}''}$, for $ \expr{N}''$ such that  $\expr{N}' \pequiv \expr{N}''$.
\end{enumerate}
\end{restatable}

\begin{proof}[Proof (Sketch)]
By induction on the structure of $\mathbb{N}$ and inspecting the rules from \figref{fig:share-reductfailure} that can be applied in $\recencodopenf{\mathbb{N}}$. The interesting cases happen when $\mathbb{N}$ is either an application  $\mathbb{N}=(M\ B)$ or an explicit substitution $\mathbb{N}=M\esubst{B}{x})$. The former is reducible when $\mathbb{N}$ is an instance of $\redlab{R:Beta}$ or when $M=\fail^{\widetilde{x}}$ and $\mathbb{N}$ is an instance of $\redlab{R:Cons_1}$. The latter, for $\mathbb{N}=M\esubst{B}{x})$, the proof is split in several subcases depending whether: (i) size of the bag $\size{B}=\#(x,M)\geq 1$, and three possible reductions can take place $\redlab{RS:lin\dash fetch}$, $\redlab{RS:Cons_3}$ and $\redlab{RS:Cont}$, depending if $M$ is a failing term or not; (ii)~$\size{B}\neq \#(x,M)$ or $\size{B}=0$,  and the proof follows either applying Rule~$\redlab{RS:Fail}$ or  by induction hypothesis. The full proof can be found in \appref{app:compandsoundone}.
\end{proof}

\subsubsection{Success Sensitiveness}\hfill

We now consider success sensitiveness, a property that complements (and relies on) operational completeness and soundness. For the purposes of the proof, we consider the extension of $\lamrfail$ and $\lamrsharfail$ with dedicated constructs and predicates that specify success. 

\begin{defi}\label{def:ext_succ}
We extend the syntax of terms for $\lamrfail$ and $\lamrsharfail$ with the same $\checkmark$ construct. 
In both cases, we assume $\checkmark$ is well formed. 
Also, we define $\headf{\checkmark} = \checkmark$ and $\recencodf{\checkmark} = \checkmark$
\end{defi}

An expression $\expr{M}$ has success, denoted \succp{\expr{M}}{\checkmark}, when there is a sequence of reductions from \expr{M} that leads to an expression that includes a summand that contains an occurrence of $\checkmark$ in head position.

\begin{defi}[Success in \lamrfail and \lamrsharfail]
\label{def:app_Suc3}
In $\lamrfail$ and $\lamrsharfail $, we define 
\[\succp{\expr{M}}{\checkmark}\iff \exists M_1 , \cdots , M_k. ~\expr{M} \red^*  M_1 + \cdots + M_k \text{ and } \headf{M_j} = \checkmark,\]
for some  $j \in \{1, \ldots, k\}$.
\end{defi}

%We then have:
%\begin{nota}
% We write $\headfsum{\expr{M}} = \headf{M_i}$, for some $M_i \in \expr{M}$, when  $\forall M_i, M_j \in \expr{M}$, the equality  $\headf{M_i} = \head{M_j}$ holds.
%\end{nota}
\begin{defi}[Head of an expression]
We extend \defref{d:headshar} from terms to expressions as follows:
% We write $\headfsum{\expr{M}} \secondrev{ = \headf{M} }$ \secondrev{for $ M \in \expr{M}$} if, for every $M_i, M_j \in \expr{M}$, we have $\headf{M_i} = \headf{M_j}$. Otherwise, $\headfsum{\expr{M}}$ is undefined.
\secondrev{
 $$
    \headfsum{\expr{M}} = 
        \begin{cases}
            \headf{M_i} & \text{if } \headf{M_i} = \headf{M_j} \text{ for all } M_i,M_j \in \expr{M}\\
            \text{undefined} & \text{otherwise}
        \end{cases}
 $$
 }
\end{defi}

%  We then have the following results, whose proofs can be found in Appendix~\ref{app:sucessone}.

\begin{restatable}[Preservation of head term]{propo}{checkpres}
\label{Prop:checkpres}
The head of a term is preserved when applying the translation $\recencodopenf{\cdot}$, i.e., 
$$\forall M \in \lamrfail. ~~ \headf{M} = \checkmark \iff \headfsum{\recencodopenf{M}} = \checkmark.$$
\end{restatable}

 \begin{proof}[Proof (Sketch)]
By induction on the structure of $M$ considering the extension of the language established in~\defref{def:ext_succ}. See \appref{app:sucessone} for details.
 \end{proof}

\begin{restatable}[Success Sensitivity]{thms}{appsuccesssensce}
\label{proof:app_successsensce}
Let  \expr{M} be a well-formed $\lamrfail$-expression.
Then,
\[\expr{M} \Downarrow_{\checkmark}\iff \recencodopenf{\expr{M}} \Downarrow_{\checkmark}.\]
\end{restatable}

 \begin{proof}[Proof (Sketch)]
By induction on the structure of $\lamrfail$ and $\lamrsharfail$ expressions. The if-case  follows from operational soundness (Thm.~\ref{l:soundnessone}) by analyzing a reductions starting from $\recencodopenf{\mathbb{M}}$. Reciprocally, the only-if-case  follows by operational completeness (Thm.~\ref{l:app_completenessone}),  analyzing  reductions starting from $\mathbb{M}$. See \appref{app:sucessone} for details.
\end{proof}
We have the corollary below, which follows from Theorems~\ref{thm:preservencintolamrfail2}, 
\ref{l:app_completenessone}, 
\ref{l:soundnessone}, and 
\ref{proof:app_successsensce}:

\begin{cor}
%\label{proof:app_successsensce}
\label{cor:one}
Our translation  $ \recencodopenf{ \cdot } $ 
is a correct encoding, in the sense of \defref{d:encoding}.
\end{cor}

\subsection{From \texorpdfstring{$\lamrsharfail$}{λ̂^↯_⊕} to \texorpdfstring{$\spi$}{sπ}}
\label{ss:secondstep}\hfill

We now define our translation of $\lamrsharfail$ into $\spi$, denoted $\piencodf{\cdot}_u$, and establish its correctness. 
As usual in translations of $\lambda$ into $\pi$, we use a name $u$ to provide the behavior of the translated expression. 
In our case, $u$ is a non-deterministic session: the translated expression can be   available or not; this is signalled by prefixes `$u.\overline{\some}$'
and 
`$u.\overline{\none}$', respectively. 
Notice that every (free) variable $x$  in a $\lamrsharfail$-term $M$ becomes a name $x$ in its corresponding process $\piencodf{M}_u$ and is assigned an appropriate session type.

\subsubsection{An Auxiliary Translation}

%\begin{figure}[t!]
%    \centering
%\begin{tikzpicture}[scale=.6pt]
%%\draw[help lines] (0,-1) grid (13,6);
%\draw[rounded corners,color=teal] (0,0) rectangle (4,5);
%\draw[rounded corners,dashed, color=black, fill =gray!40, ] (1,1) rectangle (3,4);
%\node (lamrsharf) at (.5,4.5){$\lamrsharfail$};
%\node (lamrshar) at (1.5,3.5){$\lamrshar$};
%\node (fail) at (1.5,.5){$\fail$};
%%%
%\draw[rounded corners, color= violet!80!black] (8,0) rectangle (12,5);
%\draw[rounded corners, dashed, color=black, fill=black!20!purple!30] (9,1) rectangle (11,4);
%\node (spi1) at (8.5,4.5){$\spi$};
%\node (none) at (9.5,.5){$\none$};
%\draw[arrow, dashed] (3,4) -- node[anchor = south]{$\piencod{\cdot}$} (9,4);
%\draw[arrow, dashed] (3,1) -- node[anchor = north]{$\piencod{\cdot}$} (9,1);
%\draw[arrow] (3,5) to [bend left=15] node[anchor = south]{$\piencodf{\cdot}$} (9,5);
%\draw[arrow] (3,0) to [bend right=15]  node[anchor = north]{$\piencodf{\cdot}$} (9,0);
%\end{tikzpicture}
%    \caption{Summary}
%    \label{fig:sec5}
%\end{figure}

Before introducing $\piencodf{\cdot}_u$, we first discuss the translation $\piencod{\cdot}_u: \lamrshar \rightarrow \spi$, i.e., the translation in which the source language does not include failures. 
%(see \figref{fig:enc}).
This auxiliary translation, shown in \figref{fig:enc},  is given for pedagogical purposes: it  allows us to gradually discuss several key design decisions in $\piencodf{\cdot}_u$.

\begin{figure}[!t]
    \begin{align*}
    \piencod{ x }_u  &=   x.\overline{\some} ;[x \leftrightarrow u ]
    \\
    \piencod{ \lambda x . M[\widetilde{x} \leftarrow x]}_u    &=  u.\overline{\some};u(x).\piencod{ M[\widetilde{x} \leftarrow x] }_u 
    \\
        \piencod{ M\ B }_u   &=  \hspace*{-1em}\bigoplus_{\hspace*{1em}B_i \in \perm{B}} \hspace*{-1em}(\nu v)( \piencod{ M}_v \mid v.\some_{u , \lfv{B}} ; \outact{v}{x} . ( x.\some_{\lfv{B_i}}; \piencod{ B_i}_x \mid [v \leftrightarrow u] ) ) 
    \\
    \piencod{M [ x_1, \ldots, x_k \leftarrow x ] }_u   &=   x.\overline{\some};x(x_1). \cdots. x(x_k).x.\close; \encod{M}{u} 
    \\
    \piencod{ M[ \leftarrow x] }_u   &=  x.\overline{\some};x.\close; \encod{M}{u} 
    \\
    \piencod{\bag{M}\cdot B}_x   &=   \outact{x}{x_1}. (x_1.\some_{\lfv{B}};\encod{M}{x_1} \sep \encod{B}{x} )  
    \\
    \piencod{\oneb }_x   &=   x. \overline{\close} 
    \\
    \piencod{ M[\widetilde{x} \leftarrow x] \esubst{ B }{ x} }_u   &=  \bigoplus_{B_i \in \perm{B}} (\nu x)( \piencod{ M[\widetilde{x} \leftarrow x]}_u \mid x.\some_{\lfv{B_i}}; \piencod{ B_i}_x )  
    \\
    \piencod{ M \linexsub{N / x}  }_u   &=  (\nu x) ( \piencod{ M }_u \mid   x.\some_{\lfv{N}};\piencod{ N }_x  )  
    \\
    \piencod{\expr{M}+\expr{N} }_u    &=  \piencod{ \expr{M} }_u \oplus \piencod{ \expr{N} }_u  
\end{align*}
    \caption{An auxiliary translation of \lamrshar into \spi, without failures}
    \label{fig:enc}
\end{figure}
We describe each case of the translation $\piencod{\cdot}_u$, focusing on the r\^{o}le of non-deterministic sessions (expressed using prefixes `$x.\overline{\some}$' and `$x.\some_{(w_1, \cdots, w_n)}$' in \spi):
\begin{itemize}
    \item $\piencod{ x }_u$: Because sessions are non-deterministically available, the translation first confirms that the behavior along $x$ is available; subsequently, the forwarder process induces a substitution $\subst{x}{u}$.

    \item $\piencod{ \lambda x . M[\widetilde{x} \leftarrow x]}_u$: As in the case of variables, the translation first confirms the behavior along $u$ before receiving a name, which will be used in the translation of $M[\widetilde{x} \leftarrow x]$, discussed next.

        \item $\piencod{ M\ B }_u$: This process  models the application of $M$ to bag $B$ as a non-deterministic choice in the order in which the elements of $B$ are substituted into $M$. Substituting each $B_i$ involves a protocol in which the translation of a term  $\lambda x . M'[\widetilde{x} \leftarrow x]$ within $M$ confirms its own availability, before and after the exchange of the name $x$, on which the translation of $B_i$ is spawned. This protocol uses the fact that $M\ B$ does not reduce to failure, i.e., there is no lack or excess of resources in $B$.

    \item $\piencod{M [ x_1, \ldots, x_k \leftarrow x ] }_u$: The translation first confirms the availability of the behavior along $x$. 
    Then, it receives along $x$ a name for each $x_i$: these received names will be used to synchronize with the translation of bags (see below). Subsequently, the protocol on $x$ safely terminates and the translation of $M$ is executed.
        
    \item $ \piencod{ M[ \leftarrow x] }_u$: When there are no variables to be shared with $x$, the translation simply confirms the behavior on $x$, close the protocol immediately after, and executes the translation of $M$.

    \item $\piencod{\bag{M}\cdot B}_x$: The translation of a non-empty bag essentially makes each element available in its corresponding order. 
    This way, for the first element $M$ a name $x_1$ is sent over $x$; the translation of ${M[x_1,\cdots , x_n \leftarrow x]}$, discussed above, must send a confirmation on $x_1$ before the translation of $M$ is executed. After these exchanges, the translation of the rest of the bag is spawned.

    \item $\piencod{\oneb }_x$: In line with the previous item, the translation of the empty bag simply closes the name $x$; this signals that there are no (further) elements in the bag and that all synchronizations are complete.

    \item $\piencod{ M[\widetilde{x} \leftarrow x] \esubst{ B }{ x} }_u$: In this case, the translation is a sum involving the parallel composition of (i)~the translation of each element $B_i$ in the bag and (ii)~the translation of $M$. Observe that a fresh name $x$ is created to enable synchronization between these two processes. Also, as in previous cases, notice how the translation of $B_i$ must first confirm its availability along $x$.
    
    \item $\piencod{ M \linexsub{N / x}  }_u$: This translation essentially executes the translations of $M$ and $N$ in parallel, with a caveat: the translation of $N$ depends on the availability of a behavior along $x$, to be produced within the translation of $M$. 
    
    \item $\piencod{\expr{M}+\expr{N} }_u$: This translation  homomorphically preserves the non-determinism between $M$ and $N$. 
    
\end{itemize}

\begin{exa}
\label{ex:failfreered}
Consider the \lamrshar-term
$M_0 = ( \lambda x. M [x_1, x_2 \leftarrow x]) \bag{N_1,N_2}$. 
Writing $\lfv{B}$ to denote the free variables in $N_1$ and $N_2$, the process 
$\piencod{M_0}_u$  is as follows:
\begin{align*}
  \piencod{M_0}_u=  & \piencod{ ( \lambda x. M [x_1, x_2 \leftarrow x] ) \bag{N_1,N_2} }_u 
    \\
     = & (\nu v)( \piencod{ \lambda x. M [x_1, x_2 \leftarrow x] }_v \mid 
      \underbrace{v.\some_{u , \lfv{B}} ; \outact{v}{x} . ( x.\some_{\lfv{B}}; \piencod{ \bag{N_1,N_2}  }_x \mid [v \leftrightarrow u] )}_{P_1} ) 
     \\
    & \oplus
    \\ 
    & (\nu v)( \piencod{ \lambda x. M [x_1, x_2 \leftarrow x] }_v \mid 
    \underbrace{v.\some_{u , \lfv{B}} ; \outact{v}{x} . ( x.\some_{\lfv{B}}; \piencod{ \bag{N_2,N_1}  }_x \mid [v \leftrightarrow u] )}_{P_2} ) 
    \\
    =~ & (\nu v)( v.\overline{\some};v(x).x.\overline{\some};x(x_1). x(x_2).x.\close; \piencod{M}{v}  \mid P_1)
    \\
    & \oplus 
    \\
    & (\nu v)( v.\overline{\some};v(x).x.\overline{\some};x(x_1). x(x_2).x.\close; \piencod{M}{v}   \mid P_2)\\
\end{align*}
%Here the application and the bag   communicate over a channel $v$ which communicates the channel name for substitution of the bag, this channel name being the variable in the abstraction term.
The translation immediately opens up a non-deterministic choice with two alternatives, corresponding to the bag of size 2.
Because of non-collapsing non-determinism, after some reductions, this amounts to accounting for the two different orders in which $N_1$ and $N_2$ can be extracted from the bag.

\[
\begin{aligned}
 \revd{B19}{
 \piencod{M_0}_u} \red^*~ & (\nu  x)( x(x_1). x(x_2).x.\close; \revd{B19}{\piencod{M}_{u}}  \mid   \piencod{ \bag{N_1,N_2}  }_x  ) 
    \\
    & \oplus \\
    & (\nu  x)( x(x_1). x(x_2).x.\close; \piencod{M}_{u}   \mid   \piencod{ \bag{N_2,N_1}  }_x  )
    \end{aligned}
\]
We show further reductions for one of the processes, which we will denote $R$,  for $R=(\nu  x)( x(x_1). x(x_2).x.\close; \piencod{M}_{u}  \mid   \piencod{ \bag{N_1,N_2}  }_x  )$, in the resulting sum (reductions for the other process are similar):
\begin{align*}
 R= &  (\nu  x)( x(x_1). x(x_2).x.\close; \piencod{M}_{u}  \mid   \piencod{ \bag{N_1,N_2}  }_x  )
  \\
    =~ & 
   (\nu  x)( x(x_1). x(x_2).x.\close; \piencod{M}_{u}  \mid   \outact{x}{x_1}. (x_1.\some_{\lfv{N_1}};\piencod{N_1}_{x_1} \sep 
   \\
   & \qquad \qquad \qquad \qquad \outact{x}{x_2}. (x_2.\some_{\lfv{N_2}};\piencod{N_2}_{x_2} \sep x. \overline{\close} ) )     )\\
   \red^*~ & (\nu  x_1,x_2)( \piencod{M}_{u}  \mid x_1.\some_{\lfv{N_1}};\piencod{N_1}_{x_1} \sep 
    x_2.\some_{\lfv{N_2}};\piencod{N_2}_{x_2}  )
\end{align*}
\end{exa}

\subsubsection{The Translation}\label{ss:second_trans}
\hfill 

The translation $\piencod{\cdot}_{x}$ leverages non-deterministic sessions in \spi to give a concurrent interpretation of \lamrshar, the non-deterministic (but fail-free) sub-calculus of \lamrsharfail. In a nutshell, non-deterministic sessions entail the explicit confirmation of the availability of  a name's behavior, via synchronizations of a prefix `$x.\some_{(w_1, \cdots, w_n)}$'   with a corresponding prefix `$x.\overline{\some}$'. Clearly, $\piencod{\cdot}_{x}$ under-utilizes the expressivity of \spi: in processes resulting from $\piencod{\cdot}_{x}$, no prefix `$x.\some_{(w_1, \cdots, w_n)}$' will ever synchronize with a prefix `$x.\overline{\none}$'. 
Indeed, because terms in \lamrshar never reduce to failure,  $\piencod{\cdot}_{x}$ should not account for such failures.

We may now introduce $\piencodf{\cdot }_u$, our translation of the fail-prone calculus $\lamrsharfail$ into $\spi$. It builds upon the structure of $\piencod{\cdot}_{x}$ to account for failures in expressions due to the lack or excess of resources. To this end, as we will see, $\piencodf{\cdot }_u$ does exploit prefixes `$x.\overline{\none}$' to signal failures.

\paragraph{Translating Expressions}
We introduce the translation  $\piencodf{\cdot }_u$, which will be shown to be a correct encoding, according to the criteria given in \secref{ss:criteria}.

\begin{defi}[From $\lamrsharfail$ into $\spi$: Expressions]
\label{def:enc_lamrsharpifail}
Let $u$ be a name.
The translation $\piencodf{\cdot }_u: \lamrsharfail \rightarrow \spi$ is defined in \figref{fig:encfail}.
\end{defi}

%We discuss key intuitions of the translation in \figref{fig:encfail}.

We discuss the most interesting aspects of the translation in \figref{fig:encfail}, in particular how the possibility of failure (lack or excess of resources in bags) induces differences with respect to the translation in \figref{fig:enc}.

\begin{figure}[t!]

{%\small 
\begin{align*}	
   \piencodf{x}_u & = x.\overline{\some} ; [x \leftrightarrow u]  
  \\[1mm] 
   \piencodf{\lambda x.M[\widetilde{x} \leftarrow x]}_u & = u.\overline{\some}; u(x).\piencodf{M[\widetilde{x} \leftarrow x]}_u
%\end{array}
\\[1mm]
  \piencodf{M\, B}_u & = \bigoplus_{B_i \in \perm{B}} (\nu v)(\piencodf{M}_v \mid v.\some_{u , \lfv{B}} ; \outact{v}{x} . ([v \leftrightarrow u] \mid \piencodf{B_i}_x ) )  
     \\[1mm]
      \piencodf{ M[\widetilde{x} \leftarrow x] \esubst{ B }{ x} }_u & =  \bigoplus_{B_i \in \perm{B}} (\nu x)( \piencodf{ M[\widetilde{x} \leftarrow x]}_u \mid \piencodf{ B_i}_x )  
      \\[1mm]
    \piencodf{M[x_1, x_2 \leftarrow x]}_u & = 
       x.\overline{\some}. \outact{x}{y_1}. \Big(y_1 . \some_{\emptyset} ;y_{1}.\close 
       \!\mid\! x.\overline{\some};x.\some_{u, (\lfv{M} \setminus \{x_1 ,  x_2\} )};x(x_1) . \\%[1mm]
       & \hspace{2.0em} . x.\overline{\some}. \outact{x}{y_2} . \big(y_2 . \some_{\emptyset} ; y_{2}.\close  \!\mid\! x.\overline{\some};x.\some_{u,( \lfv{M} \setminus \{x_2\} ) };x(x_2)
      \\%[1mm]
       & \hspace{2.8em} . x.\overline{\some}; \outact{x}{y_{}}. ( y_{} . \some_{u, \lfv{M} } ;y_{}.\close; \piencodf{M}_u \mid x.\overline{\none} )~ \big)  \Big) \\[1mm]
    \piencodf{M[ \leftarrow x]}_u & = x. \overline{\some}. \outact{x}{y} . ( y . \some_{u,\lfv{M}} ;y_{}.\close; \piencodf{M}_u \mid x. \overline{\none}) \\[1mm]
    \piencodf{\bag{M} \cdot B}_x & =
    \begin{array}[t]{l}
       x.\some_{\lfv{\bag{M} \cdot B} } ; x(y_i). x.\some_{y_i, \lfv{\bag{M} \cdot B}};x.\overline{\some} ; \outact{x}{x_i}
       \\[1mm]
       \qquad . (x_i.\some_{\lfv{M}} ; \piencodf{M}_{x_i} \mid \piencodf{B}_x \mid y_i. \overline{\none})
   \end{array} 
  % \medskip
   \\[3mm]    
      \piencodf{\oneb}_x & = x.\some_{\emptyset} ; x(y). (y.\overline{\some};y. \overline{\close} \mid x.\some_{\emptyset} ; x. \overline{\none}) 
      \\[1mm]
      \piencodf{\fail^{x_1 , \cdots , x_k}}_u & = u.\overline{\none} \mid x_1.\overline{\none} \mid \cdots \mid x_k.\overline{\none} 
   \\[1mm]
           \piencodf{ M \linexsub{N / x}  }_u   & =   (\nu x) ( \piencodf{ M }_u \mid   x.\some_{\lfv{N}};\piencodf{ N }_x  )  
        \\[1mm]
     \piencodf{\expr{M}+\expr{N} }_u    & =  \piencodf{ \expr{M} }_u \oplus \piencodf{ \expr{N} }_u   
%\end{array}
%\]
\end{align*}
}
    \caption{Translating \lamrsharfail expressions into \spi processes.}
    \label{fig:encfail}
\end{figure}

Most salient differences can be explained by looking at the translation of the application $M \, B$. Indeed, the sources of failure in $\lamrsharfail$ concern a mismatch between the number of variable occurrences in $M$ and the number of resources present in $B$. Both $M$ and $B$ can fail on their own, and our translation into \spi must capture this mutual dependency. Let us recall the translation given in  \figref{fig:enc}:
$$
    \piencod{ M\ B }_u   =  \bigoplus_{B_i \in \perm{B}} (\nu v)( \piencod{ M}_v \mid v.\some_{u , \lfv{B}} ; \outact{v}{x} . ( x.\some_{\lfv{B_i}}; \piencod{ B_i}_x \mid [v \leftrightarrow u] ) ) 
$$
The corresponding translation in \figref{fig:encfail} is seemingly simpler:
$$
  \piencodf{M\, B}_u  = \bigoplus_{B_i \in \perm{B}} (\nu v)(\piencodf{M}_v \mid v.\some_{u , \lfv{B}} ; \outact{v}{x} . ([v \leftrightarrow u] \mid \piencodf{B_i}_x ) )  
$$
Indeed, the main difference is the prefix `$x.\some_{\lfv{B_i}}$', which is present in process $\piencod{ M\ B }_u$ but is not explicit in process $\piencodf{M\, B}_u$. Intuitively, such a prefix denotes the dependency of $B$ on $M$; because terms in $\lamrshar$ do not fail, we can be certain that a corresponding confirming prefix `$x.\overline{\some}$' will be available to spawn every $\piencod{B_i}_x$. When moving to $\lamrsharfail$, however, this is not the case:  $\piencodf{M}_v$ may fail to provide the expected number of corresponding confirmations. %In that case, one or more $\piencodf{B_i}_x$ will fail 
For this reason, the role of prefix `$x.\some_{\lfv{B_i}}$' in $\piencod{ M\ B }_u$ is implemented within process $\piencodf{B_i}_x$. As a consequence, the translations for sharing terms ($M[\widetilde{x} \leftarrow x]$ and $ M[ \leftarrow x]$) and for bags ($\bag{M}\cdot B$ and $\oneb$) are more involved in the case of failure.

With this motivation for $\piencodf{M\, B}_u$ in mind, we discuss the remaining entries in \figref{fig:encfail}:

\begin{itemize}
\item Translations for $x$ and $\lambda x.M[\widetilde{x} \leftarrow x]$ are exactly as in \figref{fig:enc}:
$$
    \piencodf{ x }_u  =   x.\overline{\some} ;[x \leftrightarrow u ]
    \qquad \qquad 
    \piencodf{ \lambda x . M[\widetilde{x} \leftarrow x]}_u    =  u.\overline{\some};u(x).\piencodf{ M[\widetilde{x} \leftarrow x] }_u 
$$

    \item Similarly as $\piencodf{M \, B}_u$, discussed above, the translation of $M[\widetilde{x} \leftarrow x] \esubst{ B }{ x}$ is more compact than the one in \figref{fig:enc}, because confirmations for each of the elements of the bag are handled within their respective translations:
    $$      \piencodf{ M[\widetilde{x} \leftarrow x] \esubst{ B }{ x} }_u  =  \bigoplus_{B_i \in \perm{B}} (\nu x)( \piencodf{ M[\widetilde{x} \leftarrow x]}_u \mid \piencodf{ B_i}_x )  
$$
    
    %  : These two items change similarly and as such are grouped. As this translation also introduces failure then the behavior of the bag is no longer as predictable, the bag can introduce failure due to there bring an excess or lack of resources. We move the confirmation of the bags behavior to within the translation of the bag.
    
       \item As anticipated, the translation of $M[x_1, \ldots, x_k \leftarrow x]$ is more involved than before. For simplicity, let us discuss the representative case when $k = 2$ (two shared variables):
       \begin{align*}
       	          \piencodf{M[x_1, x_2 \leftarrow x]}_u & = 
            x.\overline{\some}. \outact{x}{y_1}. \Big(y_1 . \some_{\emptyset} ;y_{1}.\close 
       \!\mid\! x.\overline{\some};x.\some_{u, (\lfv{M} \setminus \{x_1 ,  x_2\} )};x(x_1) . 
      \\%[1mm]
      & \hspace{2.0em}  x.\overline{\some}. \outact{x}{y_2} . \big(y_2 . \some_{\emptyset} ; y_{2}.\close  \!\mid\! x.\overline{\some};x.\some_{u,( \lfv{M} \setminus \{x_2\} ) };x(x_2).
      \\%[1mm]
    %   \hspace{2.0em} {} \mid x.\overline{\some};x.\some_{u,( \lfv{M} \setminus x_n ) };x(x_n) \\
      & \hspace{2.8em}  x.\overline{\some}; \outact{x}{y_{}}. ( y_{} . \some_{u, \lfv{M} } ;y_{}.\close; \piencodf{M}_u \mid x.\overline{\none} )~ \big)  \Big)
       \end{align*}
       This process is meant to synchronize with the translation of a bag. After confirming the presence of a behavior on name $x$, an auxiliary name $y_i$ is sent to signal that there are elements to be substituted. This name implements a short protocol that allows us to check for lack of resources in the bag. These steps on $y_i$ are followed by another confirmation and also a request for confirmation of behavior along $x$; this represents that the name can fail in one of two ways, capturing the mutual dependency between $M$ and the bag mentioned above. Once these two steps on $x$ have succeeded, it is finally safe for the process to receive a name $x_i$. This process is repeated for each shared variable to ensure safe communication of the elements of the bag. The last line shows the very final step: a name $y$ is communicated to ensure that there are no further elements in the bag; in such a case, $y$ fails and the failure is propagated to  $\piencodf{M}_u$. The prefix `$x.\overline{\none}$' signals the end of the shared variables, and is meant to synchronize with the translation of $\oneb$, the last element of the bag. If the bag has elements that still need to be synchronized then the failure along $x$ is propagated to the remaining resources within the translation of the  bag.
    
    \item The translation of $M[ \leftarrow x]$  corresponds to the final step in the translation  just discussed:
    $$        \piencodf{M[ \leftarrow x]}_u  = x. \overline{\some}. \outact{x}{y} . ( y . \some_{u,\lfv{M}} ;y_{}.\close; \piencodf{M}_u \mid x. \overline{\none})
$$

    %    \item Notice how we no longer simply confirm behavior along  name $x$ and close. We first send a name $y_i$ which we use as a channel to track an excess of resources in a bag that the term synchronizes. The confirmation of the behavior along $y_i$ guards the success of the the translation of $M$ as failure along $y_i$ collapses the process. If there is an excess of resources then $y_i$ fails and hence correctly collapses the translation term $M$ to its failed state. The channel $x$ no longer closes but instead fails to provide the correct behavior, thus signalling along the channel that the translated term has no more shared variables to synchronize and that any more processes means an excess of resources.

    \item The translation of the non-empty bag $\bag{M} \cdot B$ is as follows: 
    \begin{align*}
    	         \piencodf{\bag{M} \cdot B}_x & =
%         \begin{array}[t]{l}
       x.\some_{\lfv{\bag{M} \cdot B} } ; x(y_i). x.\some_{y_i, \lfv{\bag{M} \cdot B}};x.\overline{\some} ; \outact{x}{x_i}
       \\[1mm]
       & \qquad . (x_i.\some_{\lfv{M}} ; \piencodf{M}_{x_i} \mid \piencodf{B}_x \mid y_i. \overline{\none})
    \end{align*}
    Notice how this process operates hand in hand with the translation of $M[x_1, \ldots, x_k \leftarrow x]$.
    The process first waits for its behavior to be confirmed; then, the auxiliary name $y_i$ is received from the translation of $M[x_1, \ldots, x_k \leftarrow x]$. The name $y_i$ fails immediately to signal that there are more resources in the bag. Name $x$ then confirms its behavior and awaits its behavior to be confirmed. Subsequently, a name $x_i$ is sent: this is the name on which the translation of $M$ will be made available to the application. After that, name $x$ is used in the translation of $B$, the rest of the bag.
%    When a bag synchronizes it waits for its behavior to be confirmed, followed by receiving the auxiliary channel $y_i$. The dummy channel $y_i$ fails immediately signifying that there are more resources in the bag to be sent. The channel $x$ proceeds to confirm and awaits its behavior to be confirmed. Finally we send a fresh channel for communicating the first element of the bag to the translated term we are synchronizing with. The channel $x$ then continues its behavior in the translation of the rest of the bag $B$.
    
    \item The translation of $\oneb$ operates aligned with the translations just discussed, exploiting the fact that in fail-free reductions the last element of the bag must be $\oneb$:
    $$
          \piencodf{\oneb}_x  = x.\some_{\emptyset} ; x(y). (y.\overline{\some};y. \overline{\close} \mid x.\some_{\emptyset} ; x. \overline{\none}) 
    $$
    This process relays the information that the translated empty bag is no longer able to provide resources for further substitutions. It first waits upon a correct behavior followed by the reception of a name $y$. The process then confirms its behavior along $y$: this signals that there are no further resources. Concurrently, name $x$ waits for a confirmation of a behavior and ends with `$x. \overline{\none}$', thus signaling the failure of producing further behaviors.
    
    \item The explicit failure term $\fail^{x_1 , \cdots , x_k}$ is not part of $\lamrshar$ and so it was not covered in  \figref{fig:enc}. Its translation is straightforward:
    $$      \piencodf{\fail^{x_1 , \cdots , x_k}}_u  = u.\overline{\none} \mid x_1.\overline{\none} \mid \cdots \mid x_k.\overline{\none} 
$$
The failure term is translated as the non-availability of a behavior along name $u$, composed with the non-availability of sessions along the names/variables $x_1, \ldots, x_n$ encapsulated by the source failure term.

%As expected, failure is expressed as the failure to produce a behavior along name $u$, along with failures on all variables encapsulated in failures \daniele{(in nones?)}.
%\daniele{better, but too many "along's"}

\item The translations for $M \linexsub{N / x} $ and $\expr{M}+\expr{N}$ are exactly as before:
$$
           \piencodf{ M \linexsub{N / x}  }_u    =   (\nu x) ( \piencodf{ M }_u \mid   x.\some_{\lfv{N}};\piencodf{ N }_x  )  
        \qquad 
     \piencodf{\expr{M}+\expr{N} }_u     =  \piencodf{ \expr{M} }_u \oplus \piencodf{ \expr{N} }_u   
$$
\end{itemize}

%\smallskip
%\newcommand{\cnum}[1]{\raisebox{.5pt}{\textcircled{\raisebox{-.9pt} {{\footnotesize \textsf{#1}}}}}}
\newcommand{\cnum}[1]{[\mathsf{#1}]}

\subsubsection{Examples} Before presenting the session types associated to our translation $ \piencodf{\cdot}_u$, we present a series of examples that illustrate different possibilities in a step-by-step fashion:
\begin{itemize}
    \item No failure: an explicit substitution that is provided an adequate amount of resources;
    \item Failure due to excess of resources in the bag;
    \item Failure due to lack of resources in the bag.
\end{itemize}
We first discuss the translation of a term in which there is no failure. In that follows, we refer to a specific reduction by adding a number as in, e.g., `$\red_{\cnum{3}}$'.

\begin{exa}[No Failure]\label{ex:encsucc}
Let us consider the well-formed $\lamrsharfail$-term $N [x_1 \leftarrow x] \esubst{ \bag{M} }{x}$, where, for simplicity, we assume that $\revo{A12}{\lfv{N} \setminus \{ x_1 \} = \lfv{M} = \emptyset}$. 
As we have seen, $N [x_1 \leftarrow x] \esubst{ \bag{M} }{x} \red N \linexsub{M / x}$. 
We discuss reduction steps for $ \piencodf{ N[ x_1 \leftarrow x] \esubst{ \bag{M} }{ x} }_u $, highlighting in \blue{blue} relevant prefixes.
First, we have:
  \[
  \small
  \begin{aligned} 
   \piencodf{ N[ x_1 \leftarrow x] \esubst{ \bag{M} }{ x} }_u 
               =~ &  (\nu x)( \piencodf{ N[x_1 \leftarrow x]}_u \mid \piencodf{\bag{M} }_x ) \\
       =~ &  (\nu x)\big( \blue{x.\overline{\some}}. \outact{x}{y_1}. (y_1 . \some_{\emptyset} ;y_{1}.\close 
       \mid x.\overline{\some};x.\some_{u}; 
       \\
       & \hspace{2.8em} . x(x_1).x.\overline{\some}; \outact{x}{y_{}}. ( y_{} . \some_{u, x_1 } ;y_{}.\close; \piencodf{N}_u \mid x.\overline{\none} )~ ) \\
       & \hspace{2em} \mid \blue{x.\some_{\emptyset}} ; x(y_1). x.\some_{y_1};x.\overline{\some} ; \outact{x}{x_1}\\
       & \hspace{2.8em} . (x_1.\some_{\emptyset} ; \piencodf{M}_{x_1} \mid y_1. \overline{\none} \mid x.\some_{\emptyset} ; x(y).    (y.\overline{\some};y. \overline{\close} \\
       &\hspace{2.8em} \mid x.\some_{\emptyset} ; x. \overline{\none}) ) \big) 
       \end{aligned}
       \]
%Its  reductions are shown in \figref{fig:encsucc}.
A detailed description of the reduction steps follows:
\begin{itemize}
    \item Reduction $\red_{\cnum{1}}$ concerns the name $x$ confirming its behavior (see highlighted prefixes above), and reduction $\red_{\cnum{2}}$ concerns the communication of name $y_1$:
    \[
    \small
    \begin{aligned}
     \piencodf{ N[ x_1 \leftarrow x] \esubst{ \bag{M} }{ x} }_u\red_{\cnum{1}}~ &  (\nu x)( \blue{\outact{x}{y_1}}. \big(y_1 . \some_{\emptyset} ;y_{1}.\close 
       \mid x.\overline{\some};x.\some_{u};x(x_1) . 
       \\
       & \hspace{1em} . x.\overline{\some}; \outact{x}{y_{}}. ( y_{} . \some_{u,x_1} ;y_{}.\close; \piencodf{N}_u \mid x.\overline{\none} )~ \big) \\
       & \hspace{1em} \mid \blue{ x(y_1)}. x.\some_{y_1 };x.\overline{\some} ; \outact{x}{x_1}. (x_1.\some_{\emptyset} ; \piencodf{M}_{x_1}\\
       & \hspace{1em}  \mid y_1. \overline{\none} \mid x.\some_{\emptyset} ; x(y).  (y.\overline{\some};y. \overline{\close}\mid x.\some_{\emptyset} ; x. \overline{\none}) ) ) 
       \\
       \red_{\cnum{2}}~ &  (\nu x, y_1)( y_1 . \some_{\emptyset} ;y_{1}.\close 
       \mid \blue{x.\overline{\some}};x.\some_{u };x(x_1) . 
       \\
       & \hspace{1em} . x.\overline{\some}; \outact{x}{y_{}}. ( y_{} . \some_{u, x_1 } ;y_{}.\close; \piencodf{N}_u \mid x.\overline{\none} )~ \\
       & \hspace{1em} \mid \blue{ x.\some_{y_1} };x.\overline{\some} ; \outact{x}{x_1} . (x_1.\some_{\emptyset} ; \piencodf{M}_{x_1} \\
       & \hspace{1em}\mid\! y_1. \overline{\none} \!\mid\! x.\some_{\emptyset} ; x(y).    (y.\overline{\some};y. \overline{\close} \!\mid\! x.\some_{\emptyset} ; x. \overline{\none}) ) ) \ (:=P)
    \end{aligned}
    \]
    \item Reduction $\red_{\cnum{3}}$ concerns $x$ confirming its behavior, which signals that there are variables free for substitution in the translated term. In the opposite direction, 
    reduction $\red_{\cnum{4}}$ signals that there are elements in the bag which are available for substitution in the translated term.
    \[
    \small
    \begin{aligned}
     P \red_{\cnum{3}}~ &  (\nu x, y_1)( y_1 . \some_{\emptyset} ;y_{1}.\close 
       \blue{\mid x.\some_{u}};x(x_1) . 
       \\
       & \hspace{1em} . x.\overline{\some}; \outact{x}{y_{}}. ( y_{} . \some_{u, x_1 } ;y_{}.\close; \piencodf{N}_u \mid x.\overline{\none} )~ \\
       & \hspace{1em} \mid \blue{x.\overline{\some}} ; \outact{x}{x_1} . (x_1.\some_{\emptyset} ; \piencodf{M}_{x_1} \mid y_1. \overline{\none} \mid x.\some_{\emptyset} ; x(y). \\
       & \hspace{1em} (y.\overline{\some};y. \overline{\close} \mid x.\some_{\emptyset} ; x. \overline{\none}) ) ) 
       \\
     \red_{\cnum{4}}~ &  (\nu x, y_1)( y_1 . \some_{\emptyset} ;y_{1}.\close 
       \para  \blue{x(x_1)} .  x.\overline{\some}; \outact{x}{y_{}}. ( y_{} . \some_{u, x_1 } ;y_{}.\close; \\
       & \hspace{1em} \piencodf{N}_u \mid x.\overline{\none} ) \mid  \blue{\outact{x}{x_1}}. (x_1.\some_{\emptyset} ; \piencodf{M}_{x_1}\mid y_1. \overline{\none} \\
       & \hspace{2em}   \mid x.\some_{\emptyset} ; x(y).   (y.\overline{\some};y. \overline{\close} \mid x.\some_{\emptyset} ; x. \overline{\none}) ) ) \hspace{1cm}  (:= Q)
    \end{aligned}
    \]
    \item Given the confirmations in the previous two steps, reduction $\red_{\cnum{5}}$ can now safely communicate a name $x_1$. This reduction synchronizes the shared variable $x_1$ with the first element in the bag.
    \[
    \small
    \begin{aligned}
    Q   \red_{\cnum{5}}~ &  (\nu x, y_1, x_1)( y_1 . \some_{\emptyset} ;y_{1}.\close   \mid  \blue{ x.\overline{\some}}; \outact{x}{y_{}}. ( y_{} . \some_{u, x_1 } ; y_{}.\close;  \piencodf{N}_u \\
       & \hspace{2em}\mid x.\overline{\none} ) \mid  x_1.\some_{\emptyset} ; \piencodf{M}_{x_1} \mid y_1. \overline{\none} \mid \blue{x.\some_{\emptyset}} ; x(y).   (y.\overline{\some};y. \overline{\close} \\
       & \hspace{2.5em} \mid x.\some_{\emptyset} ; x. \overline{\none}) ) \hspace{1cm}  (:= R)
    \end{aligned}
    \]
    \item Reduction $\red_{\cnum{6}}$ concerns  $x$ confirming its behavior. At this point, we could have alternatively performed a reduction on name $y_1$. We chose to discuss all reductions on  $x$ first; thanks to confluence this choice has no effect on the overall behavior.  Reduction $\red_{\cnum{7}}$ communicates name $y$ along $x$. 
    \[\small
    \begin{aligned}
        R  \red_{\cnum{6}}~ &  (\nu x, y_1, x_1)( y_1 . \some_{\emptyset} ;y_{1}.\close 
       \mid \blue{ \outact{x}{y_{}}}. ( y_{} . \some_{u, x_1 } ;y_{}.\close; \piencodf{N}_u ~ \\
       & \hspace{2em} \mid x.\overline{\none} )\mid  x_1.\some_{\emptyset} ;\piencodf{M}_{x_1} \mid y_1. \overline{\none} \mid \blue{x(y)}. (y.\overline{\some};y. \overline{\close}  
       \\
       & \hspace{2em}\mid x.\some_{\emptyset} ; x. \overline{\none}) )\\
       \red_{\cnum{7}}~ &  (\nu x, y, y_1, x_1)( y_1 . \some_{\emptyset} ;y_{1}.\close 
       \mid  y_{} . \some_{u, x_1 } ;y_{}.\close; \piencodf{N}_u \mid \blue{x.\overline{\none}} ~ \\
       & \hspace{2em} \mid  x_1.\some_{\emptyset} ; \piencodf{M}_{x_1} \mid y_1. \overline{\none} \mid  y.\overline{\some};y. \overline{\close} \mid \blue{x.\some_{\emptyset}} ; x. \overline{\none} ) \hspace{.8cm} (:= S)
       \\
    \end{aligned}
    \]
    \item Reduction $\red_{\cnum{8}}$ cancels the behavior along $x$, meaning that there are no more free variables to synchronize with. Subsequently, reduction $\red_{\cnum{9}}$ cancels the behavior along $y_1$: \revd{B22}{at the beginning, when $y_1$ was received, the encoded bag had the element $M$ left to be synchronized; at this point, the failure on $y_1$ signals that the bag still has elements to be synchronized with}. 
    \[
    \small
    \begin{aligned}
    S %
       \red_{\cnum{8}}~ &   (\nu  y, y_1, x_1)(\blue{ y_1 . \some_{\emptyset}} ;y_{1}.\close 
       \mid  y_{} . \some_{u, x_1 } ;y_{}.\close; \piencodf{N}_u \mid  x_1.\some_{\emptyset} ; \piencodf{M}_{x_1}  \\
       & \hspace{2em}  \mid \blue{y_1. \overline{\none}} \mid  y.\overline{\some};y. \overline{\close} ) 
       \\
       \red_{\cnum{9}}~ & (\nu  y,x_1)(  \blue{y_{} . \some_{u, x_1 }} ;y_{}.\close; \piencodf{N}_u  \mid  x_1.\some_{\emptyset} ; \piencodf{M}_{x_1} x\mid  \blue{y.\overline{\some}};y. \overline{\close} ) \qquad (:=T)
       \\
    \end{aligned}
    \]
    \item Finally, reductions $\red_{\cnum{10}}$ and $\red_{\cnum{11}}$ concern name $y$: the former signals that the bag has no more elements to be synchronized for substitution; the latter closes the session, as it has served its purpose of correctly synchronizing the translated term. The resulting process corresponds to the translation of $N \linexsub{M / x}$.
    \[
    \begin{aligned}
    T  \red_{\cnum{10}}~ &  (\nu  y, x_1)(  \blue{y_{}.\close}; \piencodf{N}_u  \mid  x_1.\some_{\emptyset} ; \piencodf{M}_{x_1} \mid  \blue{y. \overline{\close}}  ) 
       \\
       \red_{\cnum{11}}~ &  (\nu   x_1)( \piencodf{N}_u \mid  x_1.\some_{\emptyset} ; \piencodf{M}_{x_1}  ) 
       =\piencodf{ N \linexsub{M / x}  }_u
    \end{aligned}
    \]
\end{itemize}
\end{exa}

We now discuss the translation of a term that fails due to an excess of resources.

\begin{exa}[Excess of Resources]\label{ex:encfailexcess}
Let us consider the well-formed $\lamrsharfail$-term that does not share occurrences of $x$, i.e., $N [ \leftarrow x] \esubst{ \bag{M} }{x}$, where $M, N$ are closed (i.e. $\lfv{N} = \lfv{M} = \emptyset$). This term's translation is:
%\begin{figure}[t!]
	  \[\small
    \begin{aligned}
       \piencodf{ N[ \leftarrow x] \esubst{ \bag{M} }{ x} }_u & =   (\nu x)( \piencodf{ N[ \leftarrow x]}_u \mid \piencodf{\bag{M} }_x ) \\
       & = (\nu x)( x. \overline{\some}. \outact{x}{y_1} . ( y_1 . \some_{u } ;y_1.\close; \piencodf{N}_u \mid x. \overline{\none}) \mid 
       \\
       & \qquad x.\some_{\emptyset} ; x(y_1). x.\some_{y_1};x.\overline{\some} ; \outact{x}{x_i}
      . (x_i.\some_{\emptyset} ; \piencodf{M}_{x_i} \mid \piencodf{\oneb}_x \mid y_1. \overline{\none}) )
      \\
    \end{aligned}
    \]

      \begin{itemize}
        \item Reductions $\red_{\cnum{1}}$ and $\red_{\cnum{2}}$ follow as in  Example \ref{ex:encsucc}.
    	  \[\small
    \begin{aligned}
      \piencodf{ N[ \leftarrow x] \esubst{ \bag{M} }{ x} }_u & \red_{\cnum{1}} (\nu x)(  \outact{x}{y_1} . ( y_1 . \some_{u } ;y_1.\close; \piencodf{N}_u \mid x. \overline{\none}) \mid 
       \\
       & \qquad  x(y_1). x.\some_{y_1};x.\overline{\some} ; \outact{x}{x_i}
           . (x_i.\some_{\emptyset} ; \piencodf{M}_{x_i} \mid \piencodf{\oneb}_x \mid y_1. \overline{\none}) )\\
       & \red_{\cnum{2}} (\nu x, y_1)(   y_1 . \some_{u } ;y_1.\close; \piencodf{N}_u \mid {\cred{  x. \overline{\none}}} \mid 
       \\
       & \quad \qquad  { \cred{x.\some_{y_1};}}x.\overline{\some} ; \outact{x}{x_i} . (x_i.\some_{\emptyset} ; \piencodf{M}_{x_i} \mid \piencodf{\oneb}_x \mid y_1. \overline{\none}) )\qquad (:= P)
       \\
    \end{aligned}
    \]
%            \caption{Example}
 %   \end{figure}
  Notice how the translation of the term first triggers the failure: prefix $x.\overline{\none}$ (highlighted in \cred{red}) signals that there are no (more) occurrences of $x$ within the process; nevertheless, the translation of the bag is still trying to communicate the translation of $M$. This failure along $x$ causes the chain reaction of the failure along $y_1$, which eventually triggers across the translation of $N$.
  
        \item Reduction $\red_{\cnum{3}}$ differs from $\red_{\cnum{3}}$ in Example~\ref{ex:encsucc}, as the translation of the shared variable is empty, we abort along the name $x$;  as the translated bag still contains elements to synchronize, the abortion of the bag triggers that failure of the dependant name $y_1$.
        \[
        \begin{aligned}
            P& \red_{\cnum{3}} (\nu  y_1)(   {\cred{  y_1 . \some_{u } ;}}y_1.\close; \piencodf{N}_u \mid { \red{ y_1. \overline{\none}}}) 
       \\
       & \red_{\cnum{4}}   u. \overline{\none}  = \piencodf{ \fail^{\emptyset} }_u
        \end{aligned}
        \]
        \item Reduction $\red_{\cnum{4}}$ differs from that of $\red_{\cnum{9}}$ and $\red_{\cnum{10}}$ from Example \ref{ex:encsucc}: the name $y_1$ fails signaling that there was an element in the bag that was to be sent; as the translation of the term $N$ is guarded by the confirmation along $y_1$, it aborts.
    \end{itemize}
    
  \end{exa}

Finally, we illustrate how $ \piencodf{\cdot}_u$ acts on a term that fails due to lack of resources in a bag.

\begin{exa}[Lack of Resources]\label{ex:encfaillack}
Consider the well-formed $\lamrsharfail$-term $N [ x_1 \leftarrow x] \esubst{ \oneb }{x}$, where $N$ is a closed  term (i.e. $\lfv{N} =  \emptyset$). This term's translation is:%\small 
  \[\small
    \begin{aligned}
        \piencodf{ N [ x_1 \leftarrow x] \esubst{ \oneb }{x} }_u  
       =~ &   (\nu x)( \piencodf{ N [ x_1 \leftarrow x] }_u \mid \piencodf{ \oneb }_x ) \\
        =~ & (\nu x)( x.\overline{\some}. \outact{x}{y_1}. (y_1 . \some_{\emptyset} ;y_{1}.\close \para x.\overline{\some};x.\some_{u};
       \\
       & \quad  x(x_1) .  x.\overline{\some}; \outact{x}{y_{2}}. ( y_{2} . \some_{u, x_1 } ;y_{2}.\close; \piencodf{N}_u \mid x.\overline{\none} )~ )  \para \\
       & \quad  x.\some_{\emptyset} ; x(y_1). ( y_1.\overline{\some};y_1 . \overline{\close} \mid x.\some_{\emptyset} ; x. \overline{\none}) )\qquad (:= P) \\
       \end{aligned}
      \]
        Notice how the translation of the empty bag $\oneb$ triggers the failure: prefix `$x.\overline{\none}$' signals that there are no (more) elements in the bag; however, the translated term aims to synchronize, as it (still) requires resources.
    
    \begin{itemize}
        \item Reductions $\red_{\cnum{1}}$ and $\red_{\cnum{2}}$ follow from Example \ref{ex:encsucc}.
      \[\small
    \begin{aligned}
       P \red_{\cnum{1}} ~ & (\nu x)( \outact{x}{y_1}. (y_1 . \some_{\emptyset} ;y_{1}.\close \mid  x.\overline{\some};x.\some_{u};x(x_1) .
       \\
       &  \quad   x.\overline{\some}; \outact{x}{y_{2}}. ( y_{2} . \some_{u, x_1 } ;y_{2}.\close; \piencodf{N}_u \mid x.\overline{\none} )~ )  \para \\
       &  \quad  x(y_1). ( y_1.\overline{\some};y_1 . \overline{\close} \mid x.\some_{\emptyset} ; x. \overline{\none}) ) \\
       \red_{\cnum{2}} ~ & (\nu x, y_1)(  y_1 . \some_{\emptyset} ;y_{1}.\close \mid x.\overline{\some};x.\some_{u};x(x_1) . 
       \\
       &  \quad   x.\overline{\some}; \outact{x}{y_{2}}. ( y_{2} . \some_{u, x_1 } ;y_{2}.\close; \piencodf{N}_u \mid x.\overline{\none} )  \para \\
       &  \quad  y_1.\overline{\some};y_1 . \overline{\close} \mid x.\some_{\emptyset} ; x. \overline{\none} ) 
    \end{aligned}
    \]
  
        \item Reductions $\red_{\cnum{3}}$ and $\red_{\cnum{4}}$ follow from that of $\red_{\cnum{9}}$ and $\red_{\cnum{10}}$ in Example \ref{ex:encsucc}; as the term contains the element $x_1$ for synchronization, the encoding of $N$ is not guarded by $y_1$.
        \[\small
        \begin{aligned}
        Q    \red_{\cnum{3}} ~ & (\nu x, y_1)(  y_{1}.\close \mid x.\overline{\some};x.\some_{u};x(x_1) . 
       \\
       &  \quad  x.\overline{\some}; \outact{x}{y_{2}}. ( y_{2} . \some_{u, x_1 } ;y_{2}.\close; \piencodf{N}_u \mid x.\overline{\none} )  \para \\
       &  \quad  y_1 . \overline{\close} \mid x.\some_{\emptyset} ; x. \overline{\none} ) \\
       \red_{\cnum{4}} ~ & (\nu x )(  x.\overline{\some};x.\some_{u};x(x_1) .  x.\overline{\some}; \outact{x}{y_{2}}. ( y_{2} . \some_{u, x_1 } ;y_{2}.\close;  \\
       &  \quad  \piencodf{N}_u \mid x.\overline{\none} )  \mid x.\some_{\emptyset} ; x. \overline{\none} ) \\
       \red_{\cnum{5}} ~ & (\nu x )(  { \cred{ x.\some_{u}}};x(x_1) .  x.\overline{\some}; \outact{x}{y_{2}}. ( y_{2} . \some_{u, x_1 } ;y_{2}.\close; \piencodf{N}_u \mid x.\overline{\none} )  \mid  { \cred{  x. \overline{\none} }} ) \\
       \red_{\cnum{6}}~ &  u. \overline{\none}
       =~   \piencodf{ \fail^{\emptyset}  }_u
        \end{aligned}
        \]
        \item Reduction $\red_{\cnum{5}}$ follows from reduction $ \red_{\cnum{3}} $ in Example \ref{ex:encsucc}.
        \item Reduction $\red_{\cnum{6}}$ differs from that of $\red_{\cnum{4}}$ from Example \ref{ex:encsucc}: the bag contains no elements, and signals this by aborting along the name $x$; still, the term expects to receive an element of the bag, and prematurely aborts.
    \end{itemize}
    
  \end{exa}

\paragraph{Translating Types}
\srev{In describing our translation $\piencodf{\cdot}_{-}$ we have informally referred to (non-deterministic) session protocols in \spi that implement (non-deterministic) expressions in \lamrsharfail. We are actually able to make these intuitions precise and give a translation of intersection types (for \lamrsharfail, cf. \defref{d:typeslamrfail}) into session types (for \spi, cf. \defref{d:sts}). This provides the protocol-oriented interpretation of intersections mentioned earlier. Intuitively speaking, given an intersection type $\pi$, we will have a  corresponding session type $\piencodf{\pi}$ that determines a protocol tied to the evaluation of a (fail-prone, non-deterministic) expression with type $\pi$.}

\begin{defi}[From $\lamrsharfail$ into $\spi$: Types]
\label{def:enc_sestypfail}
%Encoding intersection types (for \lamrfail) into session types (for \spi):
The translation  $\piencodf{\cdot}$ on types is defined 
in \figref{fig:enc_sestypfail}.
 \revo{}{Let  
$\Gamma = x_1: \sigma_1, \cdots, x_m : \sigma_k, v_1: \pi_1 , \cdots , v_n: \pi_n$
%$\Gamma = x_1: \pi_1 , \cdots , x_n: \pi_n$
be as in \defref{d:tcont}.}
            
For some strict types $\tau_1,\cdots,\tau_n$ and $i_1,\cdots,i_n \geq 0$ we define: 
%\revo{}{$$\piencodf{\Gamma} = x_1: \with \overline{\piencodf{\pi_1}_{(\sigma_1, i_1)}}, \cdots , x_n: \with \overline{\piencodf{\pi_n}_{(\sigma_1, i_n)}}$$}
$$\piencodf{\Gamma} = x_1: \with \overline{\piencodf{\sigma_1}} , \cdots ,  x_k : \with \overline{\piencodf{\sigma_k}} ,
 v_1: \with \overline{\piencodf{\pi_1}_{(\tau_1, i_1)}}, \cdots , v_n: \with \overline{\piencodf{\pi_n}_{(\tau_n, i_n)}}$$
\end{defi}

\begin{figure}[!t]
    \centering

\begin{align*}
 \piencodf{\unit} & = \with \onef 
 \\[1mm]
 \piencodf{\pi \rightarrow \tau}  & = \with((  \overline{\piencodf{\pi }_{(\sigma, i)}} ) \ampy \piencodf{\tau})  \quad \text{(for some strict type $\sigma$, with $i \geq 0$)}  
 \\[1mm]
   \piencodf{ \sigma \wedge \pi }_{(\tau, i)} &= \overline{   \with(( \oplus \bot) \otimes ( \with  \oplus (( \with  \overline{\piencodf{ \sigma }} )  \ampy (\overline{\piencodf{\pi}_{(\tau, i)}})))) } \\
     & = \oplus(( \with \onef) \ampy ( \oplus  \with (( \oplus \piencodf{\sigma} ) \otimes (\piencodf{\pi}_{(\tau, i)})))) 
     \\[1mm]
      \piencodf{\omega}_{(\sigma, i)} &=  
 \begin{cases}
     \overline{\with(( \oplus \bot )\otimes ( \with \oplus \bot )))} & \text{if $i = 0$}
     \\
\overline{   \with(( \oplus \bot) \otimes ( \with  \oplus (( \with  \overline{\piencodf{ \sigma }} )  \ampy (\overline{\piencodf{\omega}_{(\sigma, i - 1)}})))) } & \text{if $i > 0$}
\end{cases}
\end{align*}

    \caption{Translating intersection types as session types.}
    \label{fig:enc_sestypfail}
\end{figure}

As we will see, given a well-formedness judgement $\Gamma \wfdash \expr{M} : \tau$, with the translations on types and assignments defined above, we will have $\piencodf{\expr{M}}_u \vdash 
\piencodf{\Gamma}, 
u : \piencodf{\tau}$; this is the content of the \emph{type preservation} property (\thmref{t:preservationtwo}).

The translation of types in \figref{fig:enc_sestypfail} leverages non-deterministic session protocols (typed with `$\with$') to represent non-deterministic fetching and fail-prone evaluation in \lamrsharfail. 
Notice that the translation of the multiset type $\pi$ depends on two arguments (a strict type $\tau$ and a number $i \geq 0$) which are left unspecified above, \secondrev{but are appropriately specified in Proposition~\ref{prop:app_aux}}.
This is crucial to represent mismatches in \lamrsharfail (i.e., sources of failures) as typable processes in \spi. 
For instance, in \figref{fig:wfsh_rules}, Rule~\redlab{FS{:}app} admits a mismatch between 
$\sigma^{j} \rightarrow \tau$ and $ \sigma^{k}$, for it allows $j \neq k$.
In our proof of type preservation, these two arguments are instantiated appropriately, enabling typability as session-typed processes.

 We are now ready to consider correctness  for $\piencodf{\cdot}_u$, in the sense of  \defref{d:encoding}.
First, the compositionality property follows directly from   \figref{fig:encfail}. 
In the following sections, we state the remaining properties in \defref{d:encoding}: type preservation, operational correspondence, and success sensitiveness.

 \subsubsection{Type Preservation}
 We prove that our translation from $\lamrsharfail$ to $\spi$  maps well-formed $\lamrsharfail$ expressions to session-typed processes in $\spi$. \secondrev{First, we show that translated multiset types can be ``lengthened'' by setting appropriate parameters to the encoding. } 
 %first correctness result for the translation is Theorem~\ref{t:preservationtwo}, stated below, whose proof can be found in Appendix \ref{app:typeprestwo}. 
%  We use the following auxiliary lemma:

% \begin{lem}
% \label{prop:app_aux}
% We have that $ \piencodf{\sigma^{j}}_{(\tau_1, m)} = \piencodf{\sigma^{k}}_{(\tau_2, n)}  $ provided that given our choice of $\tau_1$, $\tau_2$, $n$,  and $m$ we choose them as follows:
%         \begin{enumerate}
%         \item If $j > k$ then we take $\tau_1 $ to be an arbitrary type and $m = 0$; also, we take $\tau_2 $ to be $\sigma$ and $n = j-k$.
        
%         \item If $j < k$ then we take $\tau_1 $ to be $\sigma$ and $m = k-j$; also,  we take $\tau_2 $ to be an arbitrary type and $n = 0$. 
        
%         \item Otherwise, if $j = k$ then we take $m = n = 0$. Notice that $\tau_1 , \tau_2 $ are unimportant in this case.
%     \end{enumerate}
    
% \end{lem}

\begin{restatable}{propo}{appaux}
\label{prop:app_aux}
 \revo{A27}{Suppose $\sigma^j$ and $\sigma^k$ are arbitrary strict types (\defref{d:typeslamrfail}), for some $j, k \geq 0$. 
Following \figref{fig:enc_sestypfail}, consider their encoding into session types $\piencodf{\sigma^{j}}_{(\tau_1, m)} $ and $\piencodf{\sigma^{k}}_{(\tau_2, n)}$, respectively, where  $\tau_1, \tau_2$ are strict types and $n, m \geq 0$. \\  We have 
 $\piencodf{\sigma^{j}}_{(\tau_1, m)} = \piencodf{\sigma^{k}}_{(\tau_2, n)}$
  under the following conditions:}
 %We have that $ \piencodf{\sigma^{j}}_{(\tau_1, m)} = \piencodf{\sigma^{k}}_{(\tau_2, n)}  $ provided that given our choice of $\tau_1$, $\tau_2$, $n$,  and $m$ we choose them as follows:
         \begin{enumerate}
         \item If $j > k$ then we take $\tau_1 $ to be an arbitrary strict type and $m = 0$; also, we take $\tau_2 $ to be $\sigma$ and $n = j-k$.
        
         \item If $j < k$ then we take $\tau_1 $ to be $\sigma$ and $m = k-j$; also,  we take $\tau_2 $ to be an arbitrary strict type and $n = 0$. 
        
         \item Otherwise, if $j = k$ then we take $m = n = 0$. Also, $\tau_1 , \tau_2 $ are arbitrary strict types.
     \end{enumerate}
\end{restatable}

\begin{proof}
\secondrev{
Immediate by unfolding the translation. The full analysis can be found in \appref{app:typeprestwo}.
}
\end{proof}

\secondrev{
Given Proposition \ref{prop:app_aux} we now show that the translation preserves types:
%We now show that type preservation holds:
}

\begin{restatable}[Type Preservation for $\piencodf{\cdot}_u$]{thms}{preservationtwo}
\label{t:preservationtwo}
Let $B$ and $\expr{M}$ be a bag and an expression in $\lamrsharfail$, respectively.
\begin{enumerate}
\item If $\secondrev{\core{\Gamma}} \wfdash B : \pi$
then 
$\piencodf{B}_u \vdash  \piencodf{\secondrev{\core{\Gamma}}}, u : \piencodf{\pi}_{(\sigma, i)}$, \revo{A13}{for some strict type $\sigma$ and index $ i \geq 0 $.}

\item If $\secondrev{\core{\Gamma}} \wfdash \expr{M} : \tau$
then 
$\piencodf{\expr{M}}_u \vdash  \piencodf{\secondrev{\core{\Gamma}}}, u :\piencodf{\tau}$.
\end{enumerate}

\end{restatable}

\begin{proof}
 By mutual induction on the typing derivation of $B$ and $\mathbb{M}$ , with an analysis of the last rule applied in $\Gamma\wfdash B:\pi$ and in $\Gamma\wfdash \mathbb{M}:\tau$. \revd{B26}{One key aspect of this proof is the application of Proposition~\ref{prop:app_aux} to ensure duality of types. Intuitively, the conditions given by Proposition~\ref{prop:app_aux} are used to instantiate the parameters in the encoding of intersection types, so as to ensure that when intersection types have different types the smaller type can be correctly ``padded'' to match the size of the larger type---Example~\ref{ex:padd}, given below, illustrates this padding. } The full proof can be found in \appref{app:typeprestwo}.
\end{proof}

\begin{exa}[\revo{A2}{Parameters in the encoding of types}]
\label{ex:padd}
We give the dual types when encoding intersection types, namely the case of $\piencodf{ \sigma \wedge \pi }_{(\sigma, i)}$, to express the encoding of intersection typed behavior into session typed behavior. 
The application of dual types is most evident in the application of a bag into an abstraction: the bag providing the intersection type and the abstraction consuming it. In session types the interaction between these is expressed by dual session types where one channel provides a behavior and and the dual channel provides the dual session type behavior via the cut rule.
Let us consider the term \( (  \lambda x . M [x_1, x_2 \leftarrow x]  ) B \) typed with the well-formedness rules by:
   \begin{prooftree}
        \AxiomC{\( \Gamma \wfdash \lambda x . M [x_1, x_2 \leftarrow x] : (\sigma \wedge \sigma) \rightarrow \tau \quad \Delta \wfdash B : \sigma^k \)}
            \LeftLabel{\redlab{FS{:}app}}
        \UnaryInfC{\( \Gamma, \Delta \wfdash (  \lambda x . M [x_1, x_2 \leftarrow x]  ) B : \tau\)}
    \end{prooftree}
When applying the translation of \figref{fig:encfail} to the term we obtain:
\[
\bigoplus_{B_i \in \perm{B}} (\nu v)(\piencodf{\lambda x . M [x_1, x_2 \leftarrow x]}_v \mid v.\some_{u , \lfv{B}} ; \outact{v}{x} . ([v \leftrightarrow u] \mid \piencodf{B_i}_x ) )  
\]
By appealing to Type Preservation (\thmref{t:preservationtwo}) we obtain both $ \piencodf{\lambda x . M [x_1, x_2 \leftarrow x]}_v \vdash  \piencodf{\Gamma}, v :\piencodf{(\sigma \wedge \sigma) \rightarrow \tau}  $ and $\piencodf{B}_x \vdash \piencodf{\Delta} , x: \piencodf{\sigma^k }_{(\delta_2, i_2)} $.
We give the typing for one non-deterministic branch where we take an arbitrary permutation of $B$ is as follows by applying the rules of \figref{fig:trulespifull} and that $\Pi_1$ is derived to be:
\begin{prooftree}
\AxiomC{}
\LeftLabel{\redlab{Tid}}
\UnaryInfC{$[v \leftrightarrow u] \vdash   v : \overline{\piencodf{\tau}} , u:\piencodf{\tau}  $}
\AxiomC{$\piencodf{B}_x \vdash \piencodf{\Delta} , x: \piencodf{\sigma^k }_{(\delta_2, i_2)} $}
\LeftLabel{\redlab{T\otimes}}
\BinaryInfC{$\outact{v}{x} . ([v \leftrightarrow u] \mid \piencodf{B}_x ) \vdash \piencodf{\Delta}, v :   \piencodf{\sigma^k }_{(\delta_2, i_2)} \ampy \overline{\piencodf{\tau}} , u:\piencodf{\tau}  $}
\LeftLabel{\redlab{T\oplus^x_{\widetilde{w}}}}
\UnaryInfC{$ v.\some_{u , \lfv{B}} ; \outact{v}{x} . ([v \leftrightarrow u] \mid \piencodf{B}_x ) \vdash \piencodf{\Delta}, v :\dual{\piencodf{(\sigma^k) \rightarrow \tau}} , u:\piencodf{\tau} $}
\end{prooftree}
Hence we obtain the derivation:
\begin{prooftree}
\AxiomC{$ \piencodf{\lambda x . M [x_1, x_2 \leftarrow x]}_v \vdash  \piencodf{\Gamma}, v :\piencodf{(\sigma \wedge \sigma) \rightarrow \tau}  $}
\AxiomC{$ \Pi_1 $}
\LeftLabel{\redlab{Tcut}}
\BinaryInfC{$(\nu v)(\piencodf{\lambda x . M [x_1, x_2 \leftarrow x]}_v \!\mid\! v.\some_{u , \lfv{B}} ; \outact{v}{x} . ([v \leftrightarrow u] \!\mid\! \piencodf{B}_x ) \vdash \piencodf{\Gamma}, \piencodf{\Delta}, u:\piencodf{\tau} $}
\end{prooftree}
Now we shall focus on the typing of the channel $v$ and $x$ in this process as these channel describes the behavior of the encoded intersection type which we are trying to match via duality. By the translation on types from \figref{fig:enc_sestypfail} we have that

\[
\piencodf{(\sigma \wedge \sigma) \rightarrow \tau} = \with((  \overline{\piencodf{(\sigma \wedge \sigma) }_{(\delta_1, i_1)}} ) \ampy \piencodf{\tau})
\]

\begin{itemize}
    \item When $B = \oneb$ we have derivation:
    \[
        \piencodf{ \oneb }_x \wfdash  \piencodf{\Delta}, x :\piencodf{\omega }_{(\delta_2, i_2)}
    \]
    To obtain duality from Rule~\redlab{Tcut}  we must have that $\piencodf{\sigma^2 }_{(\delta_1, i_1)} = \piencodf{\omega }_{(\delta_2, i_2)} $. By Proposition \ref{prop:app_aux} we can take $\delta_1$ to be an arbitrary strict type, $i_1 = 0$, $i_2=2$ , $\delta_2 = \sigma$. We have:
    
    \[
    \begin{aligned}
        \piencodf{\omega }_{(\sigma, 2)} & = \overline{   \with(( \oplus \bot) \otimes ( \with  \oplus (( \with  \overline{\piencodf{ \sigma }} )  \ampy (\overline{\piencodf{\omega}_{(\sigma, 1)}})))) }\\
        & = \overline{   \with(( \oplus \bot) \otimes ( \with  \oplus (( \with  \overline{\piencodf{ \sigma }} )  \ampy ( \with(( \oplus \bot) \otimes ( \with  \oplus (( \with  \overline{\piencodf{ \sigma }} )  \ampy (\overline{\piencodf{\omega}_{(\sigma, 0)}})))) )))) }\\
        & = \piencodf{\sigma^2 }_{(\delta_1, i_1)}
    \end{aligned}
    \]
    
    \item When $B = \bag{N_1,N_2}$ we have derivation:
    \[
        \piencodf{ \bag{N_1,N_2} }_x \wfdash  \piencodf{\Delta}, x :\piencodf{\sigma^2 }_{(\delta_2, i_2)}
    \]
   To obtain duality from Rule~\redlab{Tcut}  we must have that $\piencodf{\sigma^2 }_{(\delta_1, i_1)} = \piencodf{\sigma^2 }_{(\delta_2, i_2)} $. By Proposition \ref{prop:app_aux} we can take $\delta_1$ and $\delta_2$ to be an arbitrary strict type and  $i_1 = i_2 = 0$ . We then obtain $\piencodf{\sigma^2 }_{(\delta_1, 0)} = \piencodf{\sigma^2 }_{(\delta_2, 0)} $, as $\piencodf{\omega}_{(\delta_1, 0)} = \piencodf{\omega}_{(\delta_2, 0)} $ for any two strict types $\delta_1,\delta_2$.
    
    \item When $B = \bag{N_1,N_2,N_3}$ we have derivation:
    \[
        \piencodf{ \bag{N_1,N_2,N_3} }_x \wfdash  \piencodf{\Delta}, x :\piencodf{\sigma^3 }_{(\delta_2, i_2)}
    \]
    To obtain duality from Rule~\redlab{Tcut}  we must have that $\piencodf{\sigma^2 }_{(\delta_1, i_1)} = \piencodf{\sigma^3 }_{(\delta_2, i_2)} $. By Proposition \ref{prop:app_aux} we can take $\delta_2$ to be an arbitrary strict type, $i_2 = 0$, $i_1=2$ , $\delta_1 = \sigma$. Then the case proceeds similarly to when $B = \oneb$.
    \end{itemize}
\end{exa}

\subsubsection{Operational Correspondence: Completeness and Soundness}\label{app:s:compsound}

We now state our operational correspondence results (completeness and soundness, cf. Fig.~\ref{f:opcs}).

\paragraph*{A Congruence} 

We will identify some $\lamrsharfail$-terms such as $M\shar{}{x}\esubst{\oneb}{x}$ and $M$. The identification is natural, as the former is a term $M$ with no occurrences of $x$ in which $x$ is going to be replaced with $\oneb$, which clearly describes a substitution that ``does nothing'', and would result in $M$ itself.  With this intuition, other terms are identified via a {\em congruence} (denoted $\pequiv  %\pcong
$) on terms and expressions that is formally defined in Fig.~\ref{fig:rsPrecongruencefailure}.
%We write $M \pequiv M'$ whenever both $M \pcong M'$ and $M' \pcong M$ hold.

\begin{figure}[!t]
\[
\begin{array}{rll}
  M [ \leftarrow x] \esubst{\oneb}{x} \!\!\!\! &\revdaniele{\pequiv} M
&
  \\
  MB \linexsub{N/x}  \!\!\!\!&\pequiv (M\linexsub{N/x})B &  \hspace*{-3em}\text{with } x \not \in \lfv{B} 
  \\
     M \linexsub{N_2/y}\linexsub{N_1/x} 
         \!\!\!\!&\pequiv M\linexsub{N_1/x}\linexsub{N_2/y} &
         \hspace*{-3em}\text{with } x \not \in \lfv{N_2},\revdaniele{ y \notin \lfv{N_1}}
         \\
  MA[\widetilde{x} \leftarrow x]\esubst{B}{x} 
         \!\!\!\!&\pequiv (M[\widetilde{x} \leftarrow x]\esubst{B}{x})A
         &  \hspace*{-3em}\text{with } x_i \in \widetilde{x} \Rightarrow x_i \not \in \lfv{A}
         \\
M[\widetilde{y} \leftarrow y]\esubst{A}{y}[\widetilde{x} \leftarrow x]\esubst{B}{x} \!\!\!\! &\revdaniele{\pequiv} 
(M[\widetilde{x} \leftarrow x]\esubst{B}{x})[\widetilde{y} \leftarrow y]\esubst{A}{y} &   \text{with } x_i \in \widetilde{x} \Rightarrow x_i \not \in \lfv{A}
\\
     C[M] \!\!\!\! & \revdaniele{\pequiv} C[M'] 
     &
        \hspace*{-3em}\text{with } M \revdaniele{\pequiv} M'
    \\
        D[\expr{M}] 
        \!\!\!\! & \revdaniele{\pequiv} D[\expr{M}'] 
         & 
         \hspace*{-3em}\text{with } \expr{M} \revdaniele{\pequiv} \expr{M}'
\end{array}
\]

\caption{Congruence in \lamrsharfail.}
    \label{fig:rsPrecongruencefailure}
\end{figure}

  \begin{exa}[Cont. Example~\ref{ex:shar-wf}]\label{ex:precong_fail}
We illustrate the congruence in case of  failure:
\[
      \begin{aligned}
            (\lambda x . x_1 [x_1 \leftarrow x]) \bag{ \fail^{\emptyset}[ \leftarrow y] \esubst{ \oneb }{y} } & \red_{\redlab{RS{:}Beta}} x_1 [x_1 \leftarrow x]  \esubst{\bag{ \fail^{\emptyset}[ \leftarrow y] \esubst{\oneb}{y} }}{x} \\
           & \red_{\redlab{RS{:}Ex \dash Sub}} x_1  \linexsub{ \fail^{\emptyset}[ \leftarrow y] \esubst{\oneb}{y}  / x_1} \\
            & \red_{\redlab{RS{:}Lin \dash Fetch}} \fail^{\emptyset}[ \leftarrow y] \esubst{\oneb}{y}   \\
            & \pequiv \fail^{\emptyset} 
        \end{aligned}
  \]
  
  In the last step, Rule~$\redlab{RS{:}Cons_2}$ cannot be  applied:
  $y$ is sharing with no shared variables and the explicit substitution involves the bag $\oneb$.
  \end{exa}

\secondrev{
\begin{restatable}[Consistency Stability Under \(\equiv\)]{thms}{consistnequiv}
\label{thm:term_consistency}
%\begin{prop}[Consistency by Typing and Reduction]\label{prop:term_consistency}
    %Consistency for $\lamrsharfail$-expressions (cf. \defref{d:consistent}) is preserved by congruence: 
    Let  ${\expr{M}}$ be a consistent $\lamrsharfail$-expression. If $\expr{M} \equiv \expr{M}'$ then ${\expr{M}'}$ is consistent.
\end{restatable}
}

\begin{proof}
\secondrev{
    By induction on the structure of $\expr{M}$; see \Cref{compandsucctwo} for details.
}
\end{proof}

\begin{figure}[t!]
\begin{tikzpicture}[scale=.9pt]
% \draw[rounded corners, color=black] (0,0) rectangle (15,8);
\node (opcom) at (4.3,7.0){Operational Completeness};
\draw[rounded corners, color=teal,fill=teal!20] (0,5.2) rectangle (14.5,6.6);
\draw[rounded corners, color=violet!80!black, fill=violet!10] (0,1.4) rectangle (14.5,2.8);
\node (lamrfail) at (1.2,6) {$\lamrsharfail$:};
\node (expr1) [right of=lamrfail, xshift=.3cm] {$\mathbb{N}$};
\node (expr2) [right of=expr1, xshift=2cm] {$\mathbb{M} \pequiv \mathbb{M}' $};
\draw[arrow] (expr1) --  (expr2);
% \node at (5.5,5.6) {\redlab{R}};
\node (lamrsharfail) at (1.2,2) {$\spi$:};
\node (transl1) [right of=lamrsharfail, xshift=.3cm] {$\piencodf{\mathbb{N}}$};
\node (transl2) [right of=transl1, xshift=2cm] {$Q=\piencodf{\mathbb{M}'}$};
\draw[arrow, dotted] (transl1) -- node[anchor= south] {$*$ }(transl2);
%\node at (5.1,1.7) {\small $\pcong$ };
\node (enc1) at (2,4) {$\piencodf{\cdot }$};
\node (refcomp) [right of=enc1, xshift=.9cm]{Thm~\ref{l:app_completenesstwo}};
\node  at (6.5,4) {$\piencodf{\cdot }$};
\draw[arrow, dotted] (expr1) -- (transl1);
\draw[arrow, dotted] (expr2) -- (transl2);
%%%
\node (opcom) at (11,7.0){Operational Soundness};
\node (expr1shar) [right of=expr2, xshift=1.8cm] {$\mathbb{N}$};
\node (expr2shar) [right of=expr1shar, xshift=2.8cm] {$\mathbb{N'}$};
\node (equivilencebit) [right of=expr1shar, xshift=2.4cm , yshift = -0.3cm] {$\pequiv$};
\draw[arrow, dotted] (expr1shar) -- node[anchor= south] {*} (expr2shar);
%\node at (12.5,5.6) {\redlab{R}};
\node (transl1shar) [right of=transl2, xshift=1.8cm] {$\piencodf{\mathbb{N}}$};
\node (transl2shar) [right of=transl1shar, xshift=1cm] {$Q$};
\node (expr3shar) [right of=transl2shar, xshift=.8cm]{$Q'=\piencodf{\mathbb{N'}}$};
\draw[arrow,dotted] (expr2shar) --  (expr3shar);
\draw[arrow,dotted] (expr1shar) -- (transl1shar);
\node (enc2) at (8.5,4) {$\piencodf{\cdot }$};
\node (refsound) [right of=enc2, xshift=1.5cm]{Thm~\ref{l:app_soundnesstwo}};
\node  at (13.8,4) {$\piencodf{\cdot }$};
\draw[arrow] (transl1shar) -- node[anchor= south] {*} (transl2shar);
\draw[arrow,dotted] (transl2shar) -- node[anchor= south] {*}(expr3shar);
%\node at (12.2,1.7) {\small $\pequiv$ };
\end{tikzpicture}
\caption{Operational Correspondence for $\piencodf{\cdot }$\label{f:opcs}}	
\end{figure}

\begin{defi}[Partially Open Terms]
We say that a $\lamrsharfail$-term $M$ is \emph{partially open} if $\forall x \in \lfv{M}$ (cf. \defref{d:fvsh}) implies that $x$ is not a sharing variable.
\end{defi}

Notice that the class of open terms (no conditions on free variables) subsumes the class of partially open terms, which in turn  subsumes the class of closed terms. Consider the following example.

\begin{exa}[Partially Open Terms]
    We give three examples of well-formed $\lamrsharfail$-terms:
    \[
    \begin{aligned}
        M_1 &= \lambda x. x_1 [x_1 \leftarrow x] 
        &
        \quad M_2 &= \lambda x. (x_1 \bag{y}) [x_1 \leftarrow x] 
        &
        \quad M_3 &= (x_1 \bag{y}) [x_1 \leftarrow x]     
    \end{aligned}
   \]
    Here the only closed term is $M_1$ as $M_2$ has one free variable (i.e., $y$)  and $M_3$ has two free variables ($y$ and $x$). 
    While $M_2$ is partially open, $M_3$ is not because $x$ is a sharing variable. 
    %Finally all terms here are open.
\end{exa}

The following proposition  will be used in the proof of operational completeness (Theorem~\ref{l:app_completenesstwo}) and operational soundness (Theorem~\ref{l:app_soundnesstwo}). The proposition relies on well-formed partially open terms; however, in the proof of operational correspondence we only consider closed terms rather then partially open terms.
%whose proofs can be found in Appendix~\ref{compandsucctwo}.

%\joe{this statement is being revised still}
\secondrev{
\begin{restatable}[]{propo}{encodingreduces}
%\label{prop:correctformfail}
\label{prop:NEEDTONAME}
    Suppose $N$ is a well-formed, partially open $\lamrsharfail$-term with $\headf{N} = x$.
     Then, there exist an index set $I$, names $\widetilde{y}$ and $n$, and processes $P_i$ such that the following four conditions hold:
    \begin{enumerate}
        \item $$\piencodf{ N }_{u} \red^* \bigoplus_{i \in I}(\nu \widetilde{y})(\piencodf{ x }_{n} \mid P_i) $$
    \end{enumerate}
    \begin{enumerate}[resume]
        \item There exists a \lamrsharfail-term $N'$ such that $N \pequiv N'$ and:
            $$\piencodf{ N' }_{u} = \bigoplus_{i \in I}(\nu \widetilde{y})(\piencodf{ x }_{n} \mid P_i) $$
    \end{enumerate}
    \begin{enumerate}[resume]
        \item For any well-formed and partially open \lamrsharfail-term $M$:
        $$\piencodf{ N\headlin{ M/x } }_{u} \red^* \bigoplus_{i \in I}(\nu \widetilde{y})(\piencodf{ M }_{n} \mid P_i) $$
    \end{enumerate}
    \begin{enumerate}[resume]
        \item There exists a \lamrsharfail-term $M'$ such that $M' \pequiv N\headlin{ M/x }$ and:
            $$\piencodf{ M' }_{u} = \bigoplus_{i \in I}(\nu \widetilde{y})(\piencodf{ M }_{n} \mid P_i)  $$
    \end{enumerate}
    %
    %
    %If $N$ is a well-formed, partially open $\lamrsharfail$-term with $\headf{N} = x$ such that
    %
    %
    %$$\piencodf{ N }_{u} \red^* \bigoplus_{i \in I}(\nu \widetilde{y})(\piencodf{ x }_{n} \mid P_i) $$ for some index set $I$, names $\widetilde{y}$ and $n$, and processes $P_i$.
    %Then there exists $N'$ such that $N' \pequiv N$ and:
    %$$\piencodf{ N' }_{u} = \bigoplus_{i \in I}(\nu \widetilde{y})(\piencodf{ x }_{n} \mid P_i) $$
    %\\
    %Similarly we have that for any term $N'$
    %$$\piencodf{ N\headlin{ N'/x } }_{u} \red^* \bigoplus_{i \in I}(\nu \widetilde{y})(\piencodf{ N' }_{n} \mid P_i) $$ 
    %\\
    %and there exists $M'$ such that $M' \pequiv N\headlin{ N'/x }$ 
    %$$\piencodf{ M' }_{u} = \bigoplus_{i \in I}(\nu \widetilde{y})(\piencodf{ N' }_{n} \mid P_i)  $$
\end{restatable}
}

\secondrev{
\begin{proof}
    By induction on the structure of $N$. We briefly sketch the strategy for proving it case below, but 
     the  complete proof can be found in \appref{compandsucctwo}.
    \begin{enumerate}
        \item The interesting cases are for $N=M\linexsub{N'/x}$ and $N=M\shar{\widetilde{y}}{y}\esubst{B}{y}$, when   $\size{B}=\size{\widetilde{y}} = 0$ and  $\headf{M}=x$. %: cases in which $N$ do not fail.
            Notice that $N=M\shar{\widetilde{y}}{y}$ is not a case, because of the definition of partially open term: $y$ is a sharing variable in $N$ and $y \in \lfv{N}$.
            The other cases follow easily by the induction hypothesis.
        \item  Reductions are only introduced by explicit weakening, which can be eliminated via the precongruence.
        \item Follows from (1) and the fact that linear head substitution can be placed deeper within the term until it reaches the head variable.
        \item Follows from (2) and (3). \qedhere
    \end{enumerate}
\end{proof}
}

%\paragraph{Completeness} We require the following auxiliary result:
%\begin{restatable}[]{propo}{correctformfail}
%\label{prop:correctformfail}
%Suppose $N$ is a well-formed, partially open $\lamrsharfail$-term with $\headf{N} = x$.
% such that $N$ does not fail, that is, there is no $M\in \lamrsharfail$ for which there is a  reduction $N  \red_{\redlab{RS{:}Fail}} M$.
%Then,

%$$\piencodf{ N }_{u} \red^* \bigoplus_{i \in I}(\nu \widetilde{y})(\piencodf{ x }_{n} \mid P_i) $$ for some index set $I$, names $\widetilde{y}$ and $n$, and processes $P_i$.
%\end{restatable}

%\begin{proof}
%By induction on the structure of $N$. The interesting cases happen when $N=M\linexsub{N'/x}$ and $N=M\shar{\widetilde{y}}{y}\esubst{B}{y}$, for  $\size{B}=\size{\widetilde{y}} = 0$,  both when  $\headf{M}=x$. %: cases in which $N$ do not fail.
%Notice that $N=M\shar{\widetilde{y}}{y}$ is not a case, because of the definition of partially open term: $y$ is a sharing variable in $N$ and $y \in \lfv{N}$.
%The other cases follow easily by induction hypothesis. The full proof can be found in \appref{compandsucctwo}.
%\end{proof}

Because of the diamond property (Proposition \ref{prop:conf1_lamrsharfail}), it suffices to consider a completeness result based on a single reduction step in $\lamrsharfail$:

\begin{restatable}[Operational Completeness]{thms}{opcomplete}
\label{l:app_completenesstwo}
Let $\expr{N} $ and $ \expr{M} $ be well-formed, \srev{partially open} $\lamrsharfail $ expressions. If $ \expr{N}\red \expr{M} $ then there exist $Q$ and $\expr{M}'$
such that $\expr{M}' \pequiv \expr{M}$, $\piencodf{\expr{N}}_u  \red^* Q = \piencodf{\expr{M}'}_u$.
%$\piencodf{\expr{N}}_u  \red^* Q$  $ \piencodf{\expr{M}}_u \red^* Q$.
\end{restatable}
\begin{proof}
By induction on the reduction rule applied to infer $\mathbb{N}\red \mathbb{M}$. The case in which
\(\mathbb{N}\red_{\redlab{RS:Lin\dash~Fetch}} \mathbb{M}\) happens for $\mathbb{N}=M\linexsub{N'/x}$ with $\headf{M}=x$, and $\mathbb{M}=M\headlin{N'/x}$. The translation of $\mathbb{N}$ is of the form (omitting details):
\[
\begin{aligned}
\piencodf{\revd{B29}{\expr{N}}}_u&= (\nu x) (\piencodf{M}_u\mid x.\some_{\lfv{N'}};\piencodf{N'}_x)\\
 &\red^* (\nu x) (\bigoplus_{i \in I}(\nu \widetilde{y})(\piencodf{ x }_{n} \mid P_i)\mid x.\some_{\lfv{N'}};\piencodf{N'}_x), ~\text{by Proposition~\ref{prop:NEEDTONAME}}\\
 &\red^* \revd{B29}{\bigoplus_{i \in I}(\nu \widetilde{y})(P_i\mid \piencodf{N'}_n )}= \piencodf{\mathbb{M}}_{u}
\end{aligned}
\]
The other cases follow by analyzing reductions from the translation of $\mathbb{N}$. The full proof can be found in \appref{compandsucctwo}.
\end{proof}

\revd{B41}{Notice how Proposition~\ref{prop:NEEDTONAME} requires a term to be partially open; however, we prove operational correspondence for closed terms. The reason for this is that we start from a source closed term in \lamrfail, which is translated by $\recencodopenf{\cdot}$ into a closed $\lamrsharfail$-term.
%apply the composed translation $\piencodf{\recencodopenf{\cdot}}_u$ to . This composition leads to closed terms in   being that when composing both encodings , operational correspondence is sufficient to hold for closed terms. This is as closed $\lamrfail$ terms are encoded into closed $\lamrsharfail$ terms.
}

\begin{exa}[Cont. Example~\ref{ex:encfailexcess}]
\label{ex:enc2}
\revd{B29}{Recall that $M$ and $N$ are well-formed with $\lfv{N} = \lfv{M} = \emptyset$}, we can verify that  $ N [ \leftarrow x] \esubst{ \bag{M} }{x}$ and $\fail^{\lfv{N} \cup \lfv{M}}$  are also well-formed. We have
% by applying rule \redlab{RS{:}Fail} from \figref{fig:share-reductfailure} one has the following reduction:
 $$
 N [ \leftarrow x] \esubst{ \bag{M} }{x} \red_{\redlab{RS{:}Fail}} \fail^{\lfv{N} \cup \lfv{M}}
 $$  
In \spi, this reduction is mimicked as $$\piencodf{ N[ \leftarrow x] \esubst{ \bag{M} }{ x} }_u \red^*  \piencodf{ \fail^{\lfv{N} \cup \lfv{M}} }_u.$$
In fact, 
\[\small
     \begin{aligned}
        \piencodf{ N[ \leftarrow x] \esubst{ \bag{M} }{ x} }_u 
        =~&   (\nu x)( \piencodf{ N[ \leftarrow x]}_u \mid \piencodf{\bag{M} }_x ) \\
        =~& (\nu x)( x. \overline{\some}. \outact{x}{y_i} . ( y_i . \some_{u} ;y_{i}.\close; \piencodf{N}_u \mid x. \overline{\none}) \mid 
        \\
        =~&  x.\some_{\emptyset } ; x(y_i). x.\some_{y_i};x.\overline{\some} ; \outact{x}{x_i}. (x_i.\some_{\emptyset} ; \piencodf{M}_{x_i} \mid \piencodf{\oneb}_x \mid y_i. \overline{\none}) )
       \\
         \red~& (\nu x)(  \outact{x}{y_i} . ( y_i . \some_{u} ;y_{i}.\close; \piencodf{N}_u \mid x. \overline{\none}) \mid \\
       & \hspace{.8cm}  x(y_i). x.\some_{y_i};x.\overline{\some} ; \outact{x}{x_i}. (x_i.\some_{\emptyset} ; \piencodf{M}_{x_i} \mid \piencodf{\oneb}_x \mid y_i. \overline{\none}) )
       \\
       \red~& (\nu x)(   y_i . \some_{u} ;y_{i}.\close; \piencodf{N}_u \mid x. \overline{\none} \mid \\
       & \hspace{.8cm}  x.\some_{y_i };x.\overline{\some} ; \outact{x}{x_i} . (x_i.\some_{\emptyset} ; \piencodf{M}_{x_i} \mid \piencodf{\oneb}_x \mid y_i. \overline{\none}) )
       \\
       \red~& (\nu x)(   y_i . \some_{u} ;y_{i}.\close; \piencodf{N}_u \mid   y_i.\overline{\none} )
       \\
       \red~& u. \overline{\none} \\
       =~&  \piencodf{ \fail^{\lfv{N} \cup \lfv{M}} }_u
     \end{aligned}
     \]    
    \end{exa}

To state soundness we rely on the congruence relation  $\pequiv$, given in \figref{fig:rsPrecongruencefailure}.

\begin{nota}
Recall the congruence $\pequiv$ for \lamrsharfail, given in Figure~\ref{fig:rsPrecongruencefailure}.
We write $N \red_{\pequiv} N'$ iff $N \pequiv N_1 \red N_2 \pequiv N' $, for some $N_1, N_2$. Then, $\red_{\pequiv}^*$ is the reflexive, transitive closure of $\red_{\pequiv}$. 
We use the notation $ M \red_{\pequiv}^i N$ to state that $M$ performs $i$ steps of $\red_{\pequiv}$ to $N$ in $i \geq 0$ steps. When $i = 0$ it refers to no reduction taking place. 
\end{nota}

%\alert{maybey here say about it (weak operational soundess)}

\begin{restatable}[Operational Soundness]{thms}{opsound}
\label{l:app_soundnesstwo}
Let $\expr{N}$ be a 
well-formed, \secondrev{partially open}  $ \lamrsharfail$ expression. 
If $ \piencodf{\expr{N}}_u \red^* Q$
then there exist $Q'$  and $\expr{N}' $ such that 
$Q \red^* Q'$, $\expr{N}  \red^*_{\pequiv} \expr{N}'$ 
and 
$\piencodf{\expr{N}'}_u = Q'$.
\end{restatable}

\begin{proof}[Proof (Sketch)]
By induction on the structure of $\mathbb{N}$ with sub-induction on the number of reduction steps in $ \piencodf{\expr{N}}_u \red^* Q$. The cases  in which $\mathbb{N}=x$, or $\mathbb{N}=\fail^{\widetilde{x}}$, or $\mathbb{N}=\lambda x. M\shar{\widetilde{x}}{x}$, are easy since there are no reductions starting from $\piencodf{\mathbb{N}}_u$, i.e., $ \piencodf{\expr{N}}_u \red^0 Q$  which implies $\piencodf{\mathbb{N}}_u=\piencodf{\mathbb{N'}}_u=Q=Q'$ and the result follows trivially. The analysis for some  cases are exhaustive, for instance, when   $\mathbb{N}=(M \ B)$ or $\mathbb{N}=M\shar{\widetilde{x}}{x} \esubst{B}{x}$, there are several sub-cases  to be considered: (i) $B$ being equal to  $\oneb$ or not; (ii) $\size{B}$   matching the number of occurrences of the variable in $M$ or not; (iii) $M$ being a failure term or not.

We now discuss one of these cases to illustrate the recurring idea used in the proof: let $\mathbb{N}=(M \ B)$ and suppose that we are able to perform $k> 1$ steps  to a process $Q$, i.e., 
 
 \begin{equation}\label{eq:enc_n}
\piencodf{\mathbb{N}}_{u}= \piencodf{(M \ B)}_{u}= \bigoplus_{B_i \in \perm{B}} (\nu v)(\piencodf{M}_v \mid v.\some_{u , \lfv{B}} ; \outact{v}{x} . ([v \leftrightarrow u] \mid \piencodf{B_i}_x ) )  \red^k Q
 \end{equation}
 
 Then there exist an  $\spi$ process $R$ and integers $n,m$ such that $k=m+n$ and
   \[
            \begin{aligned}
               \piencodf{\expr{N}}_u &\red^m  \bigoplus_{B_i \in \perm{B}} (\nu v)( R \mid v.\some_{u , \lfv{B}} ; \outact{v}{x} . ( \piencodf{ B_i}_x \mid [v \leftrightarrow u] ) ) \red^n  Q\\
            \end{aligned}
            \]
           where the first $m \geq 0$ reduction steps are  internal to $\piencodf{ M}_v$; type preservation in \spi ensures that, if they occur,  these reductions  do not discard the possibility of synchronizing with $v.\some$. Then, the first of the $n \geq 0$ reduction steps towards $Q$ is a synchronization between $R$ and $v.\some_{u, \lfv{B}}$. 
           
           We will  consider the case when $m = 0$ and $n \geq 1$. Then  $R = \piencodf{\expr{M}}_u \red^0 \piencodf{\expr{M}}_u$ and  there are two possibilities of having an unguarded $v.\overline{\some}$ or $v.\overline{\none}$ without internal reductions:
           \begin{enumerate}[(i)]
               \item $M = (\lambda x . M' [\widetilde{x} \leftarrow x]) \linexsub{N_1 / y_1} \cdots \linexsub{N_p / y_p} \qquad (p \geq 0)$
                        
            \item $M = \fail^{\widetilde{z}}$
           \end{enumerate}
           
%    \joe{here is the explanation for the special reduction}
    Firstly we use case (i) to express the need for the reduction $\expr{N}  \red^*_{\pequiv} \expr{N}'$. In this case $\mathbb{N}=((\lambda x . M' [\widetilde{x} \leftarrow x]) \linexsub{N_1 / y_1} \cdots \linexsub{N_p / y_p} \ B)$ and $\piencod{\mathbb{N}}_u$ may perform synchronizations where both $\piencod{\lambda x . M'}_v$ and $\piencod{B}_x$ synchronize across their shared channel. Here we use the congruence relation as follows:
    \[
    \begin{aligned}
        \mathbb{N} & = ((\lambda x . M' [\widetilde{x} \leftarrow x]) \linexsub{N_1 / y_1} \cdots \linexsub{N_p / y_p} \ B) \\
        & \pequiv ((\lambda x . M'  [\widetilde{x} \leftarrow x])\ B) \linexsub{N_1 / y_1} \cdots \linexsub{N_p / y_p} \\
    \end{aligned}
    \]
    This enables the abstraction $\lambda x . M'$ to synchronize with the bag $B$.
    %\joe{added part ends here }
    
    Now we will develop case (ii):
     \[
    \begin{aligned}
    \piencodf{M}_v &= \piencodf{\fail^{\widetilde{z}}}_v= \piencodf{\fail^{\widetilde{z}}}_v= v.\overline{\none} \mid \widetilde{z}.\overline{\none} \\
    \end{aligned}
    \]
 With this shape for $M$, the translation and reductions from (\ref{eq:enc_n}) become
 \begin{equation}\label{eq:red_n}
 \begin{aligned}
 \piencodf{\expr{N}}_u  = & \bigoplus_{B_i \in \perm{B}} (\nu v)( \piencodf{ M}_v \mid v.\some_{u, \lfv{B}} ; \outact{v}{x} . (  \piencodf{ B_i}_x \mid [v \leftrightarrow u] ) )\\
  = & \bigoplus_{B_i \in \perm{B}} (\nu v)(  v.\overline{\none}\mid \widetilde{z}.\overline{\none} \mid v.\some_{u, \lfv{B}} ; \outact{v}{x} . (  \piencodf{ B_i}_x \mid [v \leftrightarrow u] ) )\\
  \red & \bigoplus_{B_i \in \perm{B}}   u.\overline{\none} \mid \widetilde{z}.\overline{\none}  \mid \lfv{B}.\overline{\none} 
 %&= \bigoplus_{\perm{B}}   u.\overline{\none} \mid \widetilde{z}.\overline{\none}  \mid \lfv{B}.\overline{\none} \\
 \end{aligned}
\end{equation}
   We also have that 
    $  \expr{N} = \fail^{\widetilde{z}} \ B  \red \sum_{\perm{B}} \fail^{\widetilde{z} \cup \lfv{B} }  = \expr{M}$. 
 Furthermore, we have:
     \begin{equation}\label{eq:red_m}
     \begin{aligned}
       \piencodf{\expr{M}}_u &= \piencodf{\sum_{\perm{B}} \fail^{\widetilde{z} \cup \lfv{B} }}_u \\
       &= \bigoplus_{\perm{B}}\piencodf{ \fail^{\widetilde{z} \cup \lfv{B} }}_u\\
     &  = \bigoplus_{\perm{B}}    u.\overline{\none} \mid \widetilde{z}.\overline{\none}  \mid \lfv{B}.\overline{\none}
        \end{aligned}
    \end{equation}
From reductions in (\ref{eq:red_n}) and (\ref{eq:red_m}) one has $\piencodf{\expr{N}}_u\red \piencodf{\expr{M}}_u$, and the result follows with $n=1$ and $\piencodf{\expr{M}}_u=Q=Q'$. The full proof can be found in \appref{compandsucctwo}.
\end{proof}

\subsubsection{Success Sensitiveness}

Finally, we consider success sensitiveness. This requires extending \lamrsharfail and \spi with success predicates. 
%In \spi, we say that $P$ is unguarded if it does not occur behind a prefix.

%We prove the property following Definition~\ref{d:encoding}(4).Proof can be seen in Appendix \ref{app:succtwo}.

\begin{defi}
We extend the syntax of $\spi$ processes (Definition~\ref{d:spi}) 
with the $\checkmark$ construct, which we assume well typed. 
Also, we extend Definition~\ref{def:enc_lamrsharpifail} by 
defining $\piencodf{\checkmark}_u = \checkmark$
\end{defi}

\begin{defi}
We say that a process occurs \emph{guarded} when it occurs behind a prefix (input, output, closing of channels and non-deterministic session behavior). That is, $P$ is guarded if  $\alpha.P$ or $ \alpha;P$, where $ \alpha = \overline{x}(y), x(y), x.\overline{\close}, x.\close,$ $ x.\overline{\some}, x.\some_{(w_1, \cdots, w_n)} $. 
We say it occurs \emph{unguarded} if it is not guarded for any prefix.
\end{defi}

\begin{defi}[Success in \spi]
\label{def:Suc4}
We extend the syntax of $\spi$ processes  
with the $\checkmark$ construct, which we assume well-typed. 
We define 
$\succp{P}{\checkmark}$ to hold whenever there exists a $P'$
such that 
$P \red^* P'$
and $P'$ contains an unguarded occurrence of $\checkmark$.
\end{defi}

%\begin{nota}
%    We use the notation $\headfsum{\expr{M}}$ to be that $\forall M_i, M_j \in \expr{M}$ we have that $\headf{M_i} = \head{M_j}$ hence we say that $\headfsum{\expr{M}} = \headf{M_i}$ for some $M_i \in \expr{M}$
%\end{nota}

\begin{restatable}[Preservation of Success]{propo}{pressucctwo}
\label{Prop:checkprespi}
The $\checkmark$ at the head of a \revd{B31}{partially open} term is preserved to an unguarded occurrence of $\checkmark$ when applying the translation $\piencodf{\cdot}_u$ up to reductions and vice-versa. That is to say:
\begin{enumerate}
    \item $\forall M \in \lamrsharfail: \quad \headf{M} = \checkmark \implies \piencodf{M}_u \red^* (P \mid \checkmark) \oplus Q $
    \item $\forall  M \in \lamrsharfail: \quad \piencodf{M}_u =  (P \mid \checkmark) \oplus Q \implies \headf{M} = \checkmark$
\end{enumerate} 

\end{restatable}

\begin{proof}[Proof (Sketch)]
By induction on the structure of $M$.
For item (1),  consider the case $M=(N\ B)$ and  $\headf{N \ B} = \headf{N} = \checkmark$. This term's translation is
\[\piencodf{N \ B}_u = \bigoplus_{B_i \in \perm{B}} (\nu v)(\piencodf{N}_v \mid v.\some_{u, \lfv{B}} ; \outact{v}{x} . ([v \leftrightarrow u] \mid \piencodf{B_i}_x ) ).\]
%\joe{Tried to make this part clearer}
By the induction hypothesis,  $\checkmark$ is unguarded in $\piencodf{N}_u$ after a sequence of reductions, i.e., $ \piencodf{N}_u\red^* (\checkmark \mid P')\oplus Q'$, for some $\spi$ processes $P'$ and $Q'$. Thus, 
\[
\begin{aligned} 
\piencodf{N \ B}_u & \red^* \bigoplus_{B_i \in \perm{B}} (\nu v)((\checkmark \mid P')\oplus Q' \mid v.\some_{u, \lfv{B}} ; \outact{v}{x} . ([v \leftrightarrow u] \mid \piencodf{B_i}_x ) )\\
& \equiv  \checkmark \mid (\nu v)(P'\oplus Q' \mid v.\some_{u, \lfv{B}} ; \outact{v}{x} . ([v \leftrightarrow u] \mid \piencodf{B_j}_x ) )\\ 
& \quad \oplus \Big( 
\bigoplus_{B_i \in (\perm{B} \linsetminus B_j )  } \checkmark \mid (\nu v)(P'\oplus Q' \mid v.\some_{u, \lfv{B}} ; \outact{v}{x} . ([v \leftrightarrow u] \mid \piencodf{B_i}_x ) ) 
\Big)  \\
& \equiv  (\checkmark \mid P) \oplus Q \\ 
\end{aligned}
\]
% Let us expand and take out one of the processes from the sum for a $B_j \in \perm{B}$
% \[
% \begin{aligned} 
% & \equiv  \checkmark \mid (\nu v)(P'\oplus Q' \mid v.\some_{u, \lfv{B}} ; \outact{v}{x} . ([v \leftrightarrow u] \mid \piencodf{B_j}_x ) )\\ 
% & \quad \oplus \Big( 
% \bigoplus_{B_i \in (\perm{B} \setminus B_j )  } \checkmark \mid (\nu v)(P'\oplus Q' \mid v.\some_{u, \lfv{B}} ; \outact{v}{x} . ([v \leftrightarrow u] \mid \piencodf{B_i}_x ) ) 
% \Big) , \\
% & \equiv  (\checkmark \mid P) \oplus Q \\
% \end{aligned}
% \]
and the result follows by taking $P = (\nu v)(P'\oplus Q' \mid v.\some_{u, \lfv{B}} ; \outact{v}{x} . ([v \leftrightarrow u] \mid \piencodf{B_j}_x ) ) $ and 
$Q = \bigoplus_{B_i \in ({\perm{B}} \linsetminus B_j)} \checkmark \mid (\nu v)(P'\oplus Q' \mid v.\some_{u, \lfv{B}} ; \outact{v}{x} . ([v \leftrightarrow u] \mid \piencodf{B_i}_x ) )$.
The analysis for the other cases are similar; see \appref{app:succtwo} for details.
\end{proof}

The translation $\piencodf{\cdot}_u:\lamrsharfail \rightarrow \spi$ is success sensitive on well-formed closed expressions. 
\begin{restatable}[Success Sensitivity]{thms}{successsenscetwo}
\label{proof:successsenscetwo}
Let  \expr{M} be a closed well-formed $\lamrsharfail$-expression.
Then,
\[\expr{M} \Downarrow_{\checkmark}\iff \succp{\piencodf{\expr{M}}_u}{\checkmark} .\]
\end{restatable}

\begin{proof}[Proof (Sketch)]
 Suppose  $\expr{M} \Downarrow_{\checkmark} $.
    By Definition \ref{def:app_Suc3} there exists $\mathbb{M}'=M_1 + \cdots + M_k$ such that $\expr{M} \red^* \expr{M}'$ and
    $\headf{M_j} = \checkmark$, for some  $j \in \{1, \ldots, k\}$ and  $M_j$. By operational completeness (Theorem~\ref{l:app_completenesstwo}),  there exists $ Q$ such that $\piencodf{\expr{M}}_u  \red^* Q = \piencodf{\expr{M}'}_u$.
   Due to compositionality of $\piencodf{\cdot}$ and the homomorphic preservation of non-determinism,  we have:
\begin{itemize}
\item \(Q = \piencodf{M_1}_u \oplus \cdots \oplus \piencodf{M_k}_u\)
\item \revt{B32}{\(\piencodf{M_j}_u  = C[ \piencodf{\checkmark}_v ]= C[ \checkmark]\)}
\end{itemize}
% \daniele{I am not sure I understood how to get to the next paragraph. How do you know that $\head{M_j}=\checkmark$ or that $\checkmark$ is unguarded in the translation to apply the proposition?}
% \daniele{Above it says $\head{M_j'}=\checkmark$ and not $\head{M_j}$}.
% \joe{I have edited the definition for success in lambda (def 5.15). We didn't need the pre-congruence here. It should make more sense now}
By Proposition \ref{Prop:checkprespi}, item  (1), since $\headf{M_j} = \checkmark$ it follows that $ \piencodf{M_j}_u \red^*  P \mid \checkmark \oplus Q'$. Hence $Q$ reduces to a process that has an unguarded occurrence of $\checkmark$. The proof of the converse  is similar and can be found in \appref{app:succtwo}.
\end{proof}
As main result of this sub-section, we have the corollary below, which follows from the previously stated Theorems~\ref{t:preservationtwo}, 
\ref{l:app_completenesstwo}, 
\ref{l:app_soundnesstwo}, and 
\ref{proof:successsenscetwo}:

\begin{cor}
\label{cor:two}
Our translation  $ \piencodf{ \cdot } $ is a correct encoding, in the sense of \defref{d:encoding}.
\end{cor}

\noindent
Together, Corollary~\ref{cor:one} and Corollary~\ref{cor:two} ensure that $\lamrfail$ can be correctly translated into \spi, using \lamrsharfail as a stepping stone.
\section{Related Work}
\label{s:rw}

Closely related works have been already discussed in the introduction and throughout the paper; here we mention other related literature.

\subsubsection*{Intersection Types}
\revt{}{The first works on intersection types date back to the late 70s~(see, e.g.,~\cite{DBLP:journals/aml/CoppoD78,Pottinger80}) and consider intersections with the  \emph{idempotence} property (i.e., $\sigma \land \sigma = \sigma$). 
This formulation enables the analysis of \emph{qualitative} properties of $\lambda$-calculi, such as (strong) normalization and solvability.
By dropping idempotence, intersection types can   characterize \emph{quantitative} properties, such as, e.g., bounds on the number of steps needed to reach a normal form.
Early works on non-idempotent intersection types include~\cite{DBLP:conf/tacs/Gardner94,DBLP:journals/logcom/Kfoury00,DBLP:journals/tcs/KfouryW04}. 
The paper~\cite{DBLP:conf/lics/BonoD20}  overviews the origins, development, and applications of intersection types.}

\revd{}{Our work formally connects non-idempotent intersection types and classical linear logic extended with the modalities $\with$ and $\oplus$, interpreted in~\cite{CairesP17} as session types for non-deterministically available protocols. 
\srev{To the best of our knowledge, this is an unexplored angle. Prior connections between (non-idempotent) intersection types and linear logic arise in very different settings (see~\cite{DBLP:journals/pacmpl/MazzaPV18} and references therein). They include~\cite{DBLP:conf/icfp/NeergaardM04}, which presents a connection based on a correspondence between normalization and type inference;
the work~\cite{DBLP:journals/corr/abs-0905-4251,DBLP:journals/mscs/Carvalho18}, which shows a correspondence between the \emph{relational model} of linear logic and an non-idempotent intersection type system;} 
and~\cite{DBLP:conf/fossacs/Ehrhard20}, which concerns \emph{indexed} linear logic (cf.~\cite{DBLP:journals/apal/BucciarelliE00,DBLP:journals/apal/BucciarelliE01}).
%A rigorous semantic relationship between (colored) linear logic and non-idempotent intersection types  is addressed in~\cite{DBLP:journals/corr/GrelloisM15b,DBLP:phd/hal/Grellois16} from the (distant) perspective of \emph{higher-order model checking}.
}

\srev{The work~\cite{DBLP:journals/pacmpl/LagoVMY19} develops a type system for the $\pi$-calculus based on non-idempotent intersections. The type system ensures that processes are ``well-behaved''---they never produce run-time errors, and can always reduce to an idle process. Remarkably, they show that their type system is \emph{complete}: every well-behaved process is typable. Although their type system does not consider session types,  it is related to our work for it builds upon Mazza et al.'s~correspondence between linear logic and intersection types, given in terms of \emph{polyadic approximations}~\cite{DBLP:journals/pacmpl/MazzaPV18}.}

% \jp{Daniele: Mention differential lambda here?}

\subsubsection*{Other Resource $\lambda$-calculi}
A fine-grained treatment of duplication and erasing---similar to our design for $\lamrsharfail$---is present in Kesner and Lengrand's  $\lambda${\tt lxr}-calculus~\cite{DBLP:journals/iandc/KesnerL07}, a 
simply-typed, deterministic $\lambda$-calculus that is in correspondence with proof nets.
The $\lambda${\tt lxr}-calculus includes operators called weakening \(\mathcal{W}_{\_}(\_)\) and contraction \(\mathcal{C}_{\_}^{\_|\_}(\_)\) to deal with empty and non-empty sharing, respectively. 
In this approach, our terms 
$\lambda x. x \bag{x}$  and $\lambda x. y \bag{z}$ 
%previous examples 
would be expressed as \(\mathcal{C}_{x}^{x_1|x_2}(\lambda x. x_1\bag{x_2})\) and \(\mathcal{W}_x(\lambda x.y\bag{z})\), respectively.

Our approach is convenient when expressing the sharing of more than two occurrences of a variable in a term; as in, e.g., the $\lamrfail$-term $\lambda x.( x\bag{x,x})$ which would correspond to $\lambda x. (x_1\bag{x_2,x_3})\shar{x_1,x_2,x_3}{x}$ in $\lamrsharfail$.
In the $\lambda${\tt lxr}-calculus, contractions are binary, and so representing $\lambda x.( x\bag{x,x})$ requires the composition of two binary contractions.

More substantial differences appear at the level of types. 
As we have seen, in \lamrsharfail we use intersection types to define  well-typed and well-formed expressions (see~\figref{fig:typing_sharing} and~\figref{fig:wfsh_rules}, respectively).   In particular, recall the well-formedness rule for the sharing construct:
\begin{prooftree}
\AxiomC{\(\Gamma, x_1:\sigma, \ldots, x_k:\sigma \wfdash M:\tau \quad x \notin \dom{\Gamma} \quad  k\neq 0\)}
\LeftLabel{\redlab{FS:share}}
\UnaryInfC{\(\Gamma, x: \sigma^k \wfdash M\shar{x_1,\ldots, x_k}{x}:\tau\)}
\end{prooftree}
  %$x$ is a variable not in the context $\Gamma$, and the  type assignments $x_1:\sigma,\ldots,  x_k:\sigma$ are replaced a single type assignment 
where, as mentioned above, $\sigma^k$ denotes the intersection type $\sigma\wedge \ldots \wedge \sigma$. 
Differently, the typing rule for contraction in the $\lambda${\tt lxr}-calculus involves an arbitrary (simple) type $A$:
\begin{prooftree}
\AxiomC{\(\Gamma, y:A, z:A \vdash M:B \)}
\LeftLabel{(Cont)}
\UnaryInfC{\(\Gamma, x: A \vdash \mathcal{C}_{x}^{y|z}(M):B\)}
\end{prooftree}
%which  contracts two variables assignments with the same (simple) type $y:A,z:A$ to one type assignment $x:A$. 
Our weakening rule $\redlab{FS:weak}$ types the empty sharing term $M\shar{}{x}$ as follows:
\begin{prooftree}
\AxiomC{\(\Gamma \wfdash M:\tau\)}
\LeftLabel{\redlab{FS:weak}}
\UnaryInfC{\(\Gamma, x:\omega \wfdash M\shar{}{x}:\tau\)}
\end{prooftree}
Hence, the context $\Gamma$ is weakened with a variable assignment $x:\omega$, where $\omega$ denotes the empty type.  In contrast, weakening in the $\lambda${\tt lxr}-calculus involves a  (simple) type $A$: 
\begin{prooftree}
\AxiomC{\(\Gamma \vdash M:A \)}
\LeftLabel{(Weak)}
\UnaryInfC{\(\Gamma, x:B \vdash \mathcal{W}_x(M): A\)}
\end{prooftree}
Hence, the context can be weakened with an assignment $x:B$, where $B$ is a simple type.

\smallskip
\srev{Inspired by the multiplicative exponential fragment of linear logic, Kesner and Renaud~\cite{DBLP:journals/tcs/KesnerR11} define the so-called \emph{prismoid of resources}, a parametric framework of simply-typed $\lambda$-calculi in which each language incorporates different choices for contraction, weakening, and substitution operations. 
The prismoid defines a uniform and general setting for establishing key properties of typed terms, 
including simulation of $\beta$-reduction, confluence, and strong normalization.
One of the languages included in the prismoid is a minor variant of the $\lambda${\tt lxr}-calculus, which we have just mentioned.}

\smallskip

There are some similarities between $\lamrfail$ and the differential $\lambda$-calculus, introduced in~\cite{DBLP:journals/tcs/EhrhardR03}. Both express non-deterministic choice via sums and use linear head reduction for evaluation. In particular, our fetch rule, which consumes non-deterministically elements from a bag, is related to the derivation (which has similarities with substitution) of a differential term. However, the focus of~\cite{DBLP:journals/tcs/EhrhardR03} is not on typability nor encodings to process calculi; instead they relate the Taylor series of analysis to the linear head reduction of $\lambda$-calculus.

\subsubsection*{Functions as Processes}
A source of inspiration for our developments is the work by Boudol and Laneve~\cite{DBLP:conf/birthday/BoudolL00}. As far as we know, this 
is the only prior study that connects $\lambda$ and $\pi$ from a resource-oriented perspective, via an encoding of a $\lambda$-calculus with multiplicities into a $\pi$-calculus without sums.
%There are key technical differences between our work and~\cite{DBLP:conf/birthday/BoudolL00}.
The goal of~\cite{DBLP:conf/birthday/BoudolL00} is different from ours, as they study the discriminating power of semantics for $\lambda$ as induced by encodings into~$\pi$. In contrast, we study how typability delineates the encodability of resource-awareness across sequential and concurrent realms.
The source and target calculi in~\cite{DBLP:conf/birthday/BoudolL00} are untyped, whereas  we consider typed calculi and our encodings preserve typability. 
As a result, the encoding in~\cite{DBLP:conf/birthday/BoudolL00} is conceptually different from ours; remarkably, our encoding respects linearity and homomorphically translates sums. 

Prior works have studied encodings of  typed $\lambda$-calculi into typed $\pi$-calculi; see, e.g.,~~\cite{DBLP:journals/mscs/Sangiorgi99,DBLP:conf/birthday/BoudolL00,DBLP:books/daglib/0004377,DBLP:conf/fossacs/BergerHY03, DBLP:conf/fossacs/ToninhoCP12,DBLP:conf/rta/HondaYB14,DBLP:conf/esop/ToninhoY18}. None of these works consider {non-determinism} and {failures}; the one exception is the encoding in~\cite{CairesP17}, which involves a $\lambda$-calculus with exceptions and failures (but without non-determinism due to bags, as in \lamrfail) for which no reduction semantics is given. As a result, the encoding in~\cite{CairesP17} is different from ours, and is only shown to preserve typability:   properties such as operational completeness, operational soundness, and success sensitivity---important in our developments---are not considered.

\section{Concluding Remarks}
\label{s:disc}

\subsubsection*{Summary}
We developed a correct encoding of \lamrfail, a new resource $\lambda$-calculus in which expressions feature non-determinism and explicit failure, into \spi, a session-typed $\pi$-calculus in which behavior is non-deterministically available: session protocols may perform as stipulated but also fail. 
Our encodability result is obtained by appealing to  \lamrsharfail, an intermediate language with a \emph{sharing construct} that simplifies the treatment of variables in expressions. 
To our knowledge, we are the first to relate typed $\lambda$-calculi and typed $\pi$-calculi encompassing non-determinism and   failures, while connecting intersection types and session types, two different mechanisms for resource-awareness in sequential and concurrent settings, respectively.  

%We close the paper with a brief discussion.
 
%\subsection{Discussion}
\subsubsection*{Design of \lamrfail (and \lamrsharfail)}
\revd{}{The design of \lamrfail has been influenced by the logically justified treatment of non-determinism and explicit failure in \spi. Our correct encoding of \lamrfail into \spi makes this influence precise by connecting terms and processes but also their associated intersection types and linear logic propositions.}
We have also adopted features from previous resource $\lambda$-calculi, in particular those in~\cite{DBLP:conf/concur/Boudol93,DBLP:conf/birthday/BoudolL00,PaganiR10}.
Major similarities between \lamrfail and these calculi include: 
as in \cite{DBLP:conf/birthday/BoudolL00}, our semantics performs lazy evaluation and linear substitution on the head variable;  
as in~\cite{PaganiR10}, our reductions lead to non-deterministic sums.
A distinctive feature of \lamrfail    is its lazy treatment of failures via the   term $\fail^{\widetilde{x}}$. In contrast,  in~\cite{DBLP:conf/concur/Boudol93,DBLP:conf/birthday/BoudolL00} there is no dedicated term to represent failure. 
The non-collapsing semantics for non-de\-ter\-mi\-nism is another distinctive feature of \lamrfail. 

\srev{Our design for \lamrsharfail has been informed by the atomic $\lambda$-calculus introduced in~\cite{DBLP:conf/lics/GundersenHP13}. 
Also, our translation from \lamrfail into \lamrsharfail (\defref{def:enctolamrsharfail}) borrows insights from translations given in~\cite{DBLP:conf/lics/GundersenHP13}.}
The calculus \lamrsharfail is also loosely related to the $\lambda$-calculus with sharing in~\cite{GhilezanILL11}, which considers (idempotent) intersection types.
Notice that the calculi in~\cite{DBLP:conf/lics/GundersenHP13,GhilezanILL11} do not consider explicit failure nor non-determinism.
We  distinguish between \emph{well-typed} and \emph{well-formed} expressions: this  allows us to make fail-prone evaluation in \lamrfail explicit. It is interesting that  explicit failures can be elegantly encoded as protocols in \spi---this way, we make the most out of \spi's expressivity.

Bags in \lamrfail have \emph{linear} resources, which are used exactly once. 
In recent work, we have defined an extension of \lamrfail in which bags contain both linear and  \emph{unrestricted} resources, as in~\cite{PaganiR10}, and 
established that our approach to encodability into \spi extends to such an enriched language~\cite{DBLP:conf/types/PaulusN021}. This development requires  the full typed process framework in~\cite{CairesP17}, with replicated processes and labeled choices (not needed to encode \lamrfail). 

\subsubsection*{Future Work}

The approach and results developed here enable us to tackle open questions that go beyond the scope of this work. We comment on some of them:
\begin{itemize}
	\item 
%having considered calculi with non-collapsing non-determinism, we wish to explore whether our correct encoding can be modified to operate in a setting in which source and target languages implement \emph{collapsing} non-determinism.
%%, by developing appropriate variants of \lamrfail, \lamrsharfail, and \spi. 
%Second, 
It would be useful to investigate the \emph{relative expressiveness} of \lamrfail with respect to other resource calculi, such as those in~\cite{DBLP:conf/birthday/BoudolL00,PaganiR10}. Derived 
encodability (and non-encodability) results could potentially unlock transfer of reasoning techniques between different calculi.
%\item We would like to study the effects of \emph{asynchronous communication} on source and target languages. In a nutshell, asynchrony amounts to equipping a language's semantics with mechanisms like buffers, which provide support for in-transit messages. In this respect, the asynchronous functional calculus with session-based concurrency and exceptions in~\cite{DBLP:journals/pacmpl/FowlerLMD19} could provide a head start. 
\item \revd{}{Besides transfer of techniques, one application of encodings between sequential and concurrent calculi is in the design of functional concurrent languages with advanced features. In this respect, it should be feasible to develop a variant of Wadler's GV~\cite{DBLP:conf/icfp/Wadler12} with non-determinism, resources, explicit failure, and session communication by exploiting our correct encodings from \lamrfail to \spi.}

 \item \revdaniele{It would be relevant to investigate \emph{decidability properties} of the intersection type systems for $\lamrfail$ and  $ \lamrsharfail $. Our translation is proven correct under the assumption that we consider only well-formed $\lamrfail$-terms.  The type assignment problem for intersection type systems is, in general, undecidable~\cite{DBLP:conf/popl/Leivant83}; it would be interesting to consider decidable fragments of intersection type systems via, for instance,  ranking restrictions~\cite{DBLP:journals/tcs/Bakel95}.}
\item It would be insightful to establish \emph{full abstraction}   for our translation of  \lamrfail into \spi. 
We choose not to consider it because,  as argued in~\cite{DBLP:journals/mscs/GorlaN16}, full abstraction is not an informative criterion when it comes to an encoding's quality. Establishing full abstraction requires developing the behavioral theory of \lamrfail and \spi, which is relevant and challenging in itself.
\end{itemize}

\subsubsection*{Acknowledgements}
We are grateful to the anonymous reviewers for their detailed and helpful comments.
%Paulus and P\'{e}rez 
We gratefully acknowledge the support of the Dutch Research Council (NWO) under project No.\,016.Vidi.189.046 (Unifying Correctness for Communicating Software).
Daniele Nantes-Sobrinho has been partially funded by the EPSRC Fellowship `VeTSpec: Verified Trustworthy Software Specification' (EP/R034567/1) and Edital DPI/DPG n.~03/2020.

\bibliographystyle{alphaurl}
\bibliography{references}

\newpage
\appendix
\tableofcontents

\newlist{myEnumerate}{enumerate}{9}
\setlist[myEnumerate,1]{label=(\arabic*)}
\setlist[myEnumerate,2]{label=(\Roman*)}
\setlist[myEnumerate,3]{label=(\Alph*)}
\setlist[myEnumerate,4]{label=(\roman*)}
\setlist[myEnumerate,5]{label=(\alph*)}
\setlist[myEnumerate,6]{label=(\arabic*)}
\setlist[myEnumerate,7]{label=(\Roman*)}
\setlist[myEnumerate,8]{label=(\Alph*)}
\setlist[myEnumerate,9]{label=(\roman*)}

%%%%%%%%%%%%%%%%%%%%%%%%%%%%%%%%%%%%%%%%%%%%%%%%%%%%

\newpage

\section{Appendix to \texorpdfstring{\secref{sec:lamfailintertypes}}{§ 2.3}}
\label{app:lamfailintertypes}

\explemfailnofail*

\begin{proof}
By induction on the structure of $ M $:

    \begin{myEnumerate}
        
        \item When $M = x$ then we have $x \headlin{ N / x}   = N$ and may derive the derivation of $ \Gamma \vdash N: \tau $ with $x \not \in \dom{\Gamma}$. By taking $\Gamma_1 = \emptyset$ and $\Gamma_2 = \Gamma$ as $\Gamma = \emptyset \contexcat \Gamma$ the case follows as $ \Gamma \vdash N : \tau$ and
            
                \begin{prooftree}
                    \AxiomC{}
                    \LeftLabel{\redlab{T:var}}
                    \UnaryInfC{\( x: \sigma \vdash x : \sigma\)}
                \end{prooftree}
        
        \item When $M = (M\ B)$ then we have that $(M\ B)\headlin{ N/x}  = (M \headlin{ N/x })\ B$. Let us consider two cases:
            
            \begin{myEnumerate}
                
                \item When $x \in \lfv{M \headlin{ N/x }}$
                    \begin{prooftree}
                        \AxiomC{\( \Gamma', x:\sigma^{k-1}  \vdash M \headlin{ N/x } : \pi \rightarrow \tau' \)}
                        \AxiomC{\( \Delta \vdash B : \pi \)}
                            \LeftLabel{\redlab{T:app}}
                        \BinaryInfC{\( (\Gamma', x:\sigma^{k-1} ) \contexcat \Delta \vdash (M \headlin{ N/x })\ B : \tau'\)}
                    \end{prooftree}
                    
                    By the IH we have that $\Gamma', x:\sigma^{k-1} \vdash M \headlin{ N/x } : \pi \rightarrow \tau'$ implies that $\exists \  \Gamma_1', \Gamma_2$ such that  $\Gamma_1' , x:\sigma^k \vdash M: \tau$, and $\Gamma_2 \vdash N : \sigma$ with $\Gamma' = \Gamma_1' \contexcat \Gamma_2$.
                    
                    \begin{prooftree}
                        \AxiomC{\( \Gamma_1', x:\sigma^{k}  \vdash M  : \pi \rightarrow \tau' \)}
                        \AxiomC{\( \Delta \vdash B : \pi \)}
                            \LeftLabel{\redlab{T:app}}
                        \BinaryInfC{\( (\Gamma_1', x:\sigma^{k} ) \contexcat \Delta \vdash M\ B : \tau'\)}
                    \end{prooftree}
                    
                \item When $x \not \in \lfv{M \headlin{ N/x }}$
                    \begin{prooftree}
                      \AxiomC{\( \Gamma'  \vdash M \headlin{ N/x } : \pi \rightarrow \tau' \)}
                        \AxiomC{\( \Delta \vdash B : \pi \)}
                            \LeftLabel{\redlab{T:app}}
                        \BinaryInfC{\( \Gamma' \contexcat \Delta \vdash (M \headlin{ N/x })\ B : \tau'\)}
                    \end{prooftree}
                    
                    By the IH we have that $\Gamma' \vdash M \headlin{ N/x } : \pi \rightarrow \tau'$ implies that $\exists \  \Gamma_1', \Gamma_2$ such that  $\Gamma_1' , x:\sigma \vdash M: \tau$, and $\Gamma_2 \vdash N : \sigma$ with $\Gamma' = \Gamma_1' \contexcat \Gamma_2$.
                    
                    \begin{prooftree}
                        \AxiomC{\( \Gamma_1', x:\sigma  \vdash M \headlin{ N/x } : \pi \rightarrow \tau' \)}
                        \AxiomC{\( \Delta \vdash B : \pi \)}
                            \LeftLabel{\redlab{T:app}}
                        \BinaryInfC{\( (\Gamma_1', x:\sigma ) \contexcat \Delta \vdash M\ B : \tau'\)}
                    \end{prooftree}
                
            \end{myEnumerate}

        \item When $M = M \esubst{B}{y}$ then we have that $(M\ \esubst{B}{y})\headlin{ N/x } = (M\headlin{ N/x })\ \esubst{B}{y}$ where $x \not = y$
            
            \begin{myEnumerate}
                
                \item When $x \in \lfv{M \headlin{ N/x }}$
                    
                    \begin{prooftree}
                       \AxiomC{\( \Gamma', x:\sigma^{k-1} ,  {y}:\delta^{j} \vdash (M\headlin{ N/x }) : \tau \)}
                             \AxiomC{\( \Delta \vdash B : \delta^{j} \)}
                        \LeftLabel{\redlab{T:ex \dash sub}}    
                        \BinaryInfC{\( \Gamma',  {y}:\delta^{j} \contexcat \Delta \vdash (M\headlin{ N/x }) \esubst{ B }{ y } : \tau \)}
                    \end{prooftree}
                    
                    By the IH we have that $\Gamma', x:\sigma^{k-1} ,  {y}:\delta^{j} \vdash (M\headlin{ N/x }) : \tau$ implies that $\exists \  \Gamma_1', \Gamma_2$ such that  $\Gamma_1' , x:\sigma^k, {y}:\delta^{j} \vdash M: \tau$, and $\Gamma_2 \vdash N : \sigma$ with $\Gamma',  {y}:\delta^{j} = (\Gamma_1',  {y}:\delta^{j}) \contexcat \Gamma_2$.
                    
                    \begin{prooftree}
                       \AxiomC{\( \Gamma'_1, x:\sigma^{k} ,  {y}:\delta^{j} \vdash M : \tau \)}
                             \AxiomC{\( \Delta \vdash B : \delta^{j} \)}
                        \LeftLabel{\redlab{T:ex \dash sub}}    
                        \BinaryInfC{\( \Gamma'_1 \contexcat \Delta \vdash M \esubst{ B }{ y } : \tau \)}
                    \end{prooftree}
                    
                \item When $x \not \in \lfv{M \headlin{ N/x }}$
                    
                    \begin{prooftree}
                       \AxiomC{\( \Gamma' ,  {y}:\delta^{k} \vdash (M\headlin{ N/x }) : \tau \)}
                             \AxiomC{\( \Delta \vdash B : \delta^{k} \)}
                        \LeftLabel{\redlab{T:ex \dash sub}}    
                        \BinaryInfC{\( \Gamma' \contexcat \Delta \vdash (M\headlin{ N/x }) \esubst{ B }{ y } : \tau \)}
                    \end{prooftree}
                    
                    By the IH we have that $\Gamma' ,  {y}:\delta^{k} \vdash (M\headlin{ N/x }) : \tau$ implies that $\exists \  \Gamma_1', \Gamma_2$ such that  $\Gamma_1' , x:\sigma \vdash M: \tau$, and $\Gamma_2 \vdash N : \sigma$ with $\Gamma',  {y}:\delta^{k} = (\Gamma_1',  {y}:\delta^{k}) \contexcat \Gamma_2$.
                    
                    \begin{prooftree}
                       \AxiomC{\( \Gamma'_1, x:\sigma ,  {y}:\delta^{k} \vdash M : \tau \)}
                             \AxiomC{\( \Delta \vdash B : \delta^{k} \)}
                        \LeftLabel{\redlab{T:ex \dash sub}}    
                        \BinaryInfC{\( \Gamma'_1 \contexcat \Delta \vdash M \esubst{ B }{ y } : \tau \)}
                    \end{prooftree}
                    
            \end{myEnumerate}

        \item When $M = \lambda y . M$ then linear head substitution is undefined on this term as $\headf{M} \not = x$.

        \item When $M = \fail^{\widetilde{x}}$ then $M$ is not well typed. \qedhere

    \end{myEnumerate}

\end{proof}

\subexpone*

\begin{proof}
By induction on the reduction rule applied.
There are four possible cases. 
    
    \begin{myEnumerate}
    
        \item When $\expr{M}'$ is reduced to via the Rule~\redlab{R:Beta}
    
            \begin{prooftree}
                \AxiomC{}
                \LeftLabel{\redlab{R:Beta}}
                \UnaryInfC{\((\lambda x. M) B \red M\ \esubst{B}{x}\)}
            \end{prooftree}
            
            Then $\expr{M}' = M\ \esubst{B}{x} $ can be type as follows:
            
            \begin{prooftree}
                    \AxiomC{\( {\Gamma ,  {x}:\sigma^{k} \vdash M : \tau} \)}
                     \AxiomC{\( \Delta \vdash B : \sigma^{k} \)}
                \LeftLabel{\redlab{T:ex \dash sub}}    
                \BinaryInfC{\( {\Gamma \contexcat \Delta \vdash M \esubst{ B }{ x } : \tau} \)}
            \end{prooftree}
            
            From the typing of $\expr{M}'  $ we can deduce that $\expr{M} = (\lambda x. M) B $ may be typed by:
            
            \begin{prooftree}
                \AxiomC{\( {\Gamma , {x}: \sigma^k \vdash M : \tau} \)}
                \LeftLabel{\redlab{T:abs}}
                \UnaryInfC{\( \Gamma \vdash \lambda x . M :  \sigma^k  \rightarrow \tau \)}
                \AxiomC{\( \Delta \vdash B : \pi \)}
                    \LeftLabel{\redlab{T:app}}
                \BinaryInfC{\( {\Gamma \contexcat \Delta \vdash (\lambda x. M) B : \tau}\)}
            \end{prooftree}
        
        \item When $\expr{M}'$ is reduced to via the Rule~\redlab{R:Fetch}
            
            \begin{prooftree}
                \AxiomC{$\headf{M} = x$}
                \AxiomC{$B = \bag{N_1, \dots ,N_k} \ , \ k\geq 1 $}
                \AxiomC{$ \#(x,M) = k $}
                \LeftLabel{\redlab{R:Fetch}}
                \TrinaryInfC{\(
                M\ \esubst{ B}{x } \red M \headlin{ N_{1}/x } \esubst{ (B\linsetminus N_1)}{ x }  + \cdots + M \headlin{ N_{k}/x } \esubst{ (B\linsetminus N_k)}{x}
                \)}
            \end{prooftree}
        
        Let us consider two cases:
        
        \begin{myEnumerate}
            
            \item The bag $B$ has $k$ elements where $k > 1$, then we type $ M \headlin{ N_{i}/x } \esubst{ (B\linsetminus N_i)}{ x } $ with the derivation $\Pi_i$ to be:
            
                \begin{prooftree}
                        \AxiomC{\( \Gamma ,  {x}:\sigma^{k-1} \vdash M \headlin{ N_{1}/x } : \tau \)}
                         \AxiomC{\( \Delta \vdash (B\linsetminus N_1) : \sigma^{k-1} \)}
                    \LeftLabel{\redlab{T:ex \dash sub}}    
                    \BinaryInfC{\( \Gamma \contexcat \Delta \vdash M \headlin{ N_{1}/x } \esubst{ (B\linsetminus N_1)}{ x }  : \tau \)}
                \end{prooftree}
            
            We can type the sum with each derivation $ \Pi_i$ to be 
            
                \hspace*{-20pt}                
                \begin{minipage}{\linewidth}
                \begin{prooftree} \hskip 17pt
                        \AxiomC{$ \Pi_1$}
                        \UnaryInfC{$ \Gamma \contexcat \Delta \vdash M \headlin{ N_{1}/x } \esubst{ (B\linsetminus N_1)}{ x } : \tau$}
                        
                        \AxiomC{$ \Pi_k$}
                        \UnaryInfC{$ \Gamma \contexcat \Delta \vdash M \headlin{ N_{k}/x } \esubst{ (B\linsetminus N_k)}{x} : \tau$}
                            %\LeftLabel{\redlab{T:\oneb}}
                        \UnaryInfC{$ \vdots $}
                    %\LeftLabel{\redlab{T:sum}}
                    \BinaryInfC{$ \Gamma \contexcat \Delta \vdash M \headlin{ N_{1}/x } \esubst{ (B\linsetminus N_1)}{ x }  + \cdots + M \headlin{ N_{k}/x } \esubst{ (B\linsetminus N_k)}{x}: \tau$}
                \end{prooftree}            
                \end{minipage}             
             
                By the anti-substitution lemma (Lemma~\ref{lem:antisubt_lem}) we have that $\exists \ \Gamma_1, \Gamma_2$  such that  $\Gamma_1 , x:\sigma^k \vdash M: \tau$, and $\Gamma_2 \vdash N_i : \sigma$ with $\Gamma = \Gamma_1 \contexcat \Gamma_2$ and finally we have:
                
                \begin{prooftree}
                        \AxiomC{\( \Gamma_1 ,  {x}:\sigma^k \vdash M : \tau \)}
                         \AxiomC{\( \Delta \contexcat \Gamma_2 \vdash B : \sigma^{k} \)}
                    \LeftLabel{\redlab{T:ex \dash sub}}    
                    \BinaryInfC{\( \Gamma \contexcat \Delta \vdash M \esubst{ B }{ x }  : \tau \)}
                \end{prooftree}
            
            notice that we make use that $\Gamma_2 \vdash N_i : \sigma$ to ensure that the bag $B$ is well typed.
            
            \item The bag $B$ has one element, then we type $ M \headlin{ N_{i}/x } \esubst{ \oneb }{ x } $ with the derivation $\Pi$ to be:
            
                \begin{prooftree}
                        \AxiomC{\( \Gamma  \vdash M \headlin{ N_{1}/x } : \tau \)}
                         \AxiomC{\( \Delta \vdash \oneb : \omega \)}
                    \LeftLabel{\redlab{T:ex \dash sub}}    
                    \BinaryInfC{\( \Gamma \contexcat \Delta \vdash M \headlin{ N_{1}/x } \esubst{ \oneb}{ x }  : \tau \)}
                \end{prooftree}
            
        \end{myEnumerate}

        By the anti-substitution lemma (Lemma~\ref{lem:antisubt_lem}) we have that $\exists \ \Gamma_1, \Gamma_2$  such that  $\Gamma_1 , x:\sigma \vdash M: \tau$ , and $\Gamma_2 \vdash N_1 : \sigma$ with $\Gamma = \Gamma_1 \contexcat \Gamma_2$ and finally we have:
                
            \begin{prooftree}
                    \AxiomC{\( \Gamma_1 ,  {x}:\sigma \vdash M : \tau \)}
                     \AxiomC{\( \Delta \contexcat \Gamma_2 \vdash \bag{N_1} : \sigma \)}
                \LeftLabel{\redlab{T:ex \dash sub}}    
                \BinaryInfC{\( \Gamma \contexcat \Delta \vdash M \headlin{ N_{1}/x } \esubst{ N }{ x }  : \tau \)}
            \end{prooftree}

        \item When $\expr{M}'$ is reduced to via the Rule~\redlab{R:TCont}

            \begin{prooftree}
                    \AxiomC{$   M \red M'_{1} + \cdots + M'_{k} $}
                    \LeftLabel{\redlab{R:TCont}}
                    \UnaryInfC{$ C[M] \red  C[M'_{1}] + \cdots +  C[M'_{k}] $}
            \end{prooftree}
        
            Hence the proof follows by the IH on $M$.
        
        \item When $\expr{M}'$ is reduced to via the Rule~\redlab{R:ECont}
        
            \begin{prooftree}
                    \AxiomC{$ \expr{M}  \red \expr{M}'  $}
                    \LeftLabel{\redlab{R:ECont}}
                    \UnaryInfC{$D[\expr{M}]  \red D[\expr{M}']  $}
            \end{prooftree} 
            
            Hence the proof follows by the IH on $M$. \qedhere
        
        % \item When $\expr{M}'$ is reduced to via the Rule~ \redlab{R:Fail}, \redlab{R:Cons_1}, \redlab{R:Cons_2}
        
        %     \begin{prooftree}
        %         \AxiomC{$\#(x,M) \neq \size{B} \qquad \widetilde{y} = (\mfv{M} \setminus x) \uplus \mfv{B} $}
        %         \LeftLabel{\redlab{R:Fail}}
        %         \UnaryInfC{\(  M\ \esubst{ B}{x } \red {}  \displaystyle\sum_{\perm{B}} \fail^{\widetilde{y}} \)}
        %     \end{prooftree}
        
        %     \begin{prooftree}
        %         \AxiomC{$  \widetilde{y} = \mfv{B} $}
        %         \LeftLabel{$\redlab{R:Cons_1}$}
        %         \UnaryInfC{\(  \fail^{\widetilde{x}}\ B \red{}  \displaystyle\sum_{\perm{B}} \fail^{\widetilde{x} \uplus \widetilde{y}} \)}
        %     \end{prooftree}
            
        %     \begin{prooftree}   
        %     \AxiomC{$\size{B} = k  \qquad \#(z , \widetilde{x}) + k  \not= 0 \qquad  \widetilde{y} = \mfv{B}$}
        %         \LeftLabel{$\redlab{R:Cons_2}$}
        %         \UnaryInfC{\( \fail^{\widetilde{x}}\ \esubst{B}{z}  \red {} \displaystyle \sum_{\perm{B}} \fail^{(\widetilde{x} \setminus z) \uplus \widetilde{y}}  \)}
        %     \end{prooftree}

        %     Hence $\expr{M}'$ isn't well-typed contradicting $\Gamma \vdash \expr{M}' : \tau $.

    \end{myEnumerate}
    
\end{proof}

% This section contains the complete proof of Lemma~\ref{lem:subt_lem} and Theorem~\ref{t:app_lamrsr}.
\subtlemfail*

\begin{proof}
By structural induction on $M$ with $\headf{M}=x$. There are three cases to be analyzed:

\begin{myEnumerate}
\item $M=x$.

This case follows trivially. First,  $x:\sigma \wfdash x:\sigma$ and $\Gamma=\emptyset$.  Second,  $x\headlin{N/x}=N$, by definition. Since $\Delta\wfdash N:\sigma$, by hypothesis, the result follows.

% \daniele{I think this case is missing.} \joe{there is a short sentence that covers this case aboue, $M=x$ is well typed in \lamr so its proved in the the previous section}

% \daniele{Ok, perhaps a sentence here and a reference to the mentioned result would suffice.}

    \item $M = M'\ B$.
    
    In this case, $\headf{M'\ B} = \headf{M'} = x$, and by inversion of the typing derivation one has the following derivation:

    \begin{prooftree}
        \AxiomC{$\Gamma_1 , x:\sigma^m \wfdash M': \delta^{j}  \rightarrow \tau$}\
        \AxiomC{$\Gamma_2 \wfdash B :  \delta^{l} $}
    	\LeftLabel{\redlab{F{:}app}}
        \BinaryInfC{$ ( \Gamma_1 , x:\sigma^m) \contexcat \Gamma_2 \wfdash M'B:\tau $}    
    \end{prooftree}
 where $\Gamma,x:\sigma^k= ( \Gamma_1 , x:\sigma^m) \contexcat \Gamma_2$, $\delta$ is a strict type, and $j,l,m$ are non-negative  integers, possibly different with $m \geq 1$.
 
 By IH, we get $\Gamma_1 \contexcat \Delta , x:\sigma^{m-1} \wfdash M'\headlin{N/x}:\delta^{j} \rightarrow \tau $, which gives the following derivation: 
    \begin{prooftree}
        \AxiomC{$\Gamma_1 \contexcat \Delta , x:\sigma^{m-1} \wfdash M'\headlin{ N / x }:  \delta^{j}  \rightarrow \tau$}\
        \AxiomC{$\Gamma_2 \wfdash B :  \delta^{l} $}
    	\LeftLabel{\redlab{F{:}app}}
        \BinaryInfC{$(\Gamma_1 \contexcat \Delta , x:\sigma^{m-1}) \contexcat \Gamma_2  \wfdash ( M'\headlin{ N / x } ) B:\tau $}    
    \end{prooftree}
    Therefore, from \defref{def:linsubfail}, one has $\Gamma \contexcat \Delta , x:\sigma^{k-1} \wfdash ( M' B) \headlin{ N / x }  :\tau $, and the result follows.
\item $M = M'\esubst{B }{ y}$.

In this case,  $\headf{M'\esubst{B }{ y}} = \headf{M'} = x$, with $x \not = y$, and by  inversion of the typing derivation one has the following derivation:

\begin{prooftree}
    \AxiomC{\( \Gamma_1 , {y}:\delta^{l} , x:\sigma^m \wfdash M' : \tau \)}
    \AxiomC{\( \Gamma_2 \wfdash B :  \delta^{j} \)}
	\LeftLabel{\redlab{F{:}ex \dash sub}}
    \BinaryInfC{\( ( \Gamma_1  ,  x:\sigma^m ) \contexcat \Gamma_2 \wfdash M' \esubst{B }{ y} : \tau \)}
\end{prooftree}
 where $\Gamma , x:\sigma^k = ( \Gamma_1  ,  x:\sigma^m ) \contexcat \Gamma_2 $, $\delta$ is a strict type and $j,l,m$ are positive integers with $m \geq 1$.
By IH, we get $ ( \Gamma_1 , {y}:\delta^{l} , x:\sigma^{m-1}) \contexcat \Delta \wfdash M'\headlin{N/x}:\tau$ and 
\begin{prooftree}
    \AxiomC{\(  ( \Gamma_1 , {y}:\delta^{l} , x:\sigma^{m-1}) \contexcat \Delta \wfdash  M' \headlin{ N / x } : \tau \)}
    \AxiomC{\( \Gamma_2 \wfdash B : \delta^{j} \)}
	\LeftLabel{\redlab{F{:}ex \dash sub}}
    \BinaryInfC{\(  ( \Gamma_1 , {y}:\delta^{l} , x:\sigma^{m-1}) \contexcat \Delta \contexcat \Gamma_2  \wfdash M' \headlin{ N / x } \esubst{ B }{ y} : \tau \)}
\end{prooftree}
\end{myEnumerate}

From \defref{def:linsubfail}, $M' \esubst{ B }{ y} \headlin{ N / x } = M' \headlin{ N / x} \esubst{ B }{ y}$, therefore, $\Gamma \contexcat \Delta , x:\sigma^{k-1} \wfdash (M'\esubst{ B }{ y})\headlin{ N / x }:\tau$ and  the result follows.
%\qed
\end{proof}

\applamrfailsr*

\begin{proof} By structural induction on the reduction rules. We proceed by analysing the rule applied in $\expr{M}$. There are seven cases:

\begin{myEnumerate}

	\item Rule~$\redlab{R:Beta}$.
	
	Then $\expr{M} = (\lambda x . M)B \red M\ \esubst{B}{x}=\expr{M}'$.
% 	and 
% 	\begin{prooftree}
% 	    \AxiomC{}
% 	    \LeftLabel{\redlab{Beta}}
% 	    \UnaryInfC{\((\lambda x. M) B \rightarrow M\ \esubst{B}{x}\)}
% 	\end{prooftree}

 	Since $\Gamma\wfdash \expr{M}:\tau$, by \revo{A15}{ inversion of the typing derivation} one has the following derivation:
	\begin{prooftree}
			\AxiomC{$ \Gamma' , {x}:\sigma^{j}  \wfdash  M: \tau $}
			\LeftLabel{\redlab{F{:}abs}}
            \UnaryInfC{$ \Gamma' \wfdash \lambda x. M: \sigma^{j} \rightarrow \tau $}
            \AxiomC{$\Delta \wfdash B: \sigma^{k} $}
			\LeftLabel{\redlab{F{:}app}}
		\BinaryInfC{$ \Gamma' \contexcat \Delta \wfdash (\lambda x. M) B:\tau $}
	\end{prooftree}
	for $\Gamma = \Gamma' \contexcat \Delta $. Notice that
    \begin{prooftree}
                \AxiomC{$ \Gamma' , {x}:\sigma^{j} \wfdash  M: \tau $}
                      \AxiomC{$\Delta \wfdash B:\sigma^{k}  $}
                \LeftLabel{\redlab{F{:}ex \dash sub}}
            \BinaryInfC{$ \Gamma' \contexcat \Delta \wfdash M \esubst{ B }{ x }:\tau $}
    \end{prooftree}
    
    Therefore, $ \Gamma\wfdash \expr{M}':\tau$ and the result follows.

	\item  Rule~$\redlab{R:Fetch}$.

	Then $ \expr{M} = M\ \esubst{B}{x}$, where $B=  \bag{N_1, \dots ,N_k}$ , $k\geq 1$, $ \#(x,M) = k $, and $\headf{M} = x$. The reduction is as follows: 
	
	\begin{prooftree}
    \AxiomC{$\headf{M} = x$}
    \AxiomC{$B = \bag{N_1, \dots ,N_k} \ , \ k\geq 1 $}
    \AxiomC{$ \#(x,M) = k $}
    \LeftLabel{\redlab{R:Fetch}}
    \TrinaryInfC{\(
    M\ \esubst{ B}{x } \rightarrow M \headlin{ N_{1}/x } \esubst{ (B\linsetminus N_1)}{ x }  + \cdots + M \headlin{ N_{k}/x } \esubst{ (B\linsetminus N_k)}{x}
    \)}
    \end{prooftree}

To simplify the proof we take $k=2$, as the case $k>2$ is similar. Therefore,  by inversion of the typing derivation and $B=\bag{N_1,N_2}$:
    
    \hspace*{-50pt}    
    \begin{minipage}{\linewidth}     
    \begin{prooftree}
    \small
             \AxiomC{\(\Gamma' , x:\sigma \wedge \sigma \wfdash  M: \tau\)}
    			\AxiomC{\(  \Delta_1 \wfdash N_1 : \sigma \)}
    				\AxiomC{\( \Delta_{2} \wfdash N_{2} : \sigma   \)}
    				\AxiomC{\(  \)}
                    \LeftLabel{\!\!\redlab{F{:}\oneb}}
                    \UnaryInfC{\( \wfdash \oneb : \omega \)}
    			\LeftLabel{\!\!\redlab{F{:}bag}}
    			\BinaryInfC{\( \Delta_2  \wfdash \bag{N_2}: \sigma   \)}
    			\LeftLabel{\!\!\redlab{F{:}bag}}
             \BinaryInfC{\(\Delta  \wfdash B: \sigma \wedge \sigma \) }
    	\LeftLabel{\redlab{F{:}ex \dash sub}}
    	\BinaryInfC{\(\Gamma' \contexcat \Delta  \wfdash   M \esubst{ B }{ x } : \tau\)}
    \end{prooftree}
    \end{minipage}
    
where $\Delta= \Delta_1 \contexcat \Delta_2$ and $\Gamma = \Gamma' \contexcat \Delta $. By the Substitution Lemma (Lemma~\ref{lem:subt_lem_fail}), there exists a derivation $\Pi_1$ of  $(\Gamma' , x:\sigma ) \contexcat \Delta_1 \wfdash   M \headlin{ N_{1}/x } : \tau $ and a derivation $\Pi_2$ of $(\Gamma' , x:\sigma ) \contexcat \Delta_2 \wfdash   M \headlin{ N_{2}/x } : \tau $. Therefore, one has the following derivation:

    \hspace*{-50pt}    
    \begin{minipage}{\linewidth}  
    \begin{prooftree}
    				\AxiomC{\( \Pi_1 \) }
                    \AxiomC{\( \Delta_2 \wfdash \bag{N_2} : \sigma \) }
                \LeftLabel{\redlab{F{:}ex \dash sub}}
    			\BinaryInfC{\( \Gamma' \contexcat \Delta  \wfdash M \headlin{ N_{1}/x } \esubst{ \bag{N_2}}{ x }  : \tau\ \) }
                    \AxiomC{\( \Pi_2 \) }
                    \AxiomC{\( \Delta_1 \wfdash \bag{N_1} : \sigma \) }
                \LeftLabel{\redlab{F{:}ex \dash sub}}
    			\BinaryInfC{\( \Gamma' \contexcat \Delta  \wfdash M \headlin{ N_2/x } \esubst{ \bag{N_1} }{x }  : \tau\  \) }
    	\LeftLabel{\redlab{F{:}sum}}
        \BinaryInfC{\(\Gamma' \contexcat \Delta  \wfdash     M \headlin{ N_{1}/x } \esubst{ \bag{N_2} }{ x } +  M \headlin{ N_{2}/x } \esubst{ \bag{N_1} }{x } : \tau\)}
    \end{prooftree}
    \end{minipage}

Assuming  $ \expr{M}'  =   \headlin{ N_{1}/x } \esubst{\bag{N_2}}{ x }  + M \headlin{ N_{2}/x } \esubst{  \bag{N_1}}{x }$, the result follows.

\item Rule~$ \redlab{R:Fail} $.

Then $\expr{M} =  M\ \esubst{ B}{x } $ where $B = \bag{N_1, \dots \cdot ,N_k} \ , \ k\geq 0 $ , $ \#(x,M) \not = k $ and we can perform the following reduction:

\begin{prooftree}
    \AxiomC{$\#(x,M) \neq \size{B}$}
    \AxiomC{\( \widetilde{y} = (\mfv{M}\setminus x )\uplus \mfv{B}\)}
    \LeftLabel{\redlab{R:Fail}}
    \BinaryInfC{\( M\ \esubst{ B}{x } \red \sum_{\perm{B}} \fail^{\widetilde{y}}\)}
\end{prooftree}   
with $\expr{M}'=\sum_{\perm{B}} \fail^{\widetilde{y}}$. By hypothesis, one has the derivation:

\begin{prooftree}
            \AxiomC{\( \Delta \wfdash B :  \sigma^{j} \)}
            \AxiomC{\( \Gamma' , {x}:\sigma^{k} \wfdash M : \tau \)}
        \LeftLabel{\redlab{F{:}ex \dash sub}}    
        \BinaryInfC{\( \Gamma' \contexcat \Delta \wfdash M \esubst{ B }{ x } : \tau \)}
    \end{prooftree}
Notice that we also have from $\#(x,M) \neq \size{B}$ that $j\neq k$.Hence $\Gamma = \Gamma' \contexcat \Delta $ and we may type the following:
    
    \begin{prooftree}
        \AxiomC{\( \)}
        \LeftLabel{\redlab{F{:}fail}}
        \UnaryInfC{$ \Gamma \wfdash \fail^{\widetilde{y}} : \tau$}
        \AxiomC{\( \cdots \)}
        \AxiomC{\( \)}
        \LeftLabel{\redlab{F{:}fail}}
        \UnaryInfC{$\Gamma \wfdash \fail^{\widetilde{y}} : \tau$}
        \LeftLabel{\redlab{F{:}sum}}
        \TrinaryInfC{$ \Gamma \wfdash \sum_{\perm{B}} \fail^{\widetilde{y}}: \tau$}
    \end{prooftree}

\item Rule~$\redlab{R:Cons_1}$.

Then $\expr{M} =   \fail^{\widetilde{x}} \ B $ where $B = \bag{N_1, \dots ,N_k} $ , $k \geq 0$ and we can perform the following reduction:

\begin{prooftree}
    \AxiomC{$\size{B} = k$}
    \AxiomC{\( \widetilde{y} = \mfv{B} \)}
    \LeftLabel{$\redlab{R:Cons_1}$}
    \BinaryInfC{\( \fail^{\widetilde{x}}\ B  \red \sum_{\perm{B}} \fail^{\widetilde{x} \uplus \widetilde{y}} \)}
\end{prooftree}
where $\expr{M}'=\sum_{\perm{B}} \fail^{\widetilde{x} \uplus \widetilde{y}}$. By hypothesis and inversion of the typing derivation, there exists the following derivation:
% and type the expression:

    \begin{prooftree}
        \AxiomC{\( \)}
        \LeftLabel{\redlab{F{:}fail}}
        \UnaryInfC{\( \Gamma' \wfdash \fail^{\widetilde{x}} : \pi' \rightarrow \tau \)}
        \AxiomC{\( \Delta \wfdash B : \pi \)}
            \LeftLabel{\redlab{F{:}app}}
        \BinaryInfC{\( \Gamma' \contexcat \Delta \wfdash \fail^{\widetilde{x}} \ B : \tau\)}
    \end{prooftree}

Hence $\Gamma = \Gamma' \contexcat \Delta $ and we may type the following:

    \begin{prooftree}
        \AxiomC{\( \)}
        \LeftLabel{\redlab{F{:}fail}}
        \UnaryInfC{$ \Gamma \wfdash \fail^{\widetilde{x} \uplus \widetilde{y}} : \tau$}
        \AxiomC{\( \cdots \)}
        \AxiomC{\( \)}
        \LeftLabel{\redlab{F{:}fail}}
        \UnaryInfC{$\Gamma \wfdash \fail^{\widetilde{x} \uplus \widetilde{y}} : \tau$}
        \LeftLabel{\redlab{F{:}sum}}
        \TrinaryInfC{$ \Gamma \wfdash \sum_{\perm{B}} \fail^{\widetilde{x} \uplus \widetilde{y}}: \tau$}
    \end{prooftree}

\item Rule~$\redlab{R:Cons_2}$.

Then $\expr{M} =   \fail^{\widetilde{z}}\ \esubst{B}{x} $ where $B = \bag{N_1, \dots ,N_k} $ , $k \geq 1$ and 
we can perform the following reduction:

\begin{prooftree}
    \AxiomC{$\size{B} = k$}
    \AxiomC{\(  \#(x , \widetilde{z}) + k  \not= 0 \)}
    \AxiomC{\( \widetilde{y} = \mfv{B} \)}
    \LeftLabel{$\redlab{R:Cons_2}$}
    \TrinaryInfC{\( \fail^{\widetilde{z}}\ \esubst{B}{x}  \red \sum_{\perm{B}} \fail^{(\widetilde{z} \setminus x) \uplus\widetilde{y}} \)}
\end{prooftree}

where $\expr{M}'=\sum_{\perm{B}} \fail^{(\widetilde{z} \setminus x) \uplus \widetilde{y}}$. By hypothesis and inversion of the typing derivation, there exists a derivation:

    \begin{prooftree}
            \AxiomC{\( \dom{\core{(\Gamma' , {x}:\sigma^{k})}}=\widetilde{z} \)}
            \LeftLabel{\redlab{F{:}fail}}
            \UnaryInfC{\( \Gamma' ,{x}:\sigma^{k}\wfdash  \fail^{\widetilde{z}} : \tau  \)}
            \AxiomC{\( \Delta \wfdash B : \sigma^{j} \)}
        \LeftLabel{\redlab{F{:}ex \dash sub}}    
        \BinaryInfC{\( \Gamma' \contexcat \Delta \wfdash \fail^{\widetilde{z}} \esubst{ B }{ x } : \tau \)}
    \end{prooftree}

Hence $\Gamma = \Gamma' \contexcat \Delta $ and we may type the following:

    \begin{prooftree}
        \AxiomC{\( \)}
        \LeftLabel{\redlab{F{:}fail}}
        \UnaryInfC{$ {\Gamma} \wfdash \fail^{(\widetilde{z} \setminus x) \uplus\widetilde{y}} : \tau$}
        \AxiomC{$ \cdots $}
         \AxiomC{\( \)}
        \LeftLabel{\redlab{F{:}fail}}
        \UnaryInfC{$ {\Gamma} \wfdash \fail^{(\widetilde{z} \setminus x) \uplus\widetilde{y}} : \tau$}
        \LeftLabel{\redlab{F{:}sum}}
        \TrinaryInfC{$ {\Gamma} \wfdash \sum_{\perm{B}} \fail^{(\widetilde{z} \setminus x) \uplus\widetilde{y}} : \tau$}
    \end{prooftree}
    
    \item Rule~$\redlab{R:TCont}$.

Then $\expr{M} = C[M]$ and the reduction is as follows:

\begin{prooftree}
        \AxiomC{$   M \red  M'_{1} + \cdots +  M'_{l} $}
        \LeftLabel{\redlab{R:TCont}}
        \UnaryInfC{$ C[M] \red  C[M'_{1}] + \cdots +  C[M'_{l}] $}
\end{prooftree}

where $\expr{M}'= C[M'_{1}] + \cdots +  C[M'_{l}]$. 
The proof proceeds by analysing the context $C$:

    \begin{enumerate}
    \item $C=[\cdot]\ B$.
    
    In this case $\expr{M}=M \ B$, for some $B$, and the following derivation holds:
\begin{prooftree}
    \AxiomC{\(  \Gamma' \wfdash  M:  \sigma^{j} \rightarrow \tau \)}
    \AxiomC{\( \Delta \wfdash  B : \sigma^{k} \)}
        \LeftLabel{\redlab{F{:}app}}
    \BinaryInfC{\( \Gamma' \contexcat \Delta \wfdash  M\ B : \tau\)}
\end{prooftree}

where $\Gamma = \Gamma' \contexcat \Delta $.

Since $\Gamma'\wfdash M: \sigma^j \rightarrow \tau$ and $M\red M_1'+\ldots + M_l'$, it follows by IH that $\Gamma'\wfdash M_1'+\ldots + M_l':\sigma^j \rightarrow \tau$. By applying \redlab{F{:}sum}, one has $\Gamma'\wfdash M_i' : \sigma^j \rightarrow \tau$, for $i=1,\ldots, l$.  Therefore, we may type the following:

\begin{prooftree}
\AxiomC{\(  \forall i \in {1 , \cdots , l} \)}
\AxiomC{\(  \Gamma' \wfdash  M'_{i}: \sigma^{j} \rightarrow \tau \)}
\AxiomC{\( \Delta \wfdash  B :  \sigma^{k} \)}
\LeftLabel{\redlab{F{:}app}}
\BinaryInfC{\(  \Gamma' \contexcat \Delta \wfdash (M'_{i}\ B):  \tau \)}
 \LeftLabel{\redlab{F{:}sum}}
    \BinaryInfC{\( \Gamma' \contexcat \Delta \wfdash (M'_{1}\ B) + \cdots +  (M'_{l} \ B) : \tau\)}
\end{prooftree}

Thus, $\Gamma\wfdash \expr{M'}: \tau$, and the result follows.
% True by the IH and similarly for all incomplete branches.

    \item $C=([\cdot])\esubst{B}{x}$.
    
    This case is similar to the previous one.
\end{enumerate}
	\item Rule~$ \redlab{R:ECont} $.
	
Then $\expr{M} = D[\expr{M}'']$ where $\expr{M}'' \rightarrow \expr{M}'''$ then we can perform the following reduction:
\begin{prooftree}
        \AxiomC{$ \expr{M}''  \red \expr{M}'''  $}
        \LeftLabel{$\redlab{R:ECont}$}
        \UnaryInfC{$D[\expr{M}'']  \red D[\expr{M}''']  $}
\end{prooftree}

Hence $\expr{M}' =  D[\expr{M}'''] $.
The proof proceeds by analysing the context $D$:

\begin{enumerate}
    \item $D= [\cdot] + \expr{N}$. In this case $\expr{M}= \expr{M}''+\expr{N}$ by inversion of the typing derivation:
\begin{prooftree}
    \AxiomC{$ \Gamma \wfdash  \expr{M}'' : \tau$}
    \AxiomC{$ \Gamma \wfdash  \expr{N} : \tau$}
    \LeftLabel{\redlab{F{:}sum}}
    \BinaryInfC{$ \Gamma \wfdash  \expr{M}''+\expr{N}: \tau$}
\end{prooftree}

Since $\Gamma\vdash \expr{M}^{''}:\tau $ and $\expr{M}^{''}\red \expr{M}^{'''}$, by IH, it follows that $\Gamma\vdash \expr{M}^{'''}:\tau$  and we may type the following:

\begin{prooftree}
    \AxiomC{$ \Gamma \wfdash  \expr{M}''' : \tau$}
    \AxiomC{$ \Gamma \wfdash  \expr{N} : \tau$}
    \LeftLabel{\redlab{F{:}sum}}
    \BinaryInfC{$ \Gamma \wfdash  \expr{M}'''+\expr{N}: \tau$}
\end{prooftree}
Therefore, $\Gamma\vdash \expr{M'}:\tau$ and the  result follows.

    \item $D= \expr{N} + [\cdot]$.
      This case is similar to the previous one. \qedhere
\end{enumerate}
\end{myEnumerate}
%\qed
\end{proof}

%%%%%%%%%%%%%%%%%%%%%%%%%%%%%%%%%%%%%%%%%%
\section{Appendix to \texorpdfstring{\secref{ss:typeshar}}{§ 3.3}}
\label{app:typeshar}

\Consistencyreductions*

%\encodingreduces*

\begin{proof}
    \secondrev{By structural induction on the reduction rules. We will consider two key reduction rules, the other cases follow analogously via  application of the IH.\\
\begin{enumerate}
 \item Rule \redlab{RS{:}Ex \dash Sub}. In this case, we have
    \begin{prooftree}
        \AxiomC{$B = \bag{M_1}
        \cdots  \bag{M_k} \qquad k \geq  1 $}
        \AxiomC{$ M \not= \fail^{\widetilde{y}} $}
        \LeftLabel{\redlab{RS{:}Ex \dash Sub}}
        \BinaryInfC{\( \!M[x_1,\ldots, x_k \leftarrow x]\esubst{ B }{ x } \red \sum_{B_i \in \perm{B}}M\linexsub{B_i(1)/x_1} \cdots \linexsub{B_i(k)/x_k}    \)}
     \end{prooftree}
    Notice that if a bag is consistent then each element in the bag is consistent, that is, for any permutation $B_i$ of the bag $B$ then each $B_i(n)$ is consistent. Then, the assumption of consistency for $(M[\widetilde{x} \leftarrow x])\esubst{ B }{ x } $,  along with each element of the bag being consistent implies consistency of $ \sum_{B_i \in \perm{B}}M\linexsub{B_i(1)/x_1} \cdots \linexsub{B_i(k)/x_k} $ for each permutation of~$B$.\\
  \item Rule \redlab{RS{:}Lin \dash Fetch}. 
  In this case, we have
    \begin{prooftree}
     \AxiomC{$ \headf{M} = x$}
         \LeftLabel{\redlab{RS{:}Lin\dash Fetch}}
         \UnaryInfC{\(  M \linexsub{N/x} \red  M \headlin{ N/x } \)}
    \end{prooftree}
    This case follows from the fact that  $M \headlin{ N/x }$  preserves consistency. The argument is by structural induction, with   base case of $M = x$ together with the fact that $N$ is consistent trivially implies that $x \headlin{ N/x }$ must also be consistent. As for the inductive step, notice that `adding' $N$ to the structure of $M$ does not break any of the consistency requirements: the consistency of $M \linexsub{N/x}$ implies that   the free variables of $M$ and $N$ are disjoint.
   \end{enumerate}
   }
\end{proof}

%\joe{working here}

\explemfailnofailshar*

\begin{proof}
By induction on the structure of $ M $:

    \begin{myEnumerate}
        
        \item When $M = x$ then we have $x \headlin{ N / x}   = N$ and may derive the derivation of $ \Gamma \vdash N: \tau $ with $x \not \in \dom{\Gamma}$. By taking $\Gamma_1 = \emptyset$ and $\Gamma_2 = \Gamma$ as $\Gamma = \emptyset \contexcat \Gamma$ the case follows as $ \Gamma \vdash N : \tau$ and
            
                \begin{prooftree}
                    \AxiomC{}
                    \LeftLabel{\redlab{TS:var}}
                    \UnaryInfC{\( x: \sigma \vdash x : \sigma\)}
                \end{prooftree}
        
        \item When $M = (M\ B)$ then we have that $(M\ B)\headlin{ N/x}  = (M \headlin{ N/x })\ B$. Let us consider two cases:
            
            \begin{myEnumerate}
                
                \item When $x \in \lfv{M \headlin{ N/x }}$
                    \begin{prooftree}
                        \AxiomC{\( \Gamma', x:\sigma^{k-1}  \vdash M \headlin{ N/x } : \pi \rightarrow \tau' \)}
                        \AxiomC{\( \Delta \vdash B : \pi \)}
                            \LeftLabel{\redlab{TS:app}}
                        \BinaryInfC{\( (\Gamma', x:\sigma^{k-1} ) \contexcat \Delta \vdash (M \headlin{ N/x })\ B : \tau'\)}
                    \end{prooftree}
                    
                    By the IH we have that $\Gamma', x:\sigma^{k-1} \vdash M \headlin{ N/x } : \pi \rightarrow \tau'$ implies that $\exists \  \Gamma_1', \Gamma_2$ such that  $\Gamma_1' , x:\sigma^k \vdash M: \tau$, and $\Gamma_2 \vdash N : \sigma$ with $\Gamma' = \Gamma_1' \contexcat \Gamma_2$.
                    
                    \begin{prooftree}
                        \AxiomC{\( \Gamma_1', x:\sigma^{k}  \vdash M \headlin{ N/x } : \pi \rightarrow \tau' \)}
                        \AxiomC{\( \Delta \vdash B : \pi \)}
                            \LeftLabel{\redlab{TS:app}}
                        \BinaryInfC{\( (\Gamma_1', x:\sigma^{k} ) \contexcat \Delta \vdash M\ B : \tau'\)}
                    \end{prooftree}
                    
                \item When $x \not \in \lfv{M \headlin{ N/x }}$
                    \begin{prooftree}
                      \AxiomC{\( \Gamma'  \vdash M \headlin{ N/x } : \pi \rightarrow \tau' \)}
                        \AxiomC{\( \Delta \vdash B : \pi \)}
                            \LeftLabel{\redlab{TS:app}}
                        \BinaryInfC{\( \Gamma' \contexcat \Delta \vdash (M \headlin{ N/x })\ B : \tau'\)}
                    \end{prooftree}
                    
                    By the IH we have that $\Gamma' \vdash M \headlin{ N/x } [\widetilde{y} \leftarrow y]: \pi \rightarrow \tau'$ implies that $\exists \  \Gamma_1', \Gamma_2$ such that  $\Gamma_1' , x:\sigma \vdash M[\widetilde{y} \leftarrow y]: \tau$, and $\Gamma_2 \vdash N : \sigma$ with $\Gamma' = \Gamma_1' \contexcat \Gamma_2$.
                    
                    \begin{prooftree}
                        \AxiomC{\( \Gamma_1', x:\sigma  \vdash M : \pi \rightarrow \tau' \)}
                        \AxiomC{\( \Delta \vdash B : \pi \)}
                            \LeftLabel{\redlab{TS:app}}
                        \BinaryInfC{\( (\Gamma_1', x:\sigma ) \contexcat \Delta \vdash M\ B : \tau'\)}
                    \end{prooftree}
                
            \end{myEnumerate}

        \item When $M = M[\widetilde{y} \leftarrow y] \esubst{B}{y}$ then we have that $(M [\widetilde{y} \leftarrow y]\esubst{B}{y})\headlin{ N/x } = (M\headlin{ N/x })\ [\widetilde{y} \leftarrow y] \esubst{B}{y}$ where $x \not = y$.
            
            \begin{myEnumerate}
                
                \item When $x \in \lfv{M \headlin{ N/x }}$:
                    
                    \begin{prooftree}
                            \AxiomC{\( \Gamma', x:\sigma^{k-1} ,  {\widetilde{y}}:\delta^{j} \vdash (M\headlin{ N/x }) : \tau \)}
                            \LeftLabel{ \redlab{TS{:}share}}
                            \UnaryInfC{\( \Gamma', x:\sigma^{k-1} ,  {y}:\delta^{j} \vdash (M\headlin{ N/x })[\widetilde{y} \leftarrow y] : \tau \)}
                       
                             \AxiomC{\( \Delta \vdash B : \delta^{j} \)}
                        \LeftLabel{\redlab{TS:ex \dash sub}}    
                        \BinaryInfC{\( \Gamma' \contexcat \Delta \vdash (M\headlin{ N/x }) [\widetilde{y} \leftarrow y]\esubst{ B }{ y } : \tau \)}
                    \end{prooftree}
                    
                    By the IH we have that $\Gamma', x:\sigma^{k-1} ,  {\widetilde{y}}:\delta^{j} \vdash (M\headlin{ N/x }) : \tau$ implies that $\exists \  \Gamma_1', \Gamma_2$ such that  $\Gamma_1' , x:\sigma^k , {\widetilde{y}}:\delta^{j}\vdash M: \tau$, and $\Gamma_2 \vdash N : \sigma$ with $\Gamma',  {y}:\delta^{j} = (\Gamma_1',  {y}:\delta^{j}) \contexcat \Gamma_2$.
                    \begin{prooftree}
                            \AxiomC{\( \Gamma'_1, x:\sigma^{k} ,  {\widetilde{y}}:\delta^{j} \vdash M : \tau \)}
                            \LeftLabel{ \redlab{TS{:}share}}
                            \UnaryInfC{\( \Gamma'_1, x:\sigma^{k} ,  {y}:\delta^{j} \vdash M[\widetilde{y} \leftarrow y] : \tau \)}
                            
                             \AxiomC{\( \Delta \vdash B : \delta^{j} \)}
                        \LeftLabel{\redlab{TS:ex \dash sub}}    
                        \BinaryInfC{\( \Gamma'_1 \contexcat \Delta \vdash (M)[\widetilde{y} \leftarrow y] \esubst{ B }{ y } : \tau \)}
                    \end{prooftree}
                    
                \item When $x \not \in \lfv{M \headlin{ N/x }}$:
                                        \begin{prooftree}
                            \AxiomC{\( \Gamma',  {\widetilde{y}}:\delta^{k} \vdash (M\headlin{ N/x }) : \tau \)}
                            \LeftLabel{ \redlab{TS{:}share}}
                            \UnaryInfC{\( \Gamma' ,  {y}:\delta^{k} \vdash (M\headlin{ N/x })[\widetilde{y} \leftarrow y] : \tau \)}
                            
                             \AxiomC{\( \Delta \vdash B : \delta^{k} \)}
                        \LeftLabel{\redlab{TS:ex \dash sub}}    
                        \BinaryInfC{\( \Gamma' \contexcat \Delta \vdash (M\headlin{ N/x })[\widetilde{y} \leftarrow y] \esubst{ B }{ y } : \tau \)}
                    \end{prooftree}
                    
                    By the IH we have that $\Gamma' ,  {y}:\delta^{k} \vdash (M\headlin{ N/x }) : \tau$ implies that $\exists \  \Gamma_1', \Gamma_2$ such that  $\Gamma_1' , x:\sigma, {\widetilde{y}}:\delta^{k} \vdash M: \tau$, and $\Gamma_2 \vdash N : \sigma$ with $\Gamma',  {y}:\delta^{k} = (\Gamma_1',  {y}:\delta^{k}) \contexcat \Gamma_2$.
                    
                    \begin{prooftree}
                            \AxiomC{\( \Gamma', x:\sigma ,  {\widetilde{y}}:\delta^{k} \vdash M : \tau \)}
                            \LeftLabel{ \redlab{TS{:}share}}
                            \UnaryInfC{\( \Gamma', x:\sigma ,  {y}:\delta^{k} \vdash M[\widetilde{y} \leftarrow y] : \tau \)}
                            
                             \AxiomC{\( \Delta \vdash B : \delta^{k} \)}
                        \LeftLabel{\redlab{TS:ex \dash sub}}    
                        \BinaryInfC{\( \Gamma' \contexcat \Delta \vdash (M[\widetilde{y} \leftarrow y]\headlin{ N/x }) \esubst{ B }{ y } : \tau \)}
                    \end{prooftree}
                    
            \end{myEnumerate}

        \item When $M = M[\widetilde{y} \leftarrow y] $ then we have that $(M [\widetilde{y} \leftarrow y])\headlin{ N/x } = (M\headlin{ N/x })\ [\widetilde{y} \leftarrow y]$ where $x \not = y$.
            
            \begin{myEnumerate}
                
                \item When $x \in \lfv{M \headlin{ N/x }}$:
                    
                    \begin{prooftree}
                            \AxiomC{\( \Gamma', x:\sigma^{k-1} ,  {\widetilde{y}}:\delta^{j} \vdash (M\headlin{ N/x }) : \tau \)}
                            \LeftLabel{ \redlab{TS{:}share}}
                            \UnaryInfC{\( \Gamma', x:\sigma^{k-1} ,  {y}:\delta^{j} \vdash (M\headlin{ N/x })[\widetilde{y} \leftarrow y] : \tau \)}
                    \end{prooftree}
                    
                    By the IH we have that $\Gamma', x:\sigma^{k-1} ,  {\widetilde{y}}:\delta^{j} \vdash (M\headlin{ N/x }) : \tau$ implies that $\exists \  \Gamma_1', \Gamma_2$ such that  $\Gamma_1' , x:\sigma^k , {\widetilde{y}}:\delta^{j}\vdash M: \tau$, and $\Gamma_2 \vdash N : \sigma$ with $\Gamma',  {y}:\delta^{j} = (\Gamma_1',  {y}:\delta^{j}) \contexcat \Gamma_2$.
                    
                    \begin{prooftree}
                            \AxiomC{\( \Gamma'_1, x:\sigma^{k} ,  {\widetilde{y}}:\delta^{j} \vdash M : \tau \)}
                            \LeftLabel{ \redlab{TS{:}share}}
                            \UnaryInfC{\( \Gamma'_1, x:\sigma^{k} ,  {y}:\delta^{j} \vdash M[\widetilde{y} \leftarrow y] : \tau \)}
                    \end{prooftree}
                    
                \item When $x \not \in \lfv{M \headlin{ N/x }}$:
                    
                    \begin{prooftree}
                            \AxiomC{\( \Gamma',  {\widetilde{y}}:\delta^{k} \vdash (M\headlin{ N/x }) : \tau \)}
                            \LeftLabel{ \redlab{TS{:}share}}
                            \UnaryInfC{\( \Gamma' ,  {y}:\delta^{k} \vdash (M\headlin{ N/x })[\widetilde{y} \leftarrow y] : \tau \)}
                    \end{prooftree}
                    
                    By the IH we have that $\Gamma' ,  {y}:\delta^{k} \vdash (M\headlin{ N/x }) : \tau$ implies that $\exists \  \Gamma_1', \Gamma_2$ such that  $\Gamma_1' , x:\sigma, {\widetilde{y}}:\delta^{k} \vdash M: \tau$, and $\Gamma_2 \vdash N : \sigma$ with $\Gamma',  {y}:\delta^{k} = (\Gamma_1',  {y}:\delta^{k}) \contexcat \Gamma_2$.
                    
                    \begin{prooftree}
                            \AxiomC{\( \Gamma', x:\sigma ,  {\widetilde{y}}:\delta^{k} \vdash M : \tau \)}
                            \LeftLabel{ \redlab{TS{:}share}}
                            \UnaryInfC{\( \Gamma', x:\sigma ,  {y}:\delta^{k} \vdash M[\widetilde{y} \leftarrow y] : \tau \)}
                    \end{prooftree}
                    
            \end{myEnumerate}

        \item When $M = M \linexsub{N /y} $ then we have that $(M \linexsub{N /y})\headlin{ N/x } = (M\headlin{ N/x })\ \linexsub{N /y}$ where $x \not = y$.
            
            \begin{myEnumerate}
                
                \item When $x \in \lfv{M \headlin{ N/x }}$:
                    
                    \begin{prooftree}
                        \AxiomC{\( \Delta \vdash N : \delta \qquad \Gamma', x:\sigma^{k-1} ,  {y}:\delta \vdash (M\headlin{ N/x }) : \tau \)}
                        \LeftLabel{\redlab{TS\!:\!ex\dash lin\dash sub}}
                        \UnaryInfC{\( \Gamma', x:\sigma^{k-1} ,  {y}:\delta , \Delta \vdash (M\headlin{ N/x }) \linexsub{N / y} : \tau \)}
                    \end{prooftree}
                    
                    By the IH we have that $\Gamma', x:\sigma^{k-1} ,  {\widetilde{y}}:\delta^{j} \vdash (M\headlin{ N/x }) : \tau$ implies that $\exists \  \Gamma_1', \Gamma_2$ such that  $\Gamma_1' , x:\sigma^k , {\widetilde{y}}:\delta^{j}\vdash M: \tau$, and $\Gamma_2 \vdash N : \sigma$ with $\Gamma',  {y}:\delta^{j} = (\Gamma_1',  {y}:\delta^{j}) \contexcat \Gamma_2$.

                    \begin{prooftree}
                        \AxiomC{\( \Delta \vdash N : \delta \qquad \Gamma'_1, x:\sigma^{k} ,  {y}:\delta \vdash M : \tau \)}
                        \LeftLabel{\redlab{TS\!:\!ex\dash lin\dash sub}}
                        \UnaryInfC{\( \Gamma'_1, x:\sigma^{k} ,  {y}:\delta , \Delta \vdash M \linexsub{N / y} : \tau \)}
                    \end{prooftree}
                    
                \item When $x \not \in \lfv{M \headlin{ N/x }}$:
                    
                    \begin{prooftree}
                        \AxiomC{\( \Delta \vdash N : \delta \qquad \Gamma',  {y}:\delta \vdash (M\headlin{ N/x }) : \tau \)}
                        \LeftLabel{\redlab{TS\!:\!ex\dash lin\dash sub}}
                        \UnaryInfC{\( \Gamma' ,  {y}:\delta , \Delta \vdash (M\headlin{ N/x }) \linexsub{N / y} : \tau \)}
                    \end{prooftree}
                    
                    By the IH we have that $\Gamma' ,  {y}:\delta^{k} \vdash (M\headlin{ N/x }) : \tau$ implies that $\exists \  \Gamma_1', \Gamma_2$ such that  $\Gamma_1' , x:\sigma, {\widetilde{y}}:\delta^{k} \vdash M: \tau$, and $\Gamma_2 \vdash N : \sigma$ with $\Gamma',  {y}:\delta^{k} = (\Gamma_1',  {y}:\delta^{k}) \contexcat \Gamma_2$.
                    
                    \begin{prooftree}
                        \AxiomC{\( \Delta \vdash N : \delta \qquad \Gamma'_1, x:\sigma ,  {y}:\delta \vdash M : \tau \)}
                        \LeftLabel{\redlab{TS\!:\!ex\dash lin\dash sub}}
                        \UnaryInfC{\( \Gamma'_1, x:\sigma ,  {y}:\delta , \Delta \vdash M \linexsub{N / y} : \tau \)}
                    \end{prooftree}
                    
            \end{myEnumerate}

        \item When $M = \lambda y . M[\widetilde{y} \leftarrow y]$ then linear head substitution is undefined on this term as $\headf{M} \not = x$.

        \item When $M = \fail^{\widetilde{x}}$ then $M$ is not well typed. \qedhere

    \end{myEnumerate}

\end{proof}

\subexponeshar*

\begin{proof}
By induction on the reduction rule applied.
There are five possible cases. 
    
    \begin{myEnumerate}
    
        \item When $\expr{M}'$ is reduced to via the Rule~\redlab{RS:Beta}:
    
            \begin{prooftree}
                \AxiomC{}
                \LeftLabel{\redlab{RS:Beta}}
                \UnaryInfC{\((\lambda x. M[\widetilde{x} \leftarrow x ]) B \red M[\widetilde{x} \leftarrow x ]\ \esubst{B}{x}\)}
            \end{prooftree}
            
            Then $\expr{M}' = M[\widetilde{x} \leftarrow x ]\ \esubst{B}{x} $ can be type as followed:
            
            \begin{prooftree}
                    \AxiomC{\( {\Gamma ,  {x}:\sigma^{k} \vdash M[\widetilde{x} \leftarrow x ] : \tau} \)}
                     \AxiomC{\( \Delta \vdash B : \sigma^{k} \)}
                \LeftLabel{\redlab{TS:ex \dash sub}}    
                \BinaryInfC{\( {\Gamma \contexcat \Delta \vdash M[\widetilde{x} \leftarrow x ] \esubst{ B }{ x } : \tau} \)}
            \end{prooftree}
            
            From the typing of $\expr{M}'  $ we can deduce that $\expr{M} = (\lambda x. M[\widetilde{x} \leftarrow x ]) B $ may be typed by:
            
            \begin{prooftree}
                \AxiomC{\( {\Gamma , {x}: \sigma^k \vdash M[\widetilde{x} \leftarrow x ] : \tau} \)}
                \LeftLabel{\redlab{TS:abs}}
                \UnaryInfC{\( \Gamma \vdash \lambda x . M[\widetilde{x} \leftarrow x ] :  \sigma^k  \rightarrow \tau \)}
                \AxiomC{\( \Delta \vdash B :  \sigma^k \)}
                    \LeftLabel{\redlab{TS:app}}
                \BinaryInfC{\( {\Gamma \contexcat \Delta \vdash (\lambda x. M[\widetilde{x} \leftarrow x ]) B : \tau}\)}
            \end{prooftree}

        \item When $\expr{M}'$ is reduced to via the Rule~\redlab{RS{:}Ex \dash Sub}:

         \hspace*{-20pt}
         \begin{minipage}{\linewidth}
         \begin{prooftree}
            \AxiomC{$B = \bag{M_1}
            \cdots  \bag{M_k} \qquad k \geq  1 $}
            \AxiomC{$ M \not= \fail^{\widetilde{y}} $}
            \LeftLabel{\redlab{RS{:}Ex \dash Sub}}
            \BinaryInfC{\( \!M[x_1,\ldots, x_k \leftarrow x]\esubst{ B }{ x } \red \sum_{B_i \in \perm{B}}M\linexsub{B_i(1)/x_1} \cdots \linexsub{B_i(k)/x_k}    \)}
         \end{prooftree}
         \end{minipage}

        Then $\expr{M}' = \!M[x_1,\ldots, x_k \leftarrow x]\esubst{ B }{ x } $ can be type as followed:

         \hspace*{-50pt}
         \begin{minipage}{\linewidth}
         \begin{prooftree}
            \AxiomC{\( \Gamma  , x_1:\sigma, \cdots , x_k:\sigma \vdash M : \tau \)}
            \AxiomC{\( \Delta_1 \vdash B_i(1) : \sigma  \)}
            \BinaryInfC{\( \vdots \)}
            \AxiomC{\( \Delta_k \vdash B_i(k) : \sigma  \)}
            \LeftLabel{\redlab{TS\!:\!ex\dash lin\dash sub}}
            \BinaryInfC{\( \Gamma , \Delta_1 , \cdots , \Delta_k \vdash M\linexsub{B_i(1)/x_1} \cdots \linexsub{B_i(k)/x_k} : \tau \)}
            \AxiomC{$ \forall B_i \in \perm{B} $}
            \LeftLabel{\redlab{TS{:}sum}}
            \BinaryInfC{$ \Gamma , \Delta_1 , \cdots , \Delta_k  \vdash \sum_{B_i \in \perm{B}}M\linexsub{B_i(1)/x_1} \cdots \linexsub{B_i(k)/x_k}: \tau$}
        \end{prooftree}
        \end{minipage}

         From the typing of $\expr{M}'  $ we can deduce that $\expr{M} = (\lambda x. M[\widetilde{x} \leftarrow x ]) B $ may be typed by:

        \begin{prooftree}
            \AxiomC{\( \Delta_1 \vdash M_1 : \sigma\)}
            \AxiomC{\( \Delta_k \vdash M_k : \sigma\)}
            \UnaryInfC{\(\ \vdots \)}
            \LeftLabel{\redlab{TS{:}bag}}
            \BinaryInfC{\( \Delta_1 , \cdots , \Delta_k \vdash B : \sigma^k \)}
            \AxiomC{\( \Gamma  , x_1:\sigma, \cdots , x_k:\sigma \vdash M : \tau \)}
            \UnaryInfC{\( \qquad \Gamma , x:\sigma^k \vdash M [\widetilde{x} \leftarrow x]: \tau \)}
            \LeftLabel{\redlab{TS\!:ex \dash sub}} 
            \BinaryInfC{\( \Gamma , \Delta_1 , \cdots , \Delta_k  \vdash M[\widetilde{x} \leftarrow x] \esubst{ B }{ x } : \tau \)}
        \end{prooftree}

         \item When $\expr{M}'$ is reduced to via the Rule~\redlab{RS{:}Lin\dash Fetch}:
        
         \begin{prooftree}
            \AxiomC{$ \headf{M} = x$}
             \LeftLabel{\redlab{RS{:}Lin\dash Fetch}}
             \UnaryInfC{\(  M \linexsub{N/x} \red  M \headlin{ N/x } \)}
        \end{prooftree}
        
        The result follow from \Cref{lem:antisubt_lem_shar}.

        \item When $\expr{M}'$ is reduced to via the Rule~\redlab{RS:TCont}:
            \begin{prooftree}
                    \AxiomC{$   M \red M'_{1} + \cdots + M'_{k} $}
                    \LeftLabel{\redlab{RS:TCont}}
                    \UnaryInfC{$ C[M] \red  C[M'_{1}] + \cdots +  C[M'_{k}] $}
            \end{prooftree}
        
            Hence the proof follows by the IH on $M$.
        
        \item When $\expr{M}'$ is reduced to via the Rule~\redlab{RS:ECont}:
            \begin{prooftree}
                    \AxiomC{$ \expr{M}  \red \expr{M}'  $}
                    \LeftLabel{\redlab{RS:ECont}}
                    \UnaryInfC{$D[\expr{M}]  \red D[\expr{M}']  $}
            \end{prooftree} 
            
            Hence the proof follows by the IH on $M$. \qedhere
        
        % \item \redlab{R:Fail}, \redlab{R:Cons_1}, \redlab{R:Cons_2}
        
        %     \begin{prooftree}
        %         \AxiomC{$\#(x,M) \neq \size{B} \qquad \widetilde{y} = (\mfv{M} \setminus x) \uplus \mfv{B} $}
        %         \LeftLabel{\redlab{R:Fail}}
        %         \UnaryInfC{\(  M\ \esubst{ B}{x } \red {}  \displaystyle\sum_{\perm{B}} \fail^{\widetilde{y}} \)}
        %     \end{prooftree}
        
        %     \begin{prooftree}
        %         \AxiomC{$  \widetilde{y} = \mfv{B} $}
        %         \LeftLabel{$\redlab{R:Cons_1}$}
        %         \UnaryInfC{\(  \fail^{\widetilde{x}}\ B \red{}  \displaystyle\sum_{\perm{B}} \fail^{\widetilde{x} \uplus \widetilde{y}} \)}
        %     \end{prooftree}
            
        %     \begin{prooftree}   
        %     \AxiomC{$\size{B} = k  \qquad \#(z , \widetilde{x}) + k  \not= 0 \qquad  \widetilde{y} = \mfv{B}$}
        %         \LeftLabel{$\redlab{R:Cons_2}$}
        %         \UnaryInfC{\( \fail^{\widetilde{x}}\ \esubst{B}{z}  \red {} \displaystyle \sum_{\perm{B}} \fail^{(\widetilde{x} \setminus z) \uplus \widetilde{y}}  \)}
        %     \end{prooftree}

    \end{myEnumerate}
    
\end{proof}

\lamrsharfailsubs*

\begin{proof}
By structural induction on $M$ with $\headf{M}=x$.
% By (wf-expr) we have that the substitution lemma holds for any term well typed in $\lamrshar$.
There are six cases to be analyzed:
\begin{myEnumerate}
\item $M=x$.

In this case, $x:\sigma \wfdash x:\sigma$ and $\Gamma=\emptyset$.  Observe that $x\headlin{N/x}=N$, since $\Delta\wfdash N:\sigma$, by hypothesis, the result follows.

    \item $M = M'\ B$.
    
    Then $\headf{M'\ B} = \headf{M'} = x$, and the derivation is the following by inversion of the typing derivation:
    \begin{prooftree}
        \AxiomC{$\Gamma_1 , x:\sigma \wfdash M': \delta^{j}  \rightarrow \tau$}\
        \AxiomC{$\Gamma_2 \wfdash B : \delta^{k} $}
    	\LeftLabel{\redlab{FS{:}app}}
        \BinaryInfC{$\Gamma_1 , \Gamma_2 , x:\sigma \wfdash M'B:\tau $}    
    \end{prooftree}

    where $\Gamma=\Gamma_1,\Gamma_2$, and  $j,k$ are non-negative integers, possibly different.  Since $\Delta \vdash N : \sigma$, by IH, the result holds for $M'$, that is,
    $$\Gamma_1 , \Delta \wfdash M'\headlin{ N / x }: \delta^{j}  \rightarrow \tau$$
    which gives the  derivation:

% and $\Delta \vdash N : \sigma$, then as

    \begin{prooftree}
        \AxiomC{$\Gamma_1 , \Delta \wfdash M'\headlin{ N / x }: \delta^{j}  \rightarrow \tau$}\
        \AxiomC{$\Gamma_2 \wfdash B : \delta^{k} $}
    	\LeftLabel{\redlab{FS{:}app}}
        \BinaryInfC{$\Gamma_1 , \Gamma_2 , \Delta \wfdash ( M'\headlin{ N / x } ) B:\tau $}    
    \end{prooftree}
    
      From \defref{def:headlinfail},   $(M'B) \headlin{ N / x } = ( M'\headlin{ N / x } ) B$, therefore, $\Gamma, \Delta \wfdash  (M' B) \headlin{ N / x } :\tau $ and the result follows.
    
    \item $M = M'[\widetilde{y} \leftarrow y] $.
    
    Then $ \headf{M'[\widetilde{y} \leftarrow y]} = \headf{M'}=x$, for  $y\neq x$. Therefore by inversion of the typing derivation, 
    \begin{prooftree}
        \AxiomC{\( \Gamma_1 , y_1: \delta, \cdots, y_k: \delta , x: \sigma \wfdash M' : \tau \quad y\notin \Gamma_1 \quad k \not = 0\)}
        \LeftLabel{ \redlab{FS{:}share}}
        \UnaryInfC{\( \Gamma_1 , y: \delta^k, x: \sigma \wfdash M'[y_1 , \cdots , y_k \leftarrow y] : \tau \)}
    \end{prooftree}
    where $\Gamma=\Gamma_1 , y: \delta^k$. 
    By IH, the result follows for $M'$, that is, 
    $$\Gamma_1 , y_1: \delta, \cdots, y_k: \delta ,\Delta \wfdash M'\headlin{N/x} : \tau $$
    
    and we have the derivation:
    
    \begin{prooftree}
        \AxiomC{\( \Gamma_1 , y_1: \delta, \cdots, y_k: \delta , \Delta \wfdash  M' \headlin{ N / x} : \tau \quad y\notin \Gamma_1 \quad k \not = 0\)}
        \LeftLabel{ \redlab{FS{:}share} }
        \UnaryInfC{\( \Gamma_1 , y: \delta^k, \Delta \wfdash M' \headlin{ N / x} [\widetilde{y} \leftarrow y] : \tau \)}
    \end{prooftree}
    From \defref{def:headlinfail} one has  $M'[\widetilde{y} \leftarrow y] \headlin{ N / x } = M' \headlin{ N / x} [\widetilde{y} \leftarrow y]$. Therefore, $\Gamma,\Delta\wfdash M'[\widetilde{y} \leftarrow y] \headlin{ N / x }:\tau$ and the result follows.
    \item $M = M'[ \leftarrow y] $.
    
    Then $ \headf{M'[ \leftarrow y]} = \headf{M'}=x$ with  $x \not  = y $, 
    \begin{prooftree}
        \AxiomC{\( \Gamma  , x: \sigma  \wfdash M : \tau\)}
        \LeftLabel{ \redlab{FS{:}weak} }
        \UnaryInfC{\(  \Gamma  , y: \omega, x: \sigma  \wfdash M[\leftarrow y]: \tau \)}
    \end{prooftree}
     and $M'[ \leftarrow y] \headlin{ N / x } = M' \headlin{ N / x} [ \leftarrow y]$. Then by the IH:
    \begin{prooftree}
        \AxiomC{\(  \Gamma , \Delta  \wfdash M \headlin{ N / x}: \tau\)}
        \LeftLabel{ \redlab{FS{:}weak}}
        \UnaryInfC{\(  \Gamma  , y: \omega, \Delta \wfdash M\headlin{ N / x}[\leftarrow y]: \tau \)}
    \end{prooftree}

\item $M = M'[\widetilde{y} \leftarrow y]\esubst{B }{ y}$.

Then $\headf{M'[\widetilde{y} \leftarrow y]\esubst{B }{ y}} = \headf{M'[\widetilde{y} \leftarrow y]} = x \not = y$  by inversion of the typing derivation we have: 

\begin{prooftree}
    \AxiomC{\( \Gamma_1 , \hat{y}:\delta^{k} , x:\sigma \wfdash M'[\widetilde{y} \leftarrow y] : \tau \)}
    \AxiomC{\( \Gamma_2 \wfdash B : \delta^{j} \)}
	\LeftLabel{\redlab{FS{:}ex \dash sub}}
    \BinaryInfC{\( \Gamma_1 , \Gamma_2 ,  x:\sigma \wfdash M' [\widetilde{y} \leftarrow y] \esubst{B }{ y} : \tau \)}
\end{prooftree}
 and $M' [\widetilde{y} \leftarrow y] \esubst{ B}{ y} \headlin{ N / x } = M' [\widetilde{y} \leftarrow y] \headlin{ N / x} \esubst{ B}{ y}$. By IH:

\begin{prooftree}
    \AxiomC{\( \Gamma_1 , \hat{y}:\delta^{k}, \Delta \wfdash  M' [\widetilde{y} \leftarrow y] \headlin{ N / x } : \tau \)}
    \AxiomC{\( \Gamma_2 \wfdash B : \delta^{j} \)}
	\LeftLabel{\redlab{FS{:}ex \dash sub}}
    \BinaryInfC{\( \Gamma_1 , \Gamma_2 ,  \Delta \wfdash M' [\widetilde{y} \leftarrow y] \headlin{ N / x } \esubst{ B }{ y} : \tau \)}
\end{prooftree}

    \item $M =  M' \linexsub {M'' /y}$.
    
    Then $\headf{M' \linexsub {M'' /y}} = \headf{M'} = x \not = y$, by inversion of the typing derivation we have: 
    
    \begin{prooftree}
        \AxiomC{\( \Delta \wfdash M'' : \delta \)}
        \AxiomC{\( \Gamma  , y:\delta , x: \sigma \wfdash M : \tau \)}
        \LeftLabel{ \redlab{FS{:}ex \dash lin \dash sub} }
        \BinaryInfC{\( \Gamma_1, \Gamma_2 , x: \sigma \wfdash M' \linexsub {M'' /y} : \tau \)}
    \end{prooftree}
     and $M' \linexsub {M'' /y}  \headlin{ N / x } = M'  \headlin{N / x } \linexsub {M'' /y}$. Then by the IH:
    
    \begin{prooftree}
        \AxiomC{\( \Delta \wfdash M'' : \delta \)}
        \AxiomC{\( \Gamma  , y:\delta , \Delta  \wfdash M'  \headlin{N / x } : \tau \)}
        \LeftLabel{ \redlab{FS{:}ex \dash lin \dash sub} }
        \BinaryInfC{\( \Gamma_1, \Gamma_2 , \Delta  \wfdash M'  \headlin{N / x } \linexsub {M'' /y} : \tau \)}
    \end{prooftree} \vspace*{-\baselineskip} \qedhere
    \end{myEnumerate}
%\qed
\end{proof}

\applamrsharfailsr*

\begin{proof} By structural induction on the reduction rule from \figref{fig:share-reductfailure} applied in $\expr{M}\red \expr{N}$. There are nine cases to be analyzed:

\begin{myEnumerate}

	\item Rule~$\redlab{RS{:}Beta}$.
	
	Then $\expr{M} = (\lambda x. M[\widetilde{x} \leftarrow x]) B $  and the reduction is:
	  \begin{prooftree}
        \AxiomC{}
        \LeftLabel{\redlab{RS{:}Beta}}
        \UnaryInfC{\((\lambda x. M[\widetilde{x} \leftarrow x]) B \red M[\widetilde{x} \leftarrow x]\ \esubst{ B }{ x }\)}
     \end{prooftree}

 	where $ \expr{M}'  =  M[\widetilde{x} \leftarrow x]\ \esubst{ B }{ x }$. Since $\Gamma\wfdash \expr{M}:\tau$ we get the following derivation by inversion of the typing derivation:
	\begin{prooftree}
			\AxiomC{$ \Gamma' , x_1:\sigma , \cdots , x_j:\sigma  \wfdash  M: \tau $}
			\LeftLabel{ \redlab{FS{:}share} }
			\UnaryInfC{$  \Gamma' , x:\sigma^{j}  \wfdash  M[\widetilde{x} \leftarrow x]: \tau $}
			\LeftLabel{ \redlab{FS{:}abs \dash sh} }
            \UnaryInfC{$ \Gamma' \wfdash \lambda x. M[\widetilde{x} \leftarrow x]: \sigma^{j} \rightarrow \tau $}
              \AxiomC{$\Delta \wfdash B: \sigma^{k} $}
			\LeftLabel{ \redlab{FS{:}app} }
		\BinaryInfC{$ \Gamma' , \Delta \wfdash (\lambda x. M[\widetilde{x} \leftarrow x]) B:\tau $}
	\end{prooftree}
	for $\Gamma = \Gamma' , \Delta $ and $x\notin \dom{\Gamma'}$. 
	Notice that: 

    \begin{prooftree}
                \AxiomC{$ \Gamma' , x_1:\sigma , \cdots , x_j:\sigma  \wfdash  M: \tau $}
			\LeftLabel{ \redlab{FS{:}share} }
			\UnaryInfC{$  \Gamma' , x:\sigma^{j}  \wfdash  M[\widetilde{x} \leftarrow x]: \tau $}
                \AxiomC{$\Delta \wfdash B:\sigma^{k}  $}
                \LeftLabel{ \redlab{FS{:}ex \dash sub} }
            \BinaryInfC{$ \Gamma' , \Delta \wfdash M[\widetilde{x} \leftarrow x]\ \esubst{ B }{ x }:\tau $}
    \end{prooftree}

    Therefore $ \Gamma',\Delta\wfdash\expr{M}' :\tau$ and the result follows.
    
    \item Rule~$ \redlab{RS{:}Ex \dash Sub}$.
    
    Then $ \expr{M} =  M[x_1, \cdots , x_k \leftarrow x]\ \esubst{ B }{ x }$ where $B=  \bag{N_1, \dots ,N_k} $. By inversion of the typing derivation the reduction is:
    
    \hspace*{-30pt}
    \begin{minipage}{\linewidth}    
    \begin{prooftree}
        \AxiomC{$B = \bag{N_1,
        \cdots ,N_k} \quad k \geq  1 $}
        \AxiomC{$ M \not= \fail^{\widetilde{y}} $}
        \LeftLabel{\redlab{RS{:}Ex \dash Sub}}
        \BinaryInfC{\( M[x_1, \cdots , x_k \leftarrow x]\ \esubst{ B }{ x } \red \sum_{B_i \in \perm{B}}M\ \linexsub{B_i(1)/x_1} \cdots \linexsub{B_i(k)/x_k}    \)}
    \end{prooftree}
    \end{minipage}    
    
    and $\expr{M'}= \sum_{B_i \in \perm{B}}M\ \linexsub{B_i(1)/x_1} \cdots \linexsub{B_i(k)/x_k}.$
    To simplify the proof we take $k=2$, as the case $k>2$ is similar. Therefore,
    
    \begin{itemize}
        \item $B=\bag{N_1,N_2}$; and
        \item $\perm{B}=\{\bag{N_1,N_2}, \bag{N_2,N_1}\}$
    \end{itemize} 
    
    Since $\Gamma\wfdash \expr{M}:\tau$ we get a derivation where we first type the bag $B$ with the derivation $\Pi$, given next:
    
    \begin{prooftree}
                    \AxiomC{\(  \Delta_1 \wfdash N_1 : \sigma \)}
                                    \AxiomC{\( \Delta_{2} \wfdash N_{2} : \sigma   \)}
        				\AxiomC{\(  \)}
                        \LeftLabel{ \redlab{FS{:}\oneb} }
                        \UnaryInfC{\( \wfdash \oneb : \omega \)}
        			\LeftLabel{ \redlab{FS{:}bag} }
        			\BinaryInfC{\( \Delta_2  \wfdash \bag{N_2}: \sigma   \)}
        			\LeftLabel{ \redlab{FS{:}bag} }
                 \BinaryInfC{\(\Delta  \wfdash B: \sigma \wedge \sigma \) }
    \end{prooftree}
The full derivation is as follows:
    \begin{prooftree}
                    \AxiomC{$ \Gamma' , x_1:\sigma ,  x_2:\sigma  \wfdash  M: \tau $}
			    \LeftLabel{ \redlab{FS{:}share} }
			    \UnaryInfC{$  \Gamma' , x:\sigma \wedge \sigma  \wfdash  M[\widetilde{x} \leftarrow x]: \tau $}
                                        \AxiomC{\(  \Pi \)}
                    \LeftLabel{ \redlab{FS{:}ex \dash sub} }
            \BinaryInfC{$ \Gamma' , \Delta \wfdash M[\widetilde{x} \leftarrow x]\ \esubst{ B }{ x }:\tau $}
    \end{prooftree}

    where $\Delta= \Delta_1,\Delta_2$ and $\Gamma = \Gamma' , \Delta $. We can build a derivation $\Pi_{1,2}$ of $ \Gamma' , \Delta \wfdash  M \linexsub{N_1/x_1} \linexsub{N_2/x_2}    : \tau$ as :
        \begin{prooftree}
        \AxiomC{\( \Gamma'  , x_1:\sigma, x_2:\sigma \wfdash M : \tau \)}
        \AxiomC{\( \Delta_1 \wfdash N_1 : \sigma \)}
        \LeftLabel{ \redlab{FS{:}ex \dash lin \dash sub} }
        \BinaryInfC{\( \Gamma , \Delta_1  , x_2:\sigma \wfdash M \linexsub{N_1 /x_1} : \tau \)}
        \AxiomC{\( \Delta_2 \wfdash N_2 : \sigma \)}
        \LeftLabel{ \redlab{FS{:}ex \dash lin \dash sub} }
        \BinaryInfC{$ \Gamma' , \Delta \wfdash  M \linexsub{N_1/x_1} \linexsub{N_2/x_2}    : \tau$}
    \end{prooftree}
    
    Similarly, we can obtain a derivation $\Pi_{2,1}$ of $ \Gamma' , \Delta \wfdash  M \linexsub{N_2/x_1} \linexsub{N_1/x_2}    : \tau$.   Finally, applying Rule~\redlab{FS{:}sum}:
    \begin{prooftree}
        \AxiomC{\( \Pi_{1,2} \)}
        \AxiomC{\( \Pi_{2,1} \)}
        \LeftLabel{ \redlab{FS{:}sum} }
        \BinaryInfC{$ \Gamma' , \Delta \wfdash  M \linexsub{N_1 /x_1} \linexsub{N_2/x_k} + M \linexsub{N_2/x_1} \linexsub{N_1/x_k}   : \tau $}
    \end{prooftree}
    
    and the result follows.
    
    \item Rule~$\redlab{RS{:}Lin \dash Fetch} $.
    
    Then $ \expr{M} =M\ \linexsub{N/x}  $ where  $\headf{M} = x$. The reduction is: 
    
    \begin{center}
     \AxiomC{$ \headf{M} = x$}
     \LeftLabel{\redlab{RS{:}Lin \dash Fetch}}
     \UnaryInfC{\(M\ \linexsub{N/x} \red  M \headlin{ N/x } \)}
     \DisplayProof
    \end{center}
    and $\expr{M'}=M\headlin{N/x}$.     Since $\Gamma\wfdash \expr{M}:\tau$ we get the following derivation by inversion of the typing derivation:
    
    \begin{prooftree}
        \AxiomC{\( \Delta \wfdash N : \sigma \)}
        \AxiomC{\( \Gamma'  , x:\sigma \wfdash M : \tau \)}
        \LeftLabel{ \redlab{FS{:}ex \dash lin \dash sub} }
        \BinaryInfC{\( \Gamma', \Delta \wfdash M \linexsub{N / x} : \tau \)}
    \end{prooftree}
        
    where $\Gamma = \Gamma' , \Delta $. By the Substitution Lemma (Lemma~\ref{l:lamrsharfailsubs}), we obtain a derivation $\Gamma' , \Delta \wfdash   M \headlin{ N/x } : \tau $, and the result follows.

\item Rule~$\redlab{RS{:}TCont}$.

Then $\expr{M} = C[M]$ and the reduction is as follows:
\begin{prooftree}
        \AxiomC{$   M \red M'_{1} + \cdots +  M'_{k} $}
        \LeftLabel{\redlab{RS{:}TCont}}
        \UnaryInfC{$ C[M] \red  C[M'_{1}] + \cdots +  C[M'_{k}] $}
\end{prooftree}
with $\expr{M'} =  C[M'_{1}] + \cdots +  C[M'_{k}] $. 
The proof proceeds by analysing the context $C$. \\
There are four cases:
% $C=[\cdot]\ B$, $C=([\cdot])\linexsub{N/x} $, $C=([\cdot])[\widetilde{x} \leftarrow x]$, and $C= ([\cdot])[ \leftarrow x]\esubst{\oneb}{ x}$.
% We content ourselves by considering only the first one:

\begin{enumerate}
    \item $C=[\cdot]\ B$.
    
    In this case $\expr{M}=M \ B$, for some $B$. Since $\Gamma\vdash \expr{M}:\tau$ by inversion of the typing derivation, one has the derivation:
\begin{prooftree}
    \AxiomC{\(  \Gamma' \wfdash  M: \sigma^{j} \rightarrow \tau \)}
    \AxiomC{\( \Delta \wfdash  B : \sigma^{k} \)}
        \LeftLabel{ \redlab{FS{:}app} }
    \BinaryInfC{\( \Gamma', \Delta \wfdash  M\ B : \tau\)}
\end{prooftree}

where $\Gamma = \Gamma' , \Delta $. From  $\Gamma'\wfdash M:\sigma^j\rightarrow\tau$ and the reduction $M \red M'_{1} + \cdots +  M'_{k} $, one has by IH that  $\Gamma'\wfdash M_1'+\ldots, M_k':\sigma^j\rightarrow\tau$, which entails $\Gamma'\wfdash M_i':\sigma^j\rightarrow\tau$, for $i=1,\ldots, k$, via Rule~\redlab{FS{:}sum}. Finally, we may type the following:
\begin{prooftree}
            \AxiomC{\(  \forall i \in {1 , \cdots , l} \)}
			    \AxiomC{\(  \Gamma' \wfdash  M'_{i}: \sigma^{j} \rightarrow \tau \)}
    				\AxiomC{\( \Delta \wfdash  B : \sigma^{k} \)}
        		\LeftLabel{ \redlab{FS{:}app} }
			\BinaryInfC{\(  \Gamma', \Delta \wfdash (M'_{i}\ B):  \tau \)}
			
        \LeftLabel{ \redlab{FS{:}sum} }
    \BinaryInfC{\( \Gamma', \Delta \wfdash (M'_{1}\ B) + \cdots +  (M'_{l} \ B) : \tau\)}
\end{prooftree}

Since $ \expr{M}'  =   (C[M'_{1}]) + \cdots +  (C[M'_{l}]) = M_1'B+\ldots+M_k'B$, the result follows.
\item  Cases $C=[\cdot]\linexsub{N/x} $ and $C=[\cdot][\widetilde{x} \leftarrow x]$ are similar to the previous one.
\item $C= [\cdot][ \leftarrow x]\esubst{\oneb}{ x}$

In this case $\expr{M}=C[M]=M[\leftarrow x] \esubst{\oneb}{x}$. Since $\Gamma\wfdash \expr{M}:\tau$ by inversion of the typing derivation, one has a derivation

\begin{prooftree}
    \AxiomC{$\Gamma\wfdash M:\tau$}
    \LeftLabel{\redlab{FS \dash weak}}
    \UnaryInfC{\(  \Gamma, x:\omega \wfdash  M[\leftarrow x] : \tau \)}
    \AxiomC{}
    \LeftLabel{\redlab{TS{:} \oneb}}
    \UnaryInfC{\(  \vdash  \oneb: \omega \)}
    \LeftLabel{\redlab{FS{:} wf \dash bag}}
    \UnaryInfC{\(  \wfdash  \oneb: \omega \)}
        \LeftLabel{ \redlab{FS{:}ex\dash sub} }
    \BinaryInfC{\(  \Gamma \wfdash  M[\leftarrow x] \esubst{\oneb}{x}:\tau \)}
\end{prooftree}

From $M\red M_1+\ldots+M_k$ and $\Gamma \wfdash M:\tau$, by the IH, it follows that $\Gamma\wfdash M_1+\ldots+ M_k:\tau$, and consequently, $\Gamma\wfdash M_i$, via application of \redlab{FS{:}sum}. Therefore, there exists a derivation

\begin{prooftree}
            \AxiomC{$\Gamma \wfdash M_i:\tau $}
            \UnaryInfC{$\Gamma, x:\omega\wfdash M_i[\leftarrow  x]:\tau $}
            \AxiomC{$\wfdash \oneb:\omega$}
            \BinaryInfC{$\Gamma \wfdash M_i[\leftarrow  x]\esubst{\oneb}{x}:\tau $}
\end{prooftree}
for each $i=1,\ldots,k$. By applying \redlab{FS{:}sum}, we obtain $\Gamma \wfdash M_1[\leftarrow  x]\esubst{\oneb}{x}+\ldots+ M_k[\leftarrow  x]\esubst{\oneb}{x}:\tau $, and the result follows.
\end{enumerate}

	\item Rule~$ \redlab{RS{:}ECont} $.
	
Then $\expr{M} = D[\expr{M}_1 ]$ where $\expr{M}_1  \red \expr{M}_2$ then we can perform the following reduction:

\begin{prooftree}
        \AxiomC{$ \expr{M}_1  \red \expr{M}_2 $}
        \LeftLabel{$\redlab{RS{:}ECont}$}
        \UnaryInfC{$D[\expr{M}_1]  \red D[\expr{M}_2]  $}
\end{prooftree}

and $\expr{M'}=D[\expr{M}_2]$. \\
The proof proceeds by analysing the context $D$. There are two cases: 
$D= [\cdot] + \expr{N}$ and $D= \expr{N} + [\cdot]$. 
We analyze only the first one:

    $D= [\cdot] + \expr{N}$. In this case $\expr{M}= \expr{M}_1+\expr{N}$ and by inversion of the typing derivation:
\begin{prooftree}
    \AxiomC{$ \Gamma \wfdash  \expr{M}_1 : \tau$}
    \AxiomC{$ \Gamma \wfdash  \expr{N} : \tau$}
    \LeftLabel{ \redlab{FS{:}sum} }
    \BinaryInfC{$ \Gamma \wfdash  \expr{M}_1+\expr{N}: \tau$}
\end{prooftree}

From $ \Gamma \wfdash \expr{M}_1 : \tau$ and $ \expr{M}_1  \red \expr{M}_2  $, by IH, one has that $\Gamma \wfdash  \expr{M}_2 : \tau$.
Hence we may type the following:
\begin{prooftree}
    \AxiomC{$ \Gamma \wfdash  \expr{M}_2 : \tau$}
    \AxiomC{$ \Gamma \wfdash  \expr{N} : \tau$}
    \LeftLabel{ \redlab{FS{:}sum} }
    \BinaryInfC{$ \Gamma \wfdash  \expr{M}_2+\expr{N}: \tau$}
\end{prooftree}
Since $\expr{M}'=D[\expr{M}_2]=\expr{M}_2+\expr{N}$,  the result follows.

\item Rule~$ \redlab{RS{:}Fail} $.

Then $\expr{M} =  M[x_1, \cdots , x_k \leftarrow x]\ \esubst{ B }{ x } $ where $B = \bag{N_1,\dots ,N_l}  $ and  the reduction is:
 \begin{prooftree}
    \AxiomC{$k \neq \size{B}$}
     \AxiomC{$ \widetilde{y} = (\lfv{M} \setminus \{  x_1, \cdots , x_k \} ) \cup \lfv{B} $}
    \LeftLabel{\redlab{RS{:}Fail}}
    \BinaryInfC{\( M[x_1, \cdots , x_k \leftarrow x]\ \esubst{ B }{ x } \red \sum_{B_i \in \perm{B}}  \fail^{\widetilde{y}} \)}
 \end{prooftree}
where $\expr{M'}=\sum_{B_i \in \perm{B}}  \fail^{\widetilde{y}}$. Since $\Gamma\wfdash \expr{M}$ and by inversion of the typing derivation, one has a derivation:
\begin{prooftree}
            \AxiomC{\( \Gamma' , x_1:\sigma,\ldots, x_k:\sigma \wfdash M: \tau \)}
            \LeftLabel{ \redlab{FS{:}ex \dash sub} }    
            \UnaryInfC{\( \Gamma' , x:\sigma^{k} \wfdash M[x_1, \cdots , x_k \leftarrow x] : \tau \)}
            \AxiomC{\( \Delta \wfdash B : \sigma^{j} \)}
        \LeftLabel{ \redlab{FS{:}ex \dash sub} }    
        \BinaryInfC{\( \Gamma', \Delta \wfdash M[x_1, \cdots , x_k \leftarrow x]\ \esubst{ B }{ x }  : \tau \)}
    \end{prooftree}

where $\Gamma = \Gamma' , \Delta $. We may type the following:
    \begin{prooftree}
        \AxiomC{\( \)}
        \LeftLabel{ \redlab{FS{:}fail}}
        \UnaryInfC{\( \Gamma' , \Delta \wfdash  \fail^{\widetilde{y}} : \tau  \)}
    \end{prooftree}
since $\Gamma',\Delta$ contain assignments on the free variables in $M$ and $B$.
% \daniele{I need to be sure that $\{x_1,\ldots, x_k\}$ does not occur in $\Delta$. Is there a place where this is guaranteed?}\joe{Yes, this is guaranteed by the sharing construct and variables only appearing once within a term}.
Therefore, $\Gamma\wfdash \fail^{\widetilde{y}}:\tau$, by applying \redlab{FS{:}sum}, it follows that $\Gamma\wfdash \sum_{B_i\in \perm{B}}\fail^{\widetilde{y}}:\tau$, as required. 
\item Rule~$\redlab{RS{:}Cons_1}$.

Then $\expr{M} =   \fail^{\widetilde{x}}\ B $ where $B = \bag{N_1, \dots ,N_k} $  and  the  reduction is:

\begin{prooftree}
    \AxiomC{\( B = \bag{N_1, \dots ,N_k} \)}
    \AxiomC{\( \widetilde{y} = \lfv{B} \)}
    \LeftLabel{$\redlab{RS{:}Cons_1}$}
    \BinaryInfC{\( \fail^{\widetilde{x}} \ B  \red \sum_{\perm{B}} \fail^{\widetilde{x} \cup \widetilde{y}} \)}
\end{prooftree}
and $ \expr{M}'  =  \sum_{\perm{B}} \fail^{\widetilde{x} \cup \widetilde{y}} $.
Since $\Gamma\wfdash \expr{M}:\tau$ and by inversion  of the typing derivation, one has the derivation: 
    \begin{prooftree}
        \AxiomC{\( \)}
        \LeftLabel{\redlab{FS{:}fail}}
        \UnaryInfC{\( \Gamma' \wfdash \fail^{\widetilde{x}} : \omega \rightarrow \tau \)}
        \AxiomC{\( \Delta \wfdash B : \pi \)}
            \LeftLabel{\redlab{FS{:}app}}
        \BinaryInfC{\( \Gamma', \Delta \wfdash \fail^{\widetilde{x}}\ B : \tau\)}
    \end{prooftree}

where $\Gamma = \Gamma' , \Delta $. After $\perm{B}$ applications of \redlab{FS{:}sum}, we obtain $ \Gamma \wfdash$\linebreak[4]$\sum_{\perm{B}} \fail^{\widetilde{x} \cup \widetilde{y}}: \tau$, and the result follows.

\item Rule~$\redlab{RS{:}Cons_2}$.

Then $\expr{M} =   (\fail^{\widetilde{x}\cup \widetilde{y}} [ \widetilde{x} \leftarrow x])\esubst{ B }{ x } $for $B = \bag{N_1, \dots ,N_k} $  and the reduction is:

\begin{prooftree}
    \AxiomC{\(  B = \bag{N_1, \dots ,N_k} \quad k+|\widetilde{x}|\neq 0 \)}
    \AxiomC{\(  \widetilde{y} = \lfv{B} \)}
    \LeftLabel{$\redlab{RS{:}Cons_2}$}
    \BinaryInfC{\( (\fail^{\widetilde{x} \cup \widetilde{y}} [ \widetilde{x} \leftarrow x])\esubst{ B }{ x } \red \sum_{\perm{B}} \fail^{\widetilde{y} \cup \widetilde{z}} \)}
\end{prooftree}
with $\expr{M'}=\sum_{\perm{B}} \fail^{\widetilde{y}\cup \widetilde{z}}$. Since $\Gamma\wfdash \expr{M}:\tau$ and by inversion  of the typing derivation, one has the derivation:

    \begin{prooftree}
              \AxiomC{\( \)}
            \LeftLabel{\redlab{FS{:}fail}}
            \UnaryInfC{\( \Delta , x_1: \sigma, \cdots, x_j: \sigma \wfdash \fail^{\widetilde{x}\cup \widetilde{y}} : \tau \quad x\notin \Delta \quad k \not = 0\)}
            \LeftLabel{ \redlab{FS{:}share}}
            \UnaryInfC{\( \Delta , x: \sigma^j \wfdash \fail^{\widetilde{x}\cup \widetilde{y}}[x_1 , \cdots , x_j \leftarrow x] : \tau \)}
             \AxiomC{\( \Delta \wfdash B : \sigma^{k} \)}
        \LeftLabel{\redlab{FS{:}ex \dash sub}}    
        \BinaryInfC{\( \Gamma, \Delta \wfdash \fail^{\widetilde{x}\cup \widetilde{y}} [\widetilde{x} \leftarrow x]\esubst{ B }{ x } : \tau \)}
    \end{prooftree}

Hence $\Gamma = \Gamma' , \Delta $ and $ \expr{M}'  =  \sum_{\perm{B}} \fail^{\widetilde{y}\cup \widetilde{z}}$ and we may type the following:
    \begin{prooftree}
        \AxiomC{\( \)}
        \LeftLabel{\redlab{FS{:}fail}}
        \UnaryInfC{$ \Gamma \wfdash \fail^{\widetilde{y}\cup \widetilde{z}} : \tau$}
        \AxiomC{$ \cdots $}
        \LeftLabel{\redlab{FS{:}sum}}
        \BinaryInfC{$ \Gamma \wfdash \sum_{\perm{B}} \fail^{\widetilde{y}\cup \widetilde{z}}: \tau$}
    \end{prooftree}
    \item Rule~$\redlab{RS{:}Cons_3}$.
    
    Then $\expr{M}=\fail^{\widetilde{y}\cup x}$ and the reduction is 
    
    \begin{prooftree}
    \AxiomC{$\widetilde{z}=\lfv{N}$}
    \LeftLabel{\redlab{RS{:}Cons_3}}
    \UnaryInfC{$\fail^{\widetilde{y}\cup x}\linexsub{N/x}\red \fail^{\widetilde{y}\cup \widetilde{z}}$}
    \end{prooftree}
    with $\expr{M'}=\fail^{\widetilde{y}\cup \widetilde{z}}$. Since $\Gamma\wfdash \expr{M}$ and by inversion of the typing derivation, one has the derivation
    
    \begin{prooftree}
    \AxiomC{}
        \LeftLabel{\redlab{FS{:}fail}}
    \UnaryInfC{$\Gamma',x:\sigma \wfdash \fail^{\widetilde{y}\cup x}:\tau$}
    \AxiomC{$\Delta\wfdash N:\sigma$}
    \LeftLabel{\redlab{FS{:}ex\dash lin\dash sub}}
    \BinaryInfC{$\Gamma',\Delta \wfdash \fail^{\widetilde{y}\cup x}\linexsub{N/x}:\tau$}
    \end{prooftree}
    where $x\notin \dom{\Gamma'}$, $\dom{\Gamma'}=\widetilde{y}$ and $\dom{\Delta}=\widetilde{z}=\lfv{N}$.
    
    We can type the following:
    \begin{prooftree}
    \AxiomC{}
    \LeftLabel{\redlab{FS{:}fail}}
    \UnaryInfC{$\Gamma',\Delta\wfdash\fail^{\widetilde{y}\cup \widetilde{z}}:\tau $}
    \end{prooftree}
    and the result follows. \qedhere
    \end{myEnumerate}
%\qed
\end{proof}

\consistencytype*

\begin{proof}
\secondrev{
By induction on the typing derivation, with a case analysis on the last applied rule (\Cref{fig:wfsh_rules}). We only consider the cases for the typing rules that relate to the sharing construct and the explicit substitution. 
    First, consider conditions~1(i)~to~1(iv), which are related to $M [ \widetilde{x} \leftarrow x ]$. The conditions are as follows (i) 
	$\widetilde{x}$ contains pairwise distinct variables; 
	(ii)~every $x_i \in \widetilde{x}$ must occur exactly once in $M$; (iii) $x_i$ is not a sharing variable;
	(iv)~$M$ is consistent. By considering rule~$\redlab{FS{:}share}$, we have:
    \begin{prooftree}
        \AxiomC{\( \Gamma , x_1: \sigma, \cdots, x_k: \sigma \wfdash M : \tau \quad x\notin \dom{\Gamma} \quad k \not = 0\)}
        \LeftLabel{ \redlab{FS{:}share}}
        \UnaryInfC{\( \Gamma , x: \sigma^{k} \wfdash M[x_1 , \cdots , x_k \leftarrow x] : \tau \)}  
    \end{prooftree}
    Condition~1(i) follows from uniquness of variables within the context. Condition~1(ii) follows from the premise, which ensures that $M$ is well-formed with a context including each $x_i$; linearity conditions imply that each $x_i$ must be consumed so it must occur in  $M$. Condition~1(iii) also follows directly from the well-formedness of $M$: each $x_i$ is typed with a strict type, and the rule ensures that the sharing variable $x$ is typed with the multiset type $ \sigma^{k}$. Finally condition~1(iv) is ensured by the IH.
    }
    
    \secondrev{
    For conditions~2(i)~to~2(iv) which are (i) the variable $x$ must occur exactly once in $M$;
	(ii) $x$ cannot be a sharing variable; 	(iii)~$M$ and $N$ are consistent; (iv)~$\lfv{M} \cap \lfv{N} = \emptyset$. Consider the case of rule $\redlab{FS{:}ex \dash lin \dash sub}$: 
    \begin{prooftree}
        \AxiomC{\( \Gamma  , x:\sigma \wfdash M : \tau \quad \Delta \wfdash N : \sigma \)}
        \LeftLabel{\redlab{FS{:}ex \dash lin \dash sub}}
        \UnaryInfC{\( \Gamma, \Delta \wfdash M \linexsub{N / x} : \tau \)}
    \end{prooftree}
    First, because $\Gamma$ and $\Delta$ are disjoint, $x$ cannot appear within $\Delta$ and $M$ must consume the type of $x:\sigma$; hence $x$ must occur in $M$, satisfying condition~2(i) and 2(iv). 
    Second, $\Gamma , x:\sigma$ ensures a strict type for $x$;       
    if $x$ were a sharing variable in $M$ then $x$ would have a multiset type $\pi$. Therefore, condition~2(ii) is satisfied. Finally, condition~(iii) is satisfied by induction on $M$ and $N$.
    }
\end{proof} 

\consistencyencode*
\begin{proof}
\secondrev{
  By induction on the structure of $\expr{M}$. 
  Notice that $\recencodf{\cdot}$ ensures consistency for bound variables: it replaces all occurrences of a bound variable (say $y$) with fresh bound variables (say, $y_1, \ldots, y_k$). Thus, the following hold for bound variables: (i) they occur once within a term and (ii) they are not shared themselves, as the sharing of variables only occurs when handling binders associated to explicit substitutions and abstractions. 
  As for free variables, the translation $\recencodopenf{\cdot}$ replaces each occurrence with a fresh variable, and does so before applying $\recencodf{\cdot}$; this ensures that free variables that are already shared are not shared again. Because of this design, the translations preserve consistency.
}
\end{proof}

%%%%%%%%%%%%%%%%%%%%%%%%%%%%%%%%%%%%%%%%%%%%
\section{Appendix to  \texorpdfstring{\secref{ss:firststep}}{§ 5.2}}
\label{app:firststep}

%%%%%%%%%%%%%%%%%%%%%%%%%%%%%%%%%%%%%%%%%%%%%%%%%
\subsection{Encoding \texorpdfstring{$\recencodf{\cdot}$}{⦇⋅⦈^•}}
\label{apen:firstencod}

%%%%%%%%%%%%%%%%%%%%%%%%%%%%%%%%%%%%%%%%%%%%%%%%%%%%%%%%%%%%%%%%%%%%%%%%%%%%%%%%%
\subsubsection[Auxiliary Encoding]{ Auxiliary Encoding: From \texorpdfstring{$\lamr$}{} into \texorpdfstring{$\lamrshar$}{} }

%The next result is a bit displaced here.

\begin{restatable}[]{prop}{}
\label{prop:linhed_enc}
The encoding commutes with linear substitution: $ \recencodf{M\headlin{N/x}}=$\linebreak[4]$\recencodf{M}\headlin{\recencodf{N}/x}$

\end{restatable}

\begin{proof}
By induction of the structure of $M$ in $M\headlin{\recencodf{N}/x}$. 
%\qed
\end{proof}

 \begin{restatable}[Well-typedness preservation for $\recencodf{-}$]{prop}{typeencintolamrfail}
\label{prop:typeencintolamrfail}
\revo{A17,A18,A19}{
Let $B$ and $\expr{M}$  be a  bag and an  expression in $\lamrfail$, respectively. 
\begin{enumerate}
\item
    If $\Gamma \vdash B:\sigma$ 
then $ \strcore{\Gamma} \vdash \recencodf{B}:\sigma$.
    \item 
    If $\Gamma \vdash \expr{M}:\sigma$  
then $ \strcore{\Gamma} \vdash \recencodf{\expr{M}}:\sigma$.
\end{enumerate}}
\end{restatable}

\begin{proof}

 By mutual induction on the typing derivations for $B$ and $\expr{M}$, with an analysis of the last rule applied. 

\noindent Part~(1) includes two cases:
    \begin{enumerate}[i)]
    \item Rule~$\redlab{T:\oneb}$: Then $B = \oneb$ and the thesis follows trivially, 
    because the encoding of terms/bags  
        (cf.  Figure~\ref{fig:auxencfail})
        ensures that $\recencodf{\oneb}=\oneb$.
                
        \item Rule~$\redlab{T:bag}$. Then $B = \bag{M}\cdot A$, where $M$ is a term and $A$ is a bag,
        and 
        \begin{prooftree}
                    \AxiomC{\( \Gamma \vdash M : \sigma\)}
                \AxiomC{\( \Delta \vdash A : \pi\)}
                \LeftLabel{$\redlab{T:bag}$}
            \BinaryInfC{\( \Gamma \contexcat \Delta \vdash \bag{M}\cdot A:\sigma \wedge \pi\)}
            \end{prooftree}
By the IHs, we have both $\strcore{\Gamma} \vdash \recencodf{M} : \sigma$ and $\strcore{\Delta} \vdash \recencodf{A} : \pi$. The thesis then  follows by applying Rule~$\redlab{TS:bag}$ in \lamrsharfail:
        \begin{prooftree}
                \AxiomC{\( \strcore{\Gamma}   \vdash \recencodf{M} : \sigma\)}
                \AxiomC{\( \strcore{\Delta}  \vdash \recencodf{A} : \pi\)}
                \LeftLabel{$\redlab{TS:bag}$}
            \BinaryInfC{\( \strcore{\Gamma} , \strcore{\Delta}  \vdash \bag{\recencodf{M}}\cdot \recencodf{A}:\sigma \wedge \pi\)}
        \end{prooftree}
        \end{enumerate}
    
    \noindent Part~(2) considers six cases:
    
    \begin{enumerate}[i)]

                \item Rule~$\redlab{T:var}$: Then $\expr{M} = x$ 
        and 
                \begin{prooftree}
            \AxiomC{}
            \LeftLabel{$\redlab{T:var}$}
            \UnaryInfC{\( x: \sigma \vdash x : \sigma\)}    
        \end{prooftree}
        By the encoding of terms
        % and contexts 
        (cf. \figref{fig:auxencfail}
        %and \defref{d:enclamcontfail}, respectively
        ), we infer $x: \sigma \vdash x : \sigma$ and so 
        the thesis holds immediately. 
        
        % \begin{prooftree}
        %     $\Bigg\{\!\!\!\Bigg\{$
        %     \AxiomC{}
        %     \LeftLabel{(var)}
        %     \UnaryInfC{\( x: \sigma \vdash x : \sigma\)}
        %     \DisplayProof
        %     $\Bigg\}\!\!\!\Bigg\}^{Judge}$
        %     \hspace{0.5cm}
        %     =
        %     \hspace{0.5cm}
        %     \AxiomC{}
        %     \LeftLabel{(var)}
        %     \UnaryInfC{\( x: \sigma \vdash x : \sigma\)}
        % \end{prooftree}
        
        \item Rule~$\redlab{T:abs}$: Then  $\expr{M} = \lambda x . M $ 
        and 
          \begin{prooftree}
    \AxiomC{\( \Gamma , x: \sigma^n \vdash M : \tau \)}
    \LeftLabel{$\redlab{T:abs}$}
    \UnaryInfC{\( \Gamma \vdash \lambda x . M :  \sigma^n  \rightarrow \tau \)}
            \end{prooftree}
            
                 By the encoding of terms (cf. \figref{fig:auxencfail}), we have 
        $\recencodf{  \expr{M}  }  =   \lambda x . \recencodf{M\linsub{x_1,\cdots, x_n}{x} } [\widetilde{x}\leftarrow x]$, where  $\#(x,M) = n$ and each $x_i$ is fresh.
        
            We work on the premise $\Gamma , {x}: \sigma^n \vdash M : \tau$ before appealing to the IH.
            % \daniele{Since the shared variables $x_1,\ldots, x_n$ have the same type as $x$, it follows that $x_i:\sigma\vdash x_i:\sigma$, for all $i=1,\ldots,n$. Does our \redlab{weak} rule give us $\Gamma,x_i:\sigma\vdash x_i:\sigma$? I do not see how, it seems that we can only weaken with the empty type \ldots. Perhaps we don't need this $\Gamma$ in the context of the derivation of $x_i:\sigma$?}\joe{ I dont understand this comment}
            
            Then, by  $n$ applications of Lemma~\ref{lem:preser_linsub} to this judgment, we obtain 
            \begin{equation}\label{eq:tp1}
             {\Gamma}, {x_1: \sigma, \cdots,  x_n: \sigma} \vdash {M}\linsub{x_1,\cdots, x_n}{x} : \tau
            \end{equation}
            By IH on \eqref{eq:tp1} we have
              \begin{equation}\label{eq:tp2}
             \strcore{\Gamma}, {x_1: \sigma, \cdots,  x_n: \sigma} \vdash \recencodf{M\linsub{x_1,\cdots, x_n}{x}} : \tau
            \end{equation}
            
              Starting from \eqref{eq:tp2}, we then have the following type derivation for  $\recencodf{\expr{M}}$, which concludes the proof for this case:
        
        % \daniele{I think we need to have a proposition about this (*) statement. } \jp{Fixed, please check.}

        \begin{prooftree}
        %\AxiomC{\(\recencodf{\Gamma}, x: \sigma, \cdots,  x: \sigma \vdash \recencodf{M} : \tau\)}
         \LeftLabel{$(*)$}
                    \AxiomC{\( \strcore{\Gamma} , x_1: \sigma, \cdots, x_n: \sigma \vdash \recencodf{M\linsub{x_1,\cdots, x_n}{x}} : \tau \)}
            \LeftLabel{$\redlab{TS:share}$}
            \UnaryInfC{\( \strcore{\Gamma} , x: \sigma^n  \vdash \recencodf{M\linsub{x_1,\cdots, x_n}{x}}[x_1 , \cdots , x_n \leftarrow x] : \tau \)}
            \LeftLabel{$\redlab{TS:abs \dash sh}$}
            \UnaryInfC{\( \strcore{\Gamma} \vdash \lambda x . (\recencodf{M\linsub{x_1,\cdots, x_n}{x}}[x_1 , \cdots , x_n \leftarrow x]) : \sigma^n \rightarrow \tau \)}
        \end{prooftree}

        % and $ \recencodf{ \Delta \vdash \lambda x . M : \sigma\wedge \sigma \wedge \cdots \wedge \sigma \rightarrow \tau }^{Judge} =  \recencodf{\Delta}^{Cont} \vdash \lambda x . (\recencodf{M\langle x_1 / x  \rangle \cdots \langle x_n / x  \rangle} 
        % [x_1 , \cdots , x_n \leftarrow x]) : \sigma\wedge \sigma \wedge \cdots \wedge \sigma \rightarrow \tau $
        
        % \begin{prooftree}
        %     $\Bigg\{\!\!\!\Bigg\{$
        %     \AxiomC{\( \Delta , x: \sigma , \cdots , x: \sigma \vdash M : \tau \quad x\notin \Delta \)}
        %     \LeftLabel{(abs)}
        %     \UnaryInfC{\( \Delta \vdash \lambda x . M : (\sigma \wedge \cdots \wedge \sigma) \rightarrow \tau \)}
        %     \DisplayProof
        %     $\Bigg\}\!\!\!\Bigg\}^{Judge}$
        %     \hspace{0.5cm}
        %     =
        %     \hspace{0.5cm}
        %     \AxiomC{\( \recencodf{\Delta}^{Cont} , x_1: \sigma, \cdots, x_n: \sigma \vdash \recencodf{M\langle x_1 / x  \rangle \cdots \langle x_n / x  \rangle} : \tau \quad x\notin \Delta \quad n \not = 0\)}
        %     \LeftLabel{ (share)}
        %     \UnaryInfC{\( \recencodf{\Delta}^{Cont} , x: \sigma \wedge \cdots \wedge \sigma \vdash \recencodf{M\langle x_1 / x  \rangle \cdots \langle x_n / x  \rangle}[x_1 , \cdots , x_n \leftarrow x] : \tau \)}
        %     \LeftLabel{ (abs-sh)}
        %     \UnaryInfC{\( \recencodf{\Delta}^{Cont} \vdash \lambda x . (\recencodf{M\langle x_1 / x  \rangle \cdots \langle x_n / x  \rangle} [x_1 , \cdots , x_n \leftarrow x]) : \sigma\wedge  \cdots \wedge \sigma \rightarrow \tau \)}
        % \end{prooftree}
        
        \item Rule~$\redlab{T:app}$: Then $\expr{M} = M\ B$ and
         \begin{prooftree}
    \AxiomC{\( \Gamma \vdash M : \pi \rightarrow \tau \)}
    \AxiomC{\( \Delta \vdash B : \pi \)}
        \LeftLabel{$\redlab{T:app}$}
    \BinaryInfC{\( \Gamma \contexcat \Delta \vdash M\ B : \tau\)}
            \end{prooftree}
            By IH we have both 
            $ \strcore{\Gamma} \vdash \recencodf{M} : \pi \rightarrow \tau$
            and 
            $ \strcore{\Delta} \vdash \recencodf{B} : \pi$,
            and the thesis follows easily by Rule~$\redlab{TS:app}$ in \lamrsharfail:
              \begin{prooftree}
                \AxiomC{\( \strcore{\Gamma} \vdash \recencodf{M} : \pi \rightarrow \tau \)}
                \AxiomC{\( \strcore{\Delta} \vdash \recencodf{B} : \pi \)}
                \LeftLabel{$\redlab{TS:app}$}
            \BinaryInfC{\( \strcore{\Gamma}, \strcore{\Delta} \vdash \recencodf{M}\ \recencodf{B} : \tau\)}
        \end{prooftree}
        
        % \daniele{I think we need a result to conclude $\recencodopen{M}\recencodopen{P}=\recencodopen{M\ P}$, we know this holds for closed terms, but I am not sure about the open terms.}
            
        % $\recencodf{\Gamma, \Delta \vdash M\ P : \tau }^{Judge} = \recencodf{\Gamma}^{Cont}, \recencodf{\Delta}^{Cont} \vdash \recencodf{M}\ \recencodf{P} : \tau $
        
        % \begin{prooftree}
        %     $\Bigg\{\!\!\!\Bigg\{$
        %     \AxiomC{\( \Gamma \vdash M : \pi \rightarrow \tau \)}
        %     \AxiomC{\( \Delta \vdash P : \pi \)}
        %         \LeftLabel{(app)}
        %     \BinaryInfC{\( \Gamma, \Delta \vdash M\ P : \tau\)}
        %     \DisplayProof
        %     $\Bigg\}\!\!\!\Bigg\}^{Judge}$
        %     \hspace{0.5cm}
        %     =
        %     \hspace{0.5cm}
        %         \AxiomC{\( \recencodf{\Gamma}^{Cont} \vdash \recencodf{M} : \pi \rightarrow \tau \)}
        %         \AxiomC{\( \recencodf{\Delta}^{Cont} \vdash \recencodf{P} : \pi \)}
        %         \LeftLabel{(app)}
        %     \BinaryInfC{\( \recencodf{\Gamma}^{Cont}, \recencodf{\Delta}^{Cont} \vdash \recencodf{M}\ \recencodf{P} : \tau\)}
        % \end{prooftree}
        
        \item Rule~$\redlab{T:ex \dash sub}$: Then $\expr{M} = M \esubst{ B }{ x }$
        and the proof is split in two cases, depending on the shape of $B$:
        
        \begin{enumerate}
            \item $B=\oneb$. In this case, $\expr{M}=M\esubst{\oneb}{x}$ and we obtain the following type derivation:
            
            \begin{prooftree}
            \LeftLabel{$\oneb$}
        \AxiomC{$\vdash \oneb :\omega$}
        \AxiomC{$ \Gamma\vdash M:\tau$}
        \LeftLabel{$\redlab{T:weak}$}
        \UnaryInfC{$ \Gamma, x:\omega \vdash M:\tau $}
        \LeftLabel{$\redlab{T:ex \dash sub}$}
        \BinaryInfC{$\Gamma  \vdash M\esubst{\oneb}{x}:\tau$}
        \end{prooftree}
        
        By IH we have both $\vdash \oneb : \omega$ and $\strcore{\Gamma} \vdash \recencodf{M}:\tau$. 
        By the encoding of terms (Figure~\ref{fig:auxencfail}), $\recencodf{M\esubst{\oneb}{x}}=\recencodf{M}[\leftarrow x]\esubst{\oneb}{x}$, and the result holds
               by the following type derivation:
        \begin{prooftree}
            \AxiomC{$\vdash \oneb :\omega$}
            \AxiomC{$\strcore{\Gamma} \vdash\recencodf{M}:\tau$}
            \LeftLabel{$\redlab{TS:weak}$}
            \UnaryInfC{ $ \strcore{\Gamma}, x:\omega \vdash\recencodf{M}[\leftarrow x]:\tau $}
            \LeftLabel{$\redlab{TS:ex \dash sub}$}
            \BinaryInfC{$ \strcore{\Gamma}  \vdash \recencodf{M}[\leftarrow x]\esubst{\oneb}{x}:\tau$}
        \end{prooftree}

            \item $B= \bag{N_1, \ldots ,N_n}$, $n\geq 1$.
            Suppose w.l.o.g. that $n=2$, then $B= \bag{N_1,N_2}$ and
        \begin{prooftree}
        \AxiomC{$\Delta_1\vdash N_1:\sigma$}
        \AxiomC{$\Delta_2\vdash N_2:\sigma$}
        \LeftLabel{$\redlab{T:bag}$}
        \BinaryInfC{$\Delta_1 \contexcat \Delta_2\vdash \bag{N_1}\cdot \bag{N_2}:\sigma^2$}
        \AxiomC{$\Gamma, {x}:\sigma^2 \vdash M:\tau $}
        \LeftLabel{$\redlab{T:ex \dash sub}$}
        \BinaryInfC{$\Gamma \contexcat \Delta_1 \contexcat \Delta_2\vdash M\esubst{B}{x}:\tau$}
        \end{prooftree}

        By IH we have 
        $ \strcore{\Delta_1} \vdash \recencodf{N_1}:\sigma$
        and
        $ \strcore{\Delta_2} \vdash \recencodf{N_2}:\sigma$.
        %By Definition~\ref{d:enclamcontfail}, 
        We can expand 
        $\Gamma, {x}:\sigma^2 \vdash M:\tau $
        into
       $\Gamma, x:\sigma \wedge \sigma \vdash M:\tau$.
       By Lemma \ref{lem:preser_linsub} and the IH on this last sequent we obtain
       $$ \strcore{\Gamma}, y_1:\sigma, y_2:\sigma \vdash \recencodf{M\linsub{y_1,y_2}{x}}:\tau$$
       where $\#(x,M)=2$ and $y_1,y_2$ are fresh variables with the same type as $x$.
       
    Now, by the encoding of terms (Figure~\ref{fig:auxencfail}), we have
    
        \[
        \begin{aligned}
             \recencodf{M\esubst{\bag{N_1,N_2}}{x}}
                = & ~\recencodf{M\linsub{y_1,y_2}{x}}\linexsub{\recencodf{N_1}/y_1
                }\linexsub{\recencodf{N_2}/y_2} + \\
                & \recencodf{M\linsub{y_1,y_2}{x}}\linexsub{\recencodf{N_1}/y_2
                }\linexsub{\recencodf{N_2}/y_1}\\
                 = & ~\expr{M}'
        \end{aligned}
        \]
        
        We give typing derivations in $\lamrsharfail$ for each summand.     First, let $\Pi_1$ be the following derivation:
       
         \begin{prooftree}
            \AxiomC{$ \strcore{\Delta_2}\vdash \recencodf{N_2}:\sigma$}
            \AxiomC{$ \strcore{\Delta_1}\vdash \recencodf{N_1}:\sigma$}
              %\AxiomC{$\recencodf{\Gamma}, x:\sigma, x:\sigma \vdash \recencodf{M}:\tau$}
              %\LeftLabel{(*)}
              \AxiomC{$ \strcore{\Gamma}, y_1:\sigma, y_2:\sigma \vdash \recencodf{M\linsub{y_1,y_2}{x}}:\tau$}     
              \BinaryInfC{$ \strcore{\Gamma}, y_2:\sigma, \strcore{\Delta_1}\vdash \recencodf{M\linsub{y_1,y_2}{x}}\linexsub{\recencodf{N_1}/y_1}:\tau$}
              \BinaryInfC{$ \strcore{\Gamma},  \strcore{\Delta_1},\strcore{\Delta_2}\vdash \recencodf{M\linsub{y_1,y_2}{x}}\linexsub{\recencodf{N_1}/y_1}\linexsub{\recencodf{N_2}/y_2}:\tau$}
         \end{prooftree}

          Similarly, we can obtain a derivation $\Pi_2$ for: 
          \[ \strcore{\Gamma},  \strcore{\Delta_1}, \strcore{\Delta_2} \vdash \recencodf{M\linsub{y_1,y_2}{x}}\linexsub{\recencodf{N_1}/y_2}\linexsub{\recencodf{N_2}/y_1}:\tau\]  
          
          From $\Pi_1$, $\Pi_2$, and Rule~$\redlab{TS:sum}$, the thesis follows:

          \begin{prooftree}
        \AxiomC{$\Pi_1 \qquad \Pi_2$}
         \LeftLabel{$\redlab{TS:sum}$}
          \UnaryInfC{$ \strcore{\Gamma},  \strcore{\Delta_1}, \strcore{\Delta_2} \vdash \expr{M}' :\tau$}
          \end{prooftree}
       
          \end{enumerate}
        \item Rule~$\redlab{T:weak}$: Then  $\expr{M}=M$ and
        \begin{prooftree}
        \AxiomC{$\Gamma\vdash M:\sigma \qquad x\notin \dom{\Gamma}$}
        \LeftLabel{$\redlab{T:weak}$}
        \UnaryInfC{$\Gamma, x:\omega \vdash M:\sigma $}
        \end{prooftree}
        Because $\redlab{TS:weak}$ is a silent typing rule in \lamrfail, we have that $ x \not \in \lfv{M}$ and so this case does not apply.

        % then $\recencodf{\Gamma, x:\omega \vdash \expr{M}: \sigma }^{Judge} = \recencodf{\Gamma}^{Cont} , x: \omega \vdash \recencodf{\expr{M}} : \tau $
        
        % \begin{prooftree}
        %     $\Bigg\{\!\!\!\Bigg\{$
        %     \AxiomC{$ \Gamma \vdash M: \sigma$}
        %     \LeftLabel{(weak)}
        %     \UnaryInfC{$ \Gamma, x:\omega \vdash M: \sigma $}
        %     \DisplayProof
        %     $\Bigg\}\!\!\!\Bigg\}^{Judge}$
        %     \hspace{0.5cm}
        %     =
        %     \hspace{0.5cm}
        %     \AxiomC{\( \recencodf{\Gamma}^{Cont}  \vdash \recencodf{M} : \tau\)}
        %     \LeftLabel{(weak)}
        %     \UnaryInfC{\( \recencodf{\Gamma}^{Cont} , x: \omega \vdash \recencodf{M}[\leftarrow x] : \tau \)}
        % \end{prooftree}

        \item Rule~$\redlab{T:sum}$:
        
        This case follows easily by IH. \qedhere
    \end{enumerate}
    %\qed

\end{proof}

%%%%%%%%%%%%%%%%%%%%%%%%%%%%%%%%%%%%
\subsubsection{Properties}
\label{app:encodingprop}

We divide the proof of well-formedness preservation: we first prove it for $\recencodf{-}$, then we extend it to $\recencodopenf{-}$.

\preservencintolamrfail*

\begin{proof}
By mutual induction on the typing derivations for $B$ and $\expr{M}$ , with an analysis of the last rule (from \figref{fig:app_wf_rules}) applied. We proceed with the following nine cases:
\begin{enumerate}
\item This case includes two subcases:
    \begin{enumerate}
    \item Rule~\redlab{F:wf \dash bag}. 
    
    Then by inversion of the typing derivation,
    \begin{prooftree}
    \AxiomC{$\Gamma\vdash B:\sigma$}     
    \LeftLabel{$\redlab{F{:}wf-bag}$}
    \UnaryInfC{$\Gamma\wfdash  B:\sigma$}
    \end{prooftree}
    
   By Propostion \ref{prop:typeencintolamrfail} we have $\Gamma\vdash B:\sigma$ implies $ \strcore{\Gamma}\vdash \recencodf{B}:\sigma$. Notice that the encoding $\recencodf{\cdot}$  given in 
   \figref{fig:auxencfail}, is a restriction of $\recencodf{\cdot}$ to $\lamr$. Therefore, $ \strcore{\Gamma} \vdash \recencodf{B}:\sigma$, and the result follows after an application of \redlab{FS{:}wf-bag}.
    
    % \textcolor{orange}{the thesis follows trivially from type preservation in $\recencod{-}$.}
    
    % \daniele{I am still not sure, have a look in example 7.7.}
                
        \item Rule~\redlab{F:bag}. 
        
        In this case $B = \bag{M}\cdot A$, where $M$ is a term and $A$ is a bag, and we have the following derivation by inversion of the typing derivation:
        \begin{prooftree}
                    \AxiomC{\( \Gamma \wfdash M : \sigma\)}
                \AxiomC{\( \Delta \wfdash A : \sigma^k  \)}
                \LeftLabel{\redlab{F:bag}}
            \BinaryInfC{\( \Gamma \contexcat \Delta \wfdash \bag{M}\cdot A:\sigma^{k+1}  \)}
            \end{prooftree}
with $\dom{\Gamma}=\lfv{M}$ and $\dom{\Delta}=\lfv{A}$. By the IHs, we have both
\begin{itemize}
    \item $ \strcore{\Gamma} \vdash \recencodf{M} : \sigma$; and
    \item $ \strcore{\Delta} \vdash \recencodf{A} : \sigma^k$.
\end{itemize}

By applying Rule~$\redlab{FS{:}bag}$ from~\figref{fig:wfsh_rules}, for $\lamrsharfail$, we obtain the following derivation:
\begin{prooftree}
\AxiomC{\( \strcore{\Gamma}  \wfdash \recencodf{M} : \sigma\)}
\AxiomC{\( \strcore{\Delta} \wfdash \recencodf{A} : \sigma^k \)}
\LeftLabel{\redlab{FS{:}bag}}
\BinaryInfC{\( \strcore{\Gamma}, \strcore{\Delta} \wfdash \bag{\recencodf{M}}\cdot \recencodf{A}:\sigma^{k+1}\)}
        \end{prooftree}
        
        Since $\recencodf{\bag{M}\cdot A}= \bag{\recencodf{M}}\cdot \recencodf{A}, $ one has $ \strcore{\Gamma}, \strcore{\Delta} \vdash \recencodf{\bag{M}\cdot A}:\sigma^{k+1}$, and the result follows.
        \end{enumerate}
    
\item    This case is divided in seven subcases:
    
    \begin{enumerate}

        \item Rule~\redlab{F:wf \dash expr}.
        
        % \daniele{I have the same issue as before.}
        
        Then the thesis follows trivially from type preservation in $\recencodopenf{-}$ of Proposition \ref{prop:typeencintolamrfail}.
        
        \item Rule~\redlab{F:weak}.
        
        In this case, $\expr{M}=M$ and by inversion of the typing derivation we have the derivation
        
        \begin{prooftree}
                    \AxiomC{$\Gamma_1\wfdash M:\tau $}
                    \LeftLabel{\redlab{F:weak}}
                    \UnaryInfC{$\Gamma_1,x:\omega\wfdash M:\tau $}
        \end{prooftree}
        
        Because $\redlab{weak}$ is a silent well-formed rule in \lamrfail, we have that $ x \not \in \lfv{M}$ and so this case does not apply.
        
        \item Rule~\redlab{F:abs}.
        
        In this case $\expr{M} = \lambda x . M $ by inversion of the typing derivation
        and we have the derivation:
        \begin{prooftree}
            \AxiomC{\( \Gamma , {x}: \sigma^n \wfdash M : \tau \)}
            \LeftLabel{\redlab{F:abs}}
            \UnaryInfC{\( \Gamma \wfdash \lambda x . M :  \sigma^n  \rightarrow \tau \)}
        \end{prooftree}
            
        By the encoding given in~\figref{fig:auxencfail}, we have 
        $\recencodf{  \expr{M}  }  =   \lambda x . \recencodf{M\linsub{x_1,\cdots, x_n}{x} } [x_1 , \cdots , x_n \leftarrow x]$, where  $\#(x,M) = n$ and each $x_i$ is fresh and has the same type as $x$.
            From ${\Gamma}, {x}: \sigma^{n} \wfdash {M} : \tau$, we obtain after $n$ applications of Proposition~\ref{prop:linhed_encfail} and Lemma~\ref{lem:preser_linsub}: 
            \begin{equation*}\label{eq:tp1fail}
                {\Gamma}, {x_1: \sigma, \cdots,  x_n: \sigma} \wfdash {M}\linsub{x_1,\cdots, x_n}{x} : \tau
            \end{equation*}
            By IH  we have:
              \begin{equation*}\label{eq:tp2fail}
             \strcore{\Gamma}, {x_1: \sigma, \cdots,  x_n: \sigma} \wfdash \recencodf{M\linsub{x_1,\cdots, x_n}{x}} : \tau
            \end{equation*}
            
            %   Starting from \eqref{eq:tp2fail}, we then have the following type derivation for  $\recencod{\expr{M}}$, which concludes the proof for this case:
        which gives us the following derivation:
        
        \begin{prooftree}
         \LeftLabel{$(*)$}
                    \AxiomC{\( \strcore{\Gamma} , x_1: \sigma, \cdots, x_n: \sigma \wfdash \recencodf{M\linsub{x_1,\cdots, x_n}{x}} : \tau \)}
            \LeftLabel{\redlab{FS{:}share}}
            \UnaryInfC{\( \strcore{\Gamma} , x : \sigma^n  \wfdash \recencodf{M\linsub{x_1,\cdots, x_n}{x}}[x_1 , \cdots , x_n \leftarrow x] : \tau \)}
            \LeftLabel{\redlab{FS{:}abs \dash sh}}
            \UnaryInfC{\( \strcore{\Gamma} \wfdash \lambda x . (\recencodf{M\linsub{x_1,\cdots, x_n}{x}}[x_1 , \cdots , x_n \leftarrow x]) : \sigma^n \rightarrow \tau \)}
        \end{prooftree}
        and the result follows.

        \item Rule~\redlab{F:app}.

        In this case $\expr{M} = M\ B$, and by inversion of the typing derivation we have the derivation:
         \begin{prooftree}
            \AxiomC{\( \Gamma \wfdash M : \sigma^{j} \rightarrow \tau \)}
            \AxiomC{\( \Delta \wfdash B : \sigma^{k} \)}
                \LeftLabel{\redlab{F:app}}
            \BinaryInfC{\( \Gamma \contexcat \Delta \wfdash M\ B : \tau\)}
        \end{prooftree}
            By IH we have both 
            $ \strcore{\Gamma} \wfdash \recencodf{M} : \sigma^{j} \rightarrow \tau$
            and 
            $ \strcore{\Delta} \wfdash \recencodf{B} : \sigma^{k}$,
            and the result follows easily by Rule~$\redlab{FS{:}app}$ in \lamrsharfail:
              \begin{prooftree}
                            \AxiomC{\( \strcore{\Gamma} \wfdash \recencodf{M} : \sigma^{j} \rightarrow \tau \)}
                \AxiomC{\( \strcore{\Delta} \wfdash \recencodf{B} : \sigma^{k} \)}
                \LeftLabel{\redlab{FS{:}app}}
            \BinaryInfC{\( \strcore{\Gamma}, \strcore{\Delta} \wfdash \recencodf{M}\ \recencodf{B} : \tau\)}
        \end{prooftree}

        \item Rule~\redlab{F:ex \dash sub}.

        %   \[
        %     \recencodf{  M \esubst{ B }{ x }  } = 
        %     \begin{cases}
        %     \sum_{B_i \in \perm{\recencodf{ B }}}\recencodf{ M \langle x_1 , \cdots , x_k / x  \rangle } \linexsub{B_i(1)/x_1} \cdots \linexsub{B_i(k)/x_k} & \mbox{where } \#(x,M) = \size{B} = k \geq 1 \\
        %     \recencodf{M\langle x_1. \cdots , x_k / x  \rangle} [x_1. \cdots , x_k \leftarrow x] \esubst{ \recencodf{B} }{ x } & \mbox{otherwise, where } \#(x,M) = k
        %      \end{cases} \\
        % \]

        Then $\expr{M} = M \esubst{ B }{ x }$
        and the proof is split in two cases, depending on the shape of~$B$:
        
  \begin{enumerate}
  \item When $\#(x,M) = \size{B} = k \geq 1 $.
  
Then we have $B= \bag{N_1, \ldots ,N_n}$, $n\geq 1$.
Suppose w.l.o.g. that $n=2$, then $B= \bag{N_1,N_2}$ and by inversion  of the typing derivation we have the following derivation:
            
\begin{prooftree}
 \AxiomC{$\Delta_1\wfdash N_1:\sigma$}
 \AxiomC{$\Delta_2\wfdash N_2:\sigma$}
 \LeftLabel{\redlab{F:bag}}      
 \BinaryInfC{$\Delta_1 \contexcat \Delta_2\wfdash \bag{N_1}\cdot \bag{N_2}:\sigma^2$}
 \AxiomC{$\Gamma_1, {x}:\sigma^2 \wfdash M:\tau $}
 \LeftLabel{\redlab{F:ex \dash sub}}
 \BinaryInfC{$\Gamma_1 \contexcat \Delta_1 \contexcat \Delta_2 \wfdash M\esubst{B}{x}:\tau$}
 \end{prooftree}     
 where $\Gamma=\Gamma_1,\Delta_1,\Delta_2$.
            By IH we have both
            \begin{itemize}
                \item $ \strcore{\Delta_1} \vdash \recencodf{N_1}:\sigma$; and 
            \item $ \strcore{\Delta_2} \vdash \recencodf{N_2}:\sigma$; and
                \item $ \strcore{\Gamma_1} , {x}:\sigma^2 \wfdash \recencodf{M}:\tau $
            \end{itemize}

   We can expand 
    $ \strcore{\Gamma_1}, {x}:\sigma^2 \wfdash \recencodf{M}:\tau $ into $ \strcore{\Gamma_1}, x:\sigma \wedge \sigma \wfdash \recencodf{M}:\tau$, which gives $\strcore{\Gamma_1}, y_1:\sigma, y_2:\sigma\wfdash \recencodf{M}\linsub{y_1,y_2}{x}:\tau$, after two applications of Proposition~\ref{prop:linhed_encfail} along with the application of Lemma \ref{lem:preser_linsub}, with  $y_1,y_2$  fresh variables of the same type as $x$. Since the encoding $\recencodf{\cdot }$ commutes with the linear substitution $\linsub{\cdot }{\cdot }$ (Proposition~\ref{prop:linhed_encfail}), if follows that, $\strcore{\Gamma_1}, y_1:\sigma, y_2:\sigma\wfdash \recencodf{M\linsub{y_1,y_2}{x}}:\tau$.
    
    % \daniele{I am not sure if we can apply IH on the last sequent because we have new variables, directly. I am sure the result holds, but we might need to prove some syntactic property.}
    
            % By applying IH on this last sequent we obtain
            % $$\recencod{\Gamma_1}, y_1:\sigma, y_2:\sigma \vdash \recencodf{M\linsub{y_1,y_2}{x}}:\tau$$
            % where $\#(x,M)=2$ and $y_1,y_2$ are fresh variables.

           Let $\Pi_1$ be the derivation obtained after  two consecutive applications of Rule~$\redlab{FS{:} ex \dash lin \dash sub}$:    
           
\hspace*{-25pt}
\begin{minipage}{\linewidth}          
\begin{prooftree}
 \AxiomC{$\strcore{\Gamma_1}, y_1:\sigma, y_2:\sigma \wfdash \recencodf{M\linsub{y_1,y_2}{x}}:\tau$}
  \AxiomC{$\strcore{\Delta_1} \wfdash \recencodf{N_1}:\sigma$}                  
  \BinaryInfC{$ \strcore{\Gamma_1}, y_2:\sigma, \strcore{\Delta_1} \wfdash \recencodf{M\linsub{y_1,y_2}{x}}\linexsub{\recencodf{N_1}/y_1}:\tau$} 
  \AxiomC{$ \strcore{\Delta_2}\wfdash \recencodf{N_2}:\sigma$}
  \BinaryInfC{$ \strcore{\Gamma_1}, \strcore{\Delta_1}, \strcore{\Delta_2} \wfdash \recencodf{M\linsub{y_1,y_2}{x}} \linexsub{\recencodf{N_1}/y_1}\linexsub{\recencodf{N_2}/y_2}:\tau$}
\end{prooftree}
\end{minipage}

 Similarly, we can obtain a derivation $\Pi_2$ for:
 $$ \strcore{\Gamma_1},  \strcore{\Delta_1}, \strcore{\Delta_2} \wfdash \recencodf{M\linsub{y_1,y_2}{x}}\linexsub{\recencodf{N_1}/y_2}\linexsub{\recencodf{N_2}/y_1}:\tau$$

 By the encoding given in Figure~\ref{fig:auxencfail}, we have
            \[
            \begin{aligned}
                \recencodf{M\esubst{\bag{N_1,N_2}}{x}}
                 = & \recencodf{M\linsub{y_1,y_2}{x}}\linexsub{\recencodf{N_1}/y_1
                }\linexsub{\recencodf{N_2}/y_2} +\\
                & \recencodf{M\linsub{y_1,y_2}{x}}\linexsub{\recencodf{N_1}/y_2
                }\linexsub{\recencodf{N_2}/y_1}
            \end{aligned}
            \]

 Therefore, 
         
  \begin{prooftree}
 \AxiomC{$\Pi_1$}
  \AxiomC{$\Pi_2$}
  \LeftLabel{$\redlab{FS{:}sum}$}
  \BinaryInfC{$\strcore{\Gamma_1}, \strcore{\Delta_1}, \strcore{\Delta_2} \wfdash \recencodf{M\esubst{\bag{N_1,N_2}}{x}}:\tau$}
 \end{prooftree}
            
      and the result follows.

            \item  $ \#(x,M) = k \neq \size{B}$.
            
           In this case, $\size{B}=j$ for some $j\neq k$, and by inversion of the typing derivation we have the following derivation:
            
            \begin{prooftree}
                    \AxiomC{\( \Delta \wfdash B : \sigma^{j} \)}
                    \AxiomC{\( \Gamma_1 , {x}:\sigma^{k} \wfdash M : \tau \)}
                \LeftLabel{\redlab{F:ex \dash sub}}    
                \BinaryInfC{\( \Gamma_1 \contexcat \Delta \wfdash M \esubst{ B }{ x } : \tau \)}
            \end{prooftree}
 where $\Gamma=\Gamma_1 \contexcat \Delta$. By IH we have both
            
  \begin{itemize}
     \item $ \strcore{\Delta} \wfdash \recencodf{B} : \sigma^{j} $; and 
                \item  $ \strcore{\Gamma_1}, \hat{x}:\sigma^{k} \wfdash \recencodf{M}:\tau$. 
            \end{itemize}

We analyse two cases, depending on the number $k$ of occurrences of $x$ in $M$:
\begin{enumerate} 

\item $k=0$.
 
 From $\Gamma_1, {x:\omega} \wfdash M:\tau$, which we get $\Gamma_1 \wfdash M:\tau$, via Rule~$\redlab{F:weak}$. The IH gives $ \strcore{\Gamma_1} \wfdash \recencodf{M}:\tau$, which entails:
\begin{prooftree}
                \AxiomC{\( \strcore{\Gamma_1}  \wfdash \recencodf{M} : \tau\)}
                \LeftLabel{ \redlab{FS{:}weak}}
                \UnaryInfC{\( \strcore{\Gamma_1}, x: \omega \wfdash \recencodf{M}[\leftarrow x]: \tau \)}
                \AxiomC{$ \strcore{\Delta} \wfdash \recencodf{B} : \sigma^{j} $}
                \LeftLabel{ \redlab{FS{:}ex \dash sub}}
                \BinaryInfC{$ \strcore{\Gamma_1} ,\strcore{\Delta} \wfdash \recencodf{M}[\leftarrow x]\esubst{B}{x}: \tau$}
\end{prooftree}

By the encoding given in Figure~\ref{fig:auxencfail},  $\recencodf{M\esubst{B}{x}}=\recencodf{M}[\leftarrow x]\esubst{\recencodf{B}}{x}$, and the result follows.
  \item  $k> 0$.
  
   By applying Proposition~\ref{prop:linhed_encfail} in \( \strcore{\Gamma_1} , {x}:\sigma^{k} \wfdash \recencodf{ M} : \tau \), we obtain 
  \[ \strcore{\Gamma_1} , x_1:\sigma,\ldots,  x_k:\sigma \wfdash \recencodf{M}\linsub{x_1,\ldots, x_k}{x} : \tau \]
 From Proposition~\ref{prop:linhed_encfail} and Lemma \ref{lem:preser_linsub}, it follows that  $ \strcore{ \Gamma_1} , x_1:\sigma,\ldots,  x_k:\sigma \wfdash \recencodf{M\linsub{x_1,\ldots, x_k}{x}} : \tau $, which entails $x\notin \dom{\Gamma_1}$ since $k\neq 0$. First we give $\Pi$:
 
\begin{prooftree}
    \AxiomC{$ \strcore{\Gamma_1}\ , x_1:\sigma,\ldots,  x_k:\sigma \wfdash \recencodf{M\linsub{x_1,\ldots, x_k}{x}} : \tau $}
    \LeftLabel{\redlab{FS{:}share}}    
    \UnaryInfC{\( \strcore{\Gamma_1} , x:\sigma^{k} \wfdash \recencodf{M\linsub{x_1. \cdots , x_k}{x}} [x_1. \cdots , x_k \leftarrow x]  : \tau \)}
\end{prooftree}

finally we give the full derivation:

\begin{prooftree}
    \AxiomC{\( \strcore{\Delta} \wfdash \recencodf{B} : \sigma^{j} \)}
    \AxiomC{$ \Pi $}
    \LeftLabel{\redlab{FS{:}ex \dash sub}}    
    \BinaryInfC{\( \strcore{\Gamma}, \strcore{\Delta} \wfdash \recencodf{M\langle x_1. \cdots , x_k / x  \rangle} [x_1. \cdots , x_k \leftarrow x] \esubst{ \recencodf{B} }{ x } : \tau \)}
\end{prooftree}
            
% \item   Let us consider two cases, when $k \geq 1$ and when $k = 0$. We type the following cases respectively as:
%            
            %\end{prooftree}
    \end{enumerate}         
            
        % Hence we have,
        %   \begin{prooftree}
        %     \AxiomC{$\Pi_1 \qquad \Pi_2$}
        %      \LeftLabel{\redlab{FS{:}sum}}
        %       \UnaryInfC{$\recencod{\Gamma},  \recencod{\Delta_1},\recencod{\Delta_2}\wfdash \recencodf{M\linsub{y_1,y_2}{x}}\linexsub{\recencodf{N_1}/y_1}\linexsub{\recencodf{N_2}/y_2}+\recencodf{M\linsub{y_1,y_2}{x}}\linexsub{\recencodf{N_1}/y_2}\linexsub{\recencodf{N_2}/y_1}:\tau$}
        %   \end{prooftree}
       
          \end{enumerate}

        \item Rule~\redlab{F:fail}.

        The result follows trivially, because the encoding of failure in~\figref{fig:auxencfail} is such that $\recencodf{\fail^{\widetilde{x}}}=\fail^{\widetilde{x}} $.
                
        \item Rule~\redlab{F:sum}.
        
        This case follows easily by IH. \qedhere
    \end{enumerate}
    \end{enumerate}
%\qed

\end{proof}

\preservencintolamrfailtwo*

\begin{proof}
By mutual induction on the typing derivations $\Gamma\wfdash B:\sigma$ and $\Gamma\wfdash \expr{M}:\sigma$, exploiting both Proposition~\ref{prop:linhed_encfail} and Lemma~\ref{lem:preser_linsub}. The analysis for bags (Part 1) follows directly from the IHs and will be omitted.

As for Part 2, there are two main cases to consider:
\begin{enumerate}[i)]
    \item $\expr{M} = M$. 
    
Without loss of generality, assume $\lfv{M} = \{x,y\}$. Then, 
\begin{equation}\label{eq:thmpres1fail}
{x}:\sigma_1^j, {y}:\sigma_2^k \wfdash M : \tau    
\end{equation}
where  $\#(x,M)=j$ and 
$\#(y,M)=k$, for some positive integers $j$ and $k$.

After $j+k$ applications of Lemma~\ref{lem:preser_linsub}
% on  \eqref{eq:thmpres1fail},
we obtain:
\begin{equation*}\label{eq:thmpres2fail}
x_1:\sigma_1, \cdots, x_j:\sigma_1, y_1:\sigma_2, \cdots, y_k:\sigma_2 \wfdash M\linsub{\widetilde{x}}{x}\linsub{\widetilde{y}}{y} : \tau    
\end{equation*}
 where  $\widetilde{x}=x_{1},\ldots, x_{j}$
   and $\widetilde{y}=y_{1},\ldots, y_{k}$. 
From Proposition~\ref{prop:linhed_encfail} and Lemma~\ref{lem:preser_linsub}
%(1) 
%on \eqref{eq:thmpres2fail}
one has
\begin{equation*}\label{eq:thmpres3fail}
{x_1:\sigma_1, \ldots, x_j:\sigma_1, y_1:\sigma_2, \ldots, y_k:\sigma_2} \wfdash \recencodf{M\linsub{\widetilde{x}}{x}\linsub{\widetilde{y}}{y}} : \tau  
\end{equation*}
Since ${x_1:\sigma_1, \ldots, x_j:\sigma_1, y_1:\sigma_2, \ldots, y_k:\sigma_2}= x_1:\sigma_1, \ldots, x_j:\sigma_1, y_1:\sigma_2, \ldots, y_k:\sigma_2$, 
% Recall that the encoding $\recencodf{-}$ on types/contexts is the identity.
% From \eqref{eq:thmpres3fail} 
we have the following derivation:
\begin{prooftree}
    \AxiomC{${x_1:\sigma_1, \cdots, x_j:\sigma_1, y_1:\sigma_2, \cdots, y_k:\sigma_2} \wfdash \recencodf{M\linsub{\widetilde{x}}{x}\linsub{\widetilde{y}}{y}} : \tau$}
    \LeftLabel{$\redlab{FS{:}share}$}
    \UnaryInfC{$x:\sigma_1^j, y_1:\sigma_2, \cdots, y_k:\sigma_2 \wfdash \recencodf{M\linsub{\widetilde{x}}{x}\linsub{\widetilde{y}}{y}}[\widetilde{x}\leftarrow x] : \tau$}
\LeftLabel{$\redlab{FS{:}share}$}
    \UnaryInfC{$x:\sigma_1^j, y:\sigma_2^k \wfdash \recencodf{M\linsub{\widetilde{x}}{x}\linsub{\widetilde{y}}{y}}[\widetilde{x}\leftarrow x][\widetilde{y}\leftarrow y] : \tau$}
    \end{prooftree}
    
    By expanding~\defref{def:enctolamrsharfail}, we have 
$$
 \recencodopenf{M} = 
\recencodf{M\linsub{\widetilde{x}}{x}\linsub{\widetilde{y}}{y}}[\widetilde{x}\leftarrow x][\widetilde{y}\leftarrow y], 
$$

    which completes the proof for this case. 
    
    \item $\expr{M} = M_1 + \cdots + M_n$.
    
    This case proceeds easily by IH, using Rule~$\redlab{FS{:}sum}$. \qedhere
    \end{enumerate}
%\qed
\end{proof}

%%%%%%%%%%%%%%%%%%%%%%%%%%%%%%%%%%%%%%%
\subsection{Completeness and Soundness}
\label{app:compandsoundone}

\appcompletenessone*

\begin{proof}
By induction on the rule from~\figref{fig:reductions_lamrfail} applied to infer $\expr{N}\red \expr{M}$, distinguishing  three cases. Below $\widetilde{[x_{1k}\leftarrow x_{1k}]}$ abbreviates $[\widetilde{x_1}\leftarrow x_1]\ldots [\widetilde{x_k}\leftarrow x_k]$:

\begin{myEnumerate}

    \item The rule applied is $\redlab{R}=\redlab{R:Beta}$. 
    
    In this case, 
        $\expr{N}= (\lambda x. M) B$, the reduction is 
    \begin{prooftree}
        \AxiomC{}
        \LeftLabel{\redlab{R:Beta}}
        \UnaryInfC{\((\lambda x. M) B \red M\ \esubst{B}{x}\)}
    \end{prooftree}

          and $\expr{M}= M\esubst{B}{x}$. Below we assume $\lfv{\expr{N}}=\{x_1,\ldots, x_k\}$ and $\widetilde{x_i}=x_{i_1},\ldots, x_{i_{j_i}}$, where $j_i= \#(x_i, N)$, for $1\leq i\leq k$.
        On the one hand, we have:
        \begin{equation}\label{eq:beta1fail}
            \begin{aligned}
            \recencodopenf{\expr{N}}&= \recencodopenf{(\lambda x. M)B}\\
            &=  \recencodf{((\lambda x. M)B) \langle{\widetilde{x_1}/x_1}\rangle\cdots \langle{\widetilde{x_k}/x_k}\rangle}\widetilde{[x_{1k}\leftarrow x_{1k}]}
            \\
            &=  \recencodf{(\lambda x. M^{'})B'}\widetilde{[x_{1k}\leftarrow x_{1k}]}\\
            &=  (\recencodf{\lambda x. M^{'}}\recencodf{B'})\widetilde{[x_{1k}\leftarrow x_{1k}]} \\
            &=  ((\lambda x.\recencodf{ M^{'}\langle{\widetilde{y}/x}\rangle}[\widetilde{y}\leftarrow x])\recencodf{B'})\widetilde{[x_{1k}\leftarrow x_{1k}]} \\
            &\red_{\redlab{RS{:}Beta}} (\recencodf{ M^{'} \langle{\widetilde{y}/x} \rangle} [\widetilde{y} \leftarrow x] \esubst{\recencodf{B'}}{x}) \widetilde{[x_{1k}\leftarrow x_{1k}]}=\expr{L}
            \end{aligned}
        \end{equation}
      
        \revo{A24}{where we define $M'$ and $B'$ to be $M$ and $B$ after the substitutions of $\langle \widetilde{x_1}/x_1\rangle\cdots$\linebreak[4]$\langle \widetilde{x_k}/x_k\rangle$.} On the other hand, we have:
        \begin{equation}\label{eq:beta2fail}
            \begin{aligned}
               \recencodopenf{\expr{M}}&=\recencodopenf{M\esubst{B}{x}}\\
               &=\recencodf{M\esubst{B}{x}\langle{\widetilde{x_1}/x_1}\rangle\cdots \langle{\widetilde{x_k}/x_k}\rangle} \widetilde{[x_{1k}\leftarrow x_{1k}]}\\
               &=\recencodf{M^{'}\esubst{B'}{x}} \widetilde{[x_{1k}\leftarrow x_{1k}]}
            \end{aligned}
        \end{equation}

        We need to analyse two sub-cases: either $\#(x,M') = \size{B} = k \geq 1 $ or $\#(x,M') = k$ and our first sub-case is not met.
        \begin{enumerate}[i)]
            \item If $\#(x,M') = \size{B} = k \geq 1$ then we can reduce $\expr{L}$ using Rule~\redlab{RS:Ex-sub}:
            \begin{equation*}
                \begin{aligned}
            \expr{L}\red \sum_{B_i\in \perm{\recencodf{B}}}\recencodf{M^{'}\linsub{\widetilde{y}}{x}}\linexsub{B_i(1)/y_1}\cdots \linexsub{B_i(n)/y_n} \widetilde{[x_{1k}\leftarrow x_{1k}]} =\recencodopenf{\expr{M}}
                \end{aligned}
            \end{equation*}
            
            From \eqref{eq:beta1fail} and \eqref{eq:beta2fail} and $\widetilde{y}=y_1\ldots y_n$, one has the desired result.

            \item Otherwise,  $\#(x,M) = k$ (either $k=0$ or $k\neq \size{B}$).
            
            % From reduction \eqref{eq:beta1fail}:

            % \[
            % \expr{L} = (\recencodf{ M^{''} \langle{\widetilde{y}/x} \rangle} [\widetilde{y} \leftarrow x] \esubst{\recencodf{B'}}{x}) [\widetilde{x_1}\leftarrow x_1]\ldots [\widetilde{x_k}\leftarrow x_k]
            % \]

            Expanding  the encoding in \eqref{eq:beta2fail} :
            \begin{align*}
                \recencodopenf{M}&= \recencodf{M^{'}\esubst{B'}{x}} \widetilde{[x_{1k}\leftarrow x_{1k}]}\\
                & = (\recencodf{ M^{'} \langle{\widetilde{y}/x} \rangle} [\widetilde{y} \leftarrow x] \esubst{\recencodf{B'}}{x}) \widetilde{[x_{1k}\leftarrow x_{1k}]}
            \end{align*}
           Therefore  $\recencodopenf{M} =\expr{L}$ and $\recencodopenf{\expr{N}}\red \recencodopenf{\expr{M}}$.
        
        \end{enumerate}
        
    %   The reduction via $ \redlab{R:Beta} $ could occur inside a context (cf. Rules $\redlab{R:TCont}$ and $\redlab{R:ECont}$). 
    %   We consider only the case when the contextual rule used is $\redlab{R:TCont}$. We have $\expr{N} = C[N]$. By induction, $N \red M$ and $C[N] \red C[M] $. It is clear that both when $C[-] = [-]B$ and $([\cdot])\esubst{B}{x} $ the thesis holds.

    \item The rule applied is $\redlab{R}=\redlab{R:Fetch}$.
    % then necessarily we have $\redlab{R}$ to be one of four reductions: $\redlab{R:Fetch}$ ,  $\redlab{R:Fail}$ , $\redlab{R:Consume_1}$ and  $\redlab{R:Consume_2}$. 
    
Then $\expr{N}=M\esubst{B}{x}$ and  the reduction is 
            \begin{prooftree}
     \AxiomC{$\headf{M} = x \quad B = \bag{N_1, \dots ,N_n}, \ n\geq 1 \quad  \#(x,M) = n $}
      \LeftLabel{\redlab{R:Fetch}}                        \UnaryInfC{$M\esubst{B}{x} \red \sum_{i=1}^n M\headlin{N_i/x}\esubst{B\linsetminus N_i}{x} $}
            \end{prooftree}
            
            with $\expr{M}= \sum_{i=1}^n M\headlin{N_i/x}\esubst{B\linsetminus N_i}{x}$.

            Below we assume $\lfv{\expr{N}}=\lfv{M\esubst{\bag{N1}}{x}}=\{x_1,\ldots, x_k\}$. We distinguish two cases:

            \begin{enumerate}
                    \item $n = 1$.
                    
                   Then $B = \bag{N_1}$ and  $\expr{N}=M\esubst{\bag{N_1}}{x}\red  M\headlin{N_1/x} \esubst{\oneb}{x} = \expr{M} $.

                On the one hand, we have:
\begin{equation*}\label{eq:fetch3fail}
  \begin{aligned}
  \recencodopenf{\expr{N}}&= \recencodopenf{M\esubst{\bag{N_1}}{x}}\\                      &=    \recencodf{(M\esubst{\bag{N_1}}{x})\langle \widetilde{x_1}/x_1\rangle\cdots \langle \widetilde{x_k}/x_k\rangle }\widetilde{[x_{1k}\leftarrow x_{1k}]}\\
  &=  \recencodf{M'\esubst{\bag{N_1'}}{x}}\widetilde{[x_{1k}\leftarrow x_{1k}]}\\
& = \recencodf{M'\langle y_1 / x \rangle}\linexsub{\recencodf{N_1'}/y_1} \widetilde{[x_{1k}\leftarrow x_{1k}]}, \text{ notice that }\headf{M'}=y_1\\
& = \recencodf{M^{''}}\linexsub{\recencodf{N_1'}/y_1} \widetilde{[x_{1k}\leftarrow x_{1k}]}\\
    &\red_{\redlab{RS{:}Lin\dash Fetch}} \recencodf{M''}\headlin{\recencodf{N_1'}/y_1}\widetilde{[x_{1k}\leftarrow x_{1k}]}\\
  \end{aligned}
\end{equation*}
\revo{A24}{where we define $M'$ and $N_1'$ to be $M$ and $N_1$ after the substitutions of $\langle \widetilde{x_1}/x_1\rangle\cdots$\linebreak[4]$\langle \widetilde{x_k}/x_k\rangle$; similarly, we define $M''$ to be $M'$after the substitution of $y_1$ for $x$}. On the other hand,                    
\begin{equation*}\label{eq:fetch4fail}
 \begin{aligned}
 \recencodopenf{\expr{M}}&= \recencodopenf{M\headlin{N_1/x} \esubst{\oneb}{x}}\\
  &= \recencodf{M\headlin{N_1/x} \esubst{\oneb}{x} \langle \widetilde{x_1}/x_1\rangle\cdots \langle \widetilde{x_k}/x_k\rangle }\widetilde{[x_{1k}\leftarrow x_{1k}]}\\
  &= \recencodf{M'\headlin{N_1'/x} \esubst{\oneb}{x}} \widetilde{[x_{1k}\leftarrow x_{1k}]}\\
  &= \recencodf{M'\headlin{N_1'/x}} [\leftarrow x]  \esubst{\oneb}{x}\widetilde{[x_{1k}\leftarrow x_{1k}]}\\
                     \end{aligned}
                    \end{equation*}
            
            By the congruence defined in~\figref{fig:rPrecongruencefail} for $\lamrfail$, one has  $M\esubst{\oneb}{x}\pequiv M$.
            
            Therefore, $\expr{M} = M\headlin{N_1/x}  \esubst{\oneb}{x} \pequiv M\headlin{N_1/x} =  \expr{M}'$. Expanding $\recencodopenf{\expr{M}'}$ we have:
 \begin{equation*}\label{eq:fetch5fail}                   
 \begin{aligned}
 \recencodopenf{\expr{M}'}&=\recencodopenf{M\headlin{N_1/x}} \\
 &=    \recencodf{M\headlin{N_1/x} \langle \widetilde{x_1}/x_1\rangle\cdots \langle \widetilde{x_j}/x_j\rangle } \widetilde{[x_{1k}\leftarrow x_{1k}]}\\
 &=  \recencodf{M' \headlin{N_1'/x}  } \widetilde{[x_{1k}\leftarrow x_{1k}]}\\
 &=  \recencodf{M' }\headlin{\recencodf{N_1'}/x}   \widetilde{[x_{1k}\leftarrow x_{1k}]}\\
 \end{aligned}
  \end{equation*}
                Hence, $\recencodopenf{\expr{N}}\red\recencodopenf{\expr{M'}}$ and the result follows.
        
                    \item   $n> 1$ 
                    
                     To simplify the proof, we take $n=2$ (the analysis when $n>2$ is similar). Then $B=\bag{N_1,N_2}$ and the reduction is
                     
                $$\expr{N}=M\esubst{B}{x}\red  M\headlin{N_1/x} \esubst{\bag{N_2}}{x}+M\headlin{N_2/x} \esubst{\bag{N_1}}{x}=\expr{M}$$
                    Notice that $\#(x,M)=2$, we take   $y_1,y_2$ fresh variables.
            On the one hand, we have:  
  \begin{equation}\label{eq:fetch1fail}
 \begin{aligned}
\recencodopenf{\expr{N}}&=\recencodopenf{M\esubst{B}{x}}
=    \recencodf{M\esubst{B}{x}\langle \widetilde{x_1}/x_1\rangle\cdots \langle \widetilde{x_k}/x_k\rangle } \widetilde{[x_{1k}\leftarrow x_{1k}]}\\
 &=  \recencodf{M'\esubst{B'}{x}}\widetilde{[x_{1k}\leftarrow x_{1k}]}\\
 & = ( \recencodf{M'\langle y_1, y_2 / x \rangle}\linexsub{\recencodf{N_1'}/y_1} \linexsub{\recencodf{N_2'}/y_2}\\
 &  \qquad +
 \recencodf{M'\langle y_1, y_2 / x \rangle}\linexsub{\recencodf{N_2'}/y_1} \linexsub{\recencodf{N_1'}/y_2})\widetilde{[x_{1k}\leftarrow x_{1k}]}\\
& =( \recencodf{M^{''}}\linexsub{\recencodf{N_1'}/y_1} \linexsub{\recencodf{N_2'}/y_2}\\
 & \qquad + \recencodf{M^{''}}\linexsub{\recencodf{N_2'}/y_1} \linexsub{\recencodf{N_1'}/y_2})\widetilde{[x_{1k}\leftarrow x_{1k}]}\\
&\red^2_{\redlab{RS{:}Lin\dash Fetch}}
     (\recencodf{M''}\headlin{\recencodf{N_1'}/y_1}\linexsub{\recencodf{N_2'}/y_2} \\
     & \qquad + \recencodf{M''}\headlin{\recencodf{N_2'}/y_1}\linexsub{\recencodf{N_1'}/y_2})\widetilde{[x_{1k}\leftarrow x_{1k}]}\\
& = \expr{L}.
\end{aligned}
\end{equation}
                    
\revo{A24}{where we define $M'$ and $B'$ to be $M$ and $B$ after the substitutions of $\langle \widetilde{x_1}/x_1\rangle\cdots$\linebreak[4]$\langle \widetilde{x_k}/x_k\rangle$ and $N_1' , N_2'$ are the the elements of the bag $B'$. 
Similarly, we define $M''$ to be $M'$ after the substitution $\langle y_1, y_2 / x \rangle$}.
On the other hand, we have:
  \begin{equation}\label{eq:fetch2fail}
  \begin{aligned}
  \recencodopenf{\expr{M}}&=\recencodopenf{M\headlin{N_1/x} \esubst{\bag{N_2}}{x}+M\headlin{N_2/x} \esubst{\bag{N_1}}{x}}\\
  &=\recencodopenf{M\headlin{N_1/x} \esubst{\bag{N_2}}{x}}+\recencodopenf{M\headlin{N_2/x} \esubst{\bag{N_1}}{x}}\\
  &=\recencodf{M\headlin{N_1/x} \esubst{N_2}{x}}\langle \widetilde{x_1}/x_1\rangle\cdots \langle \widetilde{x_k}/x_k\rangle\widetilde{[x_{1k}\leftarrow x_{1k}]}\\
 & \qquad \qquad + \recencodf{M\headlin{N_2/x} \esubst{N_1}{x}\langle \widetilde{x_1}/x_1\rangle\cdots \langle \widetilde{x_k}/x_k\rangle }\widetilde{[x_{1k}\leftarrow x_{1k}]}\\
 &=\recencodf{M'\headlin{N_1'/x} \esubst{N_2'}{x}}\widetilde{[x_{1k}\leftarrow x_{1k}]}\\
 & \qquad  \qquad + \recencodf{M'\headlin{N_2'/x} \esubst{N_1'}{x} }\widetilde{[x_{1k}\leftarrow x_{1k}]}\\
   & =\recencodf{M'\headlin{N_1'/x}} \linexsub{\recencodf{N_2'}/y_2}\widetilde{[x_{1k}\leftarrow x_{1k}]} \\
   & \qquad \qquad + \recencodf{M'\headlin{N_2'/x}} \linexsub{\recencodf{N_1'}/y_2} \widetilde{[x_{1k}\leftarrow x_{1k}]}
   \end{aligned}        
 \end{equation}
                    The reductions in \eqref{eq:fetch1fail} and \eqref{eq:fetch2fail} lead to identical expressions, up to renaming of shared variables, which are taken to be fresh by definition. In both cases, we have taken the same  fresh variables.
                    
            \end{enumerate}
         
             \item The rule applied is $ \redlab{R}\neq \redlab{R:Beta}$ and $ \redlab{R}\neq \redlab{R:Fetch}$. There are two possible cases. Below $\widetilde{[x_{1n}\leftarrow x_{1n}]}$  abbreviates $[\widetilde{x_1}\leftarrow x_1]\cdots [\widetilde{x_n}\leftarrow x_n]$:
             \begin{myEnumerate}
                 
            \item  $\redlab{R}=\redlab{R:Fail}$
            
            Then $\expr{N}=M\esubst{B}{x}$ and the reduction is 
            
            \begin{prooftree}
            \AxiomC{$\#(x,M)\neq \size{B}  \qquad  \widetilde{y} = (\mfv{M}\setminus x)\uplus \mfv{B} $}
            \LeftLabel{\redlab{R:Fail}}
            \UnaryInfC{$M\ \esubst{ B}{x } \red \sum_{\perm{B}} \fail^{\widetilde{y}}$}
            \end{prooftree}
            where $\expr{M}= \sum_{\perm{B}} \fail^{\widetilde{y}}$. Below assume $\lfv{\expr{N}}=\{x_1,\ldots, x_n\}$.

                On the one hand, we have:
            \begin{equation*}\label{eq:fail1fail}
                \begin{aligned}
                \recencodopenf{\expr{N}} &= \recencodopenf{M\esubst{B}{x}}
                = \recencodf{M\esubst{B}{x}\langle \widetilde{x_1}/x_1\rangle\cdots \langle \widetilde{x_n}/x_n\rangle } \widetilde{[x_{1n}\leftarrow x_{1n}]} \\
                 &= \recencodf{M'\esubst{B'}{x} } \widetilde{[x_{1n}\leftarrow x_{1n}]} \\
                &=  \recencodf{M'\langle y_1, \cdots , y_k / x  \rangle} [y_1, \cdots , y_k \leftarrow x] \esubst{ \recencodf{B'} }{ x } \widetilde{[x_{1n}\leftarrow x_{1n}]}\\
                & \red_{\redlab{RS{:}Fail}} \sum_{\perm{B}} \fail^{\widetilde{y}} \  \widetilde{[x_{1n}\leftarrow x_{1n}]}=\expr{L}\\
                \end{aligned}
            \end{equation*}
            
            \revo{A24}{where we define $M'$ and $B'$ to be $M$ and $B$ after the substitutions of $\langle \widetilde{x_1}/x_1\rangle\cdots$\linebreak[4]$\langle \widetilde{x_k}/x_k\rangle$}. On the other hand, we have:
            \begin{equation*}\label{eq:fail2fail}
                \begin{aligned}
                \recencodopenf{\expr{M}} &= \recencodopenf{\sum_{\perm{B}} \fail^{\widetilde{y}}}= \sum_{\perm{B}} \recencodopenf{\fail^{\widetilde{y}}}\\
                &= \sum_{\perm{B}} \recencodf{\fail^{\widetilde{y}}} \widetilde{[x_{1n}\leftarrow x_{1n}]} 
                = \sum_{\perm{B}} \fail^{\widetilde{y}} \widetilde{[x_{1n}\leftarrow x_{1n}]} = \expr{L}\\
                \end{aligned}
            \end{equation*}
            
            Therefore, $\recencodopenf{\expr{N}} \red \recencodopenf{\expr{M}} $ and the result follows.

            \item  $\redlab{R}= \redlab{R:Cons_1}$.
            
            Then $\expr{N}= \fail^{\widetilde{y}}\ B$ and the reduction is

            \begin{prooftree}
            \AxiomC{$\size{B}=k \qquad  \widetilde{z} = \mfv{B} $}        
            \LeftLabel{\redlab{R:Cons_1}}
            \UnaryInfC{$ \fail^{\widetilde{y}}\ B  \red \sum_{\perm{B}} \fail^{\widetilde{y} \uplus \widetilde{z}} $}
            \end{prooftree}
            
            %  $B = \bag{N_1}\cdot \dots \cdot \bag{N_k}, \quad k \geq 0 $ and $ \widetilde{z} = \lfv{B} $.

           and $\expr{M}'= \sum_{\perm{B}} \fail^{\widetilde{y} \uplus \widetilde{z}}$. Below we assume $\lfv{\expr{N}}=\{x_1,\ldots, x_n\}$.
           
                On the one hand, we have:
            \begin{equation*}\label{eq:consume1fail}
                \begin{aligned}
                \recencodopenf{N} &= \recencodopenf{\fail^{\widetilde{y}}\ B}
                = \recencodf{ \fail^{\widetilde{y}}\ B \langle \widetilde{x_1}/x_1\rangle\cdots \langle \widetilde{x_n}/x_n\rangle } \widetilde{[x_{1n}\leftarrow x_{1n}]}\\
                 &= \recencodf{ \fail^{\widetilde{y'}}\ B'  } \widetilde{[x_{1n}\leftarrow x_{1n}]}= \recencodf{ \fail^{\widetilde{y'}}} \ \recencodf{ B' } \widetilde{[x_{1n}\leftarrow x_{1n}]}\\
                &= \fail^{\widetilde{y'}} \ \recencodf{ B' } \widetilde{[x_{1n}\leftarrow x_{1n}]}\\
                & \red_{\redlab{RS{:}Cons_1}} \sum_{\perm{B}} \fail^{\widetilde{y'} \cup \widetilde{z'}}  \widetilde{[x_{1n}\leftarrow x_{1n}]}=\expr{L}\\
                \end{aligned}
            \end{equation*}
            
            \revo{A24}{where we define $B'$ to be $B$ after the substitutions of $\langle \widetilde{x_1}/x_1\rangle\cdots \langle \widetilde{x_k}/x_k\rangle$. Similarly, $\widetilde{y'}$ and $\widetilde{z'}$ are $\widetilde{y}$ and $\widetilde{z}$ after the substitution $\langle \widetilde{x_1}/x_1\rangle\cdots \langle \widetilde{x_k}/x_k\rangle$}. 
            On the other hand, we have:
            \begin{equation*}\label{eq:consume2fail}
                \begin{aligned}
                \recencodopenf{M} &= \recencodopenf{\sum_{\perm{B}} \fail^{\widetilde{y} \uplus \widetilde{z}}}
                = \sum_{\perm{B}} \recencodopenf{\fail^{\widetilde{y} \uplus \widetilde{z}}}\\
                &= \sum_{\perm{B}} \recencodf{\fail^{\widetilde{y'} \uplus \widetilde{z'}}}\widetilde{[x_{1n}\leftarrow x_{1n}]}
                = \sum_{\perm{B}} \fail^{\widetilde{y'} \cup \widetilde{z'}}\widetilde{[x_{1n}\leftarrow x_{1n}]}=\expr{L}\\
                \end{aligned}
            \end{equation*}
            
            Therefore, $\recencodopenf{\expr{N}}\red \expr{L}= \recencodopenf{ \expr{M}}$, and the result follows.
            % The reductions in \eqref{eq:consume1fail} and \eqref{eq:consume2fail} lead to identical expressions

            \item $\redlab{R}= \redlab{R:Cons_2}$
            
            Then $\expr{N}= \fail^{\widetilde{y}}\ \esubst{B}{x}$ and the reduction is 
        \begin{prooftree}
            \AxiomC{$\size{B} = k$}
            \AxiomC{\(  \#(x , \widetilde{y}) + k  \not= 0 \)}
            \AxiomC{\( \widetilde{z} = \mfv{B} \)}
            \LeftLabel{$\redlab{R:Cons_2}$}
            \TrinaryInfC{\( \fail^{\widetilde{y}}\ \esubst{B}{x}  \red \sum_{\perm{B}} \fail^{(\widetilde{y} \setminus x) \uplus\widetilde{z}} \)}
        \end{prooftree}
            
      and $\expr{M}=\sum_{\perm{B}} \fail^{(\widetilde{y} \setminus x) \uplus\widetilde{z}}$. 
            Below we assume $\lfv{\expr{N}}=\{x_1,\ldots, x_n\}$.

            On the one hand, we have: (below $\widetilde{y}=y_1,\ldots, y_m$)
            \begin{equation}\label{eq:consume3fail}
                \begin{aligned}
                \recencodopenf{\expr{N}} &= \recencodopenf{\fail^{\widetilde{y}}\ \esubst{B}{x}}\\
                &= \recencodf{ \fail^{\widetilde{y}}\ \esubst{B}{x} \langle \widetilde{x_1}/x_1\rangle\cdots \langle \widetilde{x_n}/x_n\rangle } \widetilde{[x_{1n}\leftarrow x_{1n}]}\\
                &= \recencodf{ \fail^{\widetilde{y'}}\langle y_1 / x  \rangle \cdots \langle y_m / x  \rangle}[\widetilde{y}\leftarrow x] \ \esubst{ \recencodf{ B' } }{x} \widetilde{[x_{1n}\leftarrow x_{1n}]}\\
                &= \fail^{\widetilde{y''}} [\widetilde{y}\leftarrow x] \ \esubst{ \recencodf{ B' } }{x} \widetilde{[x_{1n}\leftarrow x_{1n}]}\\
                & \red_{\redlab{RS{:}Cons_2}} \fail^{(\widetilde{y'} \setminus x) \cup\widetilde{z'}} \widetilde{[x_{1n}\leftarrow x_{1n}]}\\
                \end{aligned}
            \end{equation}
            
            As $\widetilde{y}$ consists of free variables,  in $\fail^{\widetilde{y}}\ \esubst{B}{x} \langle \widetilde{x_1}/x_1\rangle\cdots \langle \widetilde{x_n}/x_n\rangle$ the substitutions also occur on $ \widetilde{y}$ resulting in a new $\widetilde{y'}$ where all $x_i$'s are replaced with their fresh components in $\widetilde{x_i}$. Similarly for \revo{A24}{$z'$ and $B'$ as well as } $\widetilde{y''}$ being $\widetilde{y'}$ with each $x$ replaced with a fresh $y_i$. On the other hand, we have:
            \begin{equation}\label{eq:consume4fail}
                \begin{aligned}
                \recencodopenf{M} &= \recencodopenf{\sum_{\perm{B}} \fail^{(\widetilde{y} \setminus x) \uplus\widetilde{z}}}
                = \sum_{\perm{B}} \recencodopenf{\fail^{(\widetilde{y} \setminus x) \uplus\widetilde{z}}}\\
                &= \sum_{\perm{B}} \recencodf{\fail^{(\widetilde{y'} \setminus x) \uplus\widetilde{z}}} \widetilde{[x_{1n}\leftarrow x_{1n}]}
                 = \fail^{(\widetilde{y'} \setminus x) \cup \widetilde{z'}} \widetilde{[x_{1n}\leftarrow x_{1n}]}\\
                \end{aligned}
            \end{equation}
            
            The reductions in \eqref{eq:consume3fail} and \eqref{eq:consume4fail} lead to identical expressions.
        \end{myEnumerate}
    \end{myEnumerate}
    
     As before, the reduction via Rule~$\redlab{R} $ could occur inside a context (cf. Rules $\redlab{R:TCont}$ and $\redlab{R:ECont}$). We consider only the case when the contextual rule used is $\redlab{R:TCont}$. We have $\expr{N} = C[N]$. When we have $C[N] \red_{\redlab{R}} C[M] $ such that $N \red_{\redlab{R}} M$ we need to show that $\recencodopenf{ C[N]} \red^j \recencodopenf{ C[M] }$for some $j$ dependent on ${\redlab{R}}$. 
     Firstly, let us assume ${\redlab{R}} = \redlab{R:Cons_2}$  then we take $j = 1$. Let us take $C[\cdot]$ to be $[\cdot]B$ and $\lfv{NB} = \{ x_1, \cdots , x_k  \}$ then 
     \[
        \begin{aligned}
            \recencodopenf{ N B}  & = \recencodf{NB\linsub{\widetilde{x_{1}}}{x_1}\cdots \linsub{\widetilde{x_k}}{x_k}}\widetilde{[x_{1k}\leftarrow x_{1k}]} \\
             & = \recencodf{N' B'}\widetilde{[x_{1k}\leftarrow x_{1k}]} = \recencodf{N'}\recencodf{ B'}\widetilde{[x_{1k}\leftarrow x_{1k}]} \\
        \end{aligned}
     \]
    We take $N'B'= NB\linsub{\widetilde{x_{1}}}{x_1}\cdots \linsub{\widetilde{x_k}}{x_k}$, and by the IH that $ \recencodf{N}\red \recencodf{M}$ and hence we can deduce that $\recencodf{N'}\red \recencodf{M'}$ where $M'B'= MB\linsub{\widetilde{x_{1}}}{x_1}\cdots \linsub{\widetilde{x_k}}{x_k}$. Finally,
    \[
        \begin{aligned}
            \recencodf{N'}\recencodf{ B'}\widetilde{[x_{1k}\leftarrow x_{1k}]}\red \recencodf{M'}\recencodf{ B'}\widetilde{[x_{1k}\leftarrow x_{1k}]}
        \end{aligned}
    \]
     and hence $ \recencodopenf{C[N]} \red \recencodopenf{C[M]} $.
\end{proof}

\soundnessone*

\begin{proof}
By induction on the structure of $\expr{N}$ with the following six cases given below, where $\widetilde{[x_{1k}\leftarrow x_{1k}]}$ abbreviates $[\widetilde{x_1}\leftarrow x_1]\ldots [\widetilde{x_k}\leftarrow x_k]$:

\begin{enumerate}[i)]
    \item $\expr{N} = x$:
    
    Then $\recencodopenf{x} = x_1 [x_1 \leftarrow x]$, and no reductions can be performed.
    
    \item $\expr{N} = \lambda x. N$:

    Suppose $\lfv{N} = \{ x_1, \cdots , x_k\}$. Then,
    \[
    \begin{aligned}
        \recencodopenf{\lambda x. N} &= \recencodf{\lambda x. N\langle \widetilde{x_1} / x_1 \rangle \cdots \langle \widetilde{x_k} / x_k \rangle} \widetilde{[x_{1k}\leftarrow x_{1k}]}\\
        &= \recencodf{\lambda x. N'}\widetilde{[x_{1k}\leftarrow x_{1k}]}= \lambda x. \recencodf{N'\langle \widetilde{y} / x \rangle} [\widetilde{y} \leftarrow x]  \widetilde{[x_{1k}\leftarrow x_{1k}]},
    \end{aligned}
    \]
    \revo{A25}{where $N'$ is $N$ after the substitutions $\langle \widetilde{x_1} / x_1 \rangle \cdots \langle \widetilde{x_k} / x_k \rangle$} and no reductions can be performed.

    \item $\expr{N} = N B$:
    
    Suppose  $\lfv{NB} = \{ x_1, \cdots , x_n\}$. Then
    \begin{equation}\label{eq:app_npfail}
    \begin{aligned}
        \recencodopenf{\expr{N}}&=\recencodopenf{NB}  = \recencodf{NB\langle \widetilde{x_1} / x_1 \rangle \cdots \langle \widetilde{x_n} / x_n \rangle} \widetilde{[x_{1n}\leftarrow x_{1n}]}\\
        & = \recencodf{N' B'} \widetilde{[x_{1n}\leftarrow x_{1n}]}
         = \recencodf{N'} \recencodf{B'} \widetilde{[x_{1n}\leftarrow x_{1n}]}
    \end{aligned}
    \end{equation}
    where $\widetilde{x_i}=x_{i1},\ldots, x_{ij_i}$, for $1\leq i \leq n$ and \revo{A25}{ $N', B'$ are $N$ and $B$ after performing the substitutions $\langle \widetilde{x_1} / x_1 \rangle \cdots \langle \widetilde{x_k} / x_k \rangle$ }.
    By the reduction rules in~\figref{fig:share-reductfailure} there are three possible reductions starting in $\expr{N}$:
  
    \begin{enumerate}
        \item $\recencodf{N'}\recencodf{B'}\widetilde{[x_{1n}\leftarrow x_{1n}]}$ reduces via rule $\redlab{RS{:}Beta}$.
        
        In this case  $N=\lambda x. N_1$, and the encoding in (\ref{eq:app_npfail}) gives $N'= N\langle \widetilde{x_1} / x_1 \rangle \cdots \langle \widetilde{x_n} / x_n \rangle$, which  implies $N' =\lambda x. N_1^{'}$ and the following holds:
        \begin{equation*}
        \begin{aligned}
            \recencodf{N'}=\recencodf{(\lambda x. N'_1)} &= (\lambda x. \recencodf{N'_1 \langle \widetilde{y} / x \rangle} [\widetilde{y} \leftarrow x])
             = (\lambda x. \recencodf{N^{''}} [\widetilde{y} \leftarrow x])
        \end{aligned}
        \end{equation*}
        
        Thus, we have the following $\redlab{RS{:}Beta}$ reduction from \eqref{eq:app_npfail}:
        \begin{equation}\label{eq:sound.appfail}
            \begin{aligned}
                \recencodopenf{\expr{N}} &= \recencodf{N'} \recencodf{B'}\widetilde{[x_{1n}\leftarrow x_{1n}]}=(\lambda x. \recencodf{N''} [\widetilde{y} \leftarrow x] \recencodf{B'}) \widetilde{[x_{1n}\leftarrow x_{1n}]}\\
                &\red_{\redlab{RS{:}Beta}}  \recencodf{N^{''}} [\widetilde{y} \leftarrow x] \esubst{\recencodf{B'}}{x}  \widetilde{[x_{1n}\leftarrow x_{1n}]} =\expr{L}
            \end{aligned}
        \end{equation}

        \revo{A25}{where $ N'' $ is $N'$ after the substitutions $\langle \widetilde{y} / x \rangle$}. Notice that the expression $\expr{N}$ can perform the following $\redlab{R:Beta}$-reduction:
        \[\expr{N}=(\lambda x. N_1) B\red_{\redlab{R:Beta}} N_1 \esubst{B}{x} \]

        Assuming $\expr{N'}=N_1 \esubst{B}{x}$, there are two cases:
        % Firstly where $\#(x,N_1) = \size{B} = k \geq 1$, otherwise $\#(x,N_1) = k$
        
        \begin{enumerate}
            
            \item$\#(x,M) = \size{B} = k \geq 1$.

            On the one hand:
            \begin{equation*}\label{eq:sound_appn1fail}
            \begin{aligned}
                \recencodopenf{\expr{N'}}&=\recencodopenf{N_1 \esubst{B}{x}}\\
                &= \recencodf{N_1 \esubst{B}{x}\langle \widetilde{x_1} / x_1 \rangle \cdots \langle \widetilde{x_n} / x_n \rangle} \widetilde{[x_{1n}\leftarrow x_{1n}]}\\
                & = \recencodf{N_1' \esubst{B'}{x}}\widetilde{[x_{1n}\leftarrow x_{1n}]}\\
                & = \sum_{B_i \in \perm{\recencodf{ B }}}\recencodf{ N_1' \langle y_1 , \cdots , y_k / x  \rangle } \linexsub{B_i(1)/y_1} \cdots \linexsub{B_i(k)/y_k} \widetilde{[x_{1n}\leftarrow x_{1n}]}\\
                & = \sum_{B_i \in \perm{\recencodf{ B }}}\recencodf{ N_1''} \linexsub{B_i(1)/y_1} \cdots \linexsub{B_i(k)/y_k} \widetilde{[x_{1n}\leftarrow x_{1n}]}\\
            \end{aligned}
            \end{equation*}
            \revo{A25}{where $ N_1'' $ is $N_1'$ after the substitution $\langle \widetilde{y} / x \rangle$}. 
            
            On the other hand, after an application of Rule~\redlab{RS:Ex-Sub}:
            \begin{equation*}\label{eq:sound_appn2fail}
            \begin{aligned}
                \expr{L} &= \recencodf{N''} [\widetilde{y} \leftarrow x] \esubst{\recencodf{B'}}{x}  \widetilde{[x_{1n}\leftarrow x_{1n}]} \\
                &\red\sum_{B_i \in \perm{\recencodf{ B }}}\recencodf{ N_1''} \linexsub{B_i(1)/y_1} \cdots \linexsub{B_i(k)/y_k} \widetilde{[x_{1n}\leftarrow x_{1n}]}\\
                &= \recencodopenf{\expr{N}'}
            \end{aligned}
            \end{equation*}
            
         and the result follows.
            
            \item Otherwise, either $\#(x,N_1) = k=0$ or $\#(x,N_1)\neq \size{B}$.
            In this case:
            \begin{equation*}\label{eq:sound_appn3fail}
            \begin{aligned}
                \recencodopenf{\expr{N'}}&=\recencodopenf{N_1 \esubst{B}{x}}\\
                & = \recencodf{N_1 \esubst{B}{x}\langle \widetilde{x_1} / x_1 \rangle \cdots \langle \widetilde{x_n} / x_n \rangle}\widetilde{[x_{1n}\leftarrow x_{1n}]}\\
                & = \recencodf{N_1' \esubst{B'}{x}}\widetilde{[x_{1n}\leftarrow x_{1n}]}\\
                & =  \recencodf{N^{''}} [\widetilde{y} \leftarrow x] \esubst{\recencodf{B'}}{x} \widetilde{[x_{1n}\leftarrow x_{1n}]}=\expr{L} \\
            \end{aligned}
            \end{equation*}
            
            From (\ref{eq:sound.appfail}):  $\recencodopenf{\expr{N}}\red \expr{L}=\recencodopenf{\expr{N'}}$ and the result follows.

        \end{enumerate}
        
        \item $\recencodf{N'}\recencodf{B'}\widetilde{[x_{1n}\leftarrow x_{1n}]}$ reduces via rule $\redlab{RS{:} Cons_1}$.

        In this case we would have  $N=\fail^{\widetilde{y}}$, and the encoding in (\ref{eq:app_npfail}) gives $N'= N\linsub{\widetilde{x_1}}{x_1}\ldots \linsub{\widetilde{x_n}}{x_n}$, which implies $N'
        =\fail^{\widetilde{y'}} $, we let $\size{B} = k $ and the following:

        \begin{equation}\label{eq:sound.consumfail}
            \begin{aligned}
                \recencodopenf{\expr{N}} &= \recencodf{N'} \recencodf{B'}\widetilde{[x_{1n}\leftarrow x_{1n}]}
                = \recencodf{\fail^{\widetilde{y'}}} \recencodf{B'}\widetilde{[x_{1n}\leftarrow x_{1n}]}\\
                &= \fail^{\widetilde{y'}} \recencodf{B'}\widetilde{[x_{1n}\leftarrow x_{1n}]}\\
                & \red \sum_{\perm{B}} \fail^{\widetilde{y'} \uplus \widetilde{z}} \widetilde{[x_{1n}\leftarrow x_{1n}]}, \text{ where } \widetilde{z} = \lfv{B'}\\
            \end{aligned}
        \end{equation}
        
        The expression $\expr{N}$ can perform the following $\redlab{R}=\redlab{R:Cons_1}$-reduction:
        
        \begin{equation}\label{eq:sound.2consumfail}
            \expr{N}=\fail^{\widetilde{y}} \  B\red_{\redlab{R}} \sum_{\perm{B}} \fail^{\widetilde{y}\uplus \widetilde{z}} \text{  where } \widetilde{z} = \mfv{B}
        \end{equation}
        
        From (\ref{eq:sound.consumfail}) and (\ref{eq:sound.2consumfail}), we infer  that  $\expr{L}=\recencodopenf{\expr{N'}}$ and so the result follows.
        
        \item Suppose that $\recencodf{N'} \red \recencodf{N''}$. \\
        This case follows from the IH.
    \end{enumerate}

    \item  $\expr{N} = N \esubst{B}{x}$:
    
    Suppose  $\lfv{N \esubst{B}{x}} = \{ x_1, \cdots , x_k\}$. Then,

    \begin{equation}\label{eq:sound_expsubfail}
    \begin{aligned}
    \recencodopenf{\expr{N}}= \recencodopenf{N \esubst{B}{x}}
         &= \recencodf{N \esubst{B}{x}\langle \widetilde{x_1} / x_1 \rangle \cdots \langle \widetilde{x_k} / x_k \rangle} \widetilde{[x_{1k}\leftarrow x_{1k}]}\\
        &= \recencodf{N' \esubst{B'}{x}} \widetilde{[x_{1k}\leftarrow x_{1k}]}
        \end{aligned}
    \end{equation}
    
    \revo{A25}{where $N', B'$ are $N$ and $B$ after performing the substitutions $\langle \widetilde{x_1} / x_1 \rangle \cdots \langle \widetilde{x_k} / x_k \rangle$ }. Let us consider the two possibilities of the encoding:
    
    \begin{myEnumerate}
        
        \item $ \#(x,M) = \size{B} = k \geq 1 $.
        
        Then we continue equation \eqref{eq:sound_expsubfail} as follows:
        \begin{equation}\label{eq:sound_expsub2fail}
            \begin{aligned}
                \recencodopenf{\expr{N}} &= \recencodf{N' \esubst{B'}{x}} \widetilde{[x_{1k}\leftarrow x_{1k}]} \\
                &=  \sum_{B_i \in \perm{\recencodf{ B' }}}\recencodf{ N' \langle y_1 , \cdots , y_n / x  \rangle } \linexsub{B_i(1)/y_1} \cdots \linexsub{B_i(n)/y_n}\widetilde{[x_{1k}\leftarrow x_{1k}]} \\
                &=  \sum_{B_i \in \perm{\recencodf{ B' }}}\recencodf{ N'' } \linexsub{B_i(1)/y_1} \cdots \linexsub{B_i(n)/y_n}\widetilde{[x_{1k}\leftarrow x_{1k}]} \\
            \end{aligned}
        \end{equation}
        
        \revo{A25}{where $N''$ is $N'$ after performing the substitutions $ \langle y_1 , \cdots , y_n / x  \rangle $ }. There are three possible reductions, these being from rules \redlab{RS{:}Lin \dash Fetch}, $\redlab{RS{:}Cons_3}$, and  \redlab{RS{:}Cont}.
        
        \begin{myEnumerate}
            
            \item Suppose that $\headf{N''} = y_1$.
            
             Then one has to consider the shape of the bag $B'$:
                \begin{myEnumerate}

                    \item When $B'$ has only one element $N_1$ then from (\ref{eq:sound_expsub2fail}) and by letting $B = \bag{N_1}$ and $B' = \bag{N'_1}$ we have
                    \begin{equation}\label{eq:sound_expsub3fail}
                        \begin{aligned}
                            \recencodopenf{\expr{N}} 
                            & = \recencodf{N^{''}}\linexsub{\recencodf{N_1'}/y_1} \widetilde{[x_{1k}\leftarrow x_{1k}]}, \text{since }\headf{M'}=y_1\\
                            &\red \recencodf{N^{''}}\headlin{\recencodf{N_1'}/y_1}\widetilde{[x_{1k}\leftarrow x_{1k}]} = \expr{L}
                        \end{aligned}
                    \end{equation}
                    
                    We also have:
                    \begin{equation}\label{eq:sound_expsub4fail}
                        \begin{aligned}
                            \expr{N} 
                            & = N\esubst{\bag{N_1}}{x}\\
                            &\red N\headlin{N_1/x}\esubst{\oneb}{x} = \expr{N}' 
                           \pequiv N\headlin{N_1/x}
                             = \expr{N}''
                            \\
                        \end{aligned}
                    \end{equation}
                    
                    From (\ref{eq:sound_expsub3fail}) and (\ref{eq:sound_expsub4fail}), we infer  that  $\expr{L}'=\recencodopenf{\expr{N'}}$ and so the result follows.

                    \item When $B'$ has more then one element. Let us say that $B =  \bag{N_1,N_2}$ and $B' =  \bag{N'_1,N'_2}$ and cases for larger bags proceed similarly then from (\ref{eq:sound_expsub2fail}). (Below we use the fact that $\headf{M'}=y_1$)
\begin{equation}\label{eq:sound_expsub33fail}
\begin{aligned}
\recencodopenf{\expr{N}} 
& = \recencodf{N''}\linexsub{\recencodf{N_1'}/y_1}\linexsub{\recencodf{N_2'}/y_2} \widetilde{[x_{1k}\leftarrow x_{1k}]} \\
& + \recencodf{N''}\linexsub{\recencodf{N_2'}/y_1}\linexsub{\recencodf{N_1'}/y_2} \widetilde{[x_{1k}\leftarrow x_{1k}]}, \\
&\red
\recencodf{N''} \headlin{\recencodf{N_1'}/y_1} \linexsub{\recencodf{N_2'}/y_2} [\widetilde{[x_{1k}\leftarrow x_{1k}]}\\
 & + \recencodf{N''} \headlin{\recencodf{N_2'}/y_1} \linexsub{\recencodf{N_1'}/y_2} \widetilde{[x_{1k}\leftarrow x_{1k}]} = \expr{L}
   \end{aligned}
\end{equation}
                    
     We also have:
\begin{equation}\label{eq:sound_expsub44fail}
 \begin{aligned}
 \expr{N} 
  & = N\esubst{\bag{N_1,N_2}}{x}\\
& \red N\headlin{N_1/x}\esubst{\bag{N_2}}{x} + N\headlin{N_2/x}\esubst{\bag{N_1}}{x} 
                             = \expr{N}' 
 \end{aligned}
\end{equation}
                    
                    From (\ref{eq:sound_expsub33fail}) and (\ref{eq:sound_expsub44fail}), we infer  that  $\expr{L}'=\recencodopenf{\expr{N'}}$ and so the result follows.
                    
                \end{myEnumerate}

            \item Suppose that $N'' = \fail^{\widetilde{z'}}$. Then we proceed similarly as from (\ref{eq:sound_expsub2fail}):
                    \begin{equation}\label{eq:sound_expsub99fail}
                        \begin{aligned}
                            \recencodopenf{\expr{N}} 
                            & = \sum_{B_i \in \perm{\recencodf{ B' }}}\fail^{\widetilde{z'}} \linexsub{B_i(1)/y_1} \cdots \linexsub{B_i(n)/y_n}\\
                            &\red^* \sum_{B_i \in \perm{\recencodf{ B' }}}\fail^{(\widetilde{z'} \setminus y_1, \cdots , y_n) \uplus\widetilde{y}}, \text{ since }\headf{M'}=y_1\\
                           &= \expr{L}'
                        \end{aligned}
                    \end{equation}
                    
         where $\widetilde{y} = \lfv{B_i(1)} \uplus \cdots \uplus \lfv{B_i(n)}$.       We also have that
                    
                    \begin{equation}\label{eq:sound_expsub11fail}
                        \begin{aligned}
                            \expr{N} 
                            & = \fail^{\widetilde{z}} \esubst{B}{x} 
                             \red \fail^{(\widetilde{z} \setminus x) \uplus\widetilde{y}} 
                            = \expr{N}' 
                        \end{aligned}
                    \end{equation}
            
            \text{ where }  $\widetilde{y} = \mfv{B}$. 
                From (\ref{eq:sound_expsub99fail}) and (\ref{eq:sound_expsub11fail}), we infer  that  $\expr{L}'=\recencodopenf{\expr{N'}}$ and so the result follows.
            
            \item  Suppose that $N'' \red N'''$
                
                This case follows by the IH.
            
        \end{myEnumerate}

        \item Otherwise we continue equation (\ref{eq:sound_expsubfail}) as follows where $\#(x,M) = k$
            \begin{equation}%\label{eq:sound_expsubotherwise1}
            \begin{aligned}
                \recencodopenf{\expr{N}} &= \recencodf{N' \esubst{B'}{x}}  \widetilde{[x_{1k}\leftarrow x_{1k}]} \\
                &=  \recencodf{N'\langle y_1. \cdots , y_k / x  \rangle} [y_1. \cdots , y_k \leftarrow x] \esubst{ \recencodf{B'} }{ x }  \widetilde{[x_{1k}\leftarrow x_{1k}]}\\
                &=  \recencodf{N''} [y_1. \cdots , y_k \leftarrow x] \esubst{ \recencodf{B'} }{ x }  \widetilde{[x_{1k}\leftarrow x_{1k}]} \\
            \end{aligned}
            \end{equation}
            
            Let us consider the two possible cases:
            
            \begin{myEnumerate}
                
                \item $ \#(x,M) = \size{B} = k = 0 $.
                
                Then we have:
                                    \begin{equation}\label{eq:sound_expsubotherwise1}
                    \begin{aligned}
                        \recencodopenf{\expr{N}} &=  \recencodf{N'}  \esubst{ 1 }{ x }  \widetilde{[x_{1k}\leftarrow x_{1k}]} \\
                    \end{aligned}
                    \end{equation}
                    
                    Reductions can only appear in $\recencodf{N'}$ and the case follows by the IH.
                    
                \item Otherwise we can perform the reduction:
                    \begin{equation}\label{eq:sound_expsubotherwise2}
                    \begin{aligned}
                        \recencodopenf{\expr{N}} &= \recencodf{N''} [y_1. \cdots , y_k \leftarrow x] \esubst{ \recencodf{B'} }{ x }  \widetilde{[x_{1k}\leftarrow x_{1k}]}\\
                        &\red \sum_{B_i \in \perm{B}}  \fail^{\widetilde{z'}} \widetilde{[x_{1k}\leftarrow x_{1k}]}
                        = \expr{L}'
                    \end{aligned}
                    \end{equation}
                 \text{ where } $\widetilde{z'} = \lfv{N''} \uplus \lfv{B'}$.
                We also have that
                    
                    \begin{equation}\label{eq:sound_expsubotherwise3}
                        \begin{aligned}
                            \expr{N} 
                            & = N \esubst{B}{x}  \red \sum_{\perm{B}} \fail^{\widetilde{z}}  = \expr{N}' 
                        \end{aligned}
                    \end{equation}
                    
                  \text{ where } $\widetilde{z} = \mfv{M} \uplus \mfv{B}$.   \\ 
                From (\ref{eq:sound_expsubotherwise2}) and (\ref{eq:sound_expsubotherwise3}), we infer  that  $\expr{L}'=\recencodopenf{\expr{N'}}$ and so the result follows.
                
            \end{myEnumerate}
        
    \end{myEnumerate}

    \item $\expr{N} = \fail^{\widetilde{y}}$
        \\
        Then $\recencodopenf{\fail^{\widetilde{y}}} = \fail^{\widetilde{y}}$, and no reductions can be performed.

    \item $\expr{N} = \expr{N}_1 + \expr{N}_2$: \\ This case holds by the IH. \qedhere
\end{enumerate}
%\qed
\end{proof}

%%%%%%%%%%%%%%%%%%%%%%%%%%%%%%%%%%
\subsection{Success Sensitiveness}
\label{app:sucessone}

\checkpres*

\begin{proof}
By induction on the structure of $M$. We only need to consider terms of the following form.

\begin{myEnumerate}

    \item When $ M = \checkmark $ the case is immediate.
    
    \item When $ M = NB $ with $\lfv{NB} = \{x_1,\cdots,x_k\}$ and  $\#(x_i,M)=j_i$ we have that: 

        \[ 
            \begin{aligned}
                \headfsum{\recencodopenf{NB}} &= \headfsum{\recencodf{NB\linsub{\widetilde{x_{1}}}{x_1}\cdots \linsub{\widetilde{x_k}}{x_k}}[\widetilde{x_1}\leftarrow x_1]\cdots [\widetilde{x_k}\leftarrow x_k]}\\
                &= \headfsum{\recencodf{NB}}
                 = \headfsum{\recencodf{N}} 
            \end{aligned}
        \]
         and $\headf{NB}= \headf{N} $, by the IH we have $\headf{N} = \checkmark \iff \headfsum{\recencodf{N}} = \checkmark$.
         
    \item When $M = N \esubst{B}{x}$, we must have that $\#(x,M) = \size{B}$ for the head of this term to be $\checkmark$. Let $\lfv{N \esubst{B}{x}} = \{x_1,\cdots,x_k\}$ and  $\#(x_i,M)=j_i$. We have that: 
        \[
            \begin{aligned}
                \headfsum{\recencodopenf{N \esubst{B}{x}}} &= \headfsum{\recencodf{N \esubst{B}{x}\linsub{\widetilde{x_{1}}}{x_1}\cdots \linsub{\widetilde{x_k}}{x_k}}[\widetilde{x_1}\leftarrow x_1]\cdots [\widetilde{x_k}\leftarrow x_k]}\\
                &= \headfsum{\recencodf{N \esubst{B}{x}}}\\
                &= \headfsum{\sum_{B_i \in \perm{\recencodf{ B }}}\recencodf{ N \langle x_1 , \cdots , x_k / x  \rangle } \linexsub{B_i(1)/x_1} \cdots \linexsub{B_i(k)/x_k}}\\
                &= \headfsum{\recencodf{ N \langle x_1 , \cdots , x_k / x  \rangle } \linexsub{B_i(1)/x_1} \cdots \linexsub{B_i(k)/x_k}}\\
                &= \headfsum{\recencodf{ N \langle x_1 , \cdots , x_k / x  \rangle } } 
            \end{aligned}
        \]
        
        and $\headf{N \esubst{B}{x}} = \headf{N}$, by the IH we have
\begin{align*}
\headf{N} = \checkmark \iff \headfsum{\recencodf{N}} = \checkmark \tag*{\qedhere}
\end{align*}        
        
\end{myEnumerate}
%\qed
\end{proof}

\appsuccesssensce*

\begin{proof}
By induction on the structure of expressions $\lamrfail$ and $\lamrsharfail$. We proceed with the proof in two parts.

\begin{myEnumerate}
    
    \item Suppose that  $\expr{M} \Downarrow_{\checkmark} $. We will prove that $\recencodopenf{\expr{M}} \Downarrow_{\checkmark}$.

    By operational completeness (\thmref{l:app_completenessone}) we have that if $\expr{M}\red_{\redlab{R}} \expr{M'}$ then

    \begin{enumerate}
        \item If $\redlab{R} =  \redlab{R:Beta}$  then $ \recencodopenf{\expr{M}}  \red^{\leq 2}\recencodopenf{\expr{M}'}$;

        \item If $\redlab{R} =\redlab{R:Fetch}$   then   $ \recencodopenf{\expr{M}}  \red^+ \recencodopenf{\expr{M}''}$, for some $ \expr{M}''$ such that  $\expr{M}' \pequiv \expr{M}''$. 
        \item If $\redlab{R} \neq  \redlab{R:Beta}$ and $\redlab{R}\neq \redlab{R:Fetch}$  then $ \recencodopenf{\expr{M}}  \red\recencodopenf{\expr{M}'}$;
    \end{enumerate}
    
    Notice that  neither our  reduction rules  (in Figure ~\ref{fig:share-reductfailure}), or our congruence $\pequiv$ (in Figure~\ref{fig:rsPrecongruencefailure}),  or  our encoding ($\recencodopenf{\checkmark }=\checkmark$)  create or destroy a $\checkmark$ occurring in the head of term. By Proposition \ref{Prop:checkpres} the encoding preserves the head of a term being $\checkmark$. The encoding acts homomorphically over sums, therefore, if a $\checkmark$ appears as the head of a term in a sum, it will stay in the encoded sum. We can iterate the operational completeness lemma and obtain the result.

    \item Suppose that $\recencodopenf{\expr{M}} \Downarrow_{\checkmark}$. We will prove that $ \expr{M} \Downarrow_{\checkmark}$. 
    
   From \defref{def:app_Suc3} we have that  $\succp{\recencodopenf{\expr{M}}}{\checkmark}\implies \exists M_1 , \cdots , M_k. ~\expr{M} \red^*  M_1 + \cdots + M_k \text{ and } \headf{M_j} = \checkmark,$
    for some  $j \in \{1, \ldots, k\}$.
    
   Notice that if $\recencodopenf{\expr{M}}$ is itself a term headed with $\checkmark$, say $\headf{\recencodopenf{\expr{M}}}=\checkmark$, then $\expr{M}$ is itself headed with $\checkmark$, from Proposition \ref{Prop:checkpres}.
   
   Based on the shape of $\recencodopenf{\expr{M}}$, we consider two cases.
   The first case, when $\recencodopenf{\expr{M}} = M_1+\ldots+M_k$, $k\geq 2$, and $\checkmark$ occurs in the head of an $M_j$, follows a similar reasoning.  Then $\expr{M}$ has one of the forms:
   \begin{enumerate}
       \item   $\expr{M}= N_1$, then $N_1$ must contain the subterm $ M\esubst{B}{x}$ and $\size{B}=\#(x,M)$. 
       Since, 
       
        $\recencodopenf{M\esubst{B}{x}}=\displaystyle\sum_{B_i \in \perm{\recencodf{ B }}}\recencodf{M\linsub{\widetilde{x}}{x}}\linexsub{B_i(1)/x_i}\ldots \linexsub{B_i(k)/x_i}$,  we can apply Proposition \ref{Prop:checkpres} as we may apply $ \headfsum{\recencodopenf{M\esubst{B}{x}}} $.
       
       \item $\expr{M}=N_1+\ldots+N_l$ for $l \geq 2$.
       
       The reasoning is similar and uses the fact  that the encoding distributes homomorphically over sums.
   \end{enumerate}

   \revo{A26}{The second case is when $\recencodopenf{\expr{M}}\red^* M_1+\ldots+M_k$, and $\headf{M_j}=\checkmark$, for some $j$ and $M_j$. By operational soundness (\thmref{l:soundnessone}) we have that if $ \recencodopenf{\expr{M}}  \red \expr{L}$ then there exist $ \expr{M}' $ such that $ \expr{M}  \red_{\redlab{R}} \expr{M}'$ and 
    \begin{enumerate}
        \item If $\redlab{R} = \redlab{R:Beta}$ then $\expr{ L } \red^{\leq 1} \recencodopenf{\expr{M}'}$;
        \item If $\redlab{R} \neq \redlab{R:Beta}$ then $\expr{ L } \red^*  \recencodopenf{\expr{M}^{''}} $, for $ \expr{M}''$ such that  $\expr{M}' \pequiv \expr{M}''$.
    \end{enumerate}
     The reasoning is similar to the previous case, since our reduction rules do not introduce/eliminate $\checkmark$ occurring in the head of terms and by taking $\expr{L}$ to be $M_1+\ldots+M_k$ with $\headf{M_j}=\checkmark$, for some $j$ and $M_j$ the result follows}. \qedhere
\end{myEnumerate}
%\qed
\end{proof}

\section{Appendix to \texorpdfstring{\secref{ss:secondstep}}{§ 5.3}}

\subsection{Type Preservation}
\label{app:typeprestwo}

\appaux*

\begin{proof}

    We shall prove the case of $(1)$ and the case of $(2)$ follows immediately. The case of $(3)$ is immediate by the encoding on types defined in Definition \ref{def:enc_sestypfail}. Hence we take $j > k$, $\tau_1 $ to be an arbitrary type and $m = 0$; also, we take $\tau_2 $ to be $\sigma$ and $n = j-k$. Hence we want to show that $ \piencodf{\sigma^{j}}_{(\tau_1, 0)} = \piencodf{\sigma^{k}}_{(\sigma, n)} $. We have the following
    \[
        \begin{aligned}
            \piencodf{\sigma^{k}}_{(\sigma, n)} &= \oplus(( \with \onef) \ampy ( \oplus  \with (( \oplus \piencodf{\sigma} ) \otimes (\piencodf{\sigma^{k-1}}_{(\sigma, n)}))))\\
            \piencodf{\sigma^{k-1}}_{(\sigma, n)} &= \oplus(( \with \onef) \ampy ( \oplus  \with (( \oplus \piencodf{\sigma} ) \otimes (\piencodf{\sigma^{k-2}}_{(\sigma, n)}))))\\
            \vdots\\
            \piencodf{\sigma^{1}}_{(\sigma, n)} &= \oplus(( \with \onef) \ampy ( \oplus  \with (( \oplus \piencodf{\sigma} ) \otimes (\piencodf{\omega}_{(\sigma, n)}))))
        \end{aligned}
    \]
    and
    \[
        \begin{aligned}
            \piencodf{\sigma^{j}}_{(\tau_1, 0)} &= \oplus(( \with \onef) \ampy ( \oplus  \with (( \oplus \piencodf{\sigma} ) \otimes (\piencodf{\sigma^{j-1}}_{(\tau_1, 0)}))))\\
            \piencodf{\sigma^{j-1}}_{(\tau_1, 0)} &= \oplus(( \with \onef) \ampy ( \oplus  \with (( \oplus \piencodf{\sigma} ) \otimes (\piencodf{\sigma^{j-2}}_{(\tau_1, 0)}))))\\
            \vdots\\
            \piencodf{\sigma^{j-k + 1}}_{(\tau_1, 0)} &= \oplus(( \with \onef) \ampy ( \oplus  \with (( \oplus \piencodf{\sigma} ) \otimes (\piencodf{\sigma^{j-k}}_{(\tau_1, 0)}))))
        \end{aligned}
    \]
    Notice that $n = j-k$, hence we wish to show that $ \piencodf{\sigma^{n}}_{(\tau_1, 0)} = \piencodf{\omega}_{(\sigma, n)} $.  Finally we have that:
    \[
        \begin{aligned}
            \piencodf{\omega}_{(\sigma, n)} & = \oplus(( \with \onef) \ampy ( \oplus  \with (( \oplus \piencodf{\sigma} ) \otimes (\piencodf{\omega}_{(\sigma, n-1)})))) \\
            \piencodf{\omega}_{(\sigma, n-1)} & = \oplus(( \with \onef) \ampy ( \oplus  \with (( \oplus \piencodf{\sigma} ) \otimes (\piencodf{\omega}_{(\sigma, n-2)})))) \\
            \vdots\\
            \piencodf{\omega}_{(\sigma, 1)} & = \oplus(( \with \onef) \ampy ( \oplus  \with (( \oplus \piencodf{\sigma} ) \otimes (\piencodf{\omega}_{(\sigma, 0)})))) \\
            \piencodf{\omega}_{(\sigma, 0)} &= \oplus(( \with \onef) \ampy ( \oplus  \with \onef ) \\
        \end{aligned}
    \]
    and 
    \begin{align*}
            \piencodf{\sigma^{n}}_{(\tau_1, 0)} &= \oplus(( \with \onef) \ampy ( \oplus  \with (( \oplus \piencodf{\sigma} ) \otimes (\piencodf{\sigma^{n-1}}_{(\tau_1, 0)}))))\\
            \piencodf{\sigma^{n-1}}_{(\tau_1, 0)} &= \oplus(( \with \onef) \ampy ( \oplus  \with (( \oplus \piencodf{\sigma} ) \otimes (\piencodf{\sigma^{n-2}}_{(\tau_1, 0)}))))\\
            \vdots\\
            \piencodf{\sigma^{1}}_{(\tau_1, 0)} &= \oplus(( \with \onef) \ampy ( \oplus  \with (( \oplus \piencodf{\sigma} ) \otimes (\piencodf{\omega}_{(\tau_1, 0)})))) \\
            \piencodf{\omega}_{(\tau_1, 0)} &= \oplus(( \with \onef) \ampy ( \oplus  \with \onef )
    \tag*{\qedhere}
    \end{align*}
%\qed    
\end{proof}

\preservationtwo*

\begin{proof}
By mutual induction on the typing derivation of $B$ and $\expr{M}$, with an analysis for the last rule applied.
Recall that the encoding of types ($\piencodf{-}$) has been given in 
Definition~\ref{def:enc_sestypfail}.
    
    \begin{enumerate}
        \item We consider two cases:
        
        \begin{enumerate}
        
        \item Rule~$\redlab{FS{:}wf \dash bag}$:
        
        In this case we have the following derivation:
        
            \begin{prooftree}
                    \AxiomC{\( \core{\Gamma} \vdash B : \pi \)}
                    \LeftLabel{\redlab{FS{:}wf \dash bag}}
                    \UnaryInfC{\( \core{\Gamma} \wfdash  B : \pi \)}
              \end{prooftree}
        
        There are two cases to be analyzed:

        \begin{enumerate}[i)]
            
            \item We may type bags with the $\redlab{TS{:}bag}$ Rule. 
            
            This case is similar to that of $\redlab{FS{:}bag}$
            
            \item We may type bags with the $\redlab{TS{:}\oneb}$ Rule.

            That is, 
            
                \begin{prooftree}
                    \AxiomC{\(  \)}
                    \RightLabel{\(\)}
                    \LeftLabel{\redlab{TS{:}\oneb}}
                    \UnaryInfC{\( \vdash \oneb : \omega \)}
                \end{prooftree}
            Our encoding gives us:
            $$\piencodf{\oneb}_x = x.\some_{\emptyset} ; x(y_n). ( y_n.\overline{\some};y_n . \overline{\close} \mid x.\some_{\emptyset} ; x. \overline{\none})$$
            
            and  the encoding of $\omega$ can be either:
            \begin{enumerate}
            \item  $\piencodf{\omega}_{(\sigma,0)} =  \overline{\with(( \oplus \bot )\otimes ( \with \oplus \bot ))}$; or
            \item $\piencodf{\omega}_{(\sigma, i)} =  \overline{   \with(( \oplus \bot) \otimes ( \with  \oplus (( \with  \overline{\piencodf{ \sigma }} )  \ampy (\overline{\piencodf{\omega}_{(\sigma, i - 1)}})))) }$
            \end{enumerate}
and one can build the following type derivation (rules from Figure~\ref{fig:trulespifull}):
        
         \hspace*{-40pt}
         \begin{minipage}{\linewidth}
            \begin{prooftree}
                    \AxiomC{\mbox{\ }}
                    \LeftLabel{\redlab{T\onef}}
                    \UnaryInfC{$y_n . \overline{\close} \vdash y_n: \onef$}
                    \LeftLabel{\redlab{T\with_d^x}}
                    \UnaryInfC{$y_n.\overline{\some};y_n . \overline{\close} \vdash  y_n :\with \onef$}

                    \AxiomC{}
                    \LeftLabel{\redlab{T\with^x}}
                    \UnaryInfC{$x.\dual{\none} \vdash x :\with A$}
                    \LeftLabel{\redlab{T\oplus^x_{\widetilde{w}}}}
                    \UnaryInfC{$x.\some_{\emptyset} ; x. \overline{\none} \vdash  x{:}\oplus \with A$}
                
                \LeftLabel{\redlab{T\mid}}
                \BinaryInfC{$( y_n.\overline{\some};y_n . \overline{\close} \mid x.\some_{\emptyset} ; x. \overline{\none}) \vdash y_n :\with \onef, x{:}\oplus \with A$}
                \LeftLabel{\redlab{T\ampy}}
                \UnaryInfC{$x(y_n). ( y_n.\overline{\some};y_n . \overline{\close} \mid x.\some_{\emptyset} ; x. \overline{\none}) \vdash  x: (\with \onef) \ampy (\oplus \with A) $}
                \LeftLabel{\redlab{T\oplus^x_{\widetilde{w}}}}
                \UnaryInfC{$x.\some_{\emptyset} ; x(y_n). ( y_n.\overline{\some};y_n . \overline{\close} \mid x.\some_{\emptyset} ; x. \overline{\none}) \vdash  x{:}\oplus ((\with \onef) \ampy (\oplus \with A))$}
            \end{prooftree}
            \end{minipage}
            
            Since $A$ is arbitrary,  we can take $A=\oneb$ for $\piencodf{\omega}_{(\sigma,0)} $ and  $A=$\linebreak[4]$ \overline{(( \with  \overline{\piencodf{ \sigma }} )  \ampy (\overline{\piencodf{\omega}_{(\sigma, i - 1)}}))}$  for $\piencodf{\omega}_{(\sigma,i)} $, in both cases, the result follows.

        \end{enumerate}

        \item Rule $\redlab{FS{:}bag}$:
    
        Then $B = \bag{M}\cdot A$ and we have the following derivation:
        
        \begin{prooftree}
            \AxiomC{\( \core{\Gamma} \wfdash M : \sigma\)}
            \AxiomC{\( \core{\Delta} \wfdash A : \sigma^{k} \)}
            \LeftLabel{\redlab{FS{:}bag}}
            \BinaryInfC{\( \core{\Gamma}, \core{\Delta} \wfdash \bag{M}\cdot A:\sigma^{k+1} \)}
        \end{prooftree}

To simplify the proof, we will consider $k=2$ (the case  $k> 2$ follows analogously).

By IH we have
\begin{align*}
    \piencodf{M}_{x_i} & \vdash \piencodf{\core{\Gamma}}, x_i: \piencodf{\sigma}
    \\
    \piencodf{A}_x & \vdash \piencodf{\core{\Delta}}, x: \piencodf{\sigma\wedge \sigma}_{(\tau, j)}
\end{align*}
% where $\piencodf{\Gamma}=\widetilde{w}:\with \Delta$, where $\Delta$ consists of the encodings of the assignments for free variables  $\widetilde{w}=\lfv{M}$ of $M$ in $\Gamma$.  
By Definition~\ref{def:enc_lamrsharpifail},

\begin{equation}
\begin{aligned}
    \piencodf{\bag{M} \cdot A}_{x} =& x.\some_{\lfv{\bag{M} \cdot A} } ; x(y_i). x.\some_{y_i, \lfv{\bag{M} \cdot A}};x.\overline{\some} ; \outact{x}{x_i}.\\
    &(x_i.\some_{\lfv{M}} ; \piencodf{M}_{x_i} \mid \piencodf{A}_{x} \mid y_i. \overline{\none})
    \end{aligned}
\end{equation} 

Let $\Pi_1$ be the derivation:

\begin{prooftree}
            \AxiomC{$\piencodf{M}_{x_i} \;{ \vdash} \piencodf{\core{\Gamma}}, x_i: \piencodf{\sigma} $}
            \LeftLabel{\redlab{T\oplus^x_{\widetilde{w}}}}
            \UnaryInfC{$x_i.\some_{\lfv{M}} ; \piencodf{M}_{x_i} \vdash \piencodf{\core{\Gamma}} ,x_i: \oplus \piencodf{\sigma} $}
            
            \AxiomC{}
            \LeftLabel{\redlab{T\with^x}}
            \UnaryInfC{$ y_i. \overline{\none} \vdash y_i :\with \onef$}
            
        \LeftLabel{\redlab{T\mid}}
        \BinaryInfC{$\underbrace{x_i.\some_{\lfv{M}} ; \piencodf{M}_{x_i} \mid y_i. \overline{\none}}_{P_1} \vdash \piencodf{\core{\Gamma}} ,x_i: \oplus \piencodf{\sigma}, y_i :\with \onef $}
\end{prooftree}

Let $ P_1 = (x_i.\some_{\lfv{M}} ; \piencodf{M}_{x_i} \mid y_i. \overline{\none})$ in the the derivation $\Pi_2$ below:

\hspace*{-30pt}
\begin{minipage}{\linewidth}
\begin{prooftree}
        \AxiomC{$ \Pi_1$} 
        
        \AxiomC{$ \piencodf{A}_{x}  \vdash  \piencodf{\core{\Delta}}, x: \piencodf{\sigma\wedge \sigma}_{(\tau, j)} $}
        
        \LeftLabel{\redlab{T\otimes}}
    \BinaryInfC{$ \outact{x}{x_i}. (P_1 \mid \piencodf{A}_{x}) \vdash  \piencodf{\core{\Gamma}}  ,  \piencodf{\core{\Delta}}, y_i :\with \onef, x: (\oplus \piencodf{\sigma})  \otimes (\piencodf{\sigma\wedge \sigma}_{(\tau, j)}) $}
    \LeftLabel{\redlab{T\with_d^x}}
    \UnaryInfC{$\underbrace{x.\overline{\some} ; \outact{x}{x_i}. (P_1 \mid \piencodf{A}_{x}  )}_{P_2} \vdash \piencodf{\core{\Gamma}}  ,  \piencodf{\core{\Delta}}, y_i :\with \onef, x: \with (( \oplus \piencodf{\sigma} ) \otimes (\piencodf{\sigma\wedge \sigma}_{(\tau, j)}))  $}
\end{prooftree}
\end{minipage}

Let $P_2 = (x.\overline{\some} ; \outact{x}{x_i}. (P_1 \mid \piencodf{A}_{x} ))$ in the derivation below
(the last two rules that were applied  are \redlab{T\oplus^x_{\widetilde{w}}} and \redlab{T\ampy}):
         
\hspace*{-60pt}
\begin{minipage}{\linewidth}
\begin{prooftree}
\small
    \AxiomC{$ \Pi_2$} 
    \noLine
    \UnaryInfC{$\vdots$}
    \noLine
    \UnaryInfC{$P_2\vdash \piencodf{\core{\Gamma}}  ,  \piencodf{\core{\Delta}}, y_i :\with \onef, x: \with (( \oplus \piencodf{\sigma} ) \otimes (\piencodf{\sigma\wedge \sigma}_{(\tau, j)}))  $}
    \LeftLabel{\redlab{T\oplus^x_{\widetilde{w}}}}
    \UnaryInfC{$x.\some_{y_i, \lfv{\bag{M} \cdot A}};P_2  \vdash \piencodf{\core{\Gamma}}  ,  \piencodf{\core{\Delta}}, y_i :\with \onef, x:\oplus  \with (( \oplus \piencodf{\sigma} ) \otimes (\piencodf{\sigma\wedge \sigma}_{(\tau, j)}))$}
    % \LeftLabel{\redlab{T\ampy}}
    \UnaryInfC{$x(y_i). x.\some_{y_i, \lfv{\bag{M} \cdot A}};P_2  \vdash \piencodf{\core{\Gamma}}  ,  \piencodf{\core{\Delta}}, x: ( \with \onef) \ampy ( \oplus  \with (( \oplus \piencodf{\sigma} ) \otimes (\piencodf{\sigma\wedge \sigma}_{(\tau, j)}))) $}
    % \LeftLabel{\redlab{T\oplus^x_{\widetilde{w}}}}
    \UnaryInfC{$\underbrace{x.\some_{\lfv{\bag{M} \cdot A} } ; x(y_i). x.\some_{y_i, \lfv{\bag{M} \cdot A}};P_2 }_{\piencodf{\bag{M}\cdot A}_x }\vdash   \piencodf{\core{\Gamma}}  ,  \piencodf{\core{\Delta}}, x: \oplus(( \with \onef) \ampy ( \oplus  \with (( \oplus \piencodf{\sigma} ) \otimes (\piencodf{\sigma\wedge \sigma}_{(\tau, j)})))) $}
\end{prooftree}
\end{minipage}

   From Definitions~\ref{def:duality} (duality) and \ref{def:enc_sestypfail}, we infer:
\begin{equation*}
    \begin{aligned}
         \oplus(( \with \onef) \ampy ( \oplus  \with (( \oplus \piencodf{\sigma} ) \otimes (\piencodf{\sigma\wedge \sigma}_{(\tau, j)})))) &=\piencodf{\sigma\wedge \sigma \wedge \sigma}_{(\tau, j)}
    \end{aligned}
\end{equation*}
Therefore, $\piencodf{\bag{M}\cdot A}_x \vdash \piencodf{\core{\Gamma},\core{\Delta}}, x: \piencodf{\sigma\wedge \sigma \wedge \sigma}_{(\tau, j)} $ and the result follows.

        \end{enumerate}
        
    \item  The proof of type preservation for expressions, relies on the analysis of nine cases:
    \begin{enumerate}
            
        \item Rule \redlab{FS{:}wf \dash expr}:

        Then we have the following derivation:
        
            \begin{prooftree}
                    \AxiomC{\( \core{\Gamma} \vdash \expr{M} : \tau \)}
                    \LeftLabel{\redlab{FS{:}wf \dash expr}}
                    \UnaryInfC{\( \core{\Gamma} \wfdash  \expr{M} : \tau \)}
            \end{prooftree}
        % \daniele{What do you mean? Can you expand a bit?}
        
            Cases follow from their corresponding case from $\redlab{FS{:}\dash}$. In the case of $ \redlab{TS{:}var} $ we have:
            
            \begin{prooftree}
                \AxiomC{}
                \LeftLabel{$\redlab{TS{:}var}$}
                \UnaryInfC{\( x: \tau\vdash x : \tau  \)}
            \end{prooftree}
            
        By Definition~\ref{def:enc_sestypfail},  $\piencodf{x:\tau}= x:\with \overline{\piencodf{\tau }}$, and by Figure~\ref{fig:encfail},  $\piencodf{x}_u=x.\overline{\some};[x\leftrightarrow u]$. The thesis holds thanks to the following derivation:
        
            \begin{prooftree}
                \AxiomC{}
                \LeftLabel{$\redlab{ (Tid)}$}
                \UnaryInfC{$ [x \leftrightarrow u ] \vdash x:  \overline{\piencodf{\tau }}  , u :  \piencodf{ \tau } $}
                \LeftLabel{$\redlab{T\with^{x}_{d})}$}
                \UnaryInfC{$ x.\overline{\some} ;[x \leftrightarrow u ] \vdash x: \with  \overline{\piencodf{ \tau }} , u :  \piencodf{ \tau } $}
            \end{prooftree}

        \item Rule $\redlab{FS{:}abs \dash sh}$:
        
        Then $\expr{M} = \lambda x . (M[\widetilde{x} \leftarrow x])$, and the derivation is:
        
        \begin{prooftree}
            \AxiomC{\( \core{\Delta} , x_1: \sigma, \cdots, x_k: \sigma \wfdash M : \tau  \)}
            \LeftLabel{ \redlab{FS{:}share}}
            \UnaryInfC{\( \core{\Delta} , x: \sigma \wedge \cdots \wedge \sigma \wfdash M[x_1 , \cdots , x_k \leftarrow x] : \tau \quad x\notin \core{\Delta}\)}
            \LeftLabel{\redlab{FS{:}abs \dash sh}}
            \UnaryInfC{\( \core{\Delta} \wfdash \lambda x . (M[\widetilde{x} \leftarrow x]) : \sigma^k  \rightarrow \tau \)}
        \end{prooftree}

\noindent To simplify the proof we will consider $k=2$ ( $k>2$ follows similarly). 

        By the IH, we have 
        \[\piencodf{M}_u\vdash  \piencodf{\core{\Delta} , x_1:\sigma, x_2:\sigma }, u:\piencodf{\tau}.\]
        
        From 
        \defref{def:enc_lamrsharpifail} and \defref{def:enc_sestypfail}, it follows that 

        \[
            \begin{aligned}
            \piencodf{ \core{\Delta} , x_1: \sigma, x_2: \sigma } &= \piencodf{\core{\Delta}}, x_1:\with\overline{\piencodf{\sigma}},  x_2:\with\overline{\piencodf{\sigma}}\\[3mm]
                \piencodf{\lambda x.M[x_1, x_2 \leftarrow x]}_u &= u.\overline{\some}; u(x).\piencodf{M[x_1, x_2 \leftarrow x]}_u \\
                & =                     u.\overline{\some}; u(x). x.\overline{\some}. \outact{x}{y_1}. (y_1 . \some_{\emptyset} ;y_{1}.\close;\zero\\
                &\quad \mid x.\overline{\some};x.\some_{u , (\lfv{M} \setminus x_1 , x_2 )};x(x_1). \textcolor{red}{x.\overline{\some}.}\\
                & \quad \textcolor{red}{\outact{x}{y_2} . (y_2 . \some_{\emptyset} ; y_{2}.\close;\zero \mid x.\overline{\some};x.\some_{u, (\lfv{M} \setminus x_2 )};} \\
                &\quad\textcolor{red}{ x(x_2).}  \textcolor{blue}{x.\overline{\some}; \outact{x}{y_{3}}. ( y_{3} . \some_{u,  \lfv{M} } ;y_{3}.\close; \piencodf{M}_u} \\
                &\quad \textcolor{blue}{\mid x.\overline{\none} ) ) )}
            \end{aligned}
            \]

We shall split the expression into three parts:
\[
\begin{aligned}
   \textcolor{blue}{ N_1} &= x.\overline{\some}; \outact{x}{y_{3}}. ( y_{3} . \some_{u,  \lfv{M} } ;y_{3}.\close; \piencodf{M}_u \mid x.\overline{\none} )\\
    \textcolor{red}{N_2} &= x.\overline{\some}. \outact{x}{y_2} . (y_2 . \some_{\emptyset} ; y_{2}.\close;\zero \mid x.\overline{\some};x.\some_{u, (\lfv{M} \setminus x_2 )};\\
    & \qquad x(x_2) . N_1)\\
    N_3 &= u.\overline{\some}; u(x). x.\overline{\some}. \outact{x}{y_1}. (y_1 . \some_{\emptyset} ;y_{1}.\close;\zero \mid x.\overline{\some};\\
    &\qquad x.\some_{u , (\lfv{M} \setminus x_1 , x_2 )};x(x_1) .N_2)
\end{aligned}
\]

and we obtain the  derivation for term $N_1$ as follows:

\hspace*{-50pt}
\begin{minipage}{\linewidth}
\begin{prooftree}
\small
        \AxiomC{$\piencodf{M}_u \vdash \piencodf{ \core{\Delta} , x_1: \sigma, x_2: \sigma }, u:\piencodf{\tau} $}
        \LeftLabel{\redlab{T\bot}}
        \UnaryInfC{$y_{3}.\close; \piencodf{M}_u \vdash \piencodf{ \core{\Delta} , x_1: \sigma, x_2: \sigma }, u:\piencodf{f\tau}, y_{3}{:}\bot$}
        \LeftLabel{\redlab{T\oplus^x_{\widetilde{w}}}}
        \UnaryInfC{$y_{3} . \some_{u,  \lfv{M} } ;y_{3}.\close; \piencodf{M}_u \vdash \piencodf{ \core{\Delta} , x_1: \sigma, x_2: \sigma }, u:\piencodf{\tau}, y_{3}{:}\oplus \bot $}
        \AxiomC{}
        \LeftLabel{\redlab{T\with^x}}
        \UnaryInfC{$x.\dual{\none} \vdash x :\with A$}
    \LeftLabel{\redlab{T\otimes}}
    \BinaryInfC{$ \outact{x}{y_{3}}. ( y_{3} . \some_{u,  \lfv{M} } ;y_{3}.\close; \piencodf{M}_u \mid x.\overline{\none} ) \vdash \piencodf{ \core{\Delta} , x_1: \sigma, x_2: \sigma }, u:\piencodf{\tau} , x: ( \oplus \bot )\otimes ( \with A ) $}
    \LeftLabel{\redlab{T\with_d^x}}
    \UnaryInfC{$\underbrace{x.\dual{\some}; \outact{x}{y_{3}}. ( y_{3} . \some_{u,  \lfv{M} } ;y_{3}.\close; \piencodf{M}_u \mid x.\overline{\none} )}_{N_1} \vdash \piencodf{ \core{\Delta} , x_1: \sigma, x_2: \sigma } , u:\piencodf{\tau} , x: \overline{\piencodf{\omega}_{(\sigma, i)}} $}
\end{prooftree}
\end{minipage}

Notice that the last rule applied \redlab{T\with_d^x} assigns $x: \with ((\oplus \bot) \otimes (\with A))$. Again, since $A$ is arbitrary,  take $A= \oplus (( \with  \overline{\piencodf{ \sigma }} )  \ampy (\overline{\piencodf{\omega}_{(\sigma, i - 1)}}))$, obtaining $x:\overline{\piencodf{\omega}_{(\sigma,i)}}$.
In order to obtain a type derivation for $N_2$, consider the derivation $\Pi_1$:

\hspace*{-50pt}
\begin{minipage}{\linewidth}
\begin{prooftree}
\small
        \AxiomC{$N_1 \vdash \piencodf{\core{\Delta}}, x_1:\with\overline{\piencodf{\sigma}},  x_2:\with\overline{\piencodf{\sigma}} , u:\piencodf{\tau}, x: \overline{\piencodf{\omega}_{(\sigma, i)}} $}
        \LeftLabel{\redlab{T\ampy}}
        \UnaryInfC{$x(x_2) . N_1  \vdash \piencodf{\core{\Delta}}, x_1:\with\overline{\piencodf{\sigma}},  u:\piencodf{\tau}, x: ( \with\overline{\piencodf{\sigma}} ) \ampy (\overline{\piencodf{\omega}_{(\sigma, i)}}) $}
        \LeftLabel{\redlab{T\oplus^x_{\widetilde{w}}}}
        \UnaryInfC{$ x.\some_{u, (\lfv{M} \setminus x_2 )};x(x_2) . N_1 \vdash \piencodf{\core{\Delta}}, x_1:\with\overline{\piencodf{\sigma}},  u:\piencodf{\tau}, x{:}\oplus (( \with\overline{\piencodf{\sigma}} ) \ampy (\overline{\piencodf{\omega}_{(\sigma, i)}}))$}
        \LeftLabel{\redlab{T\with_d^x}}
        \UnaryInfC{$ x.\overline{\some};x.\some_{u, (\lfv{M} \setminus x_2 )};x(x_2) . N_1  \vdash \piencodf{\core{\Delta}}, x_1:\with\overline{\piencodf{\sigma}},  u:\piencodf{\tau} , x :\with \oplus (( \with\overline{\piencodf{\sigma}} ) \ampy ( \overline{\piencodf{\omega}_{(\sigma, i)}} ))$}
\end{prooftree}
\end{minipage}

We take $ P_1 = x.\overline{\some};x.\some_{u, (\lfv{M} \setminus x_2 )};x(x_2) . N_1$ and $\core{\Gamma_1} =   \piencodf{\core{\Delta}}, x_1:\with\overline{\piencodf{\sigma}},  u:\piencodf{\tau} $ and continue the derivation of $ N_2 $

{\small 
\begin{prooftree}
        \AxiomC{}
        \LeftLabel{\redlab{T\cdot}}
        \UnaryInfC{$\zero\vdash $}
        \LeftLabel{\redlab{T\bot}}
        \UnaryInfC{$ y_{2}. \close;\zero \vdash y_{2} : \bot  $}
        \LeftLabel{\redlab{T\oplus^x_{\widetilde{w}}}}
        \UnaryInfC{$ y_2 . \some_{\emptyset} ; y_{2}.\close;\zero \vdash  y_2{:}\oplus \bot$}
        
        \AxiomC{$\Pi_1 $}
        \noLine
        \UnaryInfC{$\vdots$}
        \noLine
        \UnaryInfC{$P_1\vdash \core{\Gamma_1}, x:\with \oplus (( \with\overline{\piencodf{\sigma}} ) \ampy ( \overline{\piencodf{\omega}_{(\sigma, i)}} ))$}
    \LeftLabel{\redlab{T\otimes}}
    \BinaryInfC{$\outact{x}{y_2} . (y_2 . \some_{\emptyset} ; y_{2}.\close;\zero \mid P_1) \vdash \core{\Gamma_1} ,  x: (\oplus \bot)\otimes (\with \oplus (( \with\overline{\piencodf{\sigma}} ) \ampy ( \overline{\piencodf{\omega}_{(\sigma, i)}} )) ) $}
    \LeftLabel{\redlab{T\with_d^x}}
    \UnaryInfC{$\underbrace{x.\overline{\some}. \outact{x}{y_2} . (y_2 . \some_{\emptyset} ; y_{2}.\close;\zero \mid P_1)}_{N_2} \vdash \core{\Gamma_1} , x : \overline{ \piencodf{\sigma \wedge \omega}_{(\sigma, i)}} $}
\end{prooftree}
}

Finally, we type $N_3$ by first having the derivation $\Pi_2 $:

{\small 
\begin{prooftree} 
        \AxiomC{$ N_2 \vdash \piencodf{\core{\Delta}}, x_1:\with\overline{\piencodf{\sigma}},  u:\piencodf{\tau} , x : \overline{ \piencodf{\sigma \wedge \omega}_{(\sigma, i)}} $}
        \LeftLabel{\redlab{T\ampy}}
        \UnaryInfC{$x(x_1) . N_2  \vdash \piencodf{\core{\Delta}},   u:\piencodf{\tau} , x: ( \with\overline{\piencodf{\sigma}} ) \ampy \overline{\piencodf{\sigma \wedge \omega}_{(\sigma, i)}} $}
        \LeftLabel{\redlab{T\oplus^x_{\widetilde{w}}}}
        \UnaryInfC{$ x.\some_{u , (\lfv{M} \setminus x_1 , x_2 )};x(x_1) . N_2 \vdash \piencodf{\core{\Delta}}, u:\piencodf{\tau} , x{:}\oplus ( ( \with\overline{\piencodf{\sigma}} ) \ampy \overline{\piencodf{\sigma \wedge \omega}_{(\sigma, i)}} ) $}
        \LeftLabel{\redlab{T\with_d^x}}
        \UnaryInfC{${P_2} \vdash \piencodf{\core{\Delta}}, u:\piencodf{\tau} , x :\with \oplus ( ( \with\overline{\piencodf{\sigma}} ) \ampy \overline{\piencodf{\sigma \wedge \omega}_{(\sigma, i)}} )$}
\end{prooftree}
}

We let $ P_2 = x.\overline{\some};x.\some_{u , (\lfv{M} \setminus x_1 , x_2 )};x(x_1) .N_2$ and $\core{\Gamma_2} =   \piencodf{\core{\Delta}}, u:\piencodf{\tau}  $. We continue the derivation of $N_3={u.\overline{\some}; u(x). x.\overline{\some}. \outact{x}{y_1}. (y_1 . \some_{\emptyset} ;y_{1}.\close;\zero \mid P_2 )}$:

{\small 
\begin{prooftree}
        \AxiomC{}
        \LeftLabel{\redlab{T\cdot}}
        \UnaryInfC{$\zero\vdash $}
        \LeftLabel{\redlab{T\bot}}
        \UnaryInfC{$ y_{1}. \close;\zero \vdash y_{1} : \bot  $}
        \LeftLabel{\redlab{T\oplus^x_{\widetilde{w}}}}
        \UnaryInfC{$ y_1 . \some_{\emptyset} ; y_{1}.\close;\zero \vdash  y_1{:}\oplus \bot$}
        \AxiomC{$ \Pi_2 $}
    \LeftLabel{\redlab{T\otimes}}
    \BinaryInfC{$\outact{x}{y_1}. (y_1 . \some_{\emptyset} ;y_{1}.\close;\zero \mid P_2 ) \vdash  \core{\Gamma_2}, x: (\oplus \bot) \otimes ( \with \oplus ( ( \with\overline{\piencodf{\sigma}} ) \ampy \overline{\piencodf{\sigma \wedge \omega}_{(\sigma, i)}} ) ) $}
    \LeftLabel{\redlab{T\with_d^x}}
    \UnaryInfC{$x.\overline{\some}. \outact{x}{y_1}. (y_1 . \some_{\emptyset} ;y_{1}.\close;\zero \mid P_2 ) \vdash \piencodf{\core{\Delta}}, u:\piencodf{\tau} , x : \overline{\piencodf{\sigma \wedge \sigma}_{(\sigma, i)}} $}
    \LeftLabel{\redlab{T\ampy}}
    \UnaryInfC{$u(x). x.\overline{\some}. \outact{x}{y_1}. (y_1 . \some_{\emptyset} ;y_{1}.\close;\zero \mid P_2 ) \vdash \piencodf{\core{\Delta}} , u: ( \overline{\piencodf{\sigma \wedge \sigma}_{(\sigma, i)}} ) \ampy ( \piencodf{\tau} ) $}
    \LeftLabel{\redlab{T\with_d^x}}
    \UnaryInfC{${N_3} \vdash \piencodf{\core{\Delta}} , u :\with (( \overline{ \piencodf{\sigma \wedge \sigma}_{(\sigma, i)}} ) \ampy ( \piencodf{\tau} ))$}
\end{prooftree}
}

        Since $\piencodf{\sigma\wedge \sigma\rightarrow \tau}= \with(  \overline{\piencodf{\sigma \wedge \sigma}_{(\sigma, i)}}  \ampy  \piencodf{ \tau }) $, we have proven that $\piencodf{\lambda x. M[\widetilde{x}\leftarrow x]}_u\vdash \piencodf{\core{\Delta}}, u:\piencodf{\sigma\wedge \sigma \rightarrow \tau}$  and the result follows.
        \item Rule $\redlab{FS{:}app}$:

        Then $\expr{M} = M\ B$, and the derivation is

        \begin{prooftree}
            \AxiomC{\( \core{\Gamma} \wfdash M : \sigma^{j} \rightarrow \tau \)}
            \AxiomC{\( \core{\Delta} \wfdash B : \sigma^{k} \)}
                \LeftLabel{\redlab{FS{:}app}}
            \BinaryInfC{\( \core{\Gamma}, \core{\Delta} \wfdash M\ B : \tau\)}
        \end{prooftree}

    By IH, we have both

    \begin{itemize}
        \item  $\piencodf{M}_u\vdash \piencodf{\core{\Gamma}}, u:\piencodf{\sigma^{j} \rightarrow \tau}$;
        \item and $\piencodf{B}_u\vdash \piencodf{\core{\Delta}}, u:\overline{\piencodf{\sigma^{k}}_{(\tau_2, n)}}$, for some $\tau_2$ and some $n$.
    \end{itemize}
    From the fact that $\expr{M}$ is well-formed and \defref{def:enc_lamrsharpifail} and \defref{def:enc_sestypfail}, we  have:
    
    \begin{itemize}
        \item $B=\bag{N_1 , \cdots , N_k}$;
        \item $\displaystyle{\piencodf{ M\ B }_u =  \bigoplus_{B_i \in \perm{B}} (\nu v)(\piencodf{M}_v \mid v.\some_{u , \lfv{B}} ; \outact{v}{x} . ([v \leftrightarrow u] \mid \piencodf{B_i}_x ) )} $;
        \item $\piencodf{\sigma^{j} \rightarrow \tau}=\with( \overline{\piencodf{\sigma^{j}}_{(\tau_1, m)}} \ampy \piencodf{\tau})$, for some $\tau_1$ and some $m$.
    \end{itemize}
    
    Also, since $\piencodf{B}_u\vdash \piencodf{\core{\Delta}}, u:\piencodf{\sigma^{k}}_{(\tau_2, n)}$, we have the following derivation $\Pi_i$:

\hspace*{-10pt}
\begin{minipage}{\linewidth}
  \begin{prooftree}
            \AxiomC{$\piencodf{ B_i}_x\vdash \piencodf{ \core{\Delta} }, x:{\piencodf{ \sigma^{k} }_{(\tau_2, n)}} $ }
                
            \AxiomC{\(\)}                                   
            \LeftLabel{$ \redlab{Tid}$}
            \UnaryInfC{$ [v \leftrightarrow u]                   \vdash v:  \overline{\piencodf{ \tau }} , u: \piencodf{ \tau }$}
        \LeftLabel{$\redlab{T \otimes}$}
        \BinaryInfC{$\outact{v}{x} . ([v \leftrightarrow u] \mid \piencodf{B_i}_x ) \vdash \piencodf{ \core{\Delta }}, v:\piencodf{ \sigma^{k} }_{(\tau_2, n)} \otimes \overline{ \piencodf{ \tau }} , u:\piencodf{ \tau } $}
        \LeftLabel{$\redlab{T\oplus^{v}_{w}}$}
        \UnaryInfC{$  v.\some_{u , \lfv{B}} ; \outact{v}{x} . ([v \leftrightarrow u] \mid \piencodf{B_i}_x )  \vdash\piencodf{ \core{\Delta} }, v:\oplus (\piencodf{ \sigma^{k} }_{(\tau_2, n)} \otimes  \overline{\piencodf{ \tau }}), u:\piencodf{ \tau }$} 
    \end{prooftree}
\end{minipage} 
    
Notice that 
\begin{equation*}
\begin{aligned}
    \oplus (\piencodf{ \sigma^{k} }_{(\tau_2, n)} \otimes  \overline{\piencodf{ \tau }}) &=  \overline{\piencodf{\sigma^{k} \rightarrow \tau} }
\end{aligned}
\end{equation*}

    Therefore, by one application of $\redlab{Tcut}$ we obtain the derivations $\nabla_i$, for each $B_i \in \perm{B}$:

        \begin{prooftree}
        \AxiomC{\( \piencodf{M}_{v} \vdash  \piencodf{ \core{\Gamma} }, v: \with (\overline{\piencodf{ \sigma^{j}} }_{(\tau_1,m)} \ampy ( \piencodf{ \tau }))\)}
        \AxiomC{$\Pi_i$}
        \LeftLabel{\( \redlab{Tcut} \)}    
        \BinaryInfC{$ (\nu v)( \piencodf{ M}_v \mid v.\some_{u , \lfv{B}} ; \outact{v}{x} . ([v \leftrightarrow u] \mid \piencodf{B_i}_x ) ) \vdash \piencodf{ \core{\Gamma} } ,\piencodf{ \core{\Delta} } , u: \piencodf{ \tau }$}
        \end{prooftree}
        
       In order to apply \redlab{Tcut}, we must have that $\piencodf{\sigma^{j}}_{(\tau_1, m)} = \piencodf{\sigma^{k}}_{(\tau_2, n)}$, therefore, the choice of $\tau_1,\tau_2,n$ and $m$, will consider the different possibilities for $j$ and $k$, as in Proposition~\ref{prop:app_aux}.

        We can then conclude that $\piencodf{M B}_u \vdash \piencodf{ \core{\Gamma}}, \piencodf{ \core{\Delta} }, u:\piencodf{ \tau }$:
        
\hspace*{-30pt}
\begin{minipage}{\linewidth}
        \begin{prooftree}
        \AxiomC{For each $B_i \in \perm{B} \qquad  \nabla_i$}
        \LeftLabel{$\redlab{T\with}$}
        \UnaryInfC{$\displaystyle{\bigoplus_{B_i \in \perm{B}} (\nu v)( \piencodf{ M}_v \mid v.\some_{u , \lfv{B}} ; \outact{v}{x} . ([v \leftrightarrow u] \mid \piencodf{B_i}_x ) )  } \vdash \piencodf{ \core{\Gamma}}, \piencodf{ \core{\Delta} }, u:\piencodf{ \tau }$}
        \end{prooftree}
\end{minipage}
        
        and the result follows.
        
        \item Rule $\redlab{FS{:}share}$:
        
        Then $\expr{M} = M [ x_1, \dots x_k \leftarrow x ]$ and 
        
        \begin{prooftree}
            \AxiomC{\( \core{\Delta} , x_1: \sigma, \cdots, x_k: \sigma \wfdash M : \tau \quad x\notin \core{\Delta} \quad k \not = 0\)}
            \LeftLabel{ \redlab{FS{:}share}}
            \UnaryInfC{\( \core{\Delta} , x: \sigma_{k} \wfdash M[x_1 , \cdots , x_k \leftarrow x] : \tau \)}
        \end{prooftree}
            
         The proof for this case is contained within 2(b).  
        
        \item Rule $\redlab{FS{:}weak}$:
        
        Then $\expr{M} = M[ \leftarrow x]$ and
        
        \begin{prooftree}
            \AxiomC{\( \core{\Gamma}  \wfdash M : \tau\)}
            \LeftLabel{ \redlab{FS{:}weak}}
            \UnaryInfC{\( \core{\Gamma} , x: \omega \wfdash M[\leftarrow x]: \tau \)}
        \end{prooftree}

        However $\core{\Gamma} , x: \omega$ is not a core context hence we disallow the case.

%By IH, $\piencodf{M}_u\vdash \piencodf{\Gamma},u:\piencodf{\tau}$ and we have the following derivation:

%        \begin{prooftree}
%                \AxiomC{$\piencodf{M}_u \vdash \piencodf{\Gamma},u:\piencodf{\tau}$}
%                \LeftLabel{\redlab{T\bot}}
%                \UnaryInfC{$y_{i}.\close; \piencodf{M}_u  \vdash \piencodf{\Gamma},u:\piencodf{\tau}, y_{i}{:}\bot $}
%                \LeftLabel{\redlab{T\oplus^x_{\widetilde{w}}}}
%                \UnaryInfC{$ y_i . \some_{u,\lfv{M}} ;y_{i}.\close; \piencodf{M}_u \vdash \piencodf{\Gamma},u:\piencodf{\tau}, y_{i}{:}\oplus \bot$}
                
%                \AxiomC{$x. \overline{\none}  \vdash  x :\with A $}
                
 %           \LeftLabel{\redlab{T\otimes}}
 %           \BinaryInfC{ $\outact{x}{y_i} . ( y_i . \some_{u,\lfv{M}} ;y_{i}.\close; \piencodf{M}_u \mid x. \overline{\none})  \vdash  \piencodf{\Gamma},u:\piencodf{\tau}, x: ( \oplus \bot ) \otimes ( \with A )$}
 %           \LeftLabel{\redlab{T\with_d^x}}
 %           \UnaryInfC{$\piencodf{M[ \leftarrow x]}_u\vdash \piencodf{\Gamma},u:\piencodf{\tau}, x :\with (( \oplus \bot ) \otimes ( \with A ))$}
 %       \end{prooftree}

 %      By \defref{def:enc_lamrsharpifail}:
 %      \begin{enumerate}
 %          \item $\piencodf{M[ \leftarrow x]}_u =x. \overline{\some}. \outact{x}{y_i} . ( y_i . \some_{u,\lfv{M}} ;y_{i}.\close; \piencodf{M}_u \mid x. \overline{\none}) $
 %          \item and $\piencodf{\Gamma,x:\omega}=\piencodf{\Gamma}, x: \overline{\piencodf{\omega}_{(\sigma, i)}}=\piencodf{\Gamma}, x:\with (( \oplus \bot ) \otimes ( \with A ))$.
 %      \end{enumerate}
       
        \item Rule $\redlab{FS{:}ex \dash sub}$:
        
        Then $\expr{M} = M[x_1, \cdots , x_k \leftarrow x]\ \esubst{ B }{ x }$ and
        
        \begin{prooftree}
                \AxiomC{\( \core{\Delta} \wfdash B : \sigma^{j} \)}
                
                \AxiomC{\( \core{\Gamma} , x:\sigma^{k} \wfdash  M[x_1, \cdots , x_k \leftarrow x] : \tau \)}
            \LeftLabel{\redlab{FS{:}ex \dash sub}}    
            \BinaryInfC{\( \core{\Gamma}, \core{\Delta} \wfdash M[x_1, \cdots , x_k \leftarrow x]\ \esubst{ B }{ x } : \tau \)}
        \end{prooftree}
    % As we have done in the case of Rule $\redlab{FS{:}app}$ we also apply the following:
    % \begin{myEnumerate}
    %     \item when $j > k$ then in the case of $\piencodf{\sigma^{j}}_{(\tau, m)}$ we take $\tau $ to be an arbitrary type $\tau$ and $m = 0$ and in the case of $\piencodf{\sigma^{k}}_{(\tau, n)}$ we take $\tau $ to be $\sigma$ and $n = j-k$
        
    %     \item when $j < k$ then in the case of $\piencodf{\sigma^{k}}_{(\tau, n)}$ we take $\tau $ to be an arbitrary type $\tau$ and $n = 0$ and in the case of $\piencodf{\sigma^{j}}_{(\tau, m)}$ we take $\tau $ to be $\sigma$ and $m = k-j$
        
    %     \item when $j = k$ then in the case of both $\piencodf{\sigma^{j}}_{(\tau, m)}$ and $\piencodf{\sigma^{k}}_{(\tau, n)}$ we take $\tau $ to be an arbitrary type $\tau$ and $m = n = 0$
    % \end{myEnumerate}
    
    By Proposition~\ref{prop:app_aux} and IH we have both 
    \[
    \begin{aligned} 
    \piencodf{ M[x_1, \cdots , x_k \leftarrow x]}_u&\vdash \piencodf{\core{\Gamma}}, x: \overline{ \piencodf{ \sigma_k }_{(\tau, n)}} , u:\piencodf{\tau}\\
    \piencodf{B}_x&\vdash \piencodf{\core{\Delta}}, x:\piencodf{ \sigma_j }_{(\tau, m)}
    \end{aligned}
    \]
    From \defref{def:enc_lamrsharpifail}, we have 
    \begin{equation*}
        \begin{aligned}
       \piencodf{ M[\widetilde{x} \leftarrow x]\ \esubst{ B }{ x }}_u&= \bigoplus_{B_i \in \perm{B}} (\nu x)( \piencodf{ M[\widetilde{x} \leftarrow x]}_u \mid \piencodf{ B_i}_x )  \\
        \end{aligned}
    \end{equation*}

Therefore, for each $B_i \in \perm{B} $, we obtain the following derivation $\Pi_i$:
        
\begin{prooftree}
            \AxiomC{$\piencodf{ M[\widetilde{x} \leftarrow x]}_u  
            \vdash \piencodf{ \core{\Gamma} }  , x:  \overline{ \piencodf{ \sigma_k }_{(\tau, n)}}, u:  \piencodf{ \tau }$}
            \AxiomC{$    \piencodf{ B_i}_x \vdash \piencodf{ \core{\Delta} }, x: \piencodf{ \sigma_j }_{(\tau, m)}$}
        \LeftLabel{$\redlab{Tcut}$}
        \BinaryInfC{$ (\nu x)( \piencodf{ M[\widetilde{x} \leftarrow x]}_u \mid \piencodf{ B_i}_x )  \vdash \piencodf{ \core{\Gamma}} , \piencodf{ \core{\Delta} } , u: \piencodf{ \tau } $}               
        \end{prooftree}
        We must have that $\piencodf{\sigma^{j}}_{(\tau, m)} = \piencodf{\sigma^{k}}_{(\tau, n)}$ which holds by the conditions in Proposition \ref{prop:app_aux}.
        Therefore, from $\Pi_i$ and multiple applications of $\redlab{T\with}$ it follows that
        
        \begin{prooftree}
                    \AxiomC{$\forall \bigoplus_{B_i \in \perm{B}} \hspace{1cm} \Pi_i$}
                    \LeftLabel{$\redlab{T\with}$}
        \UnaryInfC{$ \bigoplus_{B_i \in \perm{B}} (\nu x)( \piencodf{ M[\widetilde{x} \leftarrow x]}_u \mid x.\some_w; \piencodf{ B_i}_x )  \vdash\piencodf{ \core{\Gamma} } , \piencodf{ \core{\Delta} } , u: \piencodf{ \tau }$}
        \end{prooftree}
        that is, $\piencodf{M[x_1,x_2\leftarrow x]\esubst{B}{x}}\vdash \piencodf{\core{\Gamma}, \core{\Delta}}, u:\piencodf{\tau}$ and the result follows.
        
        \item Rule $\redlab{FS{:}ex \dash lin \dash sub}$:

        Then $\expr{M} = M \linexsub{N / x}$ and
        
        \begin{prooftree}
        \AxiomC{\( \core{\Delta} \wfdash N : \sigma \)}
        \AxiomC{\( \core{\Gamma}  , x:\sigma \wfdash M : \tau \)}
            \LeftLabel{\redlab{FS{:}ex \dash lin \dash sub}}
        \BinaryInfC{\( \core{\Gamma}, \core{\Delta} \wfdash M \linexsub{N / x} : \tau \)}
        \end{prooftree}
    
    By IH we have both 
    \[
    \begin{aligned}
        \piencodf{N}_x&\vdash \piencodf{\core{\Delta}}, x: \piencodf{\sigma}\\
        \piencodf{M}_x&\vdash \piencodf{\core{\Gamma}}, x: \with\overline{\piencodf{\sigma}}, u:\piencodf{\tau}.
    \end{aligned}
    \]
    
    From \defref{def:enc_lamrsharpifail}, $\piencodf{M \linexsub{N / x} }_u=(\nu x) ( \piencodf{ M }_u \mid   x.\some_{\lfv{N}};\piencodf{ N }_x  )$ and 
    
\hspace*{-30pt}
\begin{minipage}{\linewidth}
    \begin{prooftree}
        \AxiomC{\( \piencodf{ M }_u \vdash  \piencodf{ \core{\Gamma} } , u :  \piencodf{ \tau } , x : \with \overline{\piencodf{ \sigma }}\)}
        \AxiomC{\( \piencodf{ N }_x  \vdash  \piencodf{ \core{\Delta} } , x : \piencodf{ \sigma } \)}
                \LeftLabel{$\redlab{T\oplus^x}$}
        \UnaryInfC{\( x.\some_{\lfv{N}};\piencodf{ N }_x \vdash  \piencodf{ \core{\Delta} } , x : \oplus \piencodf{ \sigma }\)}
        \LeftLabel{$\redlab{TCut}$}
        \BinaryInfC{\((\nu x) ( \piencodf{ M }_u \mid   x.\some_{\lfv{N}};\piencodf{ N }_x  )  \vdash  \piencodf{ \core{\Gamma}} , \piencodf{ \core{\Delta} } , u : \piencodf{ \tau }  \)}
        \end{prooftree}
\end{minipage}

        Observe that for the application of Rule~$\redlab{TCut}$ we used the fact that $\overline{\oplus\piencodf{\sigma}}=\with \overline{\piencodf{\sigma}}$. Therefore, $\piencodf{M \linexsub{N / x} }_u\vdash \piencodf{ \core{\Gamma}} , \piencodf{ \core{\Delta} } , u : \piencodf{ \tau } $ and the result follows.

        \item Rule $\redlab{FS{:}fail}$:
        
        Then $\expr{M} = M \linexsub{N / x}$ and
        
            \begin{prooftree}
                \AxiomC{\( \core{(x_1:\sigma_1, \cdots , x_n:\sigma_n)} = x_1:\sigma_1, \cdots , x_n:\sigma_n  \)}
                \LeftLabel{\redlab{FS{:}fail}}
                \UnaryInfC{\( x_1:\sigma_1, \cdots , x_n:\sigma_n  \wfdash  \fail^{x_1, \cdots , x_n} : \tau \)}
            \end{prooftree}
        
        From Definition \ref{def:enc_lamrsharpifail}, $\piencodf{\fail^{x_1, \cdots , x_n} }_u= u.\overline{\none} \mid x_1.\overline{\none} \mid \cdots \mid x_k.\overline{\none} $ and 
        
    {\small 
        \begin{prooftree}
                \AxiomC{}
                \LeftLabel{\redlab{T\with^u}}
                \UnaryInfC{$u.\overline{\none} \vdash u : \piencodf{ \tau } $}

                    \AxiomC{}
                    \LeftLabel{\redlab{T\with^{x_1}}}
                    \UnaryInfC{$x_1.\overline{\none} \vdash_1 : \with \overline{\piencodf{\sigma_1}} $}
                    
                    \AxiomC{}
                    \LeftLabel{\redlab{T\with^{x_n}}}
                    \UnaryInfC{$x_n.\overline{\none} \vdash x_n : \with \overline{\piencodf{\sigma_n}} $}
                    \UnaryInfC{$\vdots$}
                \BinaryInfC{$x_1.\overline{\none} \mid \cdots \mid x_k.\overline{\none} \vdash  x_1 : \with \overline{\piencodf{\sigma_1}}, \cdots  ,x_n : \with \overline{\piencodf{\sigma_n}}$}
            \LeftLabel{\redlab{T\mid}}
            \BinaryInfC{$u.\overline{\none} \mid x_1.\overline{\none} \mid \cdots \mid x_k.\overline{\none} \vdash x_1 : \with \overline{\piencodf{\sigma_1}}, \cdots  ,x_n : \with \overline{\piencodf{\sigma_n}}, u : \piencodf{ \tau }$}
        \end{prooftree}
        }
        Therefore, $\piencodf{\fail^{x_1, \cdots , x_n} }_u\vdash  x_1 : \with \overline{\piencodf{\sigma_1}}, \cdots  ,x_n : \with \overline{\piencodf{\sigma_n}}, u : \piencodf{ \tau } $ and the result follows.

        \item Rule $\redlab{FS{:}sum}$: 
        
        This case follows easily by IH. \qedhere
    \end{enumerate}
    \end{enumerate}
%\qed
\end{proof}

\subsection{Completeness and Soundness}
\label{compandsucctwo}

%\NEEDTONAME*

\consistnequiv*

\begin{proof}
By induction on the structure of $\expr{M}$. Let us consider first two conditions 1 and 2 as other conditions are analogous. The congruence  rules that concern the sharing construct of condition 1 are:

        \[
            \begin{array}{rll}
                \!\!\!M [ \leftarrow x] \esubst{\oneb}{x} \!\!\!\! & \pequiv M &
                \\
                \!\!\!MA[\widetilde{x} \leftarrow x]\esubst{B}{x} 
                         \!\!\!\!&\pequiv (M[\widetilde{x} \leftarrow x]\esubst{B}{x})A
                         &  \!\text{with } x_i \in \widetilde{x} \Rightarrow x_i \not \in \lfv{A}
                \\
                \!\!\!M[\widetilde{y} \leftarrow y]\esubst{A}{y}[\widetilde{x} \leftarrow x]\esubst{B}{x} \!\!\!\! & \pequiv
                (M[\widetilde{x} \leftarrow x]\esubst{B}{x})[\widetilde{y} \leftarrow y]\esubst{A}{y} &   \!\text{with } x_i \in \widetilde{x} \Rightarrow x_i \not \in \lfv{A}
                \\
            \end{array}
        \]
        Notice that these rules neither add or remove occurrences of shared variables neither do they allow shared variables to be extruded from their bindings by their side conditions. Also, they do not introduce new sharing on already shared variables. 
        Hence, conditions 1(i) to 1(iv) are preserved by these rules.
        
        Now consider the congruence rules concerning the explicit substitution of condition 2:
        \[
            \begin{array}{rll}
            MB \linexsub{N/x}  \!\!\!\!&\pequiv (M\linexsub{N/x})B &  \text{with } x \not \in \lfv{B} 
            \\
            M \linexsub{N_2/y}\linexsub{N_1/x} 
                     \!\!\!\!&\pequiv M\linexsub{N_1/x}\linexsub{N_2/y} &
                     \text{with } x \not \in \lfv{N_2},\revdaniele{ y \notin \lfv{N_1}}
            \\
            \end{array}
        \]
        As before, variables are not duplicated or eliminated from terms and by the side conditions of the rules they cannot extrude bound variables. Similarly, the rules do not introduce any sharing or new free variables.
        Hence conditions 2(i) to 2(iv) are satisfied.
\end{proof}

\encodingreduces*

\begin{proof}
%\revd{B28}{We prove for the broader statement for where $N$ has all sharing variables bound rather then being closed and we extend the result to closed terms.}
Let us consider each part:
\begin{myEnumerate}

    \item We proceed by induction on the structure of $N$.

    \begin{myEnumerate}
    
        \item  $N = x$.
        
        Then $\piencodf{x}_u$. Hence $ I = \emptyset$ and $ \widetilde{y} = \emptyset$.
    
        \item $N = (M\ B)$. 
        
        Then $\headf{M\ B} = \headf{M} = x$ and 
        \[ \piencodf{N}_u = \piencodf{M\ B}_u  = \bigoplus_{B_i \in \perm{B}} (\nu v)(\piencodf{M}_v \mid v.\some_{u, \lfv{B}} ; \outact{v}{x} . ([v \leftrightarrow u] \mid \piencodf{B_i}_x ) ) \]
        and the result follows by induction on $\piencodf{M}_u$.
        
        \item $N = M[\widetilde{y} \leftarrow y]$. Not possible due to the assumption of partially open terms.
        
        \revd{B29}{ 
        \item $N = (M[\widetilde{y} \leftarrow y])\esubst{ B }{ y }$. \\
        Then $\headf{(M[\widetilde{y} \leftarrow y])\esubst{ B }{ y }} = \headf{(M[\widetilde{y} \leftarrow y])} = x$ when $\widetilde{y} = \emptyset,\ B = \oneb $ and $\headf{M} = x$ .
        \[
        \begin{aligned}
           \piencodf{N}_u &= \piencodf{(M[ \leftarrow y])\esubst{ \oneb }{ y }}_u = (\nu y)( \piencodf{ M[ \leftarrow y]}_u \mid  \piencodf{\oneb}_y) 
           \\
           &=(\nu y)( y. \overline{\some}. \outact{y}{z} . ( z . \some_{u,\lfv{M}} ;z_{}.\close; \piencodf{M}_u \mid y. \overline{\none})  \mid \\
           & \qquad y.\some_{\emptyset} ; y(z). (z.\overline{\some};z. \overline{\close} \mid y.\some_{\emptyset} ; y. \overline{\none})) \\
           & \red^*  \piencodf{M}_u \\
        \end{aligned}
        \]
        Then the result follows by induction on $\piencodf{ M }_u $.
        }

        \item When $N = M \linexsub {N' /y}$, then $\headf{M \linexsub {N' /y}} = \headf{M } = x$ and 
        \[
        \begin{aligned}
           \piencodf{N}_u & = \piencodf{M \linexsub {N' /y}}_u 
            =  (\nu y) ( \piencodf{ M }_u \mid   x.\some_{\lfv{N'}};\piencodf{ N' }_x  ) 
        \end{aligned}
        \]
        Then true by induction on $\piencodf{ M }_u $

    \end{myEnumerate}
    %\qed

    \item In this case, notice how reductions are only introduced when $N$ has sub-term $(M[ \leftarrow y])\esubst{ \oneb }{ y }$ from case 1(IV), however from the congruence of \figref{fig:rsPrecongruencefailure} we may rewrite this sub-term to be $M$ which eliminates the need for reductions. Inductively, performing this application of $\pequiv$ provides the result.

    \item This case is similar to the first, with the clear difference that linear head substitution must also be used. However, we can inductively push the linear head substitution inside the term to reach the head variable. Consider the base case when $N = x$ and we have some well-formed partially open term $M$. Then $\piencodf{N\headlin{M/x}}_u = \piencodf{x\headlin{M/x}}_u = \piencodf{M}_u$. Hence $ I = \emptyset$ and $ \widetilde{y} = \emptyset$ matching that of case 1(i).

    Next, let us consider the case of $N\headlin{M/x} = M' \linexsub {N' /y}\headlin{M/x} = M'\headlin{M/x} \linexsub {N' /y}$. By considering 1(V) we can see the evaluating the translation of creates the same process shape up to linear head substitution. Other cases follow analogously.

    \item This is a consequence of both (2) and (3). \qedhere
    
\end{myEnumerate}

\end{proof}

\revd{B46}{
\begin{nota}
    We use the notation $\lfv{M}.\overline{\none}$ and $\widetilde{x}.\overline{\none}$ where $\lfv{M}$ or $\widetilde{x}$ are equal to $ x_1 , \cdots , x_k$ to describe a process of the form $x_1.\overline{\none} \mid \cdots \mid x_k.\overline{\none} $
\end{nota}
}

\opcomplete*

\begin{proof}
%\revd{B28}{We prove for the broader statement for where $\expr{N}$ has all sharing variables bound rather then being closed and we extend the result to closed terms.}
By induction on the reduction rule applied to infer $\expr{N}\red \expr{M}$.  
We have five cases.

    \begin{enumerate}
        \item  Case $\redlab{RS{:}Beta}$: 
              
               Then  $ \expr{N}= (\lambda x. M[\widetilde{x} \leftarrow x]) B \red M[\widetilde{x} \leftarrow x]\ \esubst{ B }{ x }=\expr{M}$.
               \\
        On the one hand, we have:
                \begin{equation}\label{eq:compl_lsbeta1fail}
        \begin{aligned}
        \piencodf{\revd{B43}{\expr{N}}}_u &= \piencodf{(\lambda x. M[\widetilde{x} \leftarrow x]) B}_u\\
        &=  \bigoplus_{B_i \in \perm{B}} (\nu v)( \piencodf{ \lambda x. M[\widetilde{x} \leftarrow x]}_v \mid v.\some_{u,\lfv{B}} ; \outact{v}{x} . ( \piencodf{ B_i}_x \mid [v \leftrightarrow u] ) )\\
        &=  \bigoplus_{B_i \in \perm{B}} (\nu v)( v.\overline{\some}; v(x).\piencodf{M[\widetilde{x} \leftarrow x]}_v \mid v.\some_{u,\lfv{B}} ; \outact{v}{x} . ( \piencodf{ B_i}_x \mid [v \leftrightarrow u] ) )\\
        & \red  \bigoplus_{B_i \in \perm{B}} (\nu v)(  v(x).\piencodf{M[\widetilde{x} \leftarrow x]}_v \mid  \outact{v}{x} . ( \piencodf{ B_i}_x \mid [v \leftrightarrow u] ) )\\
        & \red  \bigoplus_{B_i \in \perm{B}} (\nu v, x)(  \piencodf{M[\widetilde{x} \leftarrow x]}_v \mid   \piencodf{ B_i}_x \mid [v \leftrightarrow u] ) \\
        & \red  \bigoplus_{B_i \in \perm{B}} (\nu x)(  \piencodf{M[\widetilde{x} \leftarrow x]}_u \mid   \piencodf{ B_i}_x  ) \\
        \end{aligned}
        \end{equation}
        
        On the other hand, we have:
        \begin{equation}\label{eq:compl_lsbeta2fail}
            \begin{aligned}
               \piencodf{\expr{M}}_u &= \piencodf{M[\widetilde{x} \leftarrow x]\ \esubst{ B }{ x }}_u = \bigoplus_{B_i \in \perm{B}} (\nu x) (\piencodf{ M[\widetilde{x} \leftarrow x] }_u \mid  \piencodf{ B_i}_x ) \\
            \end{aligned}
        \end{equation}
        Therefore, by \eqref{eq:compl_lsbeta1fail} and \eqref{eq:compl_lsbeta2fail} the result follows.
        
        \item Case $ \redlab{RS{:}Ex \dash Sub}$: 
        
        Then $ N=M[x_1, \cdots , x_k \leftarrow x]\ \esubst{ B }{ x }$, with $B=\bag{N_1 , \ldots , N_k}$, $k\geq 1$ and $M \not= \fail^{\widetilde{y}}$. 
        The reduction is $$\expr{N} = M[x_1, \cdots , x_k \leftarrow x]\ \esubst{ B }{ x } \red \sum_{B_i \in \perm{B}}M\ \linexsub{B_i(1)/x_1} \cdots \linexsub{B_i(k)/x_k} = \expr{M}.$$
        
        We detail the encodings of $\piencodf{\expr{N}}_u$ and $\piencodf{\expr{M}}_u$. To simplify the proof, we will consider $k=1$ (the case   $k> 1$ follows analogously). 
        \\
        On the one hand, we have:
        \begin{equation}\label{eq:compl_lsbeta3fail}
        \begin{aligned}
        \piencodf{\expr{N}}_u &= \piencodf{M[x_1 \leftarrow x]\ \esubst{ B }{ x }}_u =  \bigoplus_{B_i \in \perm{B}} (\nu x)( \piencodf{ M[x_1\leftarrow x]}_u \mid  \piencodf{ B_i}_x ) 
        \\
        &=  \bigoplus_{B_i \in \perm{B}} (\nu x)( x.\overline{\some}. \outact{x}{y_1}. (y_1 . \some_{\emptyset} ;y_{1}.\close;\zero \mid x.\overline{\some};x.\some_{u, (\lfv{M} \setminus x_1 )}; \\
                & \hspace{1cm} x(x_1). x.\overline{\some}; \outact{x}{y_{2}}. ( y_{2} . \some_{u,\lfv{M}} ;y_{2}.\close; \piencodf{M}_u \mid x.\overline{\none} ) ) \mid \\
                &\hspace{1cm} x.\some_{\lfv{B_i(1)}} ; x(y_1). x.\some_{y_1,\lfv{B_i(1)}};x.\overline{\some} ; \outact{x}{x_1}. (x_1.\some_{\lfv{B_i(1)}};\\
                &\hspace{1cm}\piencodf{B_i(1)}_{x_1} \mid y_1. \overline{\none} \mid x.\some_{\emptyset} ; x(y_2). ( y_2.\overline{\some};y_2 . \overline{\close} \mid x.\some_{\emptyset} ; x. \overline{\none}) ) )
        \\
        % & \text{\daniele{I think there is a piece of the process missing from one line to the other}\joe{there was a y2.close missing. the x.none never synchronises as the x.some fails before it can.}}\\
        & \red^* \bigoplus_{B_i \in \perm{B}} (\nu x, y_1,x_1,y_2)( 
                  y_1 . \some_{\emptyset} ;y_{1}.\close;\zero \mid y_1. \overline{\none} \mid  y_{2} . \some_{u,\lfv{M}};y_{2}.\close;\\
                & \hspace{1cm}  \piencodf{M}_u \mid y_2.\overline{\some};y_2 . \overline{\close} \mid x.\overline{\none}   \mid  x.\some_{\emptyset} ; x. \overline{\none} \mid  x_1.\some_{\lfv{B_i(1)}} ; \piencodf{B_i(1)}_{x_1} )
        \\
        & \red^* \bigoplus_{B_i \in \perm{B}} (\nu x_1 )(  \piencodf{M}_u \mid  x_1.\some_{\lfv{B_i(1)}} ; \piencodf{B_i(1)}_{x_1} ) \\
        \end{aligned}
        \end{equation}
        
        On the other hand, we have:
        \begin{equation}\label{eq:compl_lsbeta4fail}
        \begin{aligned}
            \piencodf{\expr{M}}_u &= \piencodf{\sum_{B_i \in \perm{B}}M\ \linexsub{B_i(1)/x_1}}_u = \bigoplus_{B_i \in \perm{B}} \piencodf{M\ \linexsub{B_i(1)/x_1} }_u\\
            &= \bigoplus_{B_i \in \perm{B}} (\nu x_1 )( \piencodf{M}_{u} \mid x_1.\some_{\lfv{B_i(1)}};\piencodf{B_i(1)}_{x_1} )   \\
        \end{aligned}
        \end{equation}
        Therefore, by \eqref{eq:compl_lsbeta3fail}
        and  \eqref{eq:compl_lsbeta4fail} the result follows.
        
        \item Case $ \redlab{RS{:}Lin \dash Fetch}$: 
       
        Then we have 
        $\expr{N} = M\ \linexsub{N'/x}$ with $\headf{M} = x$ and $\expr{N} \red  M \headlin{ N' /x } = \expr{M}$.
        
       On the one hand, we have:
        \begin{equation}\label{eq:compl_lsbeta5fail}
        \begin{aligned}
        \piencodf{N}_u = \piencodf{M\ \linexsub{N'/x}}_u&= (\nu x) ( \piencodf{ M }_u \mid   x.\some_{\lfv{N'}};\piencodf{ N' }_x  )\\
       & \red^* (\nu x) ( \bigoplus_{i \in I}(\nu \widetilde{y})(\piencodf{ x }_{j} \mid P_i) \mid   x.\some_{\lfv{N'}};\piencodf{ N' }_x  ) \qquad (*)   
        \\
        &= (\nu x) ( \bigoplus_{i \in I}(\nu \widetilde{y})(\piencodf{ x }_{j} \mid P_i)  \mid   x.\some;\piencodf{ N' }_x  )  \\
        & \red (\nu x) ( \bigoplus_{i \in I}(\nu \widetilde{y})([x \leftrightarrow j ] \mid P_i) \mid   \piencodf{ N' }_x  ) \\
        & \red  \bigoplus_{i \in I}(\nu \widetilde{y}) ( P_i \mid   \piencodf{ N' }_j  )     =Q
        \end{aligned}
        \end{equation}
        
       where the reductions denoted by $(*)$ are inferred via Proposition~\ref{prop:NEEDTONAME}.
       
       \revd{B29}{On the other hand, we have by Proposition \ref{prop:NEEDTONAME} :
        \begin{equation}\label{eq:compl_lsbeta6fail}
        \begin{aligned}
        \piencodf{\expr{M}}_u &= \piencodf{M \headlin{ N'/x }}_u \red^*  \bigoplus_{i \in I}(\nu \widetilde{y}) ( P_i \mid   \piencodf{ N' }_j  )
        \end{aligned}
        \end{equation}
        We also have by Proposition \ref{prop:NEEDTONAME} and \eqref{eq:compl_lsbeta6fail}
        that there exists $M'$ such that $M' \pequiv M \headlin{ N'/x }$ with:
        \begin{equation}\label{eq:compl_lsbeta7fail}
            \piencodf{ M' }_{u} = \bigoplus_{i \in I}(\nu \widetilde{y}) ( P_i \mid   \piencodf{ N' }_j  )
        \end{equation}
        }
        
                Therefore, by \eqref{eq:compl_lsbeta5fail}
        and  \eqref{eq:compl_lsbeta7fail} the result follows.
        
        \item Case $\redlab{RS{:}TCont}$ and $\redlab{RS{:}ECont}$:
         These cases follow by IH.
         
        % \item Case $\redlab{RS{:}ECont}$:This case follows by IH.
        
        \item Case $\redlab{RS{:}Fail}$:
        
        Then, 
        $\expr{N} = M[x_1, \cdots , x_k \leftarrow x]\ \esubst{ B }{ x }$ with $k \neq \size{B}$ and
        
        $$ \expr{N} \red  \sum_{B_i \in \perm{B}}  \fail^{\widetilde{y} } = \expr{M},$$ where $\widetilde{y} = (\lfv{M} \setminus \{  x_1, \cdots , x_k \} ) \cup \lfv{B}$. 
        
        Let us assume that $k > l$ and we proceed similarly for $k > l$. Hence $k = l + m$ for some $m \geq 1$. \revd{B44}{On the one hand, we have \eqref{eq:compl_fail1-fail}, this can be seen in \figref{fig:proofreductions1}}.
        
        \begin{figure}[!t]
        \hrule 
            {\small
            \begin{equation}\label{eq:compl_fail1-fail}
            \begin{aligned}
                \piencodf{\revd{B45}{\expr{N}}}_u &= \piencodf{M[x_1, \cdots , x_k \leftarrow x]\ \esubst{ B }{ x }}_u\\
                &=  \bigoplus_{B_i \in \perm{B}} (\nu x)( \piencodf{ M[x_1, \cdots , x_k \leftarrow x]}_u \mid  \piencodf{ B_i}_x )  
                \\
                &=  \bigoplus_{B_i \in \perm{B}} (\nu x)( x.\overline{\some}. \outact{x}{y_1}. (y_1 . \some_{\emptyset} ;y_{1}.\close;\zero \mid x.\overline{\some};x.\some_{u,(\lfv{M} \setminus x_1 , \cdots , x_k )};\\
       & \hspace{1cm} x(x_1) . \cdots  x.\overline{\some}. \outact{x}{y_k} . (y_k . \some_{\emptyset} ; y_{k}.\close;\zero \mid x.\overline{\some};x.\some_{u,(\lfv{M} \setminus  x_k )}; \\
       & \hspace{1cm} x(x_k) . x.\overline{\some}; \outact{x}{y_{k+1}}. ( y_{k+1} . \some_{u,\lfv{M} } ;y_{k+1}.\close; \piencodf{M}_u \mid x.\overline{\none} )) \cdots )  \\
      &\hspace{1cm} \mid x.\some_{\lfv{B}} ; x(y_1). x.\some_{y_1,\lfv{B}};x.\overline{\some} ; \outact{x}{x_1}. (x_1.\some_{\lfv{B_i(1)}} ; \piencodf{B_i(1)}_{x_1} \\
       & \hspace{1cm} \mid y_1. \overline{\none} \mid \cdots  x.\some_{\lfv{B_i(l)}} ; x(y_l). x.\some_{y_l ,\lfv{B_i(l)}};x.\overline{\some} ; \outact{x}{x_l}. (x_l.\some_{\lfv{B_i(l)}} ;  \\
       &\hspace{1cm}  \piencodf{B_i(l)}_{x_l} \mid y_l. \overline{\none} \mid x.\some_{\emptyset} ; x(y_{l+1}). ( y_{l+1}.\overline{\some};y_{l+1} . \overline{\close} \mid x.\some_{\emptyset} ; x. \overline{\none}) ) ) )
                \\
     & \red^*  \bigoplus_{B_i \in \perm{B}} (\nu x, y_1, x_1, \cdots  y_l, x_l)(  y_1 . \some_{\emptyset} ;y_{1}.\close;\zero \mid \cdots \mid y_l . \some_{\emptyset} ;y_{l}.\close;\zero \\
             & \hspace{1cm}x.\overline{\some}. \outact{x}{y_{l+1}} . (y_{l+1} . \some_{\emptyset} ; y_{l+1}.\close;\zero \mid x.\overline{\some};x.\some_{u,(\lfv{M} \setminus x_{l+1} , \cdots , x_k )}; \\
             & \hspace{1cm} x(x_{l+1}). \cdots x.\overline{\some}. \outact{x}{y_k} . (y_k . \some_{\emptyset} ; y_{k}.\close;\zero \mid x.\overline{\some};x.\some_{u,(\lfv{M} \setminus  x_k )};x(x_k) . \\
            &\hspace{1cm} x.\overline{\some}; \outact{x}{y_{k+1}}. ( y_{k+1} . \some_{u,\lfv{M} } ;y_{k+1}.\close; \piencodf{M}_u \mid x.\overline{\none} )) \cdots ) \mid \\
         & \hspace{1cm}  x_1.\some_{\lfv{B_i(1)}} ; \piencodf{B_i(1)}_{x_1} \mid \cdots \mid  x_l.\some_{\lfv{B_i(l)}} ; \piencodf{B_i(l)}_{x_l} \mid  y_1. \overline{\none} \mid \cdots \mid y_l. \overline{\none}\\
          & \hspace{1.5cm} x.\some_{\emptyset} ; x(y_{l+1}). ( y_{l+1}.\overline{\some};y_{l+1} . \overline{\close} \mid x.\some_{\emptyset} ; x. \overline{\none}) ) 
                \\
                & \red^* 
                    \bigoplus_{B_i \in \perm{B}} (\nu x, x_1, \cdots  , x_l)(
           x.\some_{u,(\lfv{M} \setminus x_{l+1} , \cdots , x_k )};x(x_{l+1}) . \cdots \\
     & \hspace{1.5cm} x.\overline{\some}. \outact{x}{y_k} . (y_k . \some_{\emptyset} ; y_{k}.\close;\zero \mid x.\overline{\some};x.\some_{u,(\lfv{M} \setminus  x_k )};x(x_k) . \\
      &\hspace{1.5cm} x.\overline{\some}; \outact{x}{y_{k+1}}. ( y_{k+1} . \some_{u,\lfv{M}} ;y_{k+1}.\close; \piencodf{M}_u \mid x.\overline{\none} ) )  \mid \\
     &\hspace{1.5cm}   x_1.\some_{\lfv{B_i(1)}} ; \piencodf{B_i(1)}_{x_1} \mid \cdots \mid  x_l.\some_{\lfv{B_i(l)}} ; \piencodf{B_i(l)}_{x_l} \mid  x. \overline{\none} ) 
                \\
        & \red \bigoplus_{B_i \in \perm{B}} (\nu  x_1, \cdots  , x_l)(  u . \overline{\none} \mid x_1 . \overline{\none} \mid  \cdots \mid x_{l} . \overline{\none} \mid (\lfv{M} \setminus x_{1} , \cdots , x_k ) . \overline{\none}  \mid \\
        & \hspace{1.5cm} x_1.\some_{\lfv{B_i(1)}} ; \piencodf{B_i(1)}_{x_1} \mid \cdots \mid  x_l.\some_{\lfv{B_i(l)}} ; \piencodf{B_i(l)}_{x_l}  ) 
                \\
                & \red^*  \bigoplus_{B_i \in \perm{B}}  u . \overline{\none} \mid (\lfv{M} \setminus \{  x_1, \cdots , x_k \} ) \cup \lfv{B} . \overline{\none}  
            \end{aligned}
            \end{equation}
            }
            \hrule 
        \caption{Reductions of an encoded explicit substitution}
            \label{fig:proofreductions1}
        \end{figure}

        On the other hand, we have:
        
        \begin{equation}\label{eq:compl_fail2-fail}
        \begin{aligned}
            \piencodf{\expr{M}}_u &= \piencodf{\sum_{B_i \in \perm{B}}  \fail^{\widetilde{y}}}_u= \bigoplus_{B_i \in \perm{B}} \piencodf{\fail^{\widetilde{y}} }_u\\
            &= \bigoplus_{B_i \in \perm{B}} u . \overline{\none} \mid (\lfv{M} \setminus \{  x_1, \cdots , x_k \} ) \cup \lfv{B} . \overline{\none}  \\
        \end{aligned}
        \end{equation}
        
        Therefore, by \eqref{eq:compl_fail1-fail}
        and  \eqref{eq:compl_fail2-fail} the result follows.
        
        \item Case $\redlab{RS{:}Cons_1}$:
        
        Then, 
        $\expr{N} = \fail^{\widetilde{x}}\ B$ with $B =  \bag{N_1 , \dots , N_k} $ and $\expr{N} \red  \sum_{\perm{B}} \fail^{\widetilde{x} \cup \widetilde{y}} = \expr{M}$, where $ \widetilde{y} = \lfv{B}$. 
        
        On the one hand, we have: 
        
        \begin{equation}\label{eq:compl_cons1-fail}
        \begin{aligned}
            \piencodf{N}_u &= \piencodf{ \fail^{\widetilde{x}}\ B }_u\\
            &= \bigoplus_{B_i \in \perm{B}} (\nu v)(\piencodf{\fail^{\widetilde{x}}}_v \mid v.\some_{u,\lfv{B}} ; \outact{v}{x} . ([v \leftrightarrow u] \mid \piencodf{B_i}_x ) ) \\
            &= \bigoplus_{B_i \in \perm{B}} (\nu v)( v . \overline{\none} \mid \widetilde{x}. \overline{\none} \mid v.\some_{u, \lfv{B}} ; \outact{v}{x} . ([v \leftrightarrow u] \mid \piencodf{B_i}_x ) ) \\
            & \red \bigoplus_{B_i \in \perm{B}} u . \overline{\none} \mid \widetilde{x}. \overline{\none} \mid \widetilde{y}. \overline{\none} \\
            &= \bigoplus_{\perm{B}} u . \overline{\none} \mid \widetilde{x}. \overline{\none} \mid \widetilde{y}. \overline{\none} \\
        \end{aligned}
        \end{equation}
        
        On the other hand, we have:
        \begin{equation}\label{eq:compl_cons2-fail}
        \begin{aligned}
            \piencodf{\expr{M}}_u &= \piencodf{\sum_{\perm{B}} \fail^{\widetilde{x} \cup \widetilde{y}}}_u
            = \bigoplus_{\perm{B}} \piencodf{\fail^{\widetilde{x} \cup \widetilde{y}} }_u\\
            &= \bigoplus_{\perm{B}} u . \overline{\none} \mid \widetilde{x}. \overline{\none} \mid \widetilde{y}. \overline{\none} \\
        \end{aligned}
        \end{equation}
        
        Therefore, by \eqref{eq:compl_cons1-fail}
        and  \eqref{eq:compl_cons2-fail} the result follows.

        \item Cases $\redlab{RS{:}Cons_2}$ and $\redlab{RS{:}Cons_3}$: These cases follow by IH similarly to Case 7. \qedhere

        % \item Case $\redlab{RS{:}Cons_3}$: This case follows by IH similarly to Case 7.
    \end{enumerate}
%\qed
\end{proof}

\opsound*

\begin{proof}
By induction on the structure of $\expr{N} $ and then induction on the number of reductions of $\piencodf{\expr{N}} \red_{\pequiv}^* Q$

\begin{myEnumerate}
    \item  $\expr{N} = x$, $\expr{N} = \fail^{\emptyset}$ and $\expr{N} = \lambda x . (M[ \widetilde{x} \leftarrow x ])$.
    
    These cases are trivial since no reduction can take place. 
    
    % Then $\piencodf{\expr{N}}_u = \piencodf{x}_u = x.\overline{\some} ;[x \leftrightarrow u ]$ and no reductions can take place. \\ Hence, we have
    % $Q =  \piencodf{\expr{N}}_u \red^0 \piencodf{\expr{N}}_u = Q'$ and $x \red^0 x = \expr{N}'$.
    
    % \item $\expr{N} = \fail^{\emptyset}$.
    
    % Then $\piencodf{\expr{N}}_u = \piencodf{\fail^{\widetilde{x}}}_u = u.\overline{\none} \mid \widetilde{x}.\overline{\none}  $ and no reductions can take place. \\ Hence, we have $Q =  \piencodf{\expr{N}}_u \red^0 \piencodf{\expr{N}}_u = Q'$ and $\fail^{\widetilde{x}} \red^0 \fail^{\widetilde{x}} = \expr{N}'$.
    
    % \item  $\expr{N} = \lambda x . (M[ \widetilde{x} \leftarrow x ])$.

    %     Then $\piencodf{\expr{N}}_u = \piencodf{\lambda x . (M[ \widetilde{x} \leftarrow x ])}_u = u.\overline{\some};u(x).\piencodf{ M[\widetilde{x} \leftarrow x] }_u $  and no reductions can take place. \\
    %     Hence, we have  
    % $Q =  \piencodf{\expr{N}}_u \red^0 \piencodf{\expr{N}}_u = Q'$ and $\lambda x . (M[ \widetilde{x} \leftarrow x ]) \red^0 \lambda x . (M[ \widetilde{x} \leftarrow x ]) = \expr{N}'$.
    
    \item $\expr{N} =  (M\ B) $.

        Then, 
        $$ \piencodf{(M\ B)}_u = \bigoplus_{B_i \in \perm{B}} (\nu v)(\piencodf{M}_v \mid v.\some_{u,\lfv{B}} ; \outact{v}{x} . ([v \leftrightarrow u] \mid \piencodf{B_i}_x ) )  $$ and we are able to perform the  reductions from $\piencodf{(M\ B)}_u$. 

        We now proceed by induction on $k$, with  $\piencodf{\expr{N}}_u \red^k Q$. 
        
        % There are two main cases:
        
    %     \begin{myEnumerate}
    %         \item When $k = 0$ the thesis follows easily:
    %         we have 
    % $Q =  \piencodf{(M\ B)}_u \red^0 \piencodf{(M\ B)}_u = Q'$ and $(M\ B) \red^0 (M\ B) = \expr{N}'$.
    
            % \item 
            The interesting case is when $k \geq 1$ (the case $k=0$ is trivial).

            Then, for some process $R$ and $n, m$ such that $k = n+m$, we have the following:
            \[
            \begin{aligned}
               \piencodf{\expr{N}}_u & =  \bigoplus_{B_i \in \perm{B}} (\nu v)( \piencodf{ M}_v \mid v.\some_{u,\lfv{B}} ; \outact{v}{x} . (  \piencodf{ B_i}_x \mid [v \leftrightarrow u] ) )\\
               & \red^m  \bigoplus_{B_i \in \perm{B}} (\nu v)( R \mid v.\some_{u , \lfv{B}} ; \outact{v}{x} . ( \piencodf{ B_i}_x \mid [v \leftrightarrow u] ) ) \\
               &\red^n  Q\\
            \end{aligned}
            \]
            Thus, the first $m \geq 0$ reduction steps are  internal to $\piencodf{ M}_v$; type preservation in \spi ensures that, if they occur,  these reductions  do not discard the possibility of synchronizing with $v.\some$. Then, the first of the $n \geq 0$ reduction steps towards $Q$ is a synchronization between $R$ and $v.\some_{u, \lfv{B}}$.
            
            We consider two sub-cases, depending on the values of  $m$ and $n$:
            \begin{myEnumerate}
                \item When $m = 0$ and $n \geq 1$:
                    
                    Thus $R = \piencodf{\expr{M}}_v$, and there  are two possibilities of having an unguarded $v.\overline{\some}$ or $v.\overline{\none}$ without internal reductions.   By the diamond property (Proposition~\ref{prop:conf1_lamrsharfail}) we will be reducing each non-deterministic choice of a process simultaneously.
                    Then we have the following for each case:

                    \begin{myEnumerate}
                        \item $M = (\lambda x . M' [\widetilde{x} \leftarrow x]) \linexsub{N_1 / y_1} \cdots \linexsub{N_p / y_p} \qquad (p \geq 0)$.
                        
                           $$
                        \begin{aligned}
                            \piencodf{M}_v &= \piencodf{(\lambda x . M' [\widetilde{x} \leftarrow x]) \linexsub{N_1 / y_1} \cdots \linexsub{N_p / y_p}}_v \\
                            &= (\nu\widetilde{y}) ( \piencodf{(\lambda x . M' [\widetilde{x} \leftarrow x])}_v \mid y_1.\some_{\lfv{N_1}};  \piencodf{ N_1 }_{y_1} \mid \cdots \mid y_p.\some_{\lfv{N_p}};\piencodf{ N_p }_{y_p} )\\
                            &= (\nu \widetilde{y}) ( \piencodf{(\lambda x . M' [\widetilde{x} \leftarrow x])}_v \mid Q'' ), \text{ for } \widetilde{y} = y_1 , \cdots , y_p\\
                            &= (\nu \widetilde{y}) ( v.\overline{\some};v(x).\piencodf{ M'[\widetilde{x} \leftarrow x] }_v \mid Q'' )\\
                        \end{aligned}
                        $$
                        where 
                        $Q'' = y_1.\some_{\lfv{N_1}};\piencodf{ N_1 }_{y_1} \mid \cdots \mid y_p.\some_{\lfv{N_p}};\piencodf{ N_p }_{y_p}$.

                    With this shape for $M$, the encoding of $\mathbb{N}$ becomes:
     \[
     \small
        \begin{aligned}
        \piencodf{\expr{N}}_u & = \piencodf{(M\ B)}_u\\
        &= \bigoplus_{B_i \in \perm{B}} (\nu v)( \piencodf{ M}_v \mid v.\some_{u,\lfv{B}} ; \outact{v}{x} . (  \piencodf{ B_i}_x \mid [v \leftrightarrow u] ) )\\
        & = \bigoplus_{B_i \in \perm{B}} (\nu v)( (\nu \widetilde{y}) ( v.\overline{\some};v(x).\piencodf{ M'[\widetilde{x} \leftarrow x] }_v \mid Q'' ) \mid \\
        & \hspace{2cm} v.\some_{u,\lfv{B}} ; \outact{v}{x} . ( \piencodf{ B_i}_x \mid [v \leftrightarrow u] ) )\\
        % & = \bigoplus_{B_i \in \perm{B}} (\nu v, \widetilde{y})( v.\overline{\some};v(x).\piencodf{ M'[\widetilde{x} \leftarrow x] }_v \mid  v.\some_{u,\lfv{B}} ; \\
        % &\hspace{2cm}\outact{v}{x} . (  \piencodf{ B_i}_x \mid [v \leftrightarrow u] ) \mid Q'' )\\
        & \red \bigoplus_{B_i \in \perm{B}} (\nu v , \widetilde{y})( v(x).\piencodf{ M'[\widetilde{x} \leftarrow x] }_v \mid \outact{v}{x} . (  \piencodf{ B_i}_x \mid [v \leftrightarrow u] ) \mid Q'' )  = Q_1\\
        & \red \bigoplus_{B_i \in \perm{B}} (\nu v, \widetilde{y} ,x)( \piencodf{ M'[\widetilde{x} \leftarrow x] }_v \mid  \piencodf{ B_i}_x \mid [v \leftrightarrow u]  \mid Q'')  = Q_2\\
        & \red \bigoplus_{B_i \in \perm{B}} (\nu x, \widetilde{y})( \piencodf{ M'[\widetilde{x} \leftarrow x] }_u \mid  \piencodf{ B_i}_x   \mid Q'') = Q_3\\
        \end{aligned}
                    \]

                    We also have that 
                    
                    \[
                    \begin{aligned}
                        \expr{N} &=(\lambda x . M' [\widetilde{x} \leftarrow x]) \linexsub{N_1 / y_1} \cdots \linexsub{N_p / y_p}\ B \\
                        &\pequiv ((\lambda x . M' [\widetilde{x} \leftarrow x]) B) \linexsub{N_1 / y_1} \cdots \linexsub{N_p / y_p} \\
                       & \red  M'[\widetilde{x} \leftarrow x] \esubst{B}{x} \linexsub{N_1 / y_1} \cdots \linexsub{N_p / y_p} = \expr{M}
                    \end{aligned}
                    \]
                    
                    Furthermore, we have: 
                    \[
                     \begin{aligned}
                           \piencodf{\expr{M}'}_v &= \piencodf{M'[\widetilde{x} \leftarrow x] \esubst{B}{x}\linexsub{N_1 / y_1} \cdots \linexsub{N_p / y_p}}_v \\
                           &= \bigoplus_{B_i \in \perm{B}} (\nu x)( \piencodf{ M'[\widetilde{x} \leftarrow x]}_v \mid  \piencodf{ B_i}_x \mid Q'' ) 
                    \end{aligned}
                    \]
                        
                    We consider different possibilities for $n \geq 1$; in all of the thesis holds: \begin{myEnumerate}
                        \item $n = 1$:
                        
                        Then  $Q = Q_1$  and  $\piencodf{\expr{N}}_u \red^1 Q_1$. In addition,
                        
                        \begin{myEnumerate}
                            \item  \( Q_1 \red^2 Q_3 = Q' \), 
                            \item $\expr{N} \red^1 M'[\widetilde{x} \leftarrow x] \esubst{B}{x} = \expr{N}'$,
                            \item  $\piencodf{M'[\widetilde{x} \leftarrow x] \esubst{B}{x}}_u = Q_3$.
                        \end{myEnumerate}
                        and the result follows.

                        \item  $n = 2$:
                         
                       Then $Q = Q_2$ and $ \piencodf{\expr{N}}_u \red^2 Q_2$. In addition, 
                        
                        \begin{itemize}
                            \item $Q_2 \red^1 Q_3 = Q'$ , 
                            \item $\expr{N} \red^1 M'[\widetilde{x} \leftarrow x] \esubst{B}{x} = \expr{N}'$
                            \item $\piencodf{M'[\widetilde{x} \leftarrow x]) \esubst{B}{x}}_u = Q_3$
                        \end{itemize}
                        and the result follows.
                         
                        \item $n \geq 3$:
                        
                       Then $ \piencodf{\expr{N}}_u \red^3 Q_3 \red^l Q$, for $l \geq 0$. In addition,  $\expr{N} \red \expr{M}'$ and  $Q_3 = \piencodf{\expr{M}'}_u$. By the IH, there exist $ Q'$ and $\expr{N}'$ such that $Q \red^i Q'$, $\expr{M}' \red_{\pequiv}^j \expr{N}'$ and $\piencodf{\expr{N}'}_u = Q'$ . Finally,
                        $\piencodf{\expr{N}}_u \red^3 Q_3 \red^l Q \red^i Q'$ and $\expr{N} \rightarrow \expr{M}'  \red_{\pequiv}^j \expr{N}'$, and the result follows.
                        
                    \end{myEnumerate}
                       \item $M = \fail^{\widetilde{z}}$.

                    \[
                        \begin{aligned}
                            \piencodf{M}_v &= \piencodf{\fail^{\widetilde{z}}}_v= v.\overline{\none} \mid \widetilde{z}.\overline{\none} \\
                        \end{aligned}
                    \]
                    
                    With this shape for $M$, the encoding of $\expr{N}$ becomes:
                    \[
                    \begin{aligned}
                        \piencodf{\expr{N}}_u & = \piencodf{(M\ B)}_u\\
                        &= \bigoplus_{B_i \in \perm{B}} (\nu v)( \piencodf{ M}_v \mid v.\some_{u, \lfv{B}} ; \outact{v}{x} . (  \piencodf{ B_i}_x \mid [v \leftrightarrow u] ) )\\
                        & = \bigoplus_{B_i \in \perm{B}} (\nu v)(  v.\overline{\none}\mid \widetilde{z}.\overline{\none} \mid v.\some_{u, \lfv{B}} ;  \outact{v}{x} . (  \piencodf{ B_i}_x \mid [v \leftrightarrow u] ) )\\
                        & \red \bigoplus_{B_i \in \perm{B}}   u.\overline{\none} \mid \widetilde{z}.\overline{\none}  \mid \lfv{B}.\overline{\none} \\
                        &= \bigoplus_{\perm{B}}   u.\overline{\none} \mid \widetilde{z}.\overline{\none}  \mid \lfv{B}.\overline{\none} \\
                        \end{aligned}
                    \]

                    \end{myEnumerate}
                    
                  Also, 
                    \[  \expr{N} = \fail^{\widetilde{z}} \ B  \red \sum_{\perm{B}} \fail^{\widetilde{z} \cup \lfv{B} }  = \expr{M}.  \]
                    
                    Furthermore, 
                    \[
                     \begin{aligned}
                           \piencodf{\expr{M}}_u &= \piencodf{\sum_{\perm{B}} \fail^{\widetilde{z} \cup \lfv{B} }}_u \\
                         &  = \bigoplus_{\perm{B}}\piencodf{ \fail^{\widetilde{z} \cup \lfv{B} }}_u\\
                          & = \bigoplus_{\perm{B}}    u.\overline{\none} \mid \widetilde{z}.\overline{\none}  \mid \lfv{B}.\overline{\none}
                    \end{aligned}
                    \]

                \item When $m \geq 1$ and $ n \geq 0$, the distinguish two cases:
                    
                    \begin{myEnumerate}
                        \item $n = 0$:
                            
                            Then, 
                            $$ 
                            \begin{aligned} &\bigoplus_{B_i \in \perm{B}} (\nu v)( R \mid v.\some_{u, \lfv{B}} ; \outact{v}{x} . (  \piencodf{ B_i}_x \mid [v \leftrightarrow u] ) ) = Q,\\ 
                            &~ \text{ and }  \piencodf{M}_u \red^m R.
                            \end{aligned} 
                            $$ 
                        By the IH there exist $R'$  and $\expr{M}' $ such that $R \red^i R'$, $M \red_{\pequiv}^j \expr{M}'$, and $\piencodf{\expr{M}'}_u = R'$.  Hence, 
                    
                            \[ 
                            \begin{aligned}
                               \piencodf{\expr{N}}_u & =  \bigoplus_{B_i \in \perm{B}} (\nu v)( \piencodf{ M}_v \mid v.\some_{u, \lfv{B}} ; \outact{v}{x} . (  \piencodf{ B_i}_x \mid [v \leftrightarrow u] ) )\\
                               & \red^m  \bigoplus_{B_i \in \perm{B}} (\nu v)( R \mid v.\some_{u , \lfv{B}} ; \outact{v}{x} . (  \piencodf{ B_i}_x \mid [v \leftrightarrow u] ) ) = Q\\
                            % \end{aligned}
                            %  \]
                            % We also know that
                            % \[ 
                            % \begin{aligned}
                               & \red^i  \bigoplus_{B_i \in \perm{B}} (\nu v)( R' \mid v.\some_{u, \lfv{B}} ; \outact{v}{x} . (  \piencodf{ B_i}_x \mid [v \leftrightarrow u] ) ) = Q'\\
                            \end{aligned}
                            \]
                            and so the \lamrsharfail term can reduce as follows: $\expr{N} = (M\ B) \red_{\pequiv}^j M'\ B = \expr{N}'$ and  $\piencodf{\expr{N}'}_u = Q'$.

                        \item $n \geq 1$:
                        
                            Then  $R$ has an occurrence of an unguarded $v.\overline{\some}$ or $v.\overline{\none}$, which implies it is of the form 
                            $ \piencodf{(\lambda x . M' [\widetilde{x} \leftarrow x]) \linexsub{N_1 / y_1} \cdots \linexsub{N_p / y_p}}_v $ or $ \piencodf{\fail}_v $, and the case follows by IH.

                    \end{myEnumerate}

            \end{myEnumerate}

        % \end{myEnumerate}
        
        This concludes the analysis for the case $\expr{N} = (M \, B)$.
        
        \item $\expr{N} = M [ \widetilde{x} \leftarrow x ]$.
    
            The sharing variable $x$ is not free and the result follows by vacuity.
        
        \item $\expr{N} = (M[\widetilde{x} \leftarrow x])\esubst{ B }{ x }$. Then,
            
            \[
                \begin{aligned}
                    \piencodf{\expr{N}}_u &=\piencodf{ (M[\widetilde{x} \leftarrow x])\esubst{ B }{ x } }_u = \bigoplus_{B_i \in \perm{B}} (\nu x)( \piencodf{ M[\widetilde{x} \leftarrow x]}_u \mid \piencodf{ B_i}_x )
                \end{aligned}
            \]

            % Let us consider three cases.

            \begin{myEnumerate}
                \item  $\size{\widetilde{x}} = \size{B}$.
                
                    Then let us consider the shape of the bag $ B$.
                    
     \begin{myEnumerate}
                
            \item When $B = 1$
                         
                            We have the following
\[
\begin{aligned}
\piencodf{\expr{N}}_u &= (\nu x)( \piencodf{ M[ \leftarrow x]}_u \mid \piencodf{ 1}_x ) \\
&=   (\nu x)(  x. \overline{\some}. \outact{x}{y_i} . ( y_i . \some_{u, \lfv{M}} ;y_{i}.\close; \piencodf{M}_u \mid x. \overline{\none}) \mid\\ &\qquad x.\some_{\emptyset} ; x(y_n). ( y_n.\overline{\some};y_n . \overline{\close} \mid x.\some_{\emptyset} ; x. \overline{\none}) )  \\
& \red  (\nu x)(   \outact{x}{y_i} . ( y_i . \some_{u, \lfv{M}} ;y_{i}.\close; \piencodf{M}_u \mid x. \overline{\none}) \mid \\
& \hspace{1cm}  x(y_n). ( y_n.\overline{\some};y_n . \overline{\close} \mid x.\some_{\emptyset} ; x. \overline{\none}) )           \qquad = Q_1                      \\
 & \red  (\nu x, y_i)(     y_i . \some_{u, \lfv{M}} ;y_{i}.\close; \piencodf{M}_u \mid x. \overline{\none} \mid  y_n.\overline{\some};\\
 & \hspace{1cm}y_n . \overline{\close} \mid x.\some_{\emptyset} ; x. \overline{\none})
                 \qquad                = Q_2 \\
& \red  (\nu x, y_i)( y_{i}.\close; \piencodf{M}_u \mid x. \overline{\none} \mid   y_n . \overline{\close} \mid x.\some_{\emptyset} ;  x. \overline{\none})
                                     \qquad = Q_3
                                \\
& \red (\nu x)(  \piencodf{M}_u \mid x. \overline{\none} \mid   x.\some_{\emptyset} ; x. \overline{\none}) 
 \qquad = Q_4\\
& \red   \piencodf{M}_u   \qquad = Q_5
                            \end{aligned}
                            \]
                            Notice how $Q_2$ has a choice however the $x$ name can be closed at any time so for simplicity we only perform communication across this name once all other names have completed their reductions.
                            
                           Now proceed by induction on the number of reductions $\piencodf{\expr{N}}_u \red^k Q$.
                            
                            \begin{myEnumerate}
                                
                                \item  $k = 0$:
                                
                                This case is trivial.
                                    %  Hence we have $\piencodf{\expr{N}}_u \red^0 \piencodf{(M[ \leftarrow x])\esubst{ 1 }{ x }}_u\ = Q $, $Q \red^0  Q'$ and $ \expr{N} =  (M[ \leftarrow x])\esubst{ 1 }{ x } \red^0 (M[ \leftarrow x])\esubst{ 1 }{ x } = \expr{N}'$

                                \item  $k = 1$: ($2 \leq  k \leq 4$: is similar.)
                                
                                    Then,  $Q = Q_1$ and $ \piencodf{\expr{N}}_u \red^1 Q_1$. In addition,  $Q_1 \red^4 Q_5 = Q'$,   $\expr{N} \red^0 M[ \leftarrow x] \esubst{ 1 }{ x } \pequiv M$ and $\piencodf{ M }_u = Q_5$, and the result follows.
                                    
                                % \item $2 \leq  k \leq 4$: 
                                % Proceeds similarly to the previous case
                                
                                \item  $k \geq 5$:
                                
                                    Then $ \piencodf{\expr{N}}_u \red^5 Q_5 \red^l Q$, for $l \geq 0$. Since $Q_5 = \piencodf{ M }_u$,  by the IH it follows  that there exist $ Q' $ and $ \expr{N}' $ such that  $ Q \red^i Q' ,  M \red_{\pequiv}^j \expr{N}'$ and $\piencodf{\expr{N}'}_u = Q'$. 
                         
                         Then,  $ \piencodf{\expr{N}}_u \red^5 Q_5 \red^l Q \red^i Q'$ and by the contextual reduction  one has $\expr{N} = (M[ \leftarrow x])\esubst{ 1 }{ x } \red_{\pequiv}^j  \expr{N}' $ and the case holds.

                            \end{myEnumerate}               
                        
                        \item$B = \bag{N_1, \cdots ,N_l}$, for $l \geq 1$.
                        
                            Then,  consider the reductions in \figref{fig:proofreductions2}.
                            
                            \begin{figure}[!t]
                            \hrule 
                            { \small
                            \[
                            \begin{aligned}
                                          \piencodf{\expr{N}}_u &=\piencodf{(M[\widetilde{x} \leftarrow x])\esubst{ B }{ x }}_u\\
                             & = \bigoplus_{B_i \in \perm{B}} (\nu x)( \piencodf{ M[\widetilde{x} \leftarrow x]}_u \mid \piencodf{ B_i}_x ) \\
                                 &=  \bigoplus_{B_i \in \perm{B}} (\nu x)( x.\overline{\some}. \outact{x}{y_1}. (y_1 . \some_{\emptyset} ;y_{1}.\close;\zero \mid x.\overline{\some};x.\some_{u, (\lfv{M} \setminus x_1 , \cdots , x_l )}; 
                                 \\
                               &\qquad x(x_1) . \cdots x.\overline{\some}. \outact{x}{y_l} . (y_l . \some_{\emptyset} ; y_{l}.\close;\zero \mid x.\overline{\some};x.\some_{u,(\lfv{M} \setminus x_l )};x(x_l) . \\
                              &\qquad  x.\overline{\some}; \outact{x}{y_{l+1}}. ( y_{l+1} . \some_{u,\lfv{M}} ;y_{l+1}.\close; \piencodf{M}_u \mid x.\overline{\none} )) \cdots ) \mid \\
                             & \qquad  x.\some_{\lfv{B}} ; x(y_1). x.\some_{y_1, \lfv{B} };x.\overline{\some} ; \outact{x}{x_1}. (x_1.\some_{\lfv{B_i(1)}} ; \piencodf{B_{i}(1)}_{x_1}  \\
                             &\qquad   \mid y_1. \overline{\none}\mid \cdots x.\some_{\lfv{B_{i}(l)}} ; x(y_l). x.\some_{y_l, \lfv{B_{i}(l)}} ;x.\overline{\some} ; \outact{x}{x_l}. (x_l.\some_{\lfv{B_{i}(l)}} ;   \\
                              & \qquad \piencodf{B_{i}(l)}_{x_l}\mid y_l. \overline{\none}\mid x.\some_{\emptyset} ; x(y_{l+1}). ( y_{l+1}.\overline{\some};y_{l+1} . \overline{\close} \mid x.\some_{\emptyset} ; x. \overline{\none})
                                  )
                                  )
                                  )
                               \\
                                & \red ^{5l}
                                \bigoplus_{B_i \in \perm{B}} (\nu x , x_1,y_1, \cdots , x_l,y_1)( y_1 . \some_{\emptyset} ;y_{1}.\close;\zero \mid  \cdots  y_l . \some_{\emptyset} ; y_{l}.\close;\zero \mid  \\
                                  &\qquad x.\overline{\some}; \outact{x}{y_{l+1}}. ( y_{l+1} . \some_{u,\lfv{M}} ;y_{l+1}.\close; \piencodf{M}_u \mid x.\overline{\none} )
                                   \mid \\
                                  & \qquad  x_1.\some_{\lfv{B_{i}(1)}} ; \piencodf{B_{i}(1)}_{x_1}  \mid y_1. \overline{\none} \mid \cdots   x_l.\some_{\lfv{B_{i}(l)}} ; \piencodf{B_{i}(l)}_{x_l}  \mid y_l. \overline{\none}\mid\\
                                  & \qquad x.\some_{\emptyset} ; x(y_{l+1}). ( y_{l+1}.\overline{\some};y_{l+1} . \overline{\close} \mid x.\some_{\emptyset} ; x. \overline{\none})
                                  )
                               \\
                               & \red ^{5}
                               \bigoplus_{B_i \in \perm{B}} (\nu  x_1,y_1, \cdots , x_l,y_1 )(  y_1 . \some_{\emptyset} ;y_{1}.\close;\zero \mid  \cdots  y_l . \some_{\emptyset} ; y_{l}.\close;\zero \\
                                 & \qquad \mid \piencodf{M}_u \mid   x_1.\some_{\lfv{B_{i}(1)}} ; \piencodf{B_{i}(1)}_{x_1}  \mid y_1. \overline{\none} \mid \cdots   x_l.\some_{\lfv{B_{i}(l)}} ; \piencodf{B_{i}(l)}_{x_l}  \mid y_l. \overline{\none} )
                           \\
                               & \red ^{l}  \bigoplus_{B_i \in \perm{B}} (\nu x_1,\cdots , x_l )( \piencodf{M}_u \mid x_1.\some_{\lfv{B_{i}(1)}} ; \piencodf{B_{i}(1)}_{x_1}  \mid  \cdots \mid  x_l.\some_{\lfv{B_{i}(l)}} ; \piencodf{B_{i}(l)}_{x_l} )\\
                               &   
                        = Q_{6l + 5}\\
                            \end{aligned}
                            \]
                            }
                            \hrule 
                            \caption{Reductions of encoded explicit substitution}
                                \label{fig:proofreductions2}
                            \end{figure}

                            The proof follows by induction on the number of reductions $\piencodf{\expr{N}}_u \red^k Q$.
                            
                            \begin{myEnumerate}
                                \item $k = 0$:
                                
                                This case is trivial. Take $\piencodf{\expr{N}}_u=Q=Q'$ and $\expr{N}=\expr{N}'$.
                                    % Hence we have $\piencodf{\expr{N}}_u \red^0 \piencodf{(M[\widetilde{x} \leftarrow x])\esubst{ B }{ x }}_u\ = Q $, $Q =  \piencodf{(M[\widetilde{x} \leftarrow x])\esubst{ B }{ x }}_u \red^0 \piencodf{(M[\widetilde{x} \leftarrow x]\esubst{ B }{ x }}_u = Q'$ and $ \expr{N} = M[\widetilde{x} \leftarrow x])\esubst{ B }{ x } \red^0 (M[\widetilde{x} \leftarrow x])\esubst{ B }{ x } = \expr{N}'$
                                
                                \item  $1 \leq k \leq 6l + 5$:
                                    
                                    Then,  $ \piencodf{\expr{N}}_u \red^k Q_k$.
                                    Observing the reductions in \figref{fig:proofreductions2}, one has 
                                    $Q_k \red^{6l + 5 - k} Q_{6l + 5} = Q'$ ,

                                    $\expr{N} \red^1 \sum_{B_i \in \perm{B}}M\ \linexsub{B_i(1)/x_1} \cdots \linexsub{B_i(l)/x_l} = \expr{N}'$ and 
                                    
                                    $\piencodf{\sum_{B_i \in \perm{B}}M\ \linexsub{B_i(1)/x_1} \cdots \linexsub{B_i(l)/x_l}}_u = Q_{6l + 5}$, and the result follows.
                                
                                \item $k > 6l + 5$:
                                
                            Then, $ \piencodf{\expr{N}}_u \red^{6l + 5} Q_{6l + 5} \red^n Q$, for $n \geq 1$. In addition,
                                    
                                    $\expr{N} \red^1 \sum_{B_i \in \perm{B}}M\ \linexsub{B_i(1)/x_1} \cdots \linexsub{B_i(l)/x_l}$ and 
                                    
                                    $Q_{6l + 5} = \piencodf{\sum_{B_i \in \perm{B}}M\ \linexsub{B_i(1)/x_1} \cdots \linexsub{B_i(l)/x_l}}_u$. By the IH there exist  $ Q'$ and $\expr{N}'$ such that $  Q \red^i Q'$, $$\sum_{B_i \in \perm{B}}M\ \linexsub{B_i(1)/x_1} \cdots \linexsub{B_i(l)/x_l} \red_{\pequiv}^j \expr{N}'$$ and $\piencodf{\expr{N}'}_u = Q'$. Finally,
                                    
                                $\piencodf{\expr{N}}_u \red^{6l + 5} Q_{6l + 5} \red^n Q \red^i Q'$ and
                                
                                $\expr{N} \rightarrow \sum_{B_i \in \perm{B}}M\ \linexsub{B_i(1)/x_1} \cdots \linexsub{B_i(l)/x_l}  \red_{\pequiv}^j \expr{N}'$.

                            \end{myEnumerate}

                    \end{myEnumerate}

                \item  $\size{\widetilde{x}} > \size{B}$.
                
                    Then,
                    $\expr{N} = M[x_1, \cdots , x_k \leftarrow x]\ \esubst{ B }{ x }$ with $B = \bag{N_1 ,  \cdots , N_l},$ for $ k > l$. Also,
                    $$ \expr{N} \red  \sum_{B_i \in \perm{B}}  \fail^{\widetilde{z}} = \expr{M}\  \text{  and } \ \widetilde{z} = (\lfv{M} \setminus \{  x_1, \cdots , x_k \} ) \cup \lfv{B}. $$
                    On the one hand, we have \figref{fig:proofreductions3}.
                    Hence $k = l + m$ for some $m \geq 1$

                    \begin{figure}[!t]
                    \hrule 
                    { \small
                    
                    \[
                    \begin{aligned}
                        \piencodf{N}_u &= \piencodf{M[x_1, \cdots , x_k \leftarrow x]\ \esubst{ B }{ x }}_u \\
                        &=  \bigoplus_{B_i \in \perm{B}} (\nu x)( \piencodf{ M[x_1, \cdots , x_k \leftarrow x]}_u \mid  \piencodf{ B_i}_x ) 
                        \\
                        &=  \bigoplus_{B_i \in \perm{B}} (\nu x)(                            x.\overline{\some}. \outact{x}{y_1}. (y_1 . \some_{\emptyset} ;y_{1}.\close;\zero \mid x.\overline{\some};x.\some_{u, (\lfv{M} \setminus x_1 , \cdots , x_k )}; \\
                                &\qquad x(x_1) . \cdots x.\overline{\some}. \outact{x}{y_k} . (y_k . \some_{\emptyset} ; y_{k}.\close;\zero \mid x.\overline{\some};x.\some_{u,(\lfv{M} \setminus x_k )};x(x_k) . \\
                                &\qquad x.\overline{\some}; \outact{x}{y_{k+1}}. ( y_{k+1} . \some_{u,\lfv{M}} ;y_{k+1}.\close; \piencodf{M}_u \mid x.\overline{\none} )) \cdots ) \mid \\
                                &\qquad x.\some_{\lfv{B}} ; x(y_1). x.\some_{y_1, \lfv{B}} ;x.\overline{\some} ; \outact{x}{x_1}. (x_1.\some_{\lfv{B_i(1)}} ; \piencodf{B_i(1)}_{x_1} \\
                                &\qquad  \mid y_1. \overline{\none} \mid \cdots x.\some_{\lfv{B_i(l)}} ; x(y_l). x.\some_{y_l, \lfv{B_i(l)} };x.\overline{\some} ; \outact{x}{x_l}. (x_l.\some_{\lfv{B_i(l)}} ; \\
                                &\qquad   \piencodf{B_i(l)}_{x_l} \mid y_l. \overline{\none} \mid x.\some_{\emptyset} ; x(y_{l+1}). ( y_{l+1}.\overline{\some};y_{l+1} . \overline{\close} \mid x.\some_{\emptyset} ; x. \overline{\none}) ) 
                               ) )
                        \\
                        & \red^{5l} \bigoplus_{B_i \in \perm{B}} (\nu x, y_1, x_1, \cdots  y_l, x_l)( 
                                  y_1 . \some_{\emptyset} ;y_{1}.\close;\zero \mid \cdots \mid y_l . \some_{\emptyset} ;y_{l}.\close;\zero \\
                                &\qquad x.\overline{\some}. \outact{x}{y_{l+1}} . (y_{l+1} . \some_{\emptyset} ; y_{l+1}.\close;\zero \mid x.\overline{\some};x.\some_{u,(\lfv{M} \setminus x_{l+1} , \cdots , x_k )}; \\
                                &\qquad x(x_{l+1}) . \cdots x.\overline{\some}. \outact{x}{y_k} . (y_k . \some_{\emptyset} ; y_{k}.\close;\zero \mid x.\overline{\some};x.\some_{u,(\lfv{M} \setminus x_k )};x(x_k) . \\
                                &\qquad x.\overline{\some}; \outact{x}{y_{k+1}}. ( y_{k+1} . \some_{u,\lfv{M}} ;y_{k+1}.\close; \piencodf{M}_u \mid x.\overline{\none} )) \cdots ) \mid \\
                                &\qquad   x_1.\some_{\lfv{B_i(1)}} ; \piencodf{B_i(1)}_{x_1} \mid \cdots \mid  x_l.\some_{\lfv{B_i(l)}} ; \piencodf{B_i(l)}_{x_l} \mid \\
                                &\qquad y_1. \overline{\none} \mid \cdots \mid y_l. \overline{\none}\\
                                &\qquad x.\some_{\emptyset} ; x(y_{l+1}). ( y_{l+1}.\overline{\some};y_{l+1} . \overline{\close} \mid x.\some_{\emptyset} ; x. \overline{\none}) ) \\
                        & \red^{l+ 5} 
                           \bigoplus_{B_i \in \perm{B}} (\nu x, x_1, \cdots , x_l)( x.\some_{u,(\lfv{M} \setminus x_{l+1} , \cdots , x_k )};x(x_{l+1}) . \cdots \\
                                &\qquad x.\overline{\some}. \outact{x}{y_k} . (y_k . \some_{\emptyset} ; y_{k}.\close;\zero \mid x.\overline{\some};x.\some_{u,(\lfv{M} \setminus x_k )};x(x_k) . \\
                                &\qquad x.\overline{\some}; \outact{x}{y_{k+1}}. ( y_{k+1} . \some_{u,\lfv{M}} ;y_{k+1}.\close; \piencodf{M}_u \mid x.\overline{\none} ) )  \mid \\
                                &\qquad   x_1.\some_{\lfv{B_i(1)}} ; \piencodf{B_i(1)}_{x_1} \mid \cdots \mid  x_l.\some_{\lfv{B_i(l)}} ; \piencodf{B_i(l)}_{x_l} \mid   x. \overline{\none} ) \\
                        & \red
                          \bigoplus_{B_i \in \perm{B}} (\nu  x_1, \cdots  , x_l)(  u . \overline{\none} \mid x_1 . \overline{\none} \mid  \cdots \mid x_{l} . \overline{\none} \mid (\lfv{M} \setminus \{  x_1, \cdots , x_k \} ). \overline{\none} \mid  
                          \\
                               & \qquad \qquad x_1.\some_{\lfv{B_i(1)}} ; \piencodf{B_i(1)}_{x_1} \mid \cdots \mid  x_l.\some_{\lfv{B_i(l)}} ; \piencodf{B_i(l)}_{x_l}  ) 
                               \\ 
                        & \red^{l}  \bigoplus_{B_i \in \perm{B}}  u . \overline{\none} \mid (\lfv{M} \setminus \{  x_1, \cdots , x_k \} ). \overline{\none} \mid \lfv{B}. \overline{\none} 
                        \\
                        & = Q_{7l + 6}  
                    \end{aligned}
                    \]
                    
                    }
                    \hrule 
                    \caption{Reductions of an encoded explicit substitution that leads to failure}
                        \label{fig:proofreductions3}
                    \end{figure}

                    Now we proceed by induction on the number of reductions $\piencodf{\expr{N}}_u \red^j Q$.
                            
                            \begin{myEnumerate}
                                \item $j = 0$:
                               
                               This case is trivial. 
                                    % Hence we have $\piencodf{\expr{N}}_u \red^0 \piencodf{(M[\widetilde{x} \leftarrow x])\esubst{ B }{ x }}_u\ = Q $, $Q =  \piencodf{(M[\widetilde{x} \leftarrow x])\esubst{ B }{ x }}_u \red^0 \piencodf{(M[\widetilde{x} \leftarrow x]\esubst{ B }{ x }}_u = Q'$ and $ \expr{N} = M[\widetilde{x} \leftarrow x])\esubst{ B }{ x } \red^0 (M[\widetilde{x} \leftarrow x])\esubst{ B }{ x } = \expr{N}'$
                                
                                \item $1 \leq j \leq 7l + 6$:
                                    
                                 Then,
                                 
                                 $ \piencodf{\expr{N}}_u \red^j Q_j \red^{7l + 6 - j} Q_{7l + 6} = Q'$ , $\expr{N} \red^1 \sum_{B_i \in \perm{B}}  \fail^{\widetilde{z}} = \expr{N}'$ and $\piencodf{\sum_{B_i \in \perm{B}}  \fail^{\widetilde{z}} }_u = Q_{7l + 6}$, and the result follows.
                                
                                \item  $j > 7l + 6$:
                                
                            Then, $ \piencodf{\expr{N}}_u \red^{7l + 6} Q_{7l + 6} \red^n Q$, for $n \geq 1$. Also,  $\expr{N} \red^1 \sum_{B_i \in \perm{B}}  \fail^{\widetilde{z}}$. However no further reductions can be performed.

                            \end{myEnumerate}

                \item $\size{\widetilde{x}} < \size{B}$.    
                   
                   Proceeds similarly to the previous case.
                    
            \end{myEnumerate}

        \item  $\expr{N} = M \linexsub {N' /x}$. 
    
            Then,
            \[
            \begin{aligned}
                \piencodf{M \linexsub {N' /x}}_u &= (\nu x) ( \piencodf{ M }_u \mid   x.\some_{\lfv{N'}};\piencodf{ N' }_x  )  \\
            \end{aligned}
            \]
            
            Then we have
            \[
            \begin{aligned}
               \piencodf{\expr{N}}_u & =  (\nu x) ( \piencodf{ M }_u \mid   x.\some_{\lfv{N'}};\piencodf{ N' }_x  ) \\
               & \red^m  (\nu x) ( R \mid   x.\some_{\lfv{N'}};\piencodf{ N' }_x  )\\
               & \red^n  Q\\
            \end{aligned}
            \]
            
            for some process $R$, where $\red^n$ is a reduction that initially synchronizes with $x.\some_{\lfv{N'}}$ when $n \geq 1$, $n + m = k \geq 1$. Type preservation in \spi ensures reducing $\piencodf{ M}_v \red^m$ does not consume possible synchronizations with $x.\some$ if they occur. Let us consider the the possible sizes of both $m$ and $n$.
            
            \begin{myEnumerate}
                \item For $m = 0$ and $n \geq 1$.
                
                    In this case  $R = \piencodf{M}_u$  and  there are two possibilities of having an unguarded $x.\overline{\some}$ or $x.\overline{\none}$ without internal reductions.  
                    
                    \begin{myEnumerate}
                        \item $M = \fail^{x, \widetilde{y}}$

                            \[
                            \begin{aligned}
                               \piencodf{\expr{N}}_u & =  (\nu x) ( \piencodf{ M }_u \mid   x.\some_{\lfv{N'}};\piencodf{ N' }_x  ) \\
                               & = (\nu x) ( \piencodf{ \fail^{x, \widetilde{y}} }_u \mid   x.\some_{\lfv{N'}};\piencodf{ N' }_x  ) \\
                               & = (\nu x) ( u.\overline{\none} \mid x.\overline{\none} \mid  \widetilde{y}.\overline{\none} \mid x.\some_{\lfv{N'}};\piencodf{ N' }_x  ) \\
                               & \red u.\overline{\none} \mid \widetilde{y}.\overline{\none} \mid \lfv{N'}.\overline{\none} \\
                            \end{aligned}
                            \]
                            
                            Notice that no further reductions can be performed.
                            
                          Thus, 
                          \[ \piencodf{\expr{N}}_u \red u.\overline{\none} \mid \widetilde{y}.\overline{\none} \mid \lfv{N'}.\overline{\none}  = Q'.\]
                            We also have that  \[\expr{N} \red \fail^{ \widetilde{y} \cup \lfv{N'} } = \expr{N}' \text{ and } \piencodf{ \fail^{\widetilde{y} \cup \lfv{N'}} }_u = Q',\]
                            and the result follows.
                            
                        \item $\headf{M} = x $
                    
                            By the diamond property (Proposition~\ref{prop:conf1_lamrsharfail}) we will be reducing each non-deterministic choice of a process simultaneously.
                            \revd{B29}{Then by Proposition~\ref{prop:NEEDTONAME} we have the following:
                            \[
                            \begin{aligned}
                             \piencodf{\expr{N}}_u     &\red^* (\nu x) ( \bigoplus_{i \in I}(\nu \widetilde{y})(\piencodf{ x }_{j} \mid P_i) \mid   x.\some_{\lfv{N'}};\piencodf{ N' }_x  ) \\
                                 &= (\nu x) ( \bigoplus_{i \in I}(\nu \widetilde{y})(x.\overline{\some} ;[x \leftrightarrow j ] \mid P_i)  \mid   x.\some_{\lfv{N'}};\piencodf{ N' }_x  ) \\
                                 &\red (\nu x) ( \bigoplus_{i \in I}(\nu \widetilde{y})([x \leftrightarrow j ] \mid P_i) \mid   \piencodf{ N' }_x  ) &  = Q_1 \\
                                 &\red  \bigoplus_{i \in I}(\nu \widetilde{y})( \piencodf{ N' }_j \mid P_i ) & = Q_2 \\
                            \end{aligned}
                            \]
                            We also have that 
                            \[
                            \begin{aligned}
                                \expr{N} &=M \linexsub {N' /x} 
                                 \red  M \headlin{ N' /x }
                                 = \expr{M}'.
                            \end{aligned}
                            \]
                            where by Proposition \ref{prop:NEEDTONAME} we obtain
                            \[
                            \piencodf{M \headlin{ N'/x }}_u \red^* \bigoplus_{i \in I}(\nu \widetilde{y})( \piencodf{ N' }_j \mid P_i ) = Q_2.
                            \]
                            and finally from Proposition \ref{prop:NEEDTONAME} there exists an $\expr{M}$ with $\expr{M} \pequiv \expr{M}'$ such that:
                            \[
                            \begin{aligned}
                               \piencodf{\expr{M}}_u &= \bigoplus_{i \in I}(\nu \widetilde{y})( \piencodf{ N' }_j \mid P_i ) & = Q_2. 
                            \end{aligned}
                            \]
                            for simplicity we assume %(by Corollary \ref{cor:sounddoesntred})
                            that $\piencodf{\expr{N}}_u \red Q_1$
                           \begin{myEnumerate}
                                \item  $n = 1$:
                                Then $Q = Q_1$ and  $ \piencodf{\expr{N}}_u \red^1 Q_1$. Since,
                                $Q_1 \red^1 Q_2 = Q'$, $\expr{N} \red^1 M \headlin{ N'/x} \pequiv  \expr{M} = \expr{N}'$ and $\piencodf{\expr{M}}_u = Q_2$, the result follows.
                                 \item  $n \geq 2$:
                                 Then,  $ \piencodf{\expr{N}}_u \red^2 Q_2 \red^l Q$, for $l \geq 0$. Also,  $\expr{N} \rightarrow \expr{M}$, $Q_2 = \piencodf{\expr{M}}_u$. By the IH there exist $ Q'$ and $\expr{N}'$ such that $ Q \red^i Q'$, $\expr{M} \red_{\pequiv}^j \expr{N}'$ and $\piencodf{\expr{N}'}_u = Q'$ . Finally, $\piencodf{\expr{N}}_u \red^2 Q_2 \red^l Q \red^i Q'$ and $\expr{N} \rightarrow \expr{M}  \red_{\pequiv}^j \expr{N}' $, and the result follows.
                            \end{myEnumerate}
                            }    
                        
                    \end{myEnumerate}
 \item For $m \geq 1$ and $ n \geq 0$.
                    
            \begin{myEnumerate}
            \item $n = 0$:
            
               Then,
               \[(\nu x) ( R \mid   x.\some_{\lfv{N'}};\piencodf{ N' }_x  )  = Q ~\text{and }~ \piencodf{M}_u \red^m R.\] 
               
             By the IH there exist $R'$  and $\expr{M}' $ such that $R \red^i R'$, $M \red_{\pequiv}^j \expr{M}'$ and $\piencodf{\expr{M}'}_u = R'$. 
              Hence, 
              \[ 
               \begin{aligned}
                   \piencodf{\expr{N}}_u & = (\nu x) ( \piencodf{M}_u  \mid   x.\some_{\lfv{N'}};\piencodf{ N' }_x  )\\ &\red^m  (\nu x) ( R \mid   x.\some_{\lfv{N'}};\piencodf{ N' }_x  )  = Q.
                \end{aligned}
                 \]
                Also, 
                \[ 
                \begin{aligned}
                   Q & \red^i   (\nu x) ( R' \mid   x.\some_{\lfv{N'}};\piencodf{ N' }_x  ) = Q'\\
                \end{aligned}
                \]
                and the term can reduce as follows:
                
                $\expr{N} = M \linexsub {N' /x} \red_{\pequiv}^j \sum_{M_i' \in \expr{M}'} M_i' \linexsub {N' /x} = \expr{N}'$ and  $\piencodf{\expr{N}'}_u = Q'$.

            \item When $n \geq 1$
            
                Then $R$ has an occurrence of an unguarded $x.\overline{\some}$ or $x.\overline{\none}$, and the case follows by IH. \qedhere
                        
                    \end{myEnumerate}
            \end{myEnumerate}
    \end{myEnumerate}
%\qed
\end{proof}

\subsection{Success Sensitiveness}
\label{app:succtwo}

\pressucctwo*

\begin{proof}
In both cases, by induction on the structure of $M$. 
\begin{enumerate}
\item We only need to consider terms of the following form:

    \begin{itemize}

        \item $ M = \checkmark $.
        
        This  case is immediate.
        
        \item $M = N\ B$.
        
       By definition, $\headf{N \ B} = \headf{N}$. Hence we consider that $\headf{N} = \checkmark$. Then,
       
       $$ \piencodf{N \ B}_u = \bigoplus_{B_i \in \perm{B}} (\nu v)(\piencodf{N}_v \mid v.\some_{u, \lfv{B}} ; \outact{v}{x} . ([v \leftrightarrow u] \mid \piencodf{B_i}_x ) ) $$
       and by the IH  $\checkmark$ is unguarded in $\piencodf{N}_u$ after a sequence of reductions.

        \item  $M = (N[\widetilde{x} \leftarrow x])\esubst{ B }{ x }$.
        
        By definition, $\headf{(N[\widetilde{x} \leftarrow x])\esubst{ B }{ x }} = \headf{N[\widetilde{x} \leftarrow x]} = \headf{N} = \checkmark$ where $ \widetilde{x} = x_1 , \cdots , x_k $ and $\#(x,M) = \size{B}$.
        \[
            \begin{aligned}
                \piencodf{(N[\widetilde{x} \leftarrow x])\esubst{ B }{ x }}_u &=  \bigoplus_{B_i \in \perm{B}} (\nu x)( \piencodf{ N[\widetilde{x} \leftarrow x]}_u \mid \piencodf{ B_i}_x ) \\
                &\red^*  \bigoplus_{B_i \in \perm{B}} (\nu \widetilde{x})( \piencodf{ N}_u \mid x_1.\some_{\lfv{B_i(1)}};\\
                & \qquad \qquad  \qquad \piencodf{ B_i(1) }_{x_1} \mid \cdots \mid x_k.\some_{\lfv{B_i(k)}};\piencodf{ B_i(k) }_{x_k} ) \\
            \end{aligned} 
        \]

        and by the IH  $\checkmark$ is unguarded in $\piencodf{N}_u$ after a sequence of reductions.

        \item $M = M' \linexsub {N /x}$.
        
       By definition,  $\headf{M' \linexsub{N /x}}  = \headf{M'} \checkmark$. Then, $$\piencodf{M' \linexsub {N /x}}_u = (\nu x) ( \piencodf{ M' }_u \mid   x.\some_{\lfv{N}};\piencodf{ N }_x  )$$ and by the IH $\checkmark$ is unguarded in $\piencodf{N}_u$.

    \end{itemize}

   \item  We only need to consider terms of the following form:    
         
    \begin{itemize}
        \item  $M = \checkmark$.
        
        This case is trivial.
        % Then we have that $\piencodf{\checkmark}_u = \checkmark$ which is an unguarded occurrence of $\checkmark$ and that $\headf{\checkmark} = \checkmark$. 
        
        \item $M = N\ B$.
        
        Then,
        \[\piencodf{N \ B}_u = \bigoplus_{B_i \in \perm{B}} (\nu v)(\piencodf{N}_v \mid v.\some_{u,\lfv{B}} ; \outact{v}{x} . ([v \leftrightarrow u] \mid \piencodf{B_i}_x ) ).\] 
        The only occurrence of an unguarded $\checkmark$ is within $\piencodf{N}_v$. By the IH we have that $\headf{N} = \checkmark$ and finally $\headf{N \ B} = \headf{N}$.
        
        % \item  $ M = N [ \widetilde{x} \leftarrow x ]$.
        
        % Then $x$ is a free sharing variable.

        \item $M = (N[\widetilde{x} \leftarrow x])\esubst{ B }{ x }$.
        
        Then,
        $$
            \begin{aligned}
                \piencodf{(N[\widetilde{x} \leftarrow x])\esubst{ B }{ x }}_u &=  \bigoplus_{B_i \in \perm{B}} (\nu x)( \piencodf{ N[\widetilde{x} \leftarrow x]}_u \mid \piencodf{ B_i}_x ) \\
            \end{aligned} 
        $$ However in both $\piencodf{ N[\widetilde{x} \leftarrow x]}_u$ and $\piencodf{ B_i}_x$ we have that both are guarded and hence $\checkmark$ cannot occur without synchronizations.
        
        \item $M = M' \linexsub {N /x}$.
        
        Then,
        $$\piencodf{M' \linexsub {N /x}}_u = (\nu x) ( \piencodf{ M' }_u \mid   x.\some_{\lfv{N}};\piencodf{ N }_x  ),$$ an unguarded occurrence of $\checkmark$ can only occur within $\piencodf{ M' }_u $. By the IH we have $\headf{M'} = \checkmark$ and hence $\headf{M' \linexsub{N /x}}  = \headf{M'}$. \qedhere
           \end{itemize}    
\end{enumerate}
\end{proof}

\successsenscetwo*

\begin{proof}
We proceed with the proof in two parts.

\begin{myEnumerate}
    
    \item Suppose that  $\expr{M} \Downarrow_{\checkmark} $. We will prove that $\piencodf{\expr{M}} \Downarrow_{\checkmark}$.

    By \defref{def:app_Suc3}, there exists $ \expr{M}' = M_1 + \cdots + M_k$ such that $\expr{M} \red^* \expr{M}'$ and with
    $\headf{M_j} = \checkmark$, for some  $j \in \{1, \ldots, k\}$. By completeness there exists $Q$ such that $\piencodf{\expr{M}}_u  \red^* Q = \piencodf{\expr{M}'}_u$.
    
    We wish to show that there exists $ Q'$ such that $Q \red^* Q'$ and $Q'$ has an unguarded occurrence of $\checkmark$.
    
    Since $Q = \piencodf{\expr{M}'}_u$ and due to compositionality and the homormorphic preservation of non-determinism, we have that
    \[
        \begin{aligned}
            Q &= \piencodf{M_1}_u \oplus \cdots \oplus \piencodf{M_k}_u\\
            %\piencodf{M_j}_u & = C[ \piencodf{\checkmark}_v ]\\
            %\piencodf{M_j}_u & = C[ \checkmark]
        \end{aligned}
    \]
    
    By Proposition \ref{Prop:checkprespi} (1) we have that $$\headf{M_j} = \checkmark \implies \piencodf{M_j}_u \red^*  (P \mid \checkmark) \oplus Q''$$ 
    for some $Q''$. 
    Hence, $Q \red^*  (P \mid \checkmark) \oplus Q'' = Q'$, as wanted.
    
    %By operational completeness (Lemma~\ref{l:completenesstwo}) we have that if $\expr{N} $ and $ \expr{M} $ be well-formed $\lamrsharfail $ closed expressions. If $ \expr{N}\red \expr{M}$ then there exists $Q$ such that $\piencodf{\expr{N}}_u  \red^* Q = \piencodf{\expr{M}}_u$.
    
    %Notice that  neither our  reduction rules  (in Figure ~\ref{fig:redspi}), or structural congruence rules (of Definition~\ref{def:pistructcong}),  or  our encoding ($\piencodf{\checkmark }_u=\checkmark$)  create or destroy a $\checkmark$ occurring in the unguarded \joe{parallel position}. By Proposition \ref{Prop:checkprespi} (1) the encoding preserves the head of a term being $\checkmark$ is encoded into an unguarded occurrence of $\checkmark$. The encoding acts homomorphically over sums, therefore, if a $\checkmark$ appears as the head of a term in a sum, it will stay in the encoded unguarded in the sum. We can iterate the operational completeness lemma and obtain the result.

    \item Suppose that $\piencodf{\expr{M}}_u \Downarrow_{\checkmark}$. We will prove that $ \expr{M} \Downarrow_{\checkmark}$.

    By operational soundness (Theorem~\ref{l:app_soundnesstwo}): if $ \piencodf{\expr{N}}_u \red^* Q$
    then there exist $Q'$  and $\expr{N}' $ such that 
    $Q \red^* Q'$, $\expr{N}  \red^*_{\pequiv} \expr{N}'$ 
    and 
    $\piencodf{\expr{N}'}_u = Q'$.
   Since $\piencodf{\expr{M}}_u \red^* P_1 \oplus \ldots \oplus P_k$, and $P_j= P_j'' \mid \checkmark$, for some $j$. 
   
   Notice that if $\piencodf{\expr{M}}_u$ is itself a term with unguarded $\checkmark$, say $\piencodf{\expr{M}}_u=P \mid \checkmark$, then $\expr{M}$ is itself headed with $\checkmark$, from Proposition \ref{Prop:checkpres} (2).
   
   In the case $\piencodf{\expr{M}}_u= P_1 \oplus \ldots \oplus P_k$, $k\geq 2$, and $\checkmark$ occurs unguarded in an $P_j$, The encoding acts homomorphically over sums and the reasoning is similar. We have that $P_j = P_j' \mid \checkmark$ we apply Proposition \ref{Prop:checkpres} (2). \qedhere
\end{myEnumerate}
%\qed
\end{proof}

\end{document}